\documentclass[letter,11pt]{article}%
\usepackage{amsmath}
\usepackage{amsfonts}
\usepackage{amssymb}
\usepackage{graphicx}
\usepackage{color}
\usepackage{rotating}
\usepackage{float}
\usepackage[height=9in,left=1in,right=0.75in,bottom=1in]{geometry}
\usepackage[breaklinks=true,bookmarksopen=true,colorlinks=true,citecolor=blue]%
{hyperref}
\usepackage[authoryear]{natbib}
\usepackage[toc,page]{appendix}
\usepackage{longtable}%
\providecommand{\U}[1]{\protect\rule{.1in}{.1in}}
\numberwithin{equation}{section}
\newtheorem{theorem}{Theorem}[section]

\newtheorem{assumption}{Assumption}

\newtheorem{lemma}[theorem]{Lemma}

\newtheorem{proposition}{Proposition}[section]
\newtheorem{remark}{Remark}

\newenvironment{proof}[1][Proof]{\noindent\textbf{#1.} }{\ \rule{0.5em}{0.5em}}

\numberwithin{equation}{section}
\allowdisplaybreaks
\begin{document}

\title{A One-Covariate-at-a-Time Multiple Testing Approach to Variable Selection in
Additive Models \thanks{We would like to thank the editor, Esfandiar Maasoumi,
an associate editor, and two referees for their helpful comments, which have
helped improve the paper substantially. Su gratefully acknowledges the support
from NSFC under the grant No. 72133002. Zhang gratefully acknowledges the
financial support from NSFC (Projects No.71973141 and No.71873033). The
computation code used in this paper is available upon request or can be found
on the second author's personal website. E-mail: sulj@sem.tsinghua.edu.cn (L.
Su), tao.yang@anu.edu.au (T. Yang), yonghui.zhang@ruc.edu.cn (Y. Zhang), and
qzhou@lsu.edu (Q. Zhou).}}
\author{Liangjun Su$^{a},$ Thomas Tao Yang$^{b}$, Yonghui Zhang$^{c},$ and Qiankun
Zhou$^{d}$\\$^{a}$ School of Economics and Management, Tsinghua University, China \\$^{b}$ Research School of Economics, Australian National University, Australia\\$^{c}$ School of Economics, Renmin University of China, China\\$^{d}$ Department of Economics, Louisiana State University, USA}
\maketitle

\begin{abstract}
This paper proposes a One-Covariate-at-a-time Multiple Testing (OCMT) approach
to choose significant variables in high-dimensional nonparametric additive
regression models. Similarly to \cite{Chudik_K_P}, we consider the statistical
significance of individual nonparametric additive components one at a time and
take into account the multiple testing nature of the problem. Both one-stage
and multiple-stage procedures are considered. The former works well in terms
of the true positive rate only if the net effects of all signals are strong
enough; the latter helps to pick up hidden signals that have weak net effects.
Simulations demonstrate the good finite sample performance of the proposed
procedures. As an empirical illustration, we use the OCMT procedure to a
dataset extracted from the Longitudinal Survey on Rural Urban Migration in
China. We find that our procedure works well in terms of out-of-sample
forecast root mean square errors, compared with competing methods such as
adaptive group Lasso (AGLASSO).

\vspace{2mm}

\begin{description}
\item[Keywords:] One covariate at a time, multiple testing, model selection,
high dimensionality, nonparametric, additive model.

\item[\emph{JEL classification:}] C12, C14, C21, C52.

\end{description}
\end{abstract}

\section{Introduction}

Variable selection has been playing a pivotal role in econometrics and
statistics for statistical learning and scientific discoveries. Notable early
contributions include \cite{Akaike1973}, \cite{Akaike1974}, and
\cite{Schwartz1978}. These authors suggested a unified approach to model
selection, viz., choosing a parameter vector by minimizing the conventional
criterion function plus an $L_{0}$ penalty to penalize the model size.
However, these approaches are not feasible in high-dimensional settings. The
seminal work of \cite{Tibshirani1996}\ addressed this important issue by
substituting the $L_{0}$ penalty with an $L_{1}$ penalty, and it has sparked
extensive studies in both statistics and econometrics. Important contributions
in the statistics literature include \cite{FanLi2001}, \cite{ZhouHastie2005},
\cite{Zou2006}, \cite{FanLv2008}, \cite{ZouLi2008}, \cite{Zhang2010},
\cite{FanFengSong}, \cite{FanLv2013}, and \cite{FanTang2013}. Important
contributions in the econometric literature include \cite{BelloniEtal2012},
\cite{BelloniChernozhukov2013}, \cite{BelloniChernozhukovHansen2014},
\cite{BelloniEtal2017},\ \cite{ChernozhukovEtal2018}, and \cite{Chudik_K_P}.
For a comprehensive review, see \cite{FanLiZhangZou2020}.

In this paper, we propose a multiple testing approach to variable selection
for high-dimensional nonparametric additive models. Recently \cite{Chudik_K_P}
(CKP hereafter) have proposed a One-Covariate-at-a-time Multiple Testing
(OCMT) approach for linear regression models. CKP suggest regressing the
dependent variable on each independent variable separately, retaining only
those variables that exhibit a high correlation with the dependent variable.
This strategy is often referred to as the \textquotedblleft\textit{screening
approach}\textquotedblright. CKP's main contributions are twofold. First, they
propose a criterion for variable selection by controlling the probability of
choosing all the signals (or with some pseudo-signals) in the model. Second,
unlike the usual screening approach that may miss some important
\textquotedblleft hidden\textquotedblright\ signals whose net effects on the
dependent variable are small, the CKP's OCMT procedure is able to pick up
hidden signals with very high probability. The OCMT procedure has been applied
in various applications; see, e.g., \cite{Kozbur2020}, \cite{Chudik_P_S}, and
\cite{AhmendPesaran2022}. In particular, \cite{Kozbur2020} considers a
testing-based forward model selection (TBFMS) procedure in linear regression
models that inductively selects covariates to add predictive power into a
working statistical model. But this latter paper mainly focuses on the error
bound and shows that the proposed procedure is able to achieve estimation
rates matching those of Least Absolute Shrinkage and Selection Operator
(Lasso) and post-Lasso. Furthermore, \cite{Sharifvaghefi2023} extends OCMT to
cases with many highly correlated covariates and allows the number of pseudo
signals to grow at the same rate as the sample size.

Our paper contributes to the literature by extending the CKP's\ OCMT approach
from parametric models to nonparametric additive models. Like CKP, we estimate
the net effect of each variable on the dependent variable one by one, possibly
with some preselected variables. The selected variables are those whose net
effects exceed some threshold value, specified to ensure the probability of
selecting all the signals is very high. The statistics constructed in this
paper are much more complicated than the $t$-statistics in CKP and might not
even exhibit a well-behaved limiting distribution. In addition to
investigating a different model, our paper differs from that of CKP in some
other important aspects. First, we generalize the definition of hidden signals
(in Table 2 in Section \ref{SEC:non_1st}). Technical details are updated
accordingly. Second, CKP chose tuning parameters similar to a Bonferroni
correction of the cumulative distribution function of the standard normal.
This choice implicitly requires a certain degree of approximation of the
standard normal to the $t$-statistic distribution in their study. Instead, we
select the tuning parameters using the classic Bayesian information criterion
(BIC). Third, we add an adaptive group-Lasso-based post-OCMT step to eliminate
pseudo-signals that cannot be eliminated with very high probability in CKP.
This step adds very little computation burden because the dimension is reduced
dramatically before the last-step estimation. This additional step is not
needed in theory, but it aims at eliminating the pseudo-signals and thereby
enhances the out-of-sample forecasting performance in practice.

One competing method to ours in the literature is the adaptive group Lasso
(AGLASSO) proposed by \cite{HuangEtal2010}. The AGLASSO adds some adaptive
penalty term to the usual least squares loss function in the spirit of
\cite{Zou2006}. Even though we also use AGLASSO to eliminate the
pseudo-signals after the OCMT procedure, our approach allows much faster
computation and provides more reliable estimates than their approach.

We consider various setups in the simulation studies and compare the above
post-OCMT AGLASSO procedure with that based on the OCMT\ alone or the AGLASSO
procedure of \cite{HuangEtal2010} alone. We find that the former one generally
outperforms the latter two significantly. We apply our method on a dataset
from the Longitudinal Survey on Rural Urban Migration in China (RUMiC) and the
empirical results also demonstrate the excellent performance of our procedure
in finite samples.

The remainder of the paper is structured as follows. In the next section, we
illustrate our approach through a single-stage procedure that is silent to
hidden signals. In Section \ref{SEC:multi}, we present the more powerful
multiple-stage procedure. We investigate the finite sample properties of our
procedure through Monte Carlo experiments in Section \ref{SEC:MC} and an
empirical application in Section \ref{SEC:application}. We conclude the paper
in Section \ref{SEC:conclusion}. The proofs of all propositions and theorems
in the paper are relegated to Appendix \ref{APP:real_main_proof}. The online
supplement contains some additional technical materials that include the
proofs of the technical lemmas in Appendix \ref{APP:mainproof} and some
additional results in the simulation and application. To facilitate reading,
We present our procedure in detail for practitioners in Appendix
\ref{SEC:procedure} and the idea of the proofs in Appendices
\ref{APP:tech_one_stage} and \ref{SEC:multi_tech}.

\textit{Notation.} For a generic real matrix $\boldsymbol{A=}\left\{
a_{ij}\right\}  \boldsymbol{,}$ let $\left\Vert \boldsymbol{A}\right\Vert
=\left[  \lambda_{\max}\left(  \boldsymbol{A}^{\prime}\boldsymbol{A}\right)
\right]  ^{1/2}$ denote the spectral norm and $\left\Vert \boldsymbol{A}%
\right\Vert _{\infty}=\max_{ij}\left\vert a_{ij}\right\vert .$ When
$\boldsymbol{A}$ is symmetric, $\lambda_{\max}\left(  \boldsymbol{A}\right)
\ $and $\lambda_{\min}\left(  \boldsymbol{A}\right)  \ $denote its maximum and
minimum eigenvalues, respectively. For vector $\boldsymbol{x},$\ $\left\Vert
\boldsymbol{x}\right\Vert $ denotes its Euclidean norm. For the deterministic
series $\left\{  a_{n},b_{n}\right\}  _{n=1}^{\infty}$, we denote
$a_{n}\propto b_{n}$ if $0<C_{1}\leq\lim\inf_{n\rightarrow\infty}\left\vert
a_{n}/b_{n}\right\vert \leq\lim\sup_{n\rightarrow\infty}\left\vert a_{n}%
/b_{n}\right\vert \leq C_{2}<\infty$ for some constants $C_{1}$ and $C_{2},$
$a_{n}\lesssim b_{n}$ if $\lim\sup_{n\rightarrow\infty}\left\vert a_{n}%
/b_{n}\right\vert \leq C<\infty$ for some $C$ that does not depend on $n,$
$a_{n}\gtrsim b_{n}$ if $b_{n}\lesssim a_{n},$ $a_{n}\ll b_{n}$ if
$a_{n}=o\left(  b_{n}\right)  ,$ and $a_{n}\gg b_{n}$ if $b_{n}\ll a_{n}.$
$\mathcal{A}^{c}$ denotes the complement of the set $\mathcal{A}$.
$\overset{P}{\rightarrow}$ denotes convergence in probability. $C$ and $M$
denote some positive constants that may vary from line to line.

\section{The Model and One-Stage Procedure}

In this section we give the model and definitions of various versions of
signals and noises. Then we will provide the one-stage procedure for variable
selection, present the basic assumptions and study the asymptotic properties
of our one-stage procedure.

\subsection{The Model}

Recently, CKP proposed a powerful multiple-stage procedure for linear models.
Our paper aims to generalize the results of CKP to the nonparametric additive
models. The notations and technical details in CKP are already quite tedious,
and they would be even more so in this paper. To facilitate the exposition, we
start with a simple case where there are no pre-determined variables, and we
conduct only one-stage multiple testing. We will present the more powerful
multiple-stage procedure and show its validity in Section \ref{SEC:multi}.

Suppose the model is
\begin{equation}
Y=f^{\ast}\left(  X_{1},X_{2},\ldots,X_{p^{\ast}}\right)  +\varepsilon,
\label{EQ:model}%
\end{equation}
where $Y$ is the dependent variable, $X_{1},$ $X_{2},\ldots,$ and $X_{p^{\ast
}}$ are random independent variables, $\varepsilon$ is an unobserved error
term, and $f^{\ast}$ is an unknown smooth function. Even though we have only
$p^{\ast}$ signal variables, namely $X_{1},$ $X_{2},\ldots,$ $X_{p^{\ast}}$,
that should be included into the regression model in (\ref{EQ:model}), we do
not know this truth before the data reveal the fact. The realistic situation
is that the $p^{\ast}$ signals are contained in a set $\mathcal{S}%
_{n}=\left\{  X_{j},j=1,2,\ldots,p_{n}\right\}  ,$ where $p_{n}$ can be much
larger than the sample size $n$ and $p_{n}\propto n^{B_{p}}$ for some
$B_{p}>0$. $\mathcal{S}_{n}$ is the set for all candidate variables, and we
refer to it as the \textit{active }set. We assume that $E\left(
\varepsilon|X_{1},X_{2},\ldots,X_{p^{\ast}},X_{p^{\ast}+1},\ldots,X_{p_{n}%
}\right)  =0.$ The target of the model selection is to pick up those signals
among $\mathcal{S}_{n}.$

To avoid the curse of dimensionality in nonparametric estimation, we impose
the additive structure on $f^{\ast}$, that is,
\begin{equation}
f^{\ast}\left(  X_{1},X_{2},\ldots,X_{p^{\ast}}\right)  =\sum_{j=1}^{p^{\ast}%
}f_{j}^{\ast}\left(  X_{j}\right)  . \label{EQ:model1a}%
\end{equation}
We allow the additive components\textbf{ }$f_{l}^{\ast}\left(  X_{l}\right)  $
to change with $n$ but we suppress the dependence of $f_{l}^{\ast}\left(
X_{l}\right)  $ on $n$ for notional convenience. Obviously, the individual
functions $f_{j}^{\ast},$ $j=1,\ldots,p^{\ast},$ cannot be identified without
certain suitable normalizations. We impose the normalization by assuming that
$E[f_{j}^{\ast}\left(  X_{j}\right)  ]=0$ for each $j$ and rewrite the model
as%
\begin{equation}
Y=\mu+\sum_{j=1}^{p^{\ast}}f_{j}^{\ast}\left(  X_{j}\right)  +\varepsilon,
\label{EQ:model1b}%
\end{equation}
where $\mu=E\left(  Y\right)  .$ Since $\mu$ can be estimated by the sample
mean of the $Y$ variable at the usual $\sqrt{n}$-rate and this estimation does
not affect the estimation of the nonparametric additive component, for the
simplicity of presentation, we assume $\mu=0$ below.

Since we will focus on the $B$-spline-based nonparametric theory, it is
standard to assume compact support for each regressor (see, e.g.,
\cite{Horowitz_Mammen2014}, \cite{Chen_handbook} and \cite{HuangEtal2010}).
Without loss of generality, we assume that the support for $X_{j}$ is
$[0,1\mathcal{]}$ for all $j$. Let $\alpha_{1}$ be a non-negative integer,
$\alpha_{2}\in(0,1],$ and $d=\alpha_{1}+\alpha_{2}.$\ Denote the class of all
$\alpha_{1}$ times continuously differentiable real-valued functions on
$[0,1]$ by $C^{\alpha_{1}}\left(  [0,1]\right)  .$\ Define $d$-th smooth
real-valued functions on $[0,1]$ as
\begin{equation}
\Lambda^{d}\left(  [0,1]\right)  =\left\{  h\in C^{\alpha_{1}}\left(
[0,1]\right)  :\left\vert h^{\left(  \alpha_{1}\right)  }\left(  t_{1}\right)
-h^{\left(  \alpha_{1}\right)  }\left(  t_{2}\right)  \right\vert \leq
C\left\vert t_{1}-t_{2}\right\vert ^{\alpha_{2}}\right\}  . \label{EQ:model1c}%
\end{equation}
For notational simplicity, we will restrict our attention to the case where
$f_{j}$'s are smooth enough and belong to $\Lambda^{d}\left(  [0,1]\right)  $
with $d>1$. The case of different smoothness parameters only complicates the
notation but does not bring in any new insight.

\subsection{Signals, Hidden Signals, and Noises}

The idea of one-stage procedure is that we estimate the impact of $X_{l}$ on
$Y,$ $l=1,2,...,p_{n},$ one by one. So we run $p_{n}$ estimations in total and
will keep those variables that are significant enough. Since in each
regression we only have one covariate, we are not estimating $f_{l}^{\ast}$
when the explanatory variable is $X_{l}$. Instead, we are estimating the
conditional expectation $f_{l}\left(  X_{l}\right)  \equiv E\left(
Y|X_{l}\right)  .$ We define the \textit{net} \textit{impact} (or the net
effect) of $X_{l}$ on $Y$ as
\[
\theta_{l}\equiv\left\{  E\left[  f_{l}\left(  X_{l}\right)  ^{2}\right]
\right\}  ^{1/2}=\left\{  E\left[  \left(  \sum_{j=1}^{p^{\ast}}\sigma
_{lj}\right)  ^{2}\right]  \right\}  ^{1/2},
\]
where $\sigma_{lj}=E[f_{j}^{\ast}\left(  X_{j}\right)  |X_{l}].$ Since we
allow\textbf{ }$f_{l}^{\ast}\left(  X_{l}\right)  $ to change with $n,$
$\theta_{l}$ might change with $n$ as well. But we suppress its dependence on
$n$ for notational convenience. Here, $\sigma_{lj}$ plays the role of the
scaled covariance between $X_{j}$ and $X_{l}$ for the linear model considered
by CKP. To better understand the connection between the net impact defined
here and in CKP, we refer the readers to the results in equations
(\ref{EQ:fnl&Ul}) and (\ref{EQ:fnl_appro}), where we approximate $f_{l}\left(
X_{l}\right)  $ with certain linear functions $f_{nl}\left(  X_{l}\right)  .$
Ignoring the bias from the approximation, one can see the connection more clearly.

Obviously, $\theta_{l}$ can be 0 or close to 0 for signals, and $\theta_{l}$
can be nonzero or large for non-signals. As in CKP, we also have four
possibilities as tabulated in Table \ref{Table1}. Cases (I) and (IV) in Table
1 are desirable cases. Case (III) happens when some non-signals are not
independent of the signals.\ The hidden signals defined in Case (II) are rare
in the linear case, and it is also rare in the nonparametric case. We
generalize the definition of hidden signals to Table 2 in the next section
where $\theta_{l}$ is non-zero but small relative to the sample size.
\begin{table}[ptb]
\caption{The original definitions of signals and noises from CKP}%
\label{Table1}
\centering{}\centering{ }
\begin{tabular}
[c]{l|cc}\hline\hline
& $\theta_{l}\neq0$ & $\theta_{l}=0$\\\hline
$\left\{  E[f_{l}^{\ast}\left(  X_{l}\right)  ^{2}]\right\}  ^{1/2}\neq0$ &
\multicolumn{1}{|l}{(I) Signals with nonzero net effect} &
\multicolumn{1}{l}{(II) Hidden signals}\\
$\left\{  E[f_{l}^{\ast}\left(  X_{l}\right)  ^{2}]\right\}  ^{1/2}=0$ &
\multicolumn{1}{|l}{(III) Pseudo-signals} & \multicolumn{1}{l}{(IV) Noise
variables}\\\hline\hline
\end{tabular}
\end{table}

As we shall see, our one-stage procedure is silent on picking up hidden
signals and eliminating pseudo-signals. For the hidden signals, we will
propose a multiple-stage procedure in Section \ref{SEC:multi}\ that can
effectively pick them up. To eliminate the pseudo-signals, as a post-procedure
we propose to re-estimate the model using adaptive group Lasso in Section
\ref{SEC:Pseudo}.

We assume that there are $p^{\ast\ast}$ pseudo-signals. Without loss of
generality, we denote them to be
\[
\left\{  X_{p^{\ast}+1},X_{p^{\ast}+2},\ldots,X_{p^{\ast}+p^{\ast\ast}%
}\right\}  .
\]
Below we focus on the one-stage procedure in this section and postpone the
multi-stage case to the next section.

\subsection{The Test Statistic and One-Stage Procedure\label{SEC:1stage}}

Suppose that we have $n$ observations $\left\{  (y_{i},x_{1i},...,x_{p_{n}%
i})\right\}  _{i=1}^{n}$ that are drawn from the distribution of
$(Y,X_{1},...,X_{p_{n}}).$ The data are given in a $n\times\left(
p_{n}+1\right)  $ matrix
\[
\left(  \boldsymbol{y},\boldsymbol{x}_{1},\boldsymbol{x}_{2},\ldots
,\boldsymbol{x}_{p_{n}}\right)
\]
where $\boldsymbol{y}=\left(  y_{1},y_{2},\ldots,y_{n}\right)  ^{\prime}$ and
$\boldsymbol{x}_{l}=\left(  x_{l1},x_{l2},\ldots,x_{ln}\right)  ^{\prime}$ for
$l=1,...,p_{n}$. We propose to choose finite order (e.g., cubic) B-spline
basis functions $\left\{  \psi_{l}\left(  x\right)  \right\}  _{l=1}^{m_{n}}%
$\ on $\left[  0,1\right]  $\ to approximate the unknown functions $f_{l}$'s.

B-splines are piecewise-defined polynomial functions that can be used to
construct curves and surfaces in numerical analysis. They offer a flexible way
to model and control the shape of these curves and surfaces. A B-spline of
order $n$ is a piecewise-defined polynomial function of degree $n-1$.
B-splines of order one are piecewise constant functions, B-splines of order
two are piecewise linear functions, B-splines of order three are piecewise
quadratic functions, and so on. The points at which different polynomial
pieces connect are called knots, and the knot vector specifies where these
knots are. For the detailed definition and properties of B-spline bases, see
\cite{Stone} and \cite{deBoor}. We list some properties of the B-spline basis
functions in Lemma \ref{LE:rank}. Other popular basis functions include
polynomials, trigonometric polynomials, splines, and orthogonal wavelets. We
refer the readers to \cite{Chen_handbook} for a nice review about sieve estimation.

Since we normalize $E[f_{j}^{\ast}\left(  X_{j}\right)  ]=0,$ we similarly
normalize the basis as%
\[
\phi_{jl}\left(  x\right)  =\psi_{j}\left(  x\right)  -n^{-1}\sum_{i=1}%
^{n}\psi_{j}\left(  x_{li}\right)  .
\]
This is a standard practice; see, e.g., \cite{HuangEtal2010}. For notational
simplicity, we will write $\phi_{j}\left(  x\right)  $ for $\phi_{jl}\left(
x\right)  $. Let $P^{m_{n}}\left(  x\right)  =\left[  \phi_{1}\left(
x\right)  ,\phi_{2}\left(  x\right)  ,\ldots,\phi_{m_{n}}\left(  x\right)
\right]  ^{\prime},$ an $m_{n}\times1$\ vector. Define%
\begin{equation}
\boldsymbol{\beta}_{l}=\left\{  E\left[  P^{m_{n}}\left(  X_{l}\right)
P^{m_{n}}\left(  X_{l}\right)  ^{\prime}\right]  \right\}  ^{-1}E\left[
P^{m_{n}}\left(  X_{l}\right)  Y\right]  \text{ and }U_{l}=Y-P^{m_{n}}\left(
X_{l}\right)  ^{\prime}\boldsymbol{\beta}_{l}, \label{EQ:betanl}%
\end{equation}
which are the population coefficient and the error term in the regression of
$Y$ on $P^{m_{n}}\left(  X_{l}\right)  $. Note that we suppress the dependence
of $\boldsymbol{\beta}_{l}$ on the sample size $n$.

For the one-stage procedure, we conduct the regression of $Y$ on $P^{m_{n}%
}\left(  X_{l}\right)  ,$ $l=1,2,\ldots,p_{n},$ one by one. Let $\mathbb{X}%
_{li}=P^{m_{n}}\left(  x_{li}\right)  $ be the approximating function basis at
the $i$th observation for $X_{l}$.$\ $Let $\mathbb{X}_{l}=(\mathbb{X}%
_{l1},\mathbb{X}_{l2},\ldots,\mathbb{X}_{ln})^{\prime}$ be the $n\times m_{n}$
\textquotedblleft design\textquotedblright\ matrix for $X_{l}.$ For $X_{l},$
we regress $\boldsymbol{y}$ on $\mathbb{X}_{l}\ $to obtain%
\[
\boldsymbol{\hat{\beta}}_{l}=\left(  \mathbb{X}_{l}^{\prime}\mathbb{X}%
_{l}\right)  ^{-1}\mathbb{X}_{l}^{\prime}\boldsymbol{y}.
\]
We construct the test statistic as\footnote{As a referee has noted, one can
define an alternative test statistic $\tilde{\mathcal{X}}_{l}=n\hat
{\boldsymbol{\beta}}^{\prime}\hat{\boldsymbol{\beta}}$, which also works under
certain rank conditions (see, e.g., Assumption \ref{A:full_rank2} below). We
opted for $\hat{\mathcal{X}}_{l}$ for two reasons. First, $\hat{\mathcal{X}%
}_{l}$ resembles the usual chi-squared statistic under conditional
homoskedasticity. Because we do not want to model the conditional
heteroskedasticity of unknown form, we cannot take into account conditional
heteroskedasticity explicitly in constructing the test statistic
$\hat{\mathcal{X}}_{l}.$ Despite this, our asymptotic theory allows for
conditional heteroskedasticity in the error term. Second and more importantly,
$\hat{\mathcal{X}}_{l}$ is scale-free whereas $\tilde{\mathcal{X}}_{l}$ is
not. The latter makes it very challenging to choose the range to search the
constant $C$ in $\varsigma_{n}$ defined below.}
\begin{equation}
\mathcal{\hat{X}}_{l}=\boldsymbol{\hat{\beta}}_{l}^{\prime}\left(  \hat
{\sigma}_{l}^{-2}\mathbb{X}_{l}^{\prime}\mathbb{X}_{l}\right)
\boldsymbol{\hat{\beta}}_{l}, \label{EQ:chil_definition}%
\end{equation}
where $\hat{\sigma}_{l}^{2}=n^{-1}\sum_{i=1}^{n}\hat{u}_{li}^{2}$ and $\hat
{u}_{li}$ is the residual from the above regression. Then we define the
first-stage OCMT selection indictor as%
\begin{equation}
\widehat{\mathcal{J}}_{l}=\mathbf{1}\left(  \mathcal{\hat{X}}_{l}%
>\varsigma_{n}\right)  \text{ for }l=1,2,\ldots,p_{n}, \label{EQ:select1stage}%
\end{equation}
where $\mathbf{1}\left(  \cdot\right)  $ is the usual indicator function and
$\varsigma_{n}$ is a threshold value.

For the linear model in CKP with one covariate at a time, $\mathcal{\hat{X}%
}_{l}$ is asymptotically $\chi^{2}\left(  1\right)  $ under conditional
homoskedasticity, and one can follow their lead to consider threshold values
for the associated $t$-statistics based on the adjusted normal critical
values. Nevertheless, such a result is not available in our framework due to
the divergent dimension of regressors in the sieve estimation. In addition,
the potential presence of conditional heteroskedasticity greatly complicates
our analysis too. What we really need is to show that $\mathcal{\hat{X}}_{l}$
behaves distinctly for signals and noises so that a suitable choice of the
threshold value $\varsigma_{n}$ can help us separate the signals from the
noises. For these reasons, we do not associate $\mathcal{\hat{X}}_{l}$ with
any asymptotic distribution. Instead, we will set $\varsigma_{n}\propto
\kappa_{n}\log\left(  m_{n}\right)  m_{n}$ for a positive series $\kappa_{n}$
that diverges to infinity slowly as in Assumption \ref{A:xi_n}. For more
details, see the remark on Assumption \ref{A:xi_n} in the next subsection.

\subsection{Basic Assumptions}

To study the asymptotic properties of the one-stage procedure, we impose the
following assumptions.

\begin{assumption}
\label{A:iid}$\left\{  y_{i},x_{1i},x_{2i},\ldots,x_{p_{n}i}\right\}
_{i=1}^{n}$ are independent and identically distributed (i.i.d.) across $i;$
$E(\varepsilon|X_{1},X_{2},\ldots,X_{p^{\ast}},X_{p^{\ast}+1},\ldots,$
$X_{p_{n}})$ $=0.$
\end{assumption}

\begin{assumption}
\label{A:p} $p^{\ast}\ $is a positive integer that does not vary with
$n.$\ $p^{\ast\ast}\lesssim n^{B_{p^{\ast\ast}}}\ $and $p_{n}\propto n^{B_{p}%
}$ for some $B_{p}>B_{p^{\ast\ast}}\geq0.$
\end{assumption}

\begin{assumption}
\label{A:supp}The support for $X_{l}$ is $\left[  0,1\right]  ,$
$l=1,\ldots,p_{n}$. The density function for $X_{l}$ is bounded and bounded
away from $0.$
\end{assumption}

\begin{assumption}
\label{A:epsilon}$\Pr\left(  \left\vert \varepsilon\right\vert >t\right)  \leq
C_{1}\exp\left(  -C_{2}t^{s}\right)  $ holds for all $t>0$ and some $s,$
$C_{1},$ $C_{2}>0.$
\end{assumption}

\begin{assumption}
\label{A:fl}$f_{l}\left(  \cdot\right)  =E\left(  Y|X_{l}=\cdot\right)
\in\Lambda^{d}\left(  [0,1]\right)  \ $with $d>1$ for $l=1,\ldots,p_{n}$.
\end{assumption}

\begin{assumption}
\label{A:tech}$\left\vert f_{j}^{\ast}\left(  \cdot\right)  \right\vert ,$
$j=1,...,p^{\ast},$ are uniformly bounded. $E\left(  \varepsilon^{2}%
|X_{1},X_{2},...,X_{p_{n}}\right)  $ is uniformly bounded almost surely.
\end{assumption}

\begin{assumption}
\label{A:mn}$m_{n}\propto n^{B_{m}}$ with $1/\left(  1+2d\right)  <B_{m}<1/3.$
\end{assumption}

\begin{assumption}
\label{A:xi_n}$\varsigma_{n}\propto\kappa_{n}\log\left(  m_{n}\right)
m_{n}\ $for some $\kappa_{n}$ such that $\kappa_{n}>0,$ $\kappa_{n}%
\rightarrow\infty$, and $\kappa_{n}=O\left[  \left(  \log n\right)
^{\epsilon}\right]  $ for some small positive $\epsilon$,\ as $n\rightarrow
\infty$.$\medskip$
\end{assumption}

Assumption \ref*{A:iid} imposes an i.i.d. condition on the observations and a
standard conditional moment restriction. The extension to weakly dependent
observations is possible but left for future research. The generalization to
independently non-identically distributed (i.n.i.d.) case is straightforward
because the main inequalities in Lemmas A.1 and A.2 allow for i.n.i.d.
observations. We keep using the i.i.d. assumption for notational convenience.
In Assumption \ref*{A:p}, we assume that the number of signals is fixed; the
number of pseudo-signals is allowed to increase as $n$ increases but at a
slower rate than that of the total candidate variables. It is also possible to
allow $p^{\ast}$ to diverge to infinity (see Section \ref{SEC:divergeP}).
Assumption \ref*{A:supp} restricts the support of $X$ to be $\left[
0,1\right]  .$ This is a very common condition for nonparametric additive
models; see, e.g., \cite{Li2000} and \cite{Horowitz_Mammen2014}. Assumption
\ref*{A:epsilon} imposes some tail conditions $\varepsilon,$ which is also
assumed in CKP but weaker than the commonly used sub-exponential condition in
the variable selection literature and the one for the adaptive group Lasso in
\cite{HuangEtal2010}. The primary use of this condition is to derive
\textit{probability} bound for errors.

Assumption \ref*{A:fl} imposes fairly weak smooth condition on $f_{l}\left(
x\right)  $, which is weaker than the commonly used condition $d\geq2$. Note
that $d>1$ is needed for Assumption \ref*{A:mn}. Assumption \ref*{A:tech} is a
technical assumption needed to simplify the proof.\ Specifically, the
boundedness of $f_{j}^{\ast}$ implies $U_{l}$ defined in equation
(\ref{EQ:fnl&Ul}) has the same tail behavior as $\varepsilon$. The bounded
conditional second moment of $\varepsilon$\ is to ensure some nice properties
of $U_{l}\phi_{j}\left(  X_{l}\right)  $ and this assumption is also common in
the sieve literature (see, e.g., \cite{Newey1997}).\ This assumption is mild
given that we assume all $X_{l}$'s have compact support and the tails of
$\varepsilon$ decay exponentially fast. Assumption \ref*{A:mn} imposes
conditions on $m_{n}.$ First, we need $B_{m}<1/3$ such that $nm_{n}%
^{-3}\rightarrow\infty.$ The last condition is necessary for $\left\Vert
n^{-1}\mathbb{X}_{l}^{\prime}\mathbb{X}_{l}\right\Vert \propto m_{n}^{-1}$ to
hold with very high probability. To see why, note that Lemma \ref{LE:rank} in
Appendix \ref{APP:tech_lemmas} suggests that $\left\Vert E\left[  P^{m_{n}%
}\left(  X_{l}\right)  P^{m_{n}}\left(  X_{l}\right)  ^{\prime}\right]
\right\Vert \propto m_{n}^{-1}.$ We need $m_{n}^{-1}\gg\left(  m_{n}/n\right)
^{1/2},$ or equivalently, $nm_{n}^{-3}\rightarrow\infty,$ in order to ensure
that $n^{-1}\mathbb{X}_{l}^{\prime}\mathbb{X}_{l}$ is close to $E\left[
P^{m_{n}}\left(  X_{l}\right)  P^{m_{n}}\left(  X_{l}\right)  ^{\prime
}\right]  $ with very high probability. Second, we need $B_{m}>1/\left(
1+2d\right)  $ to ensure that the approximation bias is asymptotically
negligible in comparison with the asymptotic variance term: $m_{n}^{-d}%
\ll\left(  m_{n}/n\right)  ^{1/2}$.

Assumption \ref*{A:xi_n} imposes conditions on the threshold value
$\varsigma_{n}$ that ensures the separability of the signals from noises. If
the true value of $\boldsymbol{\beta}_{l}$ is $\boldsymbol{0}$, $\mathcal{\hat
{X}}_{l}$ in equation (\ref{EQ:chil_definition}) is $O_{P}\left(
m_{n}\right)  .$ In this case, to ensure $\widehat{\mathcal{J}}_{l}=0$ with
very high probability, we can take $\varsigma_{n}\propto\kappa_{n}\log\left(
m_{n}\right)  m_{n}$ and lose some power up to $\kappa_{n}\log\left(
m_{n}\right)  .$ Here, $\kappa_{n}$ can be any series diverging to infinity
slowly, e.g., $\left[  \log\left(  m_{n}\right)  \right]  ^{\epsilon}$ for
some small $\epsilon>0.$ The loss of the power to some degree is inevitable
because of the nature of the multiple testing procedure when the number of
tests goes to infinity. In contrast, CKP choose their threshold by the
Bonferroni correction of the standard normal. We choose not to follow them
because of the following reasons. First, given the divergence of $m_{n},$
$\mathcal{\hat{X}}_{l}$ does not converge to a chi-square distribution
asymptotically even in the homoskedastic case so that we cannot use chi-square
distribution to approximate the finite sample distribution of $\mathcal{\hat
{X}}_{l}.$ Under conditional heteroskedasticity, $\mathcal{\hat{X}}_{l}$ does
not converge to a chi-square distribution even if $m_{n}$ is held fixed. So
our procedure does not rely on the chi-square approximation. Second, even if
we can do the approximation, the cumulative density function (CDF) of a
chi-square distribution is very complicated. We do not have a rate for the
inverse of its CDF evaluated at a certain rate (e.g., $n^{-C}$) like the case
of normal CDF. For these reasons, we do our selection based on the asymptotic
results. Specifically, we will take $\varsigma_{n}=C\kappa_{n}\log\left(
m_{n}\right)  m_{n}$ for some $\kappa_{n}$ and choose the value of $C$ by the
classic BIC. The details can be found in Appendix \ref{SEC:procedure}. This
$\varsigma_{n}$ helps separate the signals from the noises with very high
probability, as demonstrated in the next section.\ Recall that in Assumption
\ref*{A:p} we assume that $p_{n}$ go to infinity at a polynomial rate of $n,$
same as $m_{n}.$ This ensures that $\log\left(  m_{n}\right)  \propto\log
p_{n}.$ Therefore, theoretically we only need to put $\log\left(
m_{n}\right)  $ in $\varsigma_{n}$ to have a control of $p^{\ast\ast}$ and
$p_{n}$ for the TPR, FDR and FPR defined after Proposition \ref{TH:main1}
below. We postpone the discussion on the comparison of technical conditions
required for our procedure and those required for the AGLASSO to Section 3.2.

\subsection{The Asymptotic Properties of $\mathcal{\hat{X}}_{l}$ and the
One-Stage Procedure\label{SEC:non_1st}}

We present the first theoretical result in this paper. It derives the
probability bounds for the \textquotedblleft Type-I\textquotedblright\ and
\textquotedblleft Type-II\textquotedblright\ errors.

\begin{proposition}
\label{TH:main1}Suppose that $Y$ is given by equation (\ref{EQ:model}) and
Assumptions \ref*{A:iid} $-$ \ref*{A:xi_n} hold.

(i) If $\theta_{l}\lesssim\log\left(  m_{n}\right)  ^{1/2}\left(
m_{n}/n\right)  ^{1/2},$ then for sufficiently large $n$ we have%
\begin{align*}
\Pr\left(  \mathcal{\hat{X}}_{l}\geq\varsigma_{n}\right)   &  \leq\exp\left(
-C_{1}m_{n}^{-1}\varsigma_{n}+\log m_{n}\right)  +C_{2}\exp\left(
-C_{3}n^{C_{4}}\right) \\
&  \leq n^{-M}+C_{2}\exp\left(  -C_{3}n^{C_{4}}\right)
\end{align*}
for any fixed constant $M>0$ and some positive constants $C_{1},$ $C_{2},$
$C_{3},$ and $C_{4}.$

(ii) If $\theta_{l}\gtrsim\kappa_{n}\log\left(  m_{n}\right)  ^{1/2}\left(
m_{n}/n\right)  ^{1/2}$ with $\kappa_{n}$ specified in Assumption
\ref*{A:xi_n}$,$ then for sufficiently large $n$ we have%
\[
\Pr\left(  \mathcal{\hat{X}}_{l}\geq\varsigma_{n}\right)  \geq1-n^{-M}%
-C_{5}\exp\left(  -C_{6}n^{C_{7}}\right)
\]
for any fixed constant $M>0$ and some positive constants $C_{5},$ $C_{6},$ and
$C_{7}.$
\end{proposition}

The proof of Proposition \ref{TH:main1} is tedious. We provide some technical
details in Appendix \ref{APP:tech_one_stage} before we formally prove the
proposition in Appendix \ref{APP:real_main_proof}. An implication of
Proposition \ref{TH:main1} is that for the well-chosen threshold value
$\varsigma_{n},$ the above one-stage procedure can separate the signals with
$\theta_{l}\gtrsim\kappa_{n}\log\left(  m_{n}\right)  ^{1/2}\left(
m_{n}/n\right)  ^{1/2}$ from the noises with $\theta_{l}=0.$ Of course, we may
have some intermediate case where $0<\theta_{l}\lesssim\log\left(
m_{n}\right)  ^{1/2}\left(  m_{n}/n\right)  ^{1/2}$ for which the above
procedure fails to do so. Note that the convergence rate for each additive
term in \cite{Stone} is $\left(  m_{n}/n\right)  ^{1/2}.$ Therefore, our
procedure loses some power up to the order of $\log\left(  n\right)  $, a
common scenario in the variable selection literature.

Following the literature and CKP, we define the true positive rates (TPR) and
the false positive rates (FPR) respectively as
\begin{align*}
\text{TPR}_{n}  &  =\frac{\sum_{l=1}^{p_{n}}\mathbf{1}\left(
\widehat{\mathcal{J}}_{l}=1\text{ and }\left\{  E\left[  f_{l}^{\ast}\left(
X_{l}\right)  ^{2}\right]  \right\}  ^{1/2}\neq0\right)  }{\sum_{l=1}^{p_{n}%
}\mathbf{1}\left(  \left\{  E\left[  f_{l}^{\ast}\left(  X_{l}\right)
^{2}\right]  \right\}  ^{1/2}\neq0\right)  },\text{ and}\\
\text{ FPR}_{n}  &  =\frac{\sum_{l=1}^{p_{n}}\mathbf{1}\left(
\widehat{\mathcal{J}}_{l}=1\text{ and }\left\{  E\left[  f_{l}^{\ast}\left(
X_{l}\right)  ^{2}\right]  \right\}  ^{1/2}=0\right)  }{\sum_{l=1}^{p_{n}%
}\mathbf{1}\left(  \left\{  E\left[  f_{l}^{\ast}\left(  X_{l}\right)
^{2}\right]  \right\}  ^{1/2}=0\right)  }.
\end{align*}
Based on the test statistic defined in equation (\ref{EQ:chil_definition}) and
its property developed in Proposition \ref{TH:main1}, we introduce the
generalized definitions of signals and noises in Table \ref{Table2}. With this
definition of hidden signals, our result in Section \ref{SEC:multi} provides
theoretical justification for the necessity of a multi-stage procedure capable
of detecting signals not identified in the first stage, yet with non-zero net
effects.\begin{table}[ptb]
\caption{The generalized definitions of signals and noises}%
\label{Table2}%
\centering{}\centering{ }
\begin{tabular}
[c]{l|ll}\hline\hline
& $\theta_{l}\gg\log\left(  m_{n}\right)  ^{1/2}\left(  m_{n}/n\right)
^{1/2}$ & $\theta_{l}\lesssim\log\left(  m_{n}\right)  ^{1/2}\left(
m_{n}/n\right)  ^{1/2}$\\\hline
$\left\{  E[f_{l}^{\ast}\left(  X_{l}\right)  ^{2}]\right\}  ^{1/2}\neq0$ &
(I) Signals & (II) Hidden signals\\
$\left\{  E[f_{l}^{\ast}\left(  X_{l}\right)  ^{2}]\right\}  ^{1/2}=0$ & (III)
Pseudo-signals & (IV) Noise variables\\\hline\hline
\end{tabular}
\end{table}

With a little abuse of notation, we continue to use $p^{\ast}$ and
$p^{\ast\ast}$ to denote the number of signals and pseudo-signals at the
sample level, and assume that they satisfy Assumption \ref*{A:p}. Adopting the
definitions in Table 2, we define the false discovery rates (FDR) as%
\[
\text{FDR}_{n}=\frac{\sum_{l=1}^{p_{n}}\mathbf{1}\left(  \widehat{\mathcal{J}%
}_{l}=1\text{, }\left\{  E\left[  f_{l}^{\ast}\left(  X_{l}\right)
^{2}\right]  \right\}  ^{1/2}=0,\text{ and }\theta_{l}\lesssim\log\left(
m_{n}\right)  ^{1/2}\left(  m_{n}/n\right)  ^{1/2}\right)  }{\sum_{l=1}%
^{p_{n}}\widehat{\mathcal{J}}_{l}+1}.
\]

Apparently, the definitions in Table 2 generalize the definitions in Table 1.

The one-stage procedure is valid in terms of TPR, only if the net effects of
all signals are strong enough. Specifically, we need the following assumption.

\begin{assumption}
\label{A;no_hidden} There are no hidden signals. That is, $\theta_{j}%
\gtrsim\kappa_{n}\log\left(  m_{n}\right)  ^{1/2}\left(  m_{n}/n\right)
^{1/2}$ for some slowly divergent series $\kappa_{n}$ as in Assumption
\ref{A:xi_n} for all $j=1,2,...,p^{\ast}.\smallskip$
\end{assumption}

We note the definition of no hidden signals in the above assumption is
equivalent to the one in Table 2, given the way $\kappa_{n}$ is defined in
Assumption 8. If Assumption \ref*{A;no_hidden} fails to hold, we need the
multiple-stage procedure to pick up the hidden signals.

Based on the results in Proposition \ref{TH:main1}, we present the results for
TPR$_{n}$, FPR$_{n}$, and FDR$_{n}$ in the following theorem.

\begin{theorem}
\label{TH:main2}Suppose that $Y$ is given by equation (\ref{EQ:model}) and
Assumptions \ref*{A:iid} $-$ \ref*{A;no_hidden}\ hold. Then after some large
$n,$

(i) $E\left(  \text{\emph{TPR}}_{n}\right)  \geq1-C_{1}\exp\left(
-C_{2}n^{C_{3}}\right)  $ for some positive constants $C_{1},C_{2}$, and
$C_{3};$

(ii) $E\left(  \text{\emph{FPR}}_{n}\right)  \leq p^{\ast\ast}/\left(
p_{n}-p^{\ast}\right)  +C_{4}n^{-M}+C_{5}\exp\left(  -C_{6}n^{C_{7}}\right)  $
for any fixed positive large constant $M$ and some positive constants $C_{4},$
$C_{5},C_{6},$ and $C_{7};$

(iii) \emph{FDR}$_{n}\overset{P}{\rightarrow}0.$
\end{theorem}

Theorem \ref{TH:main2} implies that all of TPR$_{n},$ FPR$_{n}$ and FDR$_{n}$
can be well controlled provided we assume away hidden signals at the sample
level. Note that Theorem \ref{TH:main2}(i)--(ii) focuses on the asymptotic
properties of TPR$_{n}$ and FPR$_{n}$ while the last part of Theorem
\ref{TH:main2} reveals that the false discovery rate is asymptotically
vanishing in large samples.

In the next section, we turn to the multiple-stage procedure that does not
rely on Assumption \ref*{A;no_hidden}.

\section{The Multiple-Stage Procedure\label{SEC:multi}}

In this section we propose a multiple-stage procedure to select variables for
the nonparametric additive models.

\subsection{The Test Statistic with Pre-selected Variables}

As mentioned above, we may not identify a signal $X_{l}$ whose net effect
satisfies $\theta_{l}\lesssim\kappa_{n}\log\left(  m_{n}\right)  ^{1/2}\left(
m_{n}/n\right)  ^{1/2},$ even in the case where the \textit{marginal} effect
of $X_{l}$ on $Y,$ namely, $\left\{  E\left[  f_{l}^{\ast}\left(
X_{l}\right)  ^{2}\right]  \right\}  ^{1/2},$ is large enough.\footnote{The
marginal effect defined here is slightly different from that in the
econometrics literature. For example, the marginal effect of $X_{l}$ on $Y$ is
defined as $\beta_{l}$ in the CKP's linear model: $Y=\beta_{0}+\sum
_{l=1}^{p^{\ast}}X_{l}\beta_{l}+\varepsilon,$ but it refers to $X_{l}\beta
_{l}$.in this paper.} Consequently, Theorem \ref{TH:main2} does not hold
without Assumption \ref*{A;no_hidden} which assumes away small sample hidden
signals. In contrast, Lemma \ref{LE:hidden} in Appendix \ref{APP:tech_lemmas}
underpins the result that as long as $\left\{  E\left[  f_{l}^{\ast}\left(
X_{l}\right)  ^{2}\right]  \right\}  ^{1/2}\gtrsim\kappa_{n}\log\left(
m_{n}\right)  ^{1/2}\left(  m_{n}/n\right)  ^{1/2}$ for some slowly divergent
series $\kappa_{n}$ and some full rank condition holds, the net effect of
$X_{l}$ will be strong enough to be picked up at certain stage of a
multiple-stage procedure. A by-product of this Lemma is the existence of at
least one signal for Stage 1 (see the second part of this Lemma).

To introduce the multiple-stage procedure, we need some extra notations.
Suppose at certain stage after stage 1, we have pre-selected $\iota_{n}$
variables from the active set $\mathcal{S}_{n}=\left\{  X_{j},\text{
}j=1,...,p_{n}\right\}  $ based on some selection procedure to be described
below. To avoid confusion, we denote these pre-selected variables as
$Z_{1},\ldots,Z_{\iota_{n}}$. Let $\boldsymbol{Z}\equiv\left(  Z_{1}%
,\ldots,Z_{\iota_{n}}\right)  ^{\prime}\ $and $P^{m_{n}}\left(  \boldsymbol{Z}%
\right)  =(P^{m_{n}}\left(  Z_{1}\right)  ^{\prime},\ldots,$ $P^{m_{n}}\left(
Z_{\iota_{n}}\right)  ^{\prime})^{\prime}.$ Note that $P^{m_{n}}\left(
\boldsymbol{Z}\right)  $ is an $\iota_{n}m_{n}\times1$ vector for
$\boldsymbol{Z}.$ At the next stage, we consider the nonparametric additive
regression of $Y$ on $\boldsymbol{Z}$ and an $X_{l}$ that has not been
selected so far and we do this one by one for all $p_{n}-\iota_{n}$
non-selected variables $X_{l}.$ We define the impact of $X_{l}$ on $Y$ after
controlling $\boldsymbol{Z}$ as%
\begin{align*}
\theta_{l,\boldsymbol{Z}}  &  \equiv\left\{  E\left\{  E\left[  \left.
Y-P^{m_{n}}\left(  \boldsymbol{Z}\right)  ^{\prime}\Phi_{\boldsymbol{Z}}%
^{-1}E\left[  P^{m_{n}}\left(  \boldsymbol{Z}\right)  Y\right]  \right\vert
X_{l}\right]  \right\}  ^{2}\right\}  ^{1/2}\\
&  =\left\{  E\left\{  \sum_{j=1}^{p^{\ast}}E\left[  \left.  f_{j}^{\ast
}\left(  X_{j}\right)  -P^{m_{n}}\left(  \boldsymbol{Z}\right)  ^{\prime}%
\Phi_{\boldsymbol{Z}}^{-1}E\left[  P^{m_{n}}\left(  \boldsymbol{Z}\right)
f_{j}^{\ast}\left(  X_{j}\right)  \right]  \right\vert X_{l}\right]  \right\}
^{2}\right\}  ^{1/2}\\
&  =\left\{  E\left[  \left(  \sum_{j=1}^{p^{\ast}}\mu_{lj,\boldsymbol{Z}%
}\right)  ^{2}\right]  \right\}  ^{1/2},
\end{align*}
where $\Phi_{\boldsymbol{Z}}\equiv E[P^{m_{n}}\left(  \boldsymbol{Z}\right)
P^{m_{n}}\left(  \boldsymbol{Z}\right)  ^{\prime}]$ and $\mu
_{lj,\boldsymbol{Z}}\equiv E\left\{  \left.  f_{j}^{\ast}\left(  X_{j}\right)
-P^{m_{n}}\left(  \boldsymbol{Z}\right)  ^{\prime}\Phi_{\boldsymbol{Z}}%
^{-1}E\left[  P^{m_{n}}\left(  \boldsymbol{Z}\right)  f_{j}^{\ast}\left(
X_{j}\right)  \right]  \right\vert X_{l}\right\}  $ denotes the effect of
$X_{l}$ on $f_{j}^{\ast}\left(  X_{j}\right)  $ after controlling the effects
of $\boldsymbol{Z}$. Apparently, we suppress the dependence of $\theta
_{l,\boldsymbol{Z}}$ on the sample size $n$.

At the sample level, let $\mathbb{Z}$ and $\mathbb{X}_{l}$ denote the $n\times
m_{n}\iota_{n}$ and $n\times m_{n}$ \textquotedblleft design
matrices\textquotedblright\ for $\boldsymbol{Z}$ and $X_{l},$ respectively.
That is,%
\begin{equation}
\mathbb{Z}\equiv\left(  \mathbb{Z}_{1},\ldots,\mathbb{Z}_{\iota_{n}}\right)
\text{ and }\mathbb{X}_{l}\equiv\left(  \mathbb{X}_{l1},\ldots,\mathbb{X}%
_{ln}\right)  ^{\prime}, \label{EQ:design}%
\end{equation}
where $\mathbb{Z}_{l}=\left(  \mathbb{Z}_{l1},\ldots,\mathbb{Z}_{ln}\right)
^{\prime}\ $is a $n\times m_{n}$ matrix, $\mathbb{Z}_{li}=P^{m_{n}}\left(
z_{li}\right)  \ $and $\mathbb{X}_{li}=P^{m_{n}}\left(  x_{li}\right)  .$
Define $M_{\mathbb{Z}}\equiv I_{n}-\mathbb{Z}\left(  \mathbb{Z}^{\prime
}\mathbb{Z}\right)  ^{-1}\mathbb{Z}^{\prime}.$ By the result of partitioned
regressions, the coefficient of $P^{m_{n}}\left(  X_{l}\right)  \ $is
estimated by%
\[
\boldsymbol{\hat{\beta}}_{l,\boldsymbol{Z}}=\left(  \mathbb{X}_{l}^{\prime
}M_{\mathbb{Z}}\mathbb{X}_{l}\right)  ^{-1}\mathbb{X}_{l}^{\prime
}M_{\mathbb{Z}}\boldsymbol{y}.
\]
To determine whether $X_{l}$ should be treated as a signal variable, we
propose the following test statistic
\begin{equation}
\mathcal{\hat{X}}_{l,\boldsymbol{Z}}=\boldsymbol{\hat{\beta}}%
_{l,\boldsymbol{Z}}^{\prime}\left(  \hat{\sigma}_{l,\boldsymbol{Z}}%
^{-2}\mathbb{X}_{l}^{\prime}M_{\mathbb{Z}}\mathbb{X}_{l}\right)
\boldsymbol{\hat{\beta}}_{l,\boldsymbol{Z}}=\left(  \boldsymbol{y}^{\prime
}M_{\mathbb{Z}}\mathbb{X}_{l}\right)  \left(  \hat{\sigma}_{l,\boldsymbol{Z}%
}^{2}\mathbb{X}_{l}^{\prime}M_{\mathbb{Z}}\mathbb{X}_{l}\right)  ^{-1}\left(
\mathbb{X}_{l}^{\prime}M_{\mathbb{Z}}\boldsymbol{y}\right)  \label{EQ:XlZ}%
\end{equation}
where $\hat{\sigma}_{l,\boldsymbol{Z}}^{2}=n^{-1}\sum_{i=1}^{n}\hat
{\varepsilon}_{li}^{2}$ and $\hat{\varepsilon}_{li}$ is the residual from the
regression $\boldsymbol{y}$ on $\left(  \mathbb{Z},\mathbb{X}_{l}\right)  $.

We will study the asymptotic properties of $\mathcal{\hat{X}}%
_{l,\boldsymbol{Z}}$ in the next subsection which lay down the foundation for
our multiple-stage procedure.

\subsection{The Asymptotic Properties of $\mathcal{\hat{X}}_{l,\boldsymbol{Z}%
}$}

To study the asymptotic properties of $\mathcal{\hat{X}}_{l,\boldsymbol{Z}},$
we impose the following technical conditions.\smallskip

\noindent\textbf{Assumption} \textbf{\ref*{A:p}'} $p^{\ast}\ $is a positive
integer that does not vary with $n.$\ $p^{\ast\ast}\lesssim n^{B_{p^{\ast\ast
}}}\ $and $p_{n}\propto n^{B_{p}}$ for some $B_{p}>B_{p^{\ast\ast}}\geq0.$
Further, $B_{p^{\ast\ast}}<\left(  1-3B_{m}\right)  /2.$\smallskip

\noindent\textbf{Assumption} \textbf{\ref*{A:fl}'} $f_{l}\left(  \cdot\right)
=E\left(  Y|X_{l}=\cdot\right)  \in\Lambda^{d}\left(  [0,1]\right)  \ $with
$d>1$ for $l=1,\ldots,p_{n}$. $f_{j}^{\ast}\in\Lambda^{d}\left(  [0,1]\right)
\ $with $d>1$ for $j=1,\ldots,p^{\ast}$.\smallskip

\noindent\textbf{Assumption} \textbf{\ref*{A;no_hidden}'} $\left\{
E[f_{j}^{\ast}\left(  X_{j}\right)  ^{2}]\right\}  ^{1/2}\gtrsim\kappa_{n}%
\log\left(  m_{n}\right)  ^{1/2}\left(  m_{n}/n\right)  ^{1/2}$ for some
slowly divergent series $\kappa_{n}$ as in Assumption \ref{A:xi_n}\ and for
$j=1,\ldots,p^{\ast}$.\smallskip

\begin{assumption}
\label{A:full_rank2}Let $\boldsymbol{X}_{1}^{p^{\ast}+p^{\ast\ast}}%
\equiv\left(  X_{1},\ldots,X_{p^{\ast}},X_{p^{\ast}+1},\ldots,X_{p^{\ast
}+p^{\ast\ast}}\right)  ^{\prime},$ the vector of all signals and pseudo
signals. Similarly, let $P^{m_{n}}(\boldsymbol{X}_{1}^{p^{\ast}+p^{\ast\ast}%
})\equiv\left[  P^{m_{n}}\left(  X_{1}\right)  ^{\prime},\ldots,P^{m_{n}%
}\left(  X_{p^{\ast}+p^{\ast\ast}}\right)  ^{\prime}\right]  ^{\prime}\ $and
$\Phi_{\boldsymbol{X}_{1}^{p^{\ast}+p^{\ast\ast}}}\equiv E[P^{m_{n}%
}(\boldsymbol{X}_{1}^{p^{\ast}+p^{\ast\ast}})$ $P^{m_{n}}(\boldsymbol{X}%
_{1}^{p^{\ast}+p^{\ast\ast}})^{\prime}].$ Assume that%
\[
B_{X1}m_{n}^{-1}\leq\lambda_{\min}\left(  \Phi_{\boldsymbol{X}_{1}^{p^{\ast
}+p^{\ast\ast}}}\right)  \leq\lambda_{\max}\left(  \Phi_{\boldsymbol{X}%
_{1}^{p^{\ast}+p^{\ast\ast}}}\right)  \leq B_{X2}m_{n}^{-1}%
\]
holds for some positive constants $B_{X1}$ and $B_{X2},$ and it also holds
when $\boldsymbol{X}_{1}^{p^{\ast}+p^{\ast\ast}}$ is augmented by an arbitrary
element $X_{l}$ with $l>p^{\ast}+p^{\ast\ast}.$
\end{assumption}

\begin{assumption}
\label{A:noisevariable}If the pre-selected variables $\boldsymbol{Z}$ are
either signals or pseudo-signals, the effect of noise variable defined in
Table 2 on $Y$ is still weak. That is, $\theta_{l,\boldsymbol{Z}}\lesssim
\log\left(  m_{n}\right)  ^{1/2}\left(  m_{n}/n\right)  ^{1/2}$ for all ${l}%
${$=p^{\ast}+p^{\ast\ast}+1,\ldots,p_{n},$} when all variables in
$\boldsymbol{Z}$ are either signals or pseudo-signals.\smallskip
\end{assumption}

Assumption \ref*{A:p}' strengthens Assumption \ref*{A:p} by adding one more
condition on $B_{p^{\ast\ast}}$. It is imposed to ensure the good property of
$\boldsymbol{\tilde{u}}_{l,\boldsymbol{Z}}$ defined in equation
(\ref{EQ:utidelz}) (see Lemma \ref{LE:error_bound}). This condition can be
very restrictive on $B_{p^{\ast\ast}}.$ For example, when $B_{m}=1/4,$ this
condition implies $B_{p^{\ast\ast}}<1/8$. Assumption \ref*{A:fl}' strengthens
Assumption \ref*{A:fl} so that equation (\ref{EQ:appro}) in Appendix A.2
holds. As remarked at the beginning of last subsection, we do not impose
Assumption \ref*{A;no_hidden} for the multiple-stage procedure, as long as the
marginal effect of the signal is not too weak as imposed in Assumption
\ref*{A;no_hidden}'. This assumption seems inevitable. It is analogous to
Assumption 6 in CKP for the linear regression model and similar to the
so-called `beta-min' condition that is commonly assumed in the penalized
regression literature (see, e.g., Chapter 7.4 of \cite{Buhlmann2011}).
Assumption \ref*{A:full_rank2} is the common rank condition for nonparametric
additive regressions. We also require it to hold for the case when we add one
noise variable. This seems inevitable because we will run the regression with
regressors being all signals and pseudo-signals plus one noise variable in the
multiple-stage procedure with very high probability.\ Assumption
\ref*{A:noisevariable} is also inevitable and implicitly imposed in CKP for
the linear regression models\emph{.}

It is worth mentioning that Assumption \ref*{A:full_rank2} plays a similar
role to the \textquotedblleft restrictive eigenvalues\textquotedblright%
\ condition in \cite{BickelRitovTsyvakov2009} and \cite{BelloniEtal2012}. The
restrictive-eigenvalues\ condition requires certain full rank conditions on
all possible design matrices composed of a certain number of covariates. In
contrast, our OCMT only requires full rank conditions on the design matrices
composed of signals and pseudo-signals, and permits arbitrary correlations
among the noise variables. Obviously, these two sets of conditions are
non-nested. In addition, \cite{HuangEtal2010} impose essentially the same set
of assumptions except the rank conditions discussed here. The main advantage
of the OCMT is that it does not require any numerical min-search of an
objective function, can be computed much faster, and deliver more reliable results.

The following proposition presents the probability bounds for the
\textquotedblleft Type-I\textquotedblright\ and \textquotedblleft
Type-II\textquotedblright\ errors when we have some pre-selected variables.

\begin{proposition}
\label{TH:main1_ms}Suppose that $Y$ is given by equation (\ref{EQ:model}),
Assumptions \ref*{A:iid}, \ref*{A:p}', \ref*{A:supp}, \ref*{A:epsilon},
\ref*{A:fl}', \ref*{A:tech}, \ref*{A:mn}, \ref*{A:xi_n}, and
\ref*{A:full_rank2} hold, and the pre-selected variables $\boldsymbol{Z}$ are
either signals or pseudo-signals.

(i) If $\theta_{l,\boldsymbol{Z}}\lesssim\log\left(  m_{n}\right)  ^{1/2}%
m_{n}^{1/2}n^{-1/2},$ then%
\begin{align}
\Pr\left(  \mathcal{\hat{X}}_{l,\boldsymbol{Z}}\geq\varsigma_{n}\right)   &
\leq\exp\left(  -C_{1}m_{n}^{-1}\varsigma_{n}+\log m_{n}\right)  +C_{2}%
\exp\left(  -C_{3}n^{C_{4}}\right) \label{EQ:xl_FDR}\\
&  \leq n^{-M}+C_{2}\exp\left(  -C_{3}n^{C_{4}}\right) \nonumber
\end{align}
for any fixed large constant $M>0$ and some positive constants $C_{1}%
,C_{2},C_{3},$ and $C_{4},$ after some large $n.$

(ii) If $\theta_{l,\boldsymbol{Z}}\gtrsim\kappa_{n}\log\left(  m_{n}\right)
^{1/2}m_{n}^{1/2}n^{-1/2}$ with $\kappa_{n}$ specified in Assumption
\ref{A:xi_n}$,$ then%
\[
\Pr\left(  \mathcal{\hat{X}}_{l,\boldsymbol{Z}}\geq\varsigma_{n}\right)
\geq1-n^{-M}-C_{5}\exp\left(  -C_{6}n^{C_{7}}\right)
\]
for any fixed large constant $M>0$ and some positive constants $C_{5},C_{6},$
and $C_{7},$ after some large $n.$
\end{proposition}

The proof of Proposition \ref{TH:main1_ms} is rather tedious. We provide some
technical discussions on the proof in Appendix \ref{SEC:multi_tech} before we
formally prove it in Appendix \ref{APP:real_main_proof}.

Like Proposition \ref{TH:main1}, Proposition \ref{TH:main1_ms} implies that
for the well-chosen threshold value $\varsigma_{n},$ the use of the test
statistic $\mathcal{\hat{X}}_{l,\boldsymbol{Z}}$ helps to separate variables
with large value of $\theta_{l,\boldsymbol{Z}}$ from those with small value of
$\theta_{l,\boldsymbol{Z}}.$ This observation will be used in our
multiple-stage procedure to select all signal variables.

\subsection{The Multiple-Stage Procedure\label{SEC:multi_procedure}}

We present the multiple-stage procedure as follows.

We conduct the first-stage selection as in Section \ref{SEC:1stage} by
constructing the test statistic $\mathcal{\hat{X}}_{l}$ as in equation
(\ref{EQ:chil_definition}) and using the threshold value $\varsigma_{n}$ that
satisfies the condition in Assumption \ref*{A:xi_n}. We re-label the selection
indicator in equation (\ref{EQ:select1stage}) as%
\[
\widehat{\mathcal{J}}_{l,\left(  1\right)  }=\mathbf{1}\left(  \mathcal{\hat
{X}}_{l}>\varsigma_{n}\right)  \text{ for }l=1,2,\ldots,p_{n}.
\]
We collect all the variables selected in stage 1 into the vector
$\boldsymbol{Z}_{\left(  1\right)  }$, and denote the index set of the
selected variables by $S_{\left(  1\right)  }.$ For the second stage, we
denote the index set of the active variables in stage 2 by $\Psi_{\left(
2\right)  }$ where $\Psi_{\left(  2\right)  }=\left\{  1,2,\ldots
,p_{n}\right\}  \backslash S_{\left(  1\right)  }$. In the second stage, we
regress $Y$ on $P^{m_{n}}\left(  X_{l}\right)  $ with $P^{m_{n}}\left(
\boldsymbol{Z}_{\left(  1\right)  }\right)  $ as pre-selected variables one by
one for $l\in\Psi_{\left(  2\right)  }.$ We construct the test statistic
$\mathcal{\hat{X}}_{l,\boldsymbol{Z}_{\left(  1\right)  }}$ as in equation
(\ref{EQ:XlZ}). We select the variable $X_{l}$ if $\widehat{\mathcal{J}%
}_{l,\left(  2\right)  }=1,$ where
\[
\widehat{\mathcal{J}}_{l,\left(  2\right)  }=\mathbf{1}\left(  \mathcal{\hat
{X}}_{l,\boldsymbol{Z}_{\left(  1\right)  }}>\varsigma_{n}\right)  \text{ for
}l\in\Psi_{\left(  2\right)  }.
\]
We add all the variables selected in stage 2 into the set of variables
selected in stage 1 as a new vector, and we denote it as $\boldsymbol{Z}%
_{\left(  2\right)  }.$ We denote the index set of the selected variables
($\boldsymbol{Z}_{\left(  2\right)  }$) by $S_{\left(  2\right)  }$ and the
index set of the active variables for stage 3 as $\Psi_{\left(  3\right)  },$
where $\Psi_{\left(  3\right)  }=\left\{  1,2,\ldots,p_{n}\right\}  \backslash
S_{\left(  2\right)  }.$ And so on and so forth. For stage $k,$ we denote the
pre-selected variables as $\boldsymbol{Z}_{\left(  k-1\right)  },$ and the
index set of the active variables as $\Psi_{\left(  k\right)  }.$ Then we
regress $Y$ on $P^{m_{n}}\left(  X_{l}\right)  $ with $P^{m_{n}}\left(
\boldsymbol{Z}_{\left(  k-1\right)  }\right)  $ as pre-selected variables one
by one for $l\in\Psi_{\left(  k\right)  }.$ We construct the test statistic
$\mathcal{\hat{X}}_{l,\boldsymbol{Z}_{\left(  k-1\right)  }}$ as in equation
(\ref{EQ:XlZ}). We select the variable $X_{l}$ if $\widehat{\mathcal{J}%
}_{l,\left(  k\right)  }=1,$ where%
\[
\widehat{\mathcal{J}}_{l,\left(  k\right)  }=\mathbf{1}\left(  \mathcal{\hat
{X}}_{l,\boldsymbol{Z}_{\left(  k-1\right)  }}>\varsigma_{n}\right)  \text{
for }l\in\Psi_{\left(  k\right)  }.
\]
We add all the variables selected in stage $k$ into the set of variables
selected in stage $k-1$ as a new vector, and we denote it as $\boldsymbol{Z}%
_{\left(  k\right)  }.$ We stop the procedure at a stage in which no new
variables are selected. We denote the stage, in which one or more variables
are selected but no new variables are selected after that, as $\hat{k}_{s}$.
So the OCMT procedures stops after stage $\hat{k}_{s}.$ The selection
indicator for variable $X_{l}$ of the OCMT procedure is defined as follows
\begin{equation}
\widehat{\mathcal{J}}_{l}=\sum_{k=1}^{\hat{k}_{s}}\widehat{\mathcal{J}%
}_{l,\left(  k\right)  }. \label{EQ:updatedJ}%
\end{equation}
By construction, $\widehat{\mathcal{J}}_{l}$ is either 1 or 0. It takes value
$1$ if $X_{l}$ is selected in the OCMT\ procedure and 0 otherwise.

The following theorem mainly studies the asymptotic properties of the
multiple-stage procedure in terms of TPR, FPR\ and FDR.

\begin{theorem}
\label{TH:stoppingFDR}Suppose\ that Assumptions \ref*{A:iid}, \ref*{A:p}',
\ref*{A:supp}, \ref*{A:epsilon}, \ref*{A:fl}', \ref*{A:tech}, \ref*{A:mn},
\ref*{A:xi_n}, \ref*{A;no_hidden}', \ref*{A:full_rank2}, and
\ref*{A:noisevariable} hold. Then after some large $n,$

(i) $\Pr\left(  \hat{k}_{s}>p^{\ast}\right)  \leq n^{-M_{6}}+C_{19}\exp\left(
-C_{20}n^{C_{21}}\right)  $ for some fixed large positive number $M_{6}$ and
some positive constants $C_{19},C_{20},$ and $C_{21};$

(ii) $E\left(  \text{\emph{TPR}}_{n}\right)  \geq1-C_{1}\exp\left(
-C_{2}n^{C_{3}}\right)  $ for some positive constants $C_{1},C_{2}$, and
$C_{3};$

(iii) $E\left(  \text{\emph{FPR}}_{n}\right)  \leq p^{\ast\ast}/\left(
p_{n}-p^{\ast}\right)  +C_{4}n^{-M}$ for some positive $C_{4}$ and any fixed
positive large constant $M;$

(iv) \emph{FDR}$_{n}\overset{P}{\rightarrow}0.$
\end{theorem}

Theorem \ref{TH:stoppingFDR}(i) implies the OCMT procedure can terminate at
step $p^{\ast}$ with very high probability. Theorem \ref{TH:stoppingFDR}%
(ii)-(iv) implies that all of TPR$_{n},$ FPR$_{n}$ and FDR$_{n}$ can be well
controlled. Of course, when $n\rightarrow\infty,$ we need to conduct at most
$p^{\ast}$-stage procedure to determine all signals and eliminate all noise
variables, with very high probability.

\subsection{Dealing with Pseudo-signals\label{SEC:Pseudo}}

Because of the nature of the OCMT procedure, pseudo-signals cannot be excluded
from the selection list with high probability. To eliminate pseudo-signals, we
propose to employ the adaptive group Lasso after the OCMT. Then by the
properties of the adaptive group Lasso, the event that the pseudo-signals are
excluded and the signals are kept occurs with probability approaching one
(w.p.a.1). We present this result in Theorem \ref{TH:AGLasso} below. We denote
the set of the variables selected by the OCMT procedure as $\hat
{S}_{\text{OCMT}}\equiv S_{\left(  \hat{k}_{s}\right)  },$ collect them into a
vector $\boldsymbol{Z}_{\text{OCMT}},$ and denote its dimension as $\hat
{p}_{\text{OCMT}}$. Similarly, let $\mathbb{Z}_{\text{OCMT}}$ ($n\times\hat
{p}_{\text{OCMT}}m_{n}$) denote the design matrix for $\boldsymbol{Z}%
_{\text{OCMT}}$ as in equation (\ref{EQ:design}). The post-OCMT adaptive group
Lasso procedure goes as follows:

\begin{enumerate}
\item Obtain the group Lasso estimator by searching $\boldsymbol{\beta}_{n}$
($\hat{p}_{\text{OCMT}}m_{n}\times1$)\ to minimize
\[
L_{n1}\left(  \boldsymbol{\beta}_{n},\lambda_{n1}\right)  =\left\Vert
\boldsymbol{y}-\mathbb{Z}_{\text{OCMT}}\boldsymbol{\beta}_{n}\right\Vert
^{2}+\lambda_{n1}\sum_{j=1}^{\hat{p}_{\text{OCMT}}}\left\Vert
\boldsymbol{\beta}_{nj}\right\Vert ,
\]
where $\lambda_{n1}$ is a positive tuning parameter, $\boldsymbol{\beta}%
_{n}=(\boldsymbol{\beta}_{n1}^{\prime},...,\boldsymbol{\beta}_{n\hat
{p}_{\text{OCMT}}}^{\prime})^{\prime},$ $\boldsymbol{\beta}_{nj}$ is a
$m_{n}\times1$ vector of coefficients of the B-spline basis for the $j$-th
element in $\boldsymbol{Z}_{\text{OCMT}}.$ Denote the above estimator as
$\boldsymbol{\tilde{\beta}}_{n}.$

\item The adaptive group Lasso estimator is obtained by searching
$\boldsymbol{\beta}_{n}$ to minimize%
\[
L_{n2}\left(  \boldsymbol{\beta}_{n},\lambda_{n2}\right)  =\left\Vert
\boldsymbol{y}-\mathbb{Z}_{\text{OCMT}}\boldsymbol{\beta}_{n}\right\Vert
^{2}+\lambda_{n2}\sum_{j=1}^{\hat{p}_{\text{OCMT}}}\frac{1}{\left\Vert
\boldsymbol{\tilde{\beta}}_{nj}\right\Vert }\left\Vert \boldsymbol{\beta}%
_{nj}\right\Vert ,
\]
where we use the convention that $0/0=0$. Denote the above estimator as
$\boldsymbol{\hat{\beta}}_{n}.$
\end{enumerate}

The post-selection estimation proceeds as follows. Denote the selected
regressor from the above procedure as $\boldsymbol{Z}_{\text{AGLASSO}},$ and
similarly denote its B-spline basis and design matrix as $P^{m_{n}}\left(
\boldsymbol{Z}_{\text{AGLASSO}}\right)  $ and $\mathbb{Z}_{\text{AGLASSO}},$
respectively$.$ The post selection estimator is the OLS estimator of
regressing $\boldsymbol{y}$ on $\mathbb{Z}_{\text{AGLASSO}}$, which is%
\[
\boldsymbol{\hat{\beta}}_{\text{post}}=\left(  \mathbb{Z}_{\text{AGLASSO}%
}^{\prime}\mathbb{Z}_{\text{AGLASSO}}\right)  ^{-1}\mathbb{Z}_{\text{AGLASSO}%
}^{\prime}\boldsymbol{y.}%
\]
The final fitted model is%
\[
P^{m_{n}}\left(  \boldsymbol{Z}_{\text{AGLASSO}}\right)  ^{\prime
}\boldsymbol{\hat{\beta}}_{\text{post}}.
\]

\begin{theorem}
\label{TH:AGLasso} Suppose\ that Assumptions \ref*{A:iid}, \ref*{A:p}',
\ref*{A:supp}, \ref*{A:epsilon}, \ref*{A:fl}', \ref*{A:tech}, \ref*{A:mn},
\ref*{A:xi_n}, \ref*{A;no_hidden}', \ref*{A:full_rank2}, and
\ref*{A:noisevariable} hold. Further, $\lambda_{n1}\geq C\sqrt{n\log\left(
p^{\ast\ast}m_{n}\right)  }$ for a sufficient large $C$, $\lambda_{n1}\ll
\sqrt{n/m_{n}},$ and $m_{n}^{1/2}\log\left(  p^{\ast\ast}m_{n}\right)
\ll\lambda_{n2}\ll nm_{n}^{-1/4}.$ Then

(i) All signal variables are kept and all pseudo signals or noise variables
are eliminated w.p.a.1;

(ii) The post OCMT estimation error satisfies
\[
P^{m_{n}}\left(  \boldsymbol{Z}_{\text{AGLASSO}}\right)  ^{\prime
}\boldsymbol{\hat{\beta}}_{\text{post}}-%
{\displaystyle\sum_{j=1}^{p^{\ast}}}
f_{j}^{\ast}\left(  X_{j}\right)  =O_{P}\left(  \left(  m_{n}/n\right)
^{1/2}\right)  .
\]

\end{theorem}

The above theorem is almost the same as that in \cite{HuangEtal2010} including
the requirements on $\lambda_{n1}$ and $\lambda_{n2}$, with the exception that
the procedure starts with the covariates post OCMT. Consequently, we only need
to show that the adaptive group Lasso procedure is still valid in the
post-OCMT situation, and the rest follows immediately from
\cite{HuangEtal2010}. In particular, under Assumption \ref*{A:mn}, the biases
of the post OCMT estimators of the nonparametric additive components are
asymptotically negligible so that their mean square errors (MSEs) are
dominated by their asymptotic variances that are of order $O\left(
m_{n}/n\right)  ,$ which explains the result in Theorem \ref{TH:AGLasso}(ii).

Our procedure enjoys the same theoretical property as the adaptive group
Lasso. After the OCMT, the dimension of candidate variables is reduced
dramatically. Thus the additional computation burden applying the post-OCMT
adaptive group Lasso can be almost ignored. We note that the main advantage of
our procedure is fast and reliable computation, which delivers better small
sample performance as shown from our simulation studies. An implication of the
above theorem is that the post-OCMT adaptive group procedure improves over the
post-selection estimation results in CKP (Theorem 2 in CKP) because
pseudo-signals are now eliminated with very high probability.

\subsection{Diverging $p^{\ast}$\label{SEC:divergeP}}

Allowing a diverging $p^{\ast}$ (number of true signals) is possible. The only
additional technical condition apart from the restriction on the speed of
$p^{\ast}$ is that $\sum_{j=1}^{p^{\ast}}f_{j}^{\ast}\left(  X_{j}\right)  $
is uniformly bounded. Note that this condition naturally holds for a fixed
$p^{\ast}$ due to the boundedness of $f_{j}^{\ast}$. The main reason for the
requirement of this condition is technical: the uniform boundedness of
$\sum_{j=1}^{p^{\ast}}f_{j}^{\ast}\left(  X_{j}\right)  $ ensures that $U_{l}%
$\ defined in equation (\ref{EQ:betanl}) also satisfies the exponential
decayed tail condition,\footnote{For details, see the proof of Lemma
\ref{LE:error_var}.} and the tail property is necessary to apply the main
inequalities to obtain the probability bounds. The uniform boundedness of
$\sum_{j=1}^{p^{\ast}}f_{j}^{\ast}\left(  X_{j}\right)  $ was also imposed in
\cite{FanFengSong}$.$ Propositions \ref{TH:main1} and \ref{TH:main1_ms} are on
individual $X_{l},$\ but we do need Assumption \ref*{A:p}\textquotedblright%
\ on $p^{\ast}$ so that Proposition \ref{TH:main1_ms} holds. It is to ensure
that we can have precise estimation on a diverging design matrix.

\noindent\textbf{Assumption} \textbf{\ref*{A:p}\textquotedblright} $p^{\ast
}\lesssim n^{B_{p^{\ast}}}$ for some $B_{p^{\ast}}\geq0.$\ $p^{\ast\ast
}\lesssim n^{B_{p^{\ast\ast}}}\ $and $p_{n}\propto n^{B_{p}}$ for some
$B_{p}>B_{p^{\ast\ast}},B_{p^{\ast}}\geq0.$ Further, $B_{p^{\ast}}%
+B_{p^{\ast\ast}}<\left(  1-3B_{m}\right)  /2.$

We present the main results in the following theorem.

\begin{theorem}
\label{TH:stoppingFDR2}Suppose\ that Assumptions \ref*{A:iid}, \ref*{A:p}%
\textquotedblright, \ref*{A:supp}, \ref*{A:epsilon}, \ref*{A:fl}',
\ref*{A:tech}, \ref*{A:mn}, \ref*{A:xi_n}, \ref*{A;no_hidden}',
\ref*{A:full_rank2}, and \ref*{A:noisevariable} hold. In addition, assume that
$\sum_{j=1}^{p^{\ast}}f_{j}^{\ast}\left(  X_{j}\right)  $ is uniformly
bounded. Then, the results in Theorem \ref{TH:stoppingFDR} continue to hold.
\end{theorem}

In the next section, we investigate the small sample performance of our
procedure by means of Monte Carlo experiment.

\section{Monte Carlo Simulations\label{SEC:MC}}

To investigate the finite-sample performance of our procedure, we conduct
Monte Carlo experiments in this section.

\subsection{Simulation Design}

Following \cite{HuangEtal2010}, we consider the following data generating
processes (DGPs). In what follows, we assume
\begin{align*}
f_{1}\left(  x\right)   &  =x;\text{\ }f_{2}\left(  x\right)  =(2x-1)^{2}%
;\text{\ }f_{3}\left(  x\right)  =\frac{\sin\left(  2\pi x\right)  }%
{2-\sin\left(  2\pi x\right)  };\text{ and}\\
f_{4}\left(  x\right)   &  =0.1\sin\left(  2\pi x\right)  +0.2\cos\left(  2\pi
x\right)  +0.3\sin\left(  2\pi x\right)  ^{2}+0.4\cos\left(  2\pi x\right)
^{3}+0.5\sin\left(  2\pi x\right)  ^{3}.
\end{align*}
For the errors, we assume $\varepsilon\sim$ i.i.d. $N\left(  0,1\right)  $ for
DGPs 1--6 and 9--10. In DGPs 7--8, we check the impact of heteroskedastic
errors on our methods. Specifically, we add a heteroskedastic error to the
simplest and the most complicated designs in DGPs 1--6 to form DGPs 7 and 8,
respectively. In DGPs 9--10, we consider the case with additive components of
binary variables that mimic the application in Section \ref{SEC:application}.

\textbf{DGP 1: Four independent signals only}. $Y$ is generated as follows:
\begin{equation}
Y=2.55f_{1}\left(  X_{1}\right)  +2.57f_{2}\left(  X_{2}\right)
+1.68f_{3}\left(  X_{3}\right)  +f_{4}\left(  X_{4}\right)  +\varepsilon.
\label{model.yi_simulation_D1}%
\end{equation}
Note the coefficients before $f_{i}$'s\ are set to make each signal have the
same strength in terms of variance\ for independent uniform $X_{1}%
,\ldots,X_{4}$. The covariates are generated as follows$:$%
\[
X_{j}=W_{j}\text{ for }j=1,\ldots,4,\text{ and }X_{j}=\frac{W_{j}+U_{1}}%
{2}\text{for }j\geq5,
\]
where $W_{j},$ $j=1,\ldots,p_{n}$, and $U_{1}$ are all independent draws from
$U(0,1)$. Thus, $p^{\ast}=4$ and $p^{\ast\ast}=0$ for DGP 1.\ Define the
Signal-to-noise ratio to be $r_{sn}=\frac{\mathtt{sd}\left(  f\right)
}{\mathtt{sd}\left(  \varepsilon\right)  }$, and $r_{sn}=1.5$ for DGP 1.

\textbf{DGP 2: Four independent signals and two pseudo-signals}. $Y$ is
generated from equation (\ref{model.yi_simulation_D1}). The covariates are
generated as follows$:$
\begin{align*}
X_{j}  &  =W_{j}\text{ for }j=1,\ldots,4,\text{ }X_{5}=\frac{4X_{1}+U_{1}}%
{5}\text{, }X_{6}=\frac{4X_{2}+U_{2}}{5},\text{ and }\\
X_{j}  &  =\frac{W_{j-2}+U_{3}}{2}\text{ for\ }j\geq7,
\end{align*}
where $W_{j},$ $j=1,\ldots,p_{n}-2$, $U_{1},$ and $U_{2}$\ are all independent
draws from $U(0,1)$. Thus, $p^{\ast}=4$ and $p^{\ast\ast}=2$ for DGP 2.

\textbf{DGP 3: Four signals, and one hidden signal}. $Y$ is generated from
\begin{equation}
Y=2.55f_{1}\left(  X_{1}\right)  +2.57f_{2}\left(  X_{2}\right)
+1.68f_{3}\left(  X_{3}\right)  +f_{4}\left(  X_{4}\right)  +f_{5}\left(
X_{5}\right)  +\varepsilon, \label{EQ:hidden_model}%
\end{equation}
where $f_{5}\left(  X_{5}\right)  =-\mathbb{E}\left[  2.55f_{1}\left(
X_{1}\right)  +2.57f_{2}\left(  X_{2}\right)  +1.68f_{3}\left(  X_{3}\right)
+f_{4}\left(  X_{4}\right)  |X_{5}\right]  .$ The covariates are generated as
follows $:$%
\begin{align*}
X_{j}  &  =W_{j}\text{ for }j=1,2,\text{ }X_{j}=\frac{W_{j}+U_{1}}{2}\text{
for }j=3,4,\text{ }X_{5}=U_{1}\text{, and}\\
X_{j}  &  =\frac{W_{j-1}+U_{2}}{2}\text{ for }j\geq6,
\end{align*}
where $W_{j},$ $j=1,\ldots,p_{n}-1,$ $U_{1},$ and $U_{2}$\ are independent
draws from $U(0,1)$. Then, $p^{\ast}=5$ and $p^{\ast\ast}=0$ for DGP 3, and
the fifth signal is hidden by our definition.\footnote{By the distribution of
covariates,
\begin{align*}
f_{5}\left(  x\right)   &  \approx0.97\pi-1.2\cos(\pi x)+\sin\left(  \pi
x\right)  +0.6861\pi\arctan[2\left(  \tan(\pi x/2)-1\right)  /\sqrt{3}]\\
&  -0.6861\pi\arctan[2\left(  \tan(\pi x/2+\pi/2)-1\right)  /\sqrt
{3}]+0.2778\cos^{3}(\pi x)-0.2222\sin^{3}(\pi x).
\end{align*}
}

\textbf{DGP 4: Four signals, two pseudo-signals, and one hidden signal}. $Y$
is generated from equation (\ref{EQ:hidden_model}). The covariates are
generated as follows $:$%
\begin{align*}
X_{j}  &  =W_{j}\ \text{for }j=1,2,\text{ }X_{j}=\frac{W_{j}+U_{1}}%
{2}\ \text{for }j=3,4,\text{ }X_{5}=U_{1},\text{\ }X_{6}=\frac{4X_{1}+U_{2}%
}{5}\text{,}\\
\text{ }X_{7}  &  =\frac{4X_{2}+U_{3}}{5},\text{ and }X_{j}=\frac
{W_{j-3}+U_{3}}{2}\ \text{for }j\geq8,
\end{align*}
where $W_{j},$ $j=1,\ldots,p_{n}-3,$ $U_{1},$ $U_{2},$ and $U_{3}$\ are
independent draws from $U(0,1)$. Then, $p^{\ast}=5$ and $p^{\ast\ast}=2$ for
DGP 4, and the fifth signal is a hidden signal.

\textbf{DGP 5: Four correlated signals.} $Y$ is generated from equation
(\ref{model.yi_simulation_D1}). The covariates are generated as follows$:$%
\begin{equation}
X_{j}=\frac{W_{j}+U_{1}}{2}\text{ for\ }j=1,\ldots,4,\text{ and }X_{j}%
=\frac{W_{j}+U_{2}}{2}\ \text{for\ }j\geq5,\nonumber
\end{equation}
where $W_{j},$ $j=1,\ldots,p_{n}$, $U_{1},$ and $U_{2}$\ are independent draws
from $U(0,1)$. Thus, four signals are correlated with each other, and
$p^{\ast}=4$ and $p^{\ast\ast}=0$ for DGP 5.

\textbf{DGP 6: Four signals, many pseudo-signals, and one hidden signal}. $Y$
is generated from equation (\ref{EQ:hidden_model}). The covariates are
generated as follows $:$%
\begin{align*}
X_{j}  &  =W_{j}\text{ for }j=1,2,\text{ }X_{j}=\frac{W_{j}+U_{1}}%
{2}\ \text{for }j=3,4,\text{ }X_{5}=U_{1},\\
X_{j}  &  =\frac{4X_{1}+\left(  j-5\right)  W_{j-1}}{j-1}\text{ for
}j=6,10,14,18,\ldots,\\
X_{j}  &  =\frac{4X_{2}+\left(  j-5\right)  W_{j-1}}{j-1}\text{ for
}j=7,11,15,19,\ldots,\\
X_{j}  &  =\frac{4X_{3}+\left(  j-5\right)  W_{j-1}}{j-1}\text{ for
}j=8,12,16,20,\ldots,\text{ and}\\
X_{j}  &  =\frac{4X_{4}+\left(  j-5\right)  W_{j-1}}{j-1}\text{ for
}j=9,13,17,21,\ldots
\end{align*}
where $W_{j},$ $j=1,\ldots,p_{n}-1,$ and $U_{1}$\ are independent draws from
$U(0,1)$. Then, $p^{\ast}=5$ with one hidden signal for DGP 6.

\textbf{DGP 7: Four independent signals with heteroskedastic errors}. $Y$ is
generated from equation (\ref{model.yi_simulation_D1}) with the same
covariates as in DGP 1. We assume that conditioning on $X,$ $\varepsilon$ is
normal with mean 0 and variance $0.436[1+\left(  X_{1}+X_{2}+X_{3}%
+X_{4}\right)  /4]^{2},$ and the unconditional variance of $\varepsilon$ is
approximately 1.

\textbf{DGP 8: Four signals, many pseudo-signals, and one hidden signal with
heteroskedastic errors}. $Y$ is generated from equation (\ref{EQ:hidden_model}%
) with the same covariates as in DGP 6. We assume that conditioning on $X,$
$\varepsilon$ is normal with mean 0 and variance $0.436[1+\left(  X_{1}%
+X_{2}+X_{3}+X_{4}\right)  /4]^{2}$.

\textbf{DGP 9: Four independent signals with some binary variables. }To mimic
the application, we consider the situation with some binary covariates. Note
that any function of a binary covariate can at most take two values. Without
loss of generality, we focus on the case in which those binary covariates
enter the model linearly. It is easy to see that our theoretical results
continue to hold in the presence of some linear additive components with the
main difference that they do not exhibit any approximation bias.\footnote{In
this case, the test statistics for the additive linear terms are the squares
of t-statistics, and the inequality continues to hold by setting, for example,
$\varsigma_{n}\propto\left[  \log p_{n}\right]  ^{1.1}$ for the case when
$p_{n}$ is a polynomial of $n.$} We set the threshold as $\varsigma
_{n}=C\left[  \log p_{n}\right]  ^{1.1}$. $Y$ is generated from
\[
Y=2.57f_{2}\left(  X_{1}\right)  +1.68f_{3}\left(  X_{2}\right)
+1.47X_{3}+1.47X_{4}+\varepsilon,
\]
where $X_{j}=W_{j},$ $j=1,2,$ $X_{j}=V_{j-2},$ $j=3,4,$ $W_{1}$ and $W_{2}$
are independent $U(0,1),$ and $V_{1}$ and $V_{2}$ are independent Bernoulli
random variables with equal chances of taking value 0 or 1. The remaining
covariates are generated as follows:%
\[
X_{j}=\frac{W_{j-2}+U_{1}}{2}\text{ for }j=5,6,...,\frac{p_{n}}{2},\text{ and
}X_{j}=V_{j-\frac{p_{n}}{2}+2}\text{ for }j=\frac{p_{n}}{2}+1,...,p_{n},
\]
where $W_{j},$ $j=5,6,...,\frac{p_{n}}{2}-2,$ and $U_{1},$ are independent
$U(0,1),$ $V_{j},$ $j=3,...,\frac{p_{n}}{2}+2$ are distributed the same as
$V_{1}.$ All $W$s$,V$s, and $U$ are independent of each other.

\textbf{DGP 10: Four signals with one hidden signal in the presence of some
binary variables. }For this DGP, $Y$ is generated from%
\[
Y=2.57f_{2}\left(  X_{1}\right)  +1.5X_{2}+1.5X_{3}-X_{4}+\varepsilon,
\]
where $X_{1}=W_{1}$ is a $U(0,1),$ $X_{j}=V_{j-1},$ $j=2,3,4,$ are Bernoulli
random variables with equal chances of taking 0 or 1, and Corr$\left(
V_{1},V_{3}\right)  =$Corr$\left(  V_{1},V_{2}\right)  =$Corr$\left(
V_{2},V_{3}\right)  =1/3$. $X_{j},$ $j=5,6,...,p_{n}$ are the same as those in
DGP 9.\footnote{An example of $V_{1},V_{2},$ and $V_{3}$ is that
$V_{j}=\mathbf{1(}\tilde{U}_{j}+\tilde{U}_{4}>1)$ for $j=1,2,$ and $3,$ where
$\tilde{U}_{1},...,\tilde{U}_{4}$ are independent $U\left(  0,1\right)  .$}
Some simple calculation implies that $X_{4}$ is a hidden signal.

\subsection{Tuning Parameters\label{SEC:tuning}}

The key tuning parameter in this study is the threshold $\varsigma_{n}.$\ The
CKP's Bonferroni correction strategy does not perform consistently well for
our case. We do the following instead. We set
\[
\varsigma_{n}=Cm_{n}\left[  \left(  \log p_{n}\right)  ^{1.1}+\left(  \log
m_{n}\right)  ^{1.1}\right]  \text{ and }\varsigma_{n}=C\left(  \log
p_{n}\right)  ^{1.1}%
\]
for continuous variables and binary variables, respectively. This satisfies
Assumption \ref{A:xi_n} when, in addition, Assumptions \ref{A:p} and
\ref{A:mn} hold (both $m_{n}$ and $p_{n}$ grow at the polynomial rate of $n$).
In the small samples, we set this $\varsigma_{n}$ to have control over both
$m_{n}$ and $p_{n}$. CKP suggest using a larger threshold, $\varsigma
_{n}^{\ast}$, for subsequent stages, to improve the finite sample performance.
We follow their lead and set a $\varsigma_{n}^{\ast}$ larger than
$\varsigma_{n}$\ for subsequent stages. That is, we replace $\varsigma_{n}$
with $\varsigma_{n}^{\ast}$ for $\widehat{\mathcal{J}}_{l,\left(  2\right)
},\widehat{\mathcal{J}}_{l,\left(  3\right)  },...,\widehat{\mathcal{J}%
}_{l,\left(  k\right)  }$ in Section \ref{SEC:multi_procedure}. The reason is
that the chance of including a noise variable increases quickly for subsequent
stages, and we need a larger $\varsigma_{n}^{\ast}$ for the OCMT to conclude
more easily. We set $\varsigma_{n}^{\ast}=3\varsigma_{n}$ for continuous
variables. Note that CKP take the threshold for later stages to be twice the
threshold for the first stage, and their test is based on $t$-statistics.
Since ours is based on chi-squared statistics, equivalently, we should set
$\varsigma_{n}^{\ast}$ to be $4\varsigma_{n}.$ However, some small-scale
experiments suggest that setting $\varsigma_{n}^{\ast}=3\varsigma_{n}$ can
yield better small sample performance, probably because our model is
nonparametric and our test statistic varies more than the parametric
counterpart. Note that this change does not affect our theoretical results
because $\varsigma_{n}^{\ast}$ is proportional to $\varsigma_{n}$. For binary
additive components, we continue to follow CKP and set $\varsigma_{n}^{\ast
}=4\varsigma_{n}$ because they enter the model linearly.

Another important tuning parameter is $m_{n},$ which is critical for the sieve
estimation. The optimal choice of $m_{n}$ has been studied extensively in the
literature. Popular ways to choose the value of $m_{n}$ include cross
validation, Akaike information criterion (AIC), and BIC. We refer the readers
to \cite{Chen_handbook} and \cite{Hansen2014} for a review on this important
issue. For $m_{n},$ we simply set $m_{n}=\left\lfloor n^{1/4}\right\rfloor
+1,$ where $\left\lfloor \cdot\right\rfloor $ is the floor operator$.$ This
$m_{n}$ satisfies Assumption \ref{A:mn}, if $d>3/2$. Other choices of sieve
terms such as $m_{n}=\left\lfloor n^{1/4}\right\rfloor +2$\ are also
considered in the simulations, and they yield similar results.

For the $C$ in $\varsigma_{n}=Cm_{n}\left[  \left(  \log p_{n}\right)
^{1.1}+\left(  \log m_{n}\right)  ^{1.1}\right]  $ or $\varsigma_{n}=C\left(
\log p_{n}\right)  ^{1.1},$\ we test $C$ in the range of $0.5$ to $2.5$,
specifically, $0.5,0.6,...,2.5$. We determine the value of $C$ by minimizing
the following BIC:
\begin{equation}
\text{BIC}\left(  C\right)  =n\log\left[  \text{RSS}\left(  C\right)
/n\right]  +\left(  \text{number of selected variables}\right)  \cdot
\log\left(  n\right)  , \label{EQ:BIC}%
\end{equation}
where RSS$\left(  C\right)  $ denotes the residual sum of squares from the
post-OCMT ordinary least squares regression of $Y$ on $P^{m_{n}}\left[
\boldsymbol{Z}_{\left(  \hat{k}_{s}\right)  }\right]  $, where $\boldsymbol{Z}%
_{\left(  \hat{k}_{s}\right)  }$ are the variables selected by the OCMT with
$C$ in use. The tuning parameters for the adaptive group Lasso\ are selected
as in Section \ref{SEC:tuningForLasso}. Another popular way to choose $C$ is
cross validation. As seen from the results, the BIC works well for our
procedure. In light of this, we will not pursue the procedure with cross validation.

For an easy reference, we present the implementation details in Appendix
\ref{SEC:procedure}.

\subsection{Estimators Compared\label{SEC:tuningForLasso}}

For comparison, we consider the adaptive group Lasso by \cite{HuangEtal2010},
which is designed for component selection in the nonparametric additive model.
It is a two-step approach---the first step is the usual group Lasso and the
second is the adaptive group Lasso with initial estimates from the first step.
We select tuning parameters $\lambda_{n1}$ in step 1 and $\lambda_{n2}$ in
step 2 ($\lambda_{n1}$ and $\lambda_{n2}$ are in the notations of
\cite{HuangEtal2010}) by BIC as well. Specifically, we set $\lambda
_{n1}=\lambda_{j}$ with%
\[
\lambda_{j}=\exp\left\{  \log\left(  \lambda_{\max}\right)  +\left[
\log\left(  \lambda_{\min}\right)  -\log\left(  \lambda_{\max}\right)
\right]  \frac{j}{30}\right\}
\]
for $j=0,1,...,30,$ where $\lambda_{\min}=\max\left\{  0.05,10^{-5}\left\Vert
\boldsymbol{y}\right\Vert \right\}  $ and $\lambda_{\max}=0.5\left\Vert
\boldsymbol{y}\right\Vert .$ We calculate BIC$_{j}$ based on the estimation
using $\lambda_{n1}=\lambda_{j},$ and we select the $\hat{\lambda}_{n1}$ that
minimizes the BIC$_{j}$ among $j=0,1,...,30.$ We set the estimates in the
first step as the estimates using $\hat{\lambda}_{n1}.$ For the second step,
we set $\lambda_{n2}=\lambda_{j}$ for $j=0,1,...,30$ with the initial
estimates as the estimates from the first step using $\hat{\lambda}_{n1}$. We
again calculate BIC$_{j}$ based on the estimation using $\lambda_{n2}%
=\lambda_{j},$ and we select the $\hat{\lambda}_{n2}$ that minimizes the
BIC$_{j}$ among $j=0,1,...,30.$ The final estimates are the adaptive group
Lasso estimates using $\hat{\lambda}_{n2}$. Note that the variables not
selected in the first step are not included in the second step, because the
penalty for those variables is infinity in the second step. We find the
solutions in both steps through the block coordinate descent algorithm (i.e.,
the \textquotedblleft shooting\textquotedblright\ algorithm).\ For the
algorithm details, see \cite{WuLange2008}.\footnote{One implementation in
MATLAB can be found at:
https://publish.illinois.edu/xiaohuichen/code/group-lasso-shooting/.}

In addition, we compare our method with the random forest regression and
bagging, both of which are commonly-used machine learning methods. Since there
is no variable selection criterion in random forest, we focus on the
comparison of out-sample forecasting performance.

\subsection{Estimation Results}

We consider the combinations of $n=200$ or $400$ and $p_{n}=100$, $200$, or
$1,000$ for each DGP. All results are based on 1,000 replications. We report
the results for five different methods. The first and second methods are our
post-OCMT procedure (Steps 1--6 in Appendix \ref{SEC:procedure}, denoted as
\textquotedblleft POST--OCMT\textquotedblright) and OCMT procedure (Steps 1--
5 in Appendix \ref{SEC:procedure}, denoted as \textquotedblleft
OCMT\textquotedblright), respectively. The third method is the adaptive group
Lasso (denoted as \textquotedblleft AGLASSO\textquotedblright), and the last
two methods are bagging (denoted as \textquotedblleft
BAGGING\textquotedblright) and random forest (denoted as \textquotedblleft
RF\textquotedblright), respectively.

Following the literature, we report the mean number of variables selected
(NV), the true positive rates (TPR), the false positive rates (FPR), the false
discovery rates (FDR), the percentage of correct selection (CS) for the first
three methods\footnote{That is, we precisely uncover the true model in
(\ref{EQ:model1b}) using all the true signals, but not any pseudo-signals or
noise variables.}, and the out-of-sample root mean squared forecast errors
(RMSFE) for all the five methods.\footnote{The forecasts are based on the
post-selection estimates from each method. In each replication, we
additionally independently generate 200 observations. Then, we calculate the
root mean square errors (RMSE) of the difference between the forecasts and $Y$
for the new observations. RMSFE is the average of those RMSE for 1,000
replications.} We report the average number of stages (STEP) for our OCMT
procedure only.

To save space, we report the results in the main body of this paper for only
DGPs 1 and 6 in Tables \ref{tabledgp1} and \ref{tabledgp6}, respectively.
These two designs correspond to the simplest and the most complicated designs
in the simulation, respectively. Results for DGPs 2--5 and 7--10 are provided
in Tables \ref{tabledgp2} to \ref{tabledgp10} in Appendix \ref{APP:tables}. To
showcase the advantage of the post-OCMT procedure compared to the OCMT
procedure only, we also consider a linear DGP (DGP 11) in the online
Appendix\textbf{ }\ref{APP:tables}\textbf{ }with results reported in Table
\ref{tabledgp11}.

\begin{table}[ptb]
\caption{DGP 1}%
\label{tabledgp1}%
\centering{}\centering{ } \resizebox{!}{0.67\textwidth}{
\begin{tabular}
[c]{l|ccccccc}\hline\hline
\multicolumn{8}{c}{Panel 1: $n=200$, $p=100$}\\\hline
& NV & TPR & FPR & FDR & CS & STEP & RMSFE\\\hline
POST-OCMT & 4.0080 & 0.9998 & 0.0001 & 0.0015 & 0.9910 & - & 1.0718\\
OCMT & 4.0180 & 0.9998 & 0.0002 & 0.0031 & 0.9820 & 1.0260 & 1.0886\\
AGLASSO & 4.0410 & 0.9990 & 0.0005 & 0.0075 & 0.9550 & - & 1.0753\\
BAGGING   & -      & -      & -      & -      & -      & -      & 1.4403   \\
RF        & -      & -      & -      & -      & -      & -      & 1.4846   \\\hline
\multicolumn{8}{c}{Panel 2: $n=200$, $p=200$}\\\hline
& NV & TPR & FPR & FDR & CS & STEP & RMSFE\\\hline
POST-OCMT & 4.0150 & 0.9998 & 0.0001 & 0.0027 & 0.9830 & - & 1.0740\\
OCMT & 4.0200 & 0.9998 & 0.0001 & 0.0035 & 0.9790 & 1.0390 & 1.0939\\
AGLASSO & 4.0510 & 0.9950 & 0.0004 & 0.0114 & 0.9330 & - & 1.0808\\
BAGGING   & -      & -      & -      & -      & -      & -      & 1.4690   \\
RF        & -      & -      & -      & -      & -      & -      & 1.5160   \\\hline
\multicolumn{8}{c}{Panel 3: $n=200$, $p=1000$}\\\hline
& NV & TPR & FPR & FDR & CS & STEP & RMSFE\\\hline
POST-OCMT & 4.0010 & 0.9952 & 0.0000 & 0.0034 & 0.9640 & - & 1.0773\\
OCMT & 4.0200 & 0.9952 & 0.0000 & 0.0062 & 0.9510 & 1.0880 & 1.1173\\
AGLASSO & 3.9370 & 0.9620 & 0.0001 & 0.0145 & 0.8770 & - & 1.0993\\
BAGGING   & -      & -      & -      & -      & -      & -      & 1.5281   \\
RF        & -      & -      & -      & -      & -      & -      & 1.5847   \\\hline
\multicolumn{8}{c}{Panel 4: $n=400$, $p=100$}\\\hline
& NV & TPR & FPR & FDR & CS & STEP & RMSFE\\\hline
POST-OCMT & 4.0010 & 1.0000 & 0.0000 & 0.0002 & 0.9990 & - & 1.0470\\
OCMT & 4.0010 & 1.0000 & 0.0000 & 0.0002 & 0.9990 & 1.0000 & 1.0470\\
AGLASSO & 4.0080 & 0.9990 & 0.0001 & 0.0020 & 0.9870 & - & 1.0531\\
BAGGING   & -      & -      & -      & -      & -      & -      & 1.3234   \\
RF        & -      & -      & -      & -      & -      & -      & 1.3648   \\\hline
\multicolumn{8}{c}{Panel 5: $n=400$, $p=200$}\\\hline
& NV & TPR & FPR & FDR & CS & STEP & RMSFE\\\hline
POST-OCMT & 4.0010 & 1.0000 & 0.0000 & 0.0002 & 0.9990 & - & 1.0477\\
OCMT & 4.0010 & 1.0000 & 0.0000 & 0.0002 & 0.9990 & 1.0000 & 1.0477\\
AGLASSO & 4.0210 & 1.0000 & 0.0001 & 0.0035 & 0.9790 & - & 1.0520\\
BAGGING   & -      & -      & -      & -      & -      & -      & 1.3536   \\
RF        & -      & -      & -      & -      & -      & -      & 1.3984   \\\hline
\multicolumn{8}{c}{Panel 6: $n=400$, $p=1000$}\\\hline
& NV & TPR & FPR & FDR & CS & STEP & RMSFE\\\hline
POST-OCMT & 4.0040 & 1.0000 & 0.0000 & 0.0007 & 0.9960 & - & 1.0519\\
OCMT & 4.0040 & 1.0000 & 0.0000 & 0.0007 & 0.9960 & 1.0000 & 1.0519\\
AGALSSO & 4.0313 & 1.0000 & 0.0000 & 0.0052 & 0.9688 & - &
1.0524\\
BAGGING   & -      & -      & -      & -      & -      & -      & 1.4120   \\
RF        & -      & -      & -      & -      & -      & -      & 1.4574   \\\hline\hline
\end{tabular}}\end{table}{ }

\begin{table}[ptb]
\caption{DGP 6}%
\label{tabledgp6}%
\centering{}\centering{ } \resizebox{!}{0.67\textwidth}{
\begin{tabular}
[c]{l|ccccccc}\hline\hline
\multicolumn{8}{c}{Panel 1: $n=200$, $p=100$}\\\hline
& NV & TPR & FPR & FDR & CS & STEP & RMSFE\\\hline
POST-OCMT & 3.6300 & 0.6628 & 0.0033 & 0.0650 & 0.1770 & - & 1.2296\\
OCMT & 4.8450 & 0.6834 & 0.0149 & 0.2163 & 0.0000 & 1.2650 & 1.4003\\
AGLASSO & 2.9840 & 0.5522 & 0.0023 & 0.0432 & 0.0040 & - & 1.3670\\
BAGGING   & -      & -      & -      & -      & -      & -      & 1.4522   \\
RF        & -      & -      & -      & -      & -      & -      & 1.4645   \\\hline
\multicolumn{8}{c}{Panel 2: $n=200$, $p=200$}\\\hline
& NV & TPR & FPR & FDR & CS & STEP & RMSFE\\\hline
POST-OCMT & 3.5600 & 0.6438 & 0.0017 & 0.0682 & 0.1290 & - & 1.2396\\
OCMT & 4.6700 & 0.6642 & 0.0069 & 0.2107 & 0.0000 & 1.2220 & 1.3861\\
AGLASSO & 2.6340 & 0.4866 & 0.0010 & 0.0417 & 0.0000 & - & 1.3695\\
BAGGING   & -      & -      & -      & -      & -      & -      & 1.4855   \\
RF        & -      & -      & -      & -      & -      & -      & 1.4969   \\\hline
\multicolumn{8}{c}{Panel 3: $n=200$, $p=1000$}\\\hline
& NV & TPR & FPR & FDR & CS & STEP & RMSFE\\\hline
POST-OCMT & 3.2550 & 0.5892 & 0.0003 & 0.0669 & 0.0440 & - & 1.2674\\
OCMT & 4.2310 & 0.6088 & 0.0012 & 0.2040 & 0.0000 & 1.0880 & 1.3311\\
AGLASSO & 1.9780 & 0.3678 & 0.0001 & 0.0306 & 0.0000 & - & 1.4304\\
BAGGING   & -      & -      & -      & -      & -      & -      & 1.5395   \\
RF        & -      & -      & -      & -      & -      & -      & 1.5539   \\\hline
\multicolumn{8}{c}{Panel 4: $n=400$, $p=100$}\\\hline
& NV & TPR & FPR & FDR & CS & STEP & RMSFE\\\hline
POST-OCMT & 4.8840 & 0.9448 & 0.0017 & 0.0246 & 0.6960 & - & 1.0659\\
OCMT & 7.4660 & 0.9534 & 0.0281 & 0.3056 & 0.0000 & 1.8380 & 1.4894\\
AGLASSO & 4.3480 & 0.8276 & 0.0022 & 0.0331 & 0.1060 & - & 1.1294\\
BAGGING   & -      & -      & -      & -      & -      & -      & 1.3502   \\
RF        & -      & -      & -      & -      & -      & -      & 1.3767   \\\hline
\multicolumn{8}{c}{Panel 5: $n=400$, $p=200$}\\\hline
& NV & TPR & FPR & FDR & CS & STEP & RMSFE\\\hline
POST-OCMT & 4.8870 & 0.9488 & 0.0007 & 0.0213 & 0.7080 & - & 1.0723\\
OCMT & 7.4120 & 0.9542 & 0.0135 & 0.3019 & 0.0000 & 1.8360 & 1.5018\\
AGLASSO & 4.2240 & 0.8122 & 0.0008 & 0.0270 & 0.0630 & - & 1.1410\\
BAGGING   & -      & -      & -      & -      & -      & -      & 1.3848   \\
RF        & -      & -      & -      & -      & -      & -      & 1.4118   \\\hline
\multicolumn{8}{c}{Panel 6: $n=400$, $p=1000$}\\\hline
& NV & TPR & FPR & FDR & CS & STEP & RMSFE\\\hline
POST-OCMT & 4.8380 & 0.9356 & 0.0002 & 0.0242 & 0.6740 & - & 1.0784\\
OCMT & 7.4140 & 0.9398 & 0.0027 & 0.3082 & 0.0000 & 1.8030 & 1.4863\\
AGLASSO & 4.1340 & 0.7916 & 0.0002 & 0.0299 & 0.0110 & - &
1.1509\\
BAGGING   & -      & -      & -      & -      & -      & -      & 1.4477   \\
RF        & -      & -      & -      & -      & -      & -      & 1.4689   \\\hline\hline
\end{tabular}}\end{table}

We can make several observations. First, POST--OCMT performs the best among
the three methods in most cases. POST--OCMT outperforms OCMT due to its
ability to eliminate pseudo-signals or noise variables.\ POST--OCMT even
outperforms OCMT slightly for DGPs without any pseudo-signals, especially for
$n=200$. This result confirms the necessity of conducting the additional
post-OCMT step. Moreover, the additional computation time for POST--OCMT
compared with OCMT can be almost ignored because the number of candidate
variables after OCMT is very small. Second, AGLASSO performs almost the same
as our procedures for DGP 2. Note that DGP 2 contains four independent
signals, two pseudo-signals, and no hidden signals. Then, this result is not
surprising, because Lasso performs very well for independent signals, and the
biggest challenge for our procedure is the possible presence of
pseudo-signals. When the signals are correlated in DGP 4, AGLASSO performs
less well and is outperformed by POST--OCMT. Third, the performance of all
methods improves as $n$ increases from $200$ to $400$. Fourth, OCMT is very
successful at picking up hidden signals, especially for $n=400$; see, for
example, the results for DGPs 3, 4, 6, and 8. Fifth, our methods perform well
in the presence of heteroskedastic errors for DGP 7 and 8. Sixth, the CS of
POST--OCMT performs well even for the complicated DGPs 6 and 8 at $n=400$,
whereas AGLASSO performs poorly in terms of CS for these two DGPs. Seventh,
the number of stages for OCMT basically confirms our theoretical findings. For
example, for DGP 1, the mean number of stages is slightly more than 1 for
$n=200$, and is 1 for $n=400$. In the presence of hidden signals, the mean
number of stages is approximately 2 for $n=400$ for DGPs\ 3, 4, 6, 8, and 10.
Note the mean number of stages is approximately 2 for $n=400$ for DGP\ 5 with
correlated signals and no hidden signals. The reason is that the correlation
makes the net effect of $X_{1}$ on $Y$ very small and $X_{1}$ almost behaves
like a \textquotedblleft hidden signal\textquotedblright. This result also
confirms the necessity of using multiple stages instead of a single-stage
procedure. Eighth, our method also works well for DGPs 9 and 10 that mimic the
application. An additional remark is that our procedure can be implemented
fast and is much faster than AGLASSO. For example, when $p=1,000,$ our
procedure took less than half a minute for one replication on average, whereas
the AGLASSO took hours for one replication. Finally, the first three
procedures have smaller RMSFE than Bagging and RF in almost all scenarios and
thus have better out-of-sample forecasting performance.\footnote{Note that
Bagging has slightly better out-of-sample forecasting performance than RF.
This is reasonable given that our DGPs have a finite number of signal
variables. In RF, many decision splits do not improve predictive accuracy
because they rely solely on noise variables to generate the trees.}\textbf{
}This is because they explicitly utilize the information (in the form of an
additive function) underlying the DGPs while BAGGING and RF do not. In
addition, POST--OCMT delivers the best out-of-sample forecasting performance
among all procedures.

To summarize, our methods perform well in small samples, and we view it as a
useful alternative to existing methods in the literature.

\section{An Application\label{SEC:application}}

In this section, we apply our method to a dataset extracted from
RUMiC.\footnote{RUMiC consists of three parts: the Urban Household Survey, the
Rural Household Survey, and the Migrant Household Survey. A group of
researchers at the Australian National University, the University of
Queensland, and the Beijing Normal University initiated this survey. The
Institute for the Study of Labor (IZA) supported it and provides the
Scientific Use Files. RUMiC had financial support from the Australian Research
Council, the Australian Agency for International Development, the Ford
Foundation, IZA, and the Chinese Foundation of Social Sciences. More
information on the survey can be found at
https://datasets.iza.org/dataset/58/longitudinal-survey-on-rural-urban-migration-in-china.}
The survey studied immigrants or workers moving from the rural areas of China
to its big cities. The survey asked interviewees (immigrants or workers) a
wide range of questions. For the detailed design of the survey and other
information, including on the construction of each variable, see the survey
website and \cite{Gongetal2008}. Currently, the 2008 wave data are publicly available.

Economic reforms since the late 1970s have brought significant changes to
China's economy. The government began relaxing its policy on population
mobility in the early 1980s. Gradually, peasants were allowed to leave
villages and work in big cities to earn higher incomes. Most migrant workers
may leave their spouses, children, or parents behind in their hometowns\, who
may need their financial support. This situation results in monetary
transfers, that is, remittances, from migrant workers to their family. In the
context of migration, family, and economic development, remittances are not
only an income source for recipients, but also reflect intrafamilial
relationships. Remittances clearly represent a dimension of family ties and
demonstrate high degrees of interaction between migrants and families at home.
In addition, remittances from the rural migrant workers also contribute
significantly to China's agricultural productivity (c.f.,
\cite{Rozelle_at_al.1999}). For these reasons, it has long been of interest to
model remittances to families or relatives in the hometown; see \cite{Li2001}
and \cite{Cai2003}, among others.

In this application, we take remittance as the dependent variable ($Y$); we
focus on the dataset from Guangdong Province (Guangzhou, Dongguan, and
Shenzhen cities) in the 2008 survey wave, and keep 78 covariates\ from the
dataset.\footnote{We keep covariates with relatively fewer missing
observations.} After dropping observations with missing information, the
number of observations is 456. We provide the definitions of the dependent
variable and covariates, and the associated summary statistics, in Tables
\ref{table:app_def} and \ref{table:summary}, respectively. We report the
original labels of all covariates in the survey in the first column of Table
\ref{table:app_def} for reference. Among the 78 covariates, there are some
continuous variables with most observations as 0 (Panel C in Table
\ref{table:app_def}), and some discrete variables with very limited support
(Panel D in Table \ref{table:app_def}). For those variables, the design matrix
of the sieves generated are either singular or close to singular. For this
reason, we add those variables linearly into the model and treat them the same
as\ dummy variables (Panel E in Table \ref{table:app_def}) for modeling.
Consequently, the way we fit the dataset resembles the approach we used for
DGPs 9 and 10 in the simulation. We explain the reason our theoretical results
continue to hold in this situation in the simulation section (DGP 9). We take
natural logarithms for the $Y$ and continuous $X$ variables to offset the
effect of outliers; otherwise, the forecast can easily take some\ extreme values.

We randomly select 400 observations as the training sample, and the remaining
56 observations as the test sample. The number of sieve terms is set as
$m_{n}=\left\lfloor 400^{1/4}\right\rfloor +1=5$ for the continuous variables
in Panel B of Table \ref{table:app_def}). We set $\varsigma_{n}=Cm_{n}\left[
\left(  \log p_{n}\right)  ^{1.1}+\left(  \log m_{n}\right)  ^{1.1}\right]  $
and $\varsigma_{n}^{\ast}=3\varsigma_{n}$ for continuous variables, and
$\varsigma_{n}=C\left(  \log p_{n}\right)  ^{1.1}$ and $\varsigma_{n}^{\ast
}=4\varsigma_{n}$ for terms entering the model linearly (see Section
\ref{SEC:tuning} for the reason). We set $C$ in the range of $0.5$ to $2.5$,
specifically, $0.5,0.6,...,2.5$, and choose $C$ to minimize the BIC for the
model selection$.$ The competing methods are the group Lasso (labelled as
GLASSO) and the adaptive group Lasso (AGLASSO). The tuning parameters for
GLASSO and AGLASSO\ are selected as in Section \ref{SEC:tuningForLasso}. We
evaluate the performance of all methods based on the RMSFE of the test dataset
using the fitted models from different methods. We independently repeat the
above procedure 100 times.

\begin{table}[tbh]
\caption{Performance in Terms of RMSFE (100 Cases), benchmark: AGLASSO}%
\label{table:ForecaseRMSE}
\centering{}\centering { }
\begin{tabular}
[c]{ccccccccccc}\hline\hline
\multicolumn{3}{c}{OCMT} &  & \multicolumn{3}{c}{POST-OCMT} &  &
\multicolumn{3}{c}{GLASSO}\\\cline{1-3}\cline{5-7}\cline{9-11}%
Better & Same & Worse &  & Better & Same & Worse &  & Better & Same & Worse\\
77 & 17 & 6 &  & 76 & 21 & 3 &  & 5 & 88 & 7\\\hline\hline
\end{tabular}
\end{table}

\begin{table}[tbh]
\caption{Average RMSFE ratio, benchmark: AGLASSO}%
\label{table:RMSEratio}
\centering{}\centering { }
\begin{tabular}
[c]{ccccc}\hline\hline
One Stage & OCMT & POST-OCMT & GLASSO & AGLASSO\\\hline
0.873 & 0.817 & 0.814 & 1.005 & 1\\\hline\hline
\end{tabular}
\end{table}

We report the results in terms of the out-of-sample RMSFE in Tables
\ref{table:ForecaseRMSE} and \ref{table:RMSEratio}, with AGLASSO as the
benchmark. Table \ref{table:ForecaseRMSE} shows that OCMT and POST-OCMT
outperform AGLASSO in the majority of cases. Of course, our methods do not
outperform AGLASSO all the time. To highlight the necessity of the multiple
stages, we also report the results of the one-stage procedure (we selected the
tuning parameter also by minimizing the BIC) in Table \ref{table:RMSEratio}.
The OCMT stops at the second stage for all 100 cases. We normalize the average
RMSFE of AGLASSO as 1.\ Notably, the average RMSFE of our methods are lower
than that of AGLASSO. The OCMT also improves the RMSFE over the one-stage
procedure, possibly owing to some hidden signals uncovered by our
multiple-stage procedure. We report the frequencies of variables (out of 100
cases) selected by all methods in Table \ref{table:varselected}. It appears
that G102 (monthly income) along with G133 (gifts to others, including
parents) and G137 (education cost for left-behind children) contribute most to
the model, as shown by all methods in general. We note that AGLASSO tends to
select more variables than our methods, on average, which was also CKP's
finding in their application. The \textquotedblleft
over-fitting\textquotedblright\ is the main cause of the relative inferior
performance of AGLASSO.

\section{Conclusion\label{SEC:conclusion}}

In this paper, we examine the one-covariate-at-a-time multiple testing
approach to model selection in additive models. The properties of the TPR,
FPR, and FDR of our approach are established based on some asymptotic
probability bounds of Type-I and II errors. The simulation experiments and one
application on the RUMiC dataset showcase excellent small-sample properties of
our methods. Just as stated by CKP for linear models, we view our approach as
a useful alternative to the model selection methods for additive models in the literature.

\appendix

\begin{center}
{\Large \textbf{Appendix} } {\small \ \ \ \ \ \ \ \ \ }
\end{center}

\noindent In Appendix \ref{SEC:procedure}, we present a detailed procedure for
practice. We then provide technical details for the one-stage and
multiple-stage procedures in Appendices \ref{APP:tech_one_stage} and
\ref{SEC:multi_tech}, respectively. Additionally, we include technical lemmas
used in the proofs of the main results in Appendix \ref{APP:tech_lemmas}.
Finally, the main results of the paper are proven in Appendix
\ref{APP:real_main_proof}. The proofs of the technical lemmas and some
additional simulation and application results are available in the online supplement.

\section{The Procedure, Some Technical Details and Technical
Lemmas\label{APP:mainproof}}

\subsection{The Procedure \label{SEC:procedure}}

For easy reference, we summarize how to implement our procedure in this
section. We refer readers to the justification of our selection of tuning
parameters in Section \ref{SEC:tuning}. We begin with an introduction to
B-splines. Then we move to the details of the implementation.

\subsubsection{A Gentle Introduction to B-splines}

Here, we provide a brief introduction to B-splines. For more details on
regression splines, see \cite{Racine2022}. B-splines are defined by their
\textquotedblleft order\textquotedblright\ $m$ and the number of interior
\textquotedblleft knots\textquotedblright\ $N$. The degree $(J)$ of the
B-spline polynomial is given by the spline order $m$ minus one, i.e., $J = m -
1$.

Let $t_{0}\leq t_{1}\leq\cdots\leq t_{N}\leq t_{N+1}$ be the knot set, where
$t_{0}$ and $t_{N+1}$ are two \textquotedblleft endpoints\textquotedblright%
\ knots, and $t_{1}, \ldots, t_{N}$ are $N$ interior knots. The splines with
equidistant knots are called \textquotedblleft\textit{uniform}%
\textquotedblright\ splines; otherwise, they are said to be \textquotedblleft%
\textit{nonuniform}\textquotedblright. Another popular choice of knots is the
\textquotedblleft quantile\textquotedblright\ knot sequence, where the
empirical quantiles of the observable variable are used as interior knots.
Quantile knots ensure an equal number of sample observations in each interval,
while the intervals can have different lengths.

To construct the B-spline basis function recursively, define the augmented
knot set by appending the lower and upper boundary knots $t_{0}$ and $t_{N+1}$
$J\left(  =m-1\right)  $ times:
\[
t_{-(m-1)}=\cdots=t_{-1}=t_{0}\leq t_{1}\leq\cdots\leq t_{N}\leq
t_{N+1}=t_{N+2}=\cdots=t_{N+m}.
\]
Reset the index for the augmented knot set such that the $N+2m$ augmented
knots $t_{i}$ are now indexed by $i=0,...,N+2m-1$. Then we can recursively
define basis functions $B_{i,j},$ $j=0,1,...,J$, as follows:%
\begin{align*}
B_{i,0}\left(  x\right)   &  =\left\{
\begin{array}
[c]{l}%
1\text{ if }x\in\lbrack t_{i},t_{i+1})\\
0\text{ otherwise,}%
\end{array}
\right. \\
B_{i,j+1}\left(  x\right)   &  =c_{i,j+1}\left(  x\right)  B_{i,j}\left(
x\right)  +\left[  1-c_{i+1,j+1}\left(  x\right)  \right]  B_{i+1,j}\left(
x\right)  ,
\end{align*}
where $c_{i,j}\left(  x\right)  =\frac{x-t_{i}}{t_{i+j}-t_{i}}$ if
$t_{i+j}\neq t_{i};$ $0$ otherwise.

The above recurrence relation is called the \textit{de Boor recurrence
relation}; see \cite{deBoor}. For a fixed $j,$ the functions $B_{i,j}$ are
called the $i$th B-spline basis functions of degree $j$, and the total number
of $B_{i,j}(x)$ functions is $N+j+1$, according to the recursive construction.
Finally, the B-spline basis functions used in regression are $\left\{
B_{i,m-1}(x) \right\}  _{i=0}^{N+m}$, and the total number of basis functions
is $N+m$. When $m=4$ and $N=3$, there are $7$ cubic B-spline basis functions.

In our simulation and application, we choose cubic B-splines with $m=4$ and
set the number of interior knots to be $N=m_{n}-4$, where $m_{n}=\left\lfloor
n^{1/4}\right\rfloor +1$ is the total number of basis functions. Here
$\left\lfloor \cdot\right\rfloor $ is the floor operator.

\subsubsection{The Detailed Procedure}

\paragraph{\textbf{Preparation}}

\begin{enumerate}
\item Normalize the dependent variable to have sample mean 0.

\item For each covariate, say $X_{l}$, normalize the B-spline basis to mean
0:
\[
\phi_{jl}\left(  x\right)  =\psi_{j}\left(  x\right)  -n^{-1}\sum_{i=1}
^{n}\psi_{j}\left(  x_{li}\right)  ,
\]
for $j=1,2,...,m_{n}.$

\item Collect the B-spline basis for covariate $X_{l}$ as $P^{m_{n} }\left(
X_{l}\right)  =\left[  \phi_{1l}\left(  X_{l}\right)  ,\phi_{2l}\left(
X_{l}\right)  ,...,\phi_{m_{n}l}\left(  X_{l}\right)  \right]  ^{\prime}.$
\end{enumerate}

\paragraph{\textbf{An Instruction on Regression}}

\begin{enumerate}
\item To run the regression of $Y$ on $X_{l},$ we regress the demeaned
dependent variable $Y$ on $P^{m_{n}}\left(  X_{l}\right)  .$ We construct
$\mathcal{\hat{X}}_{l}$ as in equation (\ref{EQ:chil_definition}).

\item To run the regression of $Y$ on $X_{l}$ with pre-selected variables
$\boldsymbol{Z},$ we obtain the estimator by the partitioned regression of $Y$
on $P^{m_{n}}\left(  X_{l}\right)  $ by partialling out $P^{m_{n}}\left(
\boldsymbol{Z}\right)  .$ $\mathcal{\hat{X}}_{l,\boldsymbol{Z}}$ is
constructed as in equation (\ref{EQ:XlZ}).
\end{enumerate}

\paragraph{\textbf{Tuning Parameters}}

For a fixed $n,$ we set $m_{n}=\left\lfloor n^{1/4}\right\rfloor +1$. We set
\[
\varsigma_{n}=Cm_{n}\left[  \left(  \log p_{n}\right)  ^{1.1}+\left(  \log
m_{n}\right)  ^{1.1}\right]  \text{ (used at Stage 1)}%
\]
and $\varsigma_{n}^{\ast}=3\varsigma_{n}$ (used at Stages 2 or later) for
continuous regressors. We set $\varsigma_{n}=C\left(  \log p_{n}\right)
^{1.1}$ (used at Stage 1) and $\varsigma_{n}^{\ast}=4\varsigma_{n}$ (used at
Stages 2 or later) for binary regressors. We experiment with $C$ in the range
of 0.5 to 2.5, specifically using values such as 0.5, 0.6, ..., 2.5. We choose
the optimal $C$ by the BIC.

\paragraph{\textbf{The Procedure}}

\begin{enumerate}
\item Set $\varsigma_{n}=0.5m_{n}\left[  \left(  \log p_{n}\right)
^{1.1}+\left(  \log m_{n}\right)  ^{1.1}\right]  $ and $\varsigma_{n}^{\ast
}=3\varsigma_{n}$ for continuous variables, $\varsigma_{n}=0.5\left(  \log
p_{n}\right)  ^{1.1}$ and $\varsigma_{n}^{\ast}=4\varsigma_{n}$ for binary
variables, and then conduct the OCMT procedure as described later in this section;

\item Denote the selected variables from the OCMT procedure by $\boldsymbol{Z}
_{\left(  \hat{k}_{s}\right)  };$

\item Re-estimate the model by regressing $Y$ on $P^{m_{n}}\left[
\boldsymbol{Z}_{\left(  \hat{k}_{s}\right)  }\right]  $ (which is a
post-selection estimation);

\item Compute the BIC (defined in equation (\ref{EQ:BIC})) from Step 3 and
denote it as BIC$_{1};$

\item Repeat Steps 1 to 4, replacing $C=0.5$ in the definition of
$\varsigma_{n}$ and $\varsigma_{n}^{\ast}$ by $0.6,$ $0.7,\ldots,$ $2.5$ and
obtain the corresponding BIC, denoted as BIC$_{2},$ BIC$_{3},\ldots,$
BIC$_{21};$

\item Find $\hat{\imath}=\arg\min_{1\leq\iota\leq21}$ BIC$_{\iota}$ and select
the variables obtained using the $\hat{\imath}$-th values of $\varsigma_{n}$
and $\varsigma_{n}^{\ast};$

\item Conduct the adaptive group Lasso (e.g., Section \ref{SEC:tuningForLasso}
) for those variables selected in Step 6.
\end{enumerate}

\paragraph{\textbf{The OCMT Procedure}}

This is a brief summary of Section \ref{SEC:multi_procedure} but is detailed
enough for implementation.

\begin{enumerate}
\item At Stage 1, we regress demeaned $Y$ on $P^{m_{n}}\left(  X_{l}\right)  $
one by one for $l=1,2,...,p_{n}$, and calculate the test statistic
$\mathcal{\hat{X}}_{l}$ as in equation (\ref{EQ:chil_definition}). We select
$X_{l}$ if $\widehat{\mathcal{J}}_{l,\left(  1\right)  }=\mathbf{1}\left(
\mathcal{\hat{X}}_{l}>\varsigma_{n}\right)  =1.$

\item We collect all the variables selected in stage 1 into the vector
$\boldsymbol{Z}_{\left(  1\right)  }$, and denote the index set of the
selected variables as $S_{\left(  1\right)  }$. The active variables in stage
2 are denoted by $\Psi_{\left(  2\right)  }=\left\{  1,2,\ldots,p_{n}\right\}
\backslash S_{\left(  1\right)  }.$

\item For stage $k\geq2,$ we regress $Y$ on $P^{m_{n}}\left(  X_{l}\right)  $
with $P^{m_{n}}\left(  \boldsymbol{Z}_{\left(  k-1\right)  }\right)  $ as
pre-selected variables one by one for $l\in\Psi_{\left(  k\right)  }.$ We
construct the test statistic $\mathcal{\hat{X}}_{l,\boldsymbol{Z}_{\left(
k-1\right)  }}$ as in equation (\ref{EQ:XlZ}). We select the variable $X_{l}$
if $\widehat{\mathcal{J}}_{l,\left(  k\right)  }=\mathbf{1}\left(
\mathcal{\hat{X}}_{l,\boldsymbol{Z}_{\left(  k-1\right)  }}>\varsigma
_{n}^{\ast}\right)  =1.$

\item We stop the procedure at a stage in which no new variables are selected.
This final stage is denoted as $\hat{k}_{s}$. $\boldsymbol{Z}_{\left(  \hat
{k}_{s}\right)  }$ are the variables selected by the OCMT.
\end{enumerate}

\subsection{Some Technical Details for the One-Stage
Procedure\label{APP:tech_one_stage}}

The proof of Proposition \ref{TH:main1} is tedious. We illustrate the main
idea of the proof by providing some brief technical details below.

Define
\begin{equation}
f_{nl}\left(  x\right)  =P^{m_{n}}\left(  x\right)  ^{\prime}\boldsymbol{\beta
}_{l},\text{ and }U_{l}=Y-f_{nl}\left(  X_{l}\right)  . \label{EQ:fnl&Ul}%
\end{equation}
Then $E\left[  P^{m_{n}}\left(  X_{l}\right)  U_{l}\right]  =0$ by the
definition of $f_{nl}\left(  x\right)  .$ By \cite{Stone}, it holds that%
\begin{equation}
\left\{  E\left[  f_{nl}\left(  X_{l}\right)  -f_{l}\left(  X_{l}\right)
\right]  ^{2}\right\}  ^{1/2}\leq C_{1}m_{n}^{-d}\text{ for some }C_{1}%
<\infty. \label{EQ:fnl_appro}%
\end{equation}
If the bias is asymptotically negligible, by the properties listed in Lemma
\ref{LE:rank} below,
\[
\left\Vert \boldsymbol{\beta}_{l}\right\Vert \propto m_{n}^{1/2}\left\{
E\left[  f_{nl}\left(  X_{l}\right)  ^{2}\right]  \right\}  ^{1/2}\propto
m_{n}^{1/2}\left\{  E\left[  f_{l}\left(  X_{l}\right)  ^{2}\right]  \right\}
^{1/2}=m_{n}^{1/2}\theta_{l}.
\]
As a result $m_{n}^{-1/2}\left\Vert \boldsymbol{\beta}_{l}\right\Vert $ can be
equivalently used to measure the strength of net impact of $X_{l}$ on $Y$.

Let $u_{li}=y_{i}-f_{nl}\left(  x_{li}\right)  \ $and$\ \boldsymbol{u}%
_{l}=\left(  u_{l1},u_{l2},\ldots,u_{ln}\right)  ^{\prime}$. Then the
estimator for $\boldsymbol{\beta}_{l}$ can be rewritten as%
\[
\boldsymbol{\hat{\beta}}_{l}=\left(  \mathbb{X}_{l}^{\prime}\mathbb{X}%
_{l}\right)  ^{-1}\mathbb{X}_{l}^{\prime}\boldsymbol{y}=\boldsymbol{\beta}%
_{l}+\left(  \mathbb{X}_{l}^{\prime}\mathbb{X}_{l}\right)  ^{-1}\mathbb{X}%
_{l}^{\prime}\boldsymbol{u}_{l},
\]
where recall that $\mathbb{X}_{l}=(\mathbb{X}_{l1},\mathbb{X}_{l2}%
,\ldots,\mathbb{X}_{ln})^{\prime}$ is the $n\times m_{n}$ \textquotedblleft
design\textquotedblright\ matrix for $X_{l}.$ With it, $\mathcal{\hat{X}}_{l}$
can be rewritten as
\begin{equation}
\mathcal{\hat{X}}_{l}=\boldsymbol{\hat{\beta}}_{l}^{\prime}\left(  \hat
{\sigma}_{l}^{-2}\mathbb{X}_{l}^{\prime}\mathbb{X}_{l}\right)
\boldsymbol{\hat{\beta}}_{l}=\boldsymbol{\beta}_{l}^{\prime}\left(
\hat{\sigma}_{l}^{-2}\mathbb{X}_{l}^{\prime}\mathbb{X}_{l}\right)
\boldsymbol{\beta}_{l}+2\hat{\sigma}_{l}^{-2}\boldsymbol{\beta}_{l}^{\prime
}\mathbb{X}_{l}^{\prime}\boldsymbol{u}_{l}+\boldsymbol{u}_{l}^{\prime
}\mathbb{X}_{l}\left(  \hat{\sigma}_{l}^{2}\mathbb{X}_{l}^{\prime}%
\mathbb{X}_{l}\right)  ^{-1}\mathbb{X}_{l}^{\prime}\boldsymbol{u}_{l}.
\label{EQ:chi_oneS}%
\end{equation}
Our analysis is based on the decomposition in (\ref{EQ:chi_oneS}). The idea is
to show that $\left\Vert n^{-1/2}\mathbb{X}_{l}^{\prime}\boldsymbol{u}%
_{l}\right\Vert ,$ $\hat{\sigma}_{l}^{2},$ and $\left\Vert n^{-1}%
\mathbb{X}_{l}^{\prime}\mathbb{X}_{l}\right\Vert $\ are bounded by some rate
or value with very high probability. Based on these results, the value of
$\widehat{\mathcal{J}}_{l}$ is solely determined by the strength of
$\left\Vert \boldsymbol{\beta}_{l}\right\Vert $ with very high probability.

\subsection{Some Technical Details for the Multiple-Stage
Procedure\label{SEC:multi_tech}}

The proof of Proposition \ref{TH:main1_ms} is rather tedious. However, the
main idea of the proof is straightforward, and similar to the one for
Proposition \ref{TH:main1}. To see it, we provide some technical discussion below.

To simplify notations, let $P^{m_{n}}\left(  \boldsymbol{X}_{1}^{p^{\ast}%
}\right)  \equiv\left[  P^{m_{n}}\left(  X_{1}\right)  ^{\prime}%
,\ldots,P^{m_{n}}\left(  X_{p^{\ast}}\right)  ^{\prime}\right]  ^{\prime},$ a
$p^{\ast}m_{n}\times1$ vector. Define the $n\times p^{\ast}m_{n}$ matrix
\[
\mathbb{X}_{1}^{p^{\ast}}\equiv\left(  \mathbb{X}_{1},\ldots,\mathbb{X}%
_{p^{\ast}}\right)  ,
\]
which is the design matrix for all the $p^{\ast}$ signal variables. To see the
approximation of the model more clearly, let
\[
\boldsymbol{\tilde{\beta}}_{1}^{p^{\ast}}=\left\{  E\left[  P^{m_{n}}\left(
\boldsymbol{X}_{1}^{p^{\ast}}\right)  P^{m_{n}}\left(  \boldsymbol{X}%
_{1}^{p^{\ast}}\right)  ^{\prime}\right]  \right\}  ^{-1}E\left[  P^{m_{n}%
}\left(  \boldsymbol{X}_{1}^{p^{\ast}}\right)  Y\right]  ,
\]
which is the population coefficient for the regression of $Y$ on $P^{m_{n}%
}\left(  \boldsymbol{X}_{1}^{p^{\ast}}\right)  $. Write $\boldsymbol{\tilde
{\beta}}_{1}^{p^{\ast}}=\left(  \boldsymbol{\tilde{\beta}}_{1}^{\prime}%
,\ldots,\boldsymbol{\tilde{\beta}}_{p^{\ast}}^{\prime}\right)  ^{\prime},$
where $\boldsymbol{\tilde{\beta}}_{j}$ corresponds to the coefficient of
$P^{m_{n}}\left(  X_{j}\right)  $ in the regression. Let $f_{nj}^{\ast}\left(
x\right)  =P^{m_{n}}\left(  x\right)  ^{\prime}\boldsymbol{\tilde{\beta}}%
_{j}.$ Under Assumption \ref*{A:fl}', \cite{Stone} implies%
\begin{equation}
\left\{  E\left[  P^{m_{n}}\left(  X_{j}\right)  ^{\prime}\boldsymbol{\tilde
{\beta}}_{j}-f_{j}^{\ast}\left(  X_{j}\right)  \right]  ^{2}\right\}
^{1/2}=O\left(  m_{n}^{-d}\right)  . \label{EQ:appro}%
\end{equation}
We rewrite $Y$ as%
\begin{align}
Y  &  =\sum_{j=1}^{p^{\ast}}P^{m_{n}}\left(  X_{j}\right)  ^{\prime
}\boldsymbol{\tilde{\beta}}_{j}+\sum_{j=1}^{p^{\ast}}\left[  f_{j}^{\ast
}\left(  X_{j}\right)  -P^{m_{n}}\left(  X_{j}\right)  ^{\prime}%
\boldsymbol{\tilde{\beta}}_{j}\right]  +\varepsilon\nonumber\\
&  \equiv\sum_{j=1}^{p^{\ast}}P^{m_{n}}\left(  X_{j}\right)  ^{\prime
}\boldsymbol{\tilde{\beta}}_{j}+R_{n}+\varepsilon. \label{EQ:yappro}%
\end{align}
Equation (\ref{EQ:appro}) implies $R_{n}=O_{P}\left(  m_{n}^{-d}\right)  $ for
a fixed $p^{\ast}.$ Then the model looks like a linear model with an
approximation bias term $R_{n}$ and an error term $\varepsilon$.

Now we have a closer look at $\theta_{l,\boldsymbol{Z}},$ which measures the
impact of $X_{l}$ on $Y$ after controlling $\boldsymbol{Z}$. If we can ignore
the approximation bias, a finite sample approximation to the term inside
$\theta_{l,\boldsymbol{Z}},$ namely,%
\[
E\left\{  \left.  Y-P^{m_{n}}\left(  \boldsymbol{Z}\right)  ^{\prime}\left\{
E\left[  P^{m_{n}}\left(  \boldsymbol{Z}\right)  P^{m_{n}}\left(
\boldsymbol{Z}\right)  ^{\prime}\right]  \right\}  ^{-1}E\left[  P^{m_{n}%
}\left(  \boldsymbol{Z}\right)  Y\right]  \right\vert X_{l}\right\}  ,
\]
is%
\begin{align}
&  P^{m_{n}}\left(  X_{l}\right)  ^{\prime}\Phi_{X_{l}}^{-1}\cdot E\left\{
P^{m_{n}}\left(  X_{l}\right)  \left\{  Y-P^{m_{n}}\left(  \boldsymbol{Z}%
\right)  ^{\prime}\Phi_{\boldsymbol{Z}}^{-1}E\left[  P^{m_{n}}\left(
\boldsymbol{Z}\right)  Y\right]  \right\}  \right\} \nonumber\\
&  \propto m_{n}\cdot P^{m_{n}}\left(  X_{l}\right)  ^{\prime}E\left\{
P^{m_{n}}\left(  X_{l}\right)  \left\{  Y-P^{m_{n}}\left(  \boldsymbol{Z}%
\right)  ^{\prime}\Phi_{\boldsymbol{Z}}^{-1}E\left[  P^{m_{n}}\left(
\boldsymbol{Z}\right)  Y\right]  \right\}  \right\}  , \label{EQ:finiteAppro}%
\end{align}
where $\Phi_{X_{l}}\equiv E[P^{m_{n}}\left(  X_{l}\right)  P^{m_{n}}\left(
X_{l}\right)  ^{\prime}],$ $\Phi_{\boldsymbol{Z}}\equiv E[P^{m_{n}}\left(
\boldsymbol{Z}\right)  P^{m_{n}}\left(  \boldsymbol{Z}\right)  ^{\prime}],$
and the last line holds by Lemma \ref{LE:rank}. We let
\[
\boldsymbol{\eta}_{l,\boldsymbol{Z}}\equiv E\left\{  P^{m_{n}}\left(
X_{l}\right)  \left\{  Y-P^{m_{n}}\left(  \boldsymbol{Z}\right)  ^{\prime}%
\Phi_{\boldsymbol{Z}}^{-1}E\left[  P^{m_{n}}\left(  \boldsymbol{Z}\right)
Y\right]  \right\}  \right\}  ,
\]
which appears as the numerator in the population coefficient of the regression
of $Y-P^{m_{n}}\left(  \boldsymbol{Z}\right)  ^{\prime}\Phi_{\boldsymbol{Z}%
}^{-1}$ $\cdot E\left[  P^{m_{n}}\left(  \boldsymbol{Z}\right)  Y\right]  $ on
$P^{m_{n}}\left(  X_{l}\right)  $. Then the last term in equation
(\ref{EQ:finiteAppro}) becomes $m_{n}P^{m_{n}}\left(  X_{l}\right)  ^{\prime
}\boldsymbol{\eta}_{l,\boldsymbol{Z}}$. As a result,%
\begin{equation}
\theta_{l,\boldsymbol{Z}}\propto\left\{  E\left[  \left(  m_{n}P^{m_{n}%
}\left(  X_{l}\right)  ^{\prime}\boldsymbol{\eta}_{l,\boldsymbol{Z}}\right)
^{2}\right]  \right\}  ^{1/2}=\left\{  m_{n}^{2}\boldsymbol{\eta
}_{l,\boldsymbol{Z}}^{\prime}\Phi_{X_{l}}\boldsymbol{\eta}_{l,\boldsymbol{Z}%
}\right\}  ^{1/2}\propto m_{n}^{1/2}\left\Vert \boldsymbol{\eta}%
_{l,\boldsymbol{Z}}\right\Vert \label{EQ:theta_ita}%
\end{equation}
by Lemma \ref{LE:rank}. Note $\theta_{l,\boldsymbol{Z}}$ measures the impact
of $X_{l}$ on $Y$ after controlling $\boldsymbol{Z}.$ $m_{n}^{1/2}\left\Vert
\boldsymbol{\eta}_{l,\boldsymbol{Z}}\right\Vert $ can be equivalently used to
measure this strength.

To facilitate the analysis of $\boldsymbol{\eta}_{l,\boldsymbol{Z}},$ let
$\Phi_{X_{l}\boldsymbol{Z}}\equiv E[P^{m_{n}}\left(  X_{l}\right)  P^{m_{n}%
}\left(  \boldsymbol{Z}\right)  ^{\prime}].$ By the orthogonality between the
error term and regressor in the least squares projection, we have
\begin{align}
\boldsymbol{\eta}_{l,\boldsymbol{Z}}  &  =E\left\{  \left[  P^{m_{n}}\left(
X_{l}\right)  -\Phi_{X_{l}\boldsymbol{Z}}\Phi_{\boldsymbol{Z}}^{-1}P^{m_{n}%
}\left(  \boldsymbol{Z}\right)  \right]  \left\{  Y-P^{m_{n}}\left(
\boldsymbol{Z}\right)  ^{\prime}\Phi_{\boldsymbol{Z}}^{-1}E\left[  P^{m_{n}%
}\left(  \boldsymbol{Z}\right)  Y\right]  \right\}  \right\} \nonumber\\
&  =E\left\{  \left[  P^{m_{n}}\left(  X_{l}\right)  -\boldsymbol{\gamma
}_{X_{l},\boldsymbol{Z}}^{\prime}P^{m_{n}}\left(  \boldsymbol{Z}\right)
\right]  \cdot\left[  Y-\boldsymbol{\gamma}_{Y,\boldsymbol{Z}}^{\prime
}P^{m_{n}}\left(  \boldsymbol{Z}\right)  \right]  \right\} \nonumber\\
&  =E\left(  \boldsymbol{U}_{X_{l},\boldsymbol{Z}}\cdot U_{Y,\boldsymbol{Z}%
}\right)  , \label{EQ:ita}%
\end{align}
where $\boldsymbol{\gamma}_{X_{l},\boldsymbol{Z}}\equiv\Phi_{\boldsymbol{Z}%
}^{-1}\Phi_{X_{l}\boldsymbol{Z}}^{\prime},$ $\boldsymbol{\gamma}%
_{Y,\boldsymbol{Z}}\equiv\Phi_{\boldsymbol{Z}}^{-1}E\left[  P^{m_{n}}\left(
\boldsymbol{Z}\right)  Y\right]  ,$ and $\boldsymbol{U}_{X_{l},\boldsymbol{Z}%
}$ and $U_{Y,\boldsymbol{Z}}$ denote the respective projection errors:
\begin{equation}
\boldsymbol{U}_{X_{l},\boldsymbol{Z}}\equiv P^{m_{n}}\left(  X_{l}\right)
-\boldsymbol{\gamma}_{X_{l},\boldsymbol{Z}}^{\prime}P^{m_{n}}\left(
\boldsymbol{Z}\right)  ,\text{ and }U_{Y,\boldsymbol{Z}}\equiv
Y-\boldsymbol{\gamma}_{Y,\boldsymbol{Z}}^{\prime}P^{m_{n}}\left(
\boldsymbol{Z}\right)  . \label{EQ:u_xlz}%
\end{equation}
Note that $\boldsymbol{U}_{X_{l},\boldsymbol{Z}}$ is $m_{n}\times1$ and
$U_{Y,\boldsymbol{Z}}$ is a scalar. Thus, $\boldsymbol{\eta}_{l,\boldsymbol{Z}%
}$ measures the covariance between $P^{m_{n}}\left(  X_{l}\right)  $ and $Y$
after controlling $P^{m_{n}}\left(  \boldsymbol{Z}\right)  $.

At the sample level, we use the following notations. Let $\boldsymbol{z}_{i}$
denote the $i$-th observation for $\boldsymbol{Z}$. We let $\boldsymbol{u}%
_{X_{l},\boldsymbol{Z}}\equiv\left(  \boldsymbol{u}_{X_{l},\boldsymbol{Z}%
,1},\ldots,\boldsymbol{u}_{X_{l},\boldsymbol{Z},n}\right)  ^{\prime}$ and
$\boldsymbol{u}_{Y,\boldsymbol{Z}}\equiv\left(  u_{Y,\boldsymbol{Z,}1}%
,\ldots,u_{Y,\boldsymbol{Z,}n}\right)  ^{\prime},\ $where $\boldsymbol{u}%
_{X_{l},\boldsymbol{Z},i}=P^{m_{n}}\left(  x_{li}\right)  -\boldsymbol{\gamma
}_{X_{l},\boldsymbol{Z}}^{\prime}P^{m_{n}}\left(  \boldsymbol{z}_{i}\right)  $
and $u_{Y,\boldsymbol{Z,}i}=y_{i}-\boldsymbol{\gamma}_{Y,\boldsymbol{Z}%
}^{\prime}P^{m_{n}}\left(  \boldsymbol{z}_{i}\right)  .$ The connection
between $\mathcal{\hat{X}}_{l,\boldsymbol{Z}}$ and $\boldsymbol{\eta
}_{l,\boldsymbol{Z}}$ can be seen from the following. Noting that
\[
\boldsymbol{y}=\mathbb{Z}\boldsymbol{\gamma}_{Y,\boldsymbol{Z}}+\boldsymbol{u}%
_{Y,\boldsymbol{Z}}\text{ and }\mathbb{X}_{l}=\mathbb{Z}\boldsymbol{\gamma
}_{X_{l},\boldsymbol{Z}}+\boldsymbol{u}_{X_{l},\boldsymbol{Z}},
\]
we have
\begin{align}
\mathcal{\hat{X}}_{l,\boldsymbol{Z}}  &  =\left(  \boldsymbol{u}%
_{Y,\boldsymbol{Z}}^{\prime}M_{\mathbb{Z}}\boldsymbol{u}_{X_{l},\boldsymbol{Z}%
}\right)  \left(  \hat{\sigma}_{l,\boldsymbol{Z}}^{2}\mathbb{X}_{l}^{\prime
}M_{\mathbb{Z}}\mathbb{X}_{l}\right)  ^{-1}\left(  \boldsymbol{u}%
_{X_{l},\boldsymbol{Z}}^{\prime}M_{\mathbb{Z}}\boldsymbol{u}_{Y,\boldsymbol{Z}%
}\right) \nonumber\\
&  =\left[  (\boldsymbol{u}_{X_{l},\boldsymbol{Z}}^{\prime}\boldsymbol{u}%
_{Y,\boldsymbol{Z}}-n\boldsymbol{\eta}_{l,\boldsymbol{Z}}-\boldsymbol{u}%
_{X_{l},\boldsymbol{Z}}^{\prime}Q_{\mathbb{Z}}\boldsymbol{u}_{Y,\boldsymbol{Z}%
})+n\boldsymbol{\eta}_{l,\boldsymbol{Z}}\right]  ^{\prime}\left(  \hat{\sigma
}_{l,\boldsymbol{Z}}^{2}\mathbb{X}_{l}^{\prime}M_{\mathbb{Z}}\mathbb{X}%
_{l}\right)  ^{-1}\nonumber\\
&  \cdot\left[  (\boldsymbol{u}_{X_{l},\boldsymbol{Z}}^{\prime}\boldsymbol{u}%
_{Y,\boldsymbol{Z}}-n\boldsymbol{\eta}_{l,\boldsymbol{Z}}-\boldsymbol{u}%
_{X_{l},\boldsymbol{Z}}^{\prime}Q_{\mathbb{Z}}\boldsymbol{u}_{Y,\boldsymbol{Z}%
})+n\boldsymbol{\eta}_{l,\boldsymbol{Z}}\right] \nonumber\\
&  =n^{2}\boldsymbol{\eta}_{l,\boldsymbol{Z}}^{\prime}\left(  \hat{\sigma
}_{l,\boldsymbol{Z}}^{2}\mathbb{X}_{l}^{\prime}M_{\mathbb{Z}}\mathbb{X}%
_{l}\right)  ^{-1}\boldsymbol{\eta}_{l,\boldsymbol{Z}}+2n\boldsymbol{\eta
}_{l,\boldsymbol{Z}}^{\prime}\left(  \hat{\sigma}_{l,\boldsymbol{Z}}%
^{2}\mathbb{X}_{l}^{\prime}M_{\mathbb{Z}}\mathbb{X}_{l}\right)  ^{-1}%
\boldsymbol{\tilde{u}}_{l,\boldsymbol{Z}}+\boldsymbol{\tilde{u}}%
_{l,\boldsymbol{Z}}^{\prime}\left(  \hat{\sigma}_{l,\boldsymbol{Z}}%
^{2}\mathbb{X}_{l}^{\prime}M_{\mathbb{Z}}\mathbb{X}_{l}\right)  ^{-1}%
\boldsymbol{\tilde{u}}_{l,\boldsymbol{Z}}\boldsymbol{,} \label{EQ:chi_moreS}%
\end{align}
where $Q_{\mathbb{Z}}=\mathbb{Z}\left(  \mathbb{Z}^{\prime}\mathbb{Z}\right)
^{-1}\mathbb{Z}^{\prime}$ and
\begin{equation}
\boldsymbol{\tilde{u}}_{l,\boldsymbol{Z}}\equiv\boldsymbol{u}_{X_{l}%
,\boldsymbol{Z}}^{\prime}\boldsymbol{u}_{Y,\boldsymbol{Z}}-n\boldsymbol{\eta
}_{l,\boldsymbol{Z}}-\boldsymbol{u}_{X_{l},\boldsymbol{Z}}^{\prime
}Q_{\mathbb{Z}}\boldsymbol{u}_{Y,\boldsymbol{Z}}. \label{EQ:utidelz}%
\end{equation}
behaves like the counterpart of $\mathbb{X}_{l}^{\prime}\boldsymbol{u}_{l}$ in
equation (\ref{EQ:chi_oneS}). To see the last point, first note that%
\[
E\left[  \boldsymbol{u}_{X_{l},\boldsymbol{Z}}^{\prime}\boldsymbol{u}%
_{Y,\boldsymbol{Z}}-n\boldsymbol{\eta}_{l,\boldsymbol{Z}}\right]  =0.
\]
Next, it is easy to show that $\boldsymbol{u}_{X_{l},\boldsymbol{Z}}^{\prime
}Q_{\mathbb{Z}}\boldsymbol{u}_{Y,\boldsymbol{Z}}=n^{-1/2}\boldsymbol{u}%
_{X_{l},\boldsymbol{Z}}^{\prime}\mathbb{Z}\left(  n^{-1}\mathbb{Z}^{\prime
}\mathbb{Z}\right)  ^{-1}n^{-1/2}\mathbb{Z}^{\prime}\boldsymbol{u}%
_{Y,\boldsymbol{Z}}$ is of smaller order than $\boldsymbol{u}_{X_{l}%
,\boldsymbol{Z}}^{\prime}\boldsymbol{u}_{Y,\boldsymbol{Z}}-n\boldsymbol{\eta
}_{l,\boldsymbol{Z}}$ because $E\left[  P^{m_{n}}\left(  \boldsymbol{Z}%
\right)  \boldsymbol{U}_{X_{l},\boldsymbol{Z}}^{\prime}\right]
=\boldsymbol{0}$ and $E\left[  P^{m_{n}}\left(  \boldsymbol{Z}\right)
U_{Y,\boldsymbol{Z}}\right]  =\boldsymbol{0}$ by the definitions in equation
(\ref{EQ:u_xlz}). The first term in (\ref{EQ:utidelz}) is the dominant term of
$\mathcal{\hat{X}}_{l,\boldsymbol{Z}}$ if $\boldsymbol{\eta}_{l,\boldsymbol{Z}%
}\ $is large enough.

The analysis of $\mathcal{\hat{X}}_{l}\left(  \boldsymbol{Z}\right)  $ is
based on the decomposition in equation (\ref{EQ:chi_moreS}). The basic idea is
similar to the one for analyzing $\mathcal{\hat{X}}_{l}$ in equation
(\ref{EQ:chi_oneS}): we show $\hat{\sigma}_{l,\boldsymbol{Z}}^{2},$
$n^{-1}\mathbb{X}_{l}^{\prime}M_{\mathbb{Z}}\mathbb{X}_{l},$ and$\ n^{-1/2}%
\boldsymbol{\tilde{u}}_{l,\boldsymbol{Z}}$ are bounded by some value or rate
with very high probability, and as a result, $\mathbf{1}\left(  \mathcal{\hat
{X}}_{l}\left(  \boldsymbol{Z}\right)  >\varsigma_{n}\right)  $ is determined
by the strength of $\theta_{l,\boldsymbol{Z}}$ (or equivalently, $m_{n}%
^{1/2}\left\Vert \boldsymbol{\eta}_{l,\boldsymbol{Z}}\right\Vert $) with very
high probability.

\subsection{Technical Lemmas\label{APP:tech_lemmas}}

We frequently use the inequalities presented in Remark \ref{RE:ineq}, Lemma
\ref{LE:bernstein}, and Lemma \ref{LE:main_inequality}.

\begin{remark}
[Some inequalities]\label{RE:ineq}\emph{We may apply the following
inequalities in the proofs directly without referring them back to here. First
}%
\[
\Pr\left(  X_{1}+X_{2}\geq C\right)  \leq\Pr\left(  X_{1}\geq\pi C\right)
+\Pr\left(  X_{2}\geq\left(  1-\pi\right)  C\right)
\]
\emph{for any constant }$\pi.$\emph{ That is due to the fact that }$\left\{
X_{1}+X_{2}\geq C\right\}  \subseteq\left\{  X_{1}\geq\pi C\right\}
\cup\left\{  X_{1}\geq\left(  1-\pi\right)  C\right\}  .$\emph{ Similarly, we
have }%
\[
\Pr\left(  \sum_{i=1}^{n}X_{i}\geq C\right)  \leq\sum_{i=1}^{n}\Pr\left(
X_{i}\geq Cn^{-1}\right)  \ \text{\emph{and} }\Pr\left(  \max_{1\leq i\leq
n}X_{i}\geq C\right)  \leq\sum_{i=1}^{n}\Pr\left(  X_{i}\geq C\right)  .
\]
\emph{For any positive random variables }$X_{1}$\emph{\ and }$X_{2}$\emph{ and
positive constants }$C_{1}$\emph{ and }$C_{2},$%
\[
\Pr\left(  X_{1}\cdot X_{2}\geq C_{1}\right)  \leq\Pr\left(  X_{1}\geq
C_{1}/C_{2}\right)  +\Pr\left(  X_{2}\geq C_{2}\right)  ,
\]
\emph{due to the fact that }$\left\{  X_{1}\cdot X_{2}\geq C_{1}\right\}
\subseteq\left\{  X_{1}\geq C_{1}/C_{2}\right\}  \cup\left\{  X_{2}\geq
C_{2}\right\}  .$\emph{ }
\end{remark}

\begin{lemma}
[Bernstein's Inequality]\label{LE:bernstein}For an i.i.d. series $\left\{
z_{i}\right\}  _{i=1}^{\infty}$ with zero mean and bounded support $\left[
-M,M\right]  $, we have%
\[
\Pr\left(  \left\vert \sum_{i=1}^{n}z_{i}\right\vert >x\right)  \leq
2\exp\left\{  -x^{2}/\left[  2\left(  v+Mx/3\right)  \right]  \right\}  ,
\]
for any $v\geq$\emph{Var}$\left(  \sum_{i=1}^{n}z_{i}\right)  .$
\end{lemma}

The above lemma is Lemma 2.2.9 in \cite{VaartWellner}.

\begin{lemma}
\label{LE:main_inequality}Suppose that for an i.i.d. series $\left\{
z_{i}\right\}  _{i=1}^{\infty}$, $\Pr\left(  \left\vert z_{i}\right\vert
>x\right)  \leq C_{1}\exp\left(  -C_{2}x^{\alpha}\right)  $ holds for all
$x>0$ for some $C_{1},C_{2},C_{3}>0.$ Further $E\left(  z_{i}\right)  =0$ and
$E\left(  z_{i}^{2}\right)  =\sigma_{z}^{2}$. Let $\left\{  v_{n}\right\}
_{n=1}^{\infty}$ be a deterministic series such that $v_{n}\propto n^{\lambda
}.$ Then (i) if $0<\lambda\leq\left(  \alpha+1\right)  /\left(  \alpha
+2\right)  $, we have%
\[
\Pr\left(  \left\vert \sum_{i=1}^{n}z_{i}\right\vert >v_{n}\right)  \leq
\exp\left[  \left.  -\left(  1-\pi\right)  ^{2}v_{n}^{2}\right/  \left(
2n\sigma_{z}^{2}\right)  \right]  \text{ for any }\pi\in\left(  0,1\right)  ;
\]
(ii) if $\lambda\geq\left(  \alpha+1\right)  /\left(  \alpha+2\right)  ,$ we
have%
\[
\Pr\left(  \left\vert \sum_{i=1}^{n}z_{i}\right\vert >v_{n}\right)  \leq
\exp\left(  -C_{4}v_{n}^{\alpha/\left(  \alpha+1\right)  }\right)  \text{ for
some }C_{4}>0.
\]

\end{lemma}

This lemma is Lemma A3 in CKP.

\begin{lemma}
\label{LE:rank}Recall that $\Phi_{X_{l}}=E[P^{m_{n}}\left(  X_{l}\right)
P^{m_{n}}\left(  X_{l}\right)  ^{\prime}].$ Suppose that Assumption
\ref*{A:supp} holds. Then there exist some positive constants $B_{1},$
$B_{2},$ $B_{3},$ and $B_{4}$ such that
\begin{align}
\text{ }B_{1}m_{n}^{-1}  &  \leq\lambda_{\min}\left(  \Phi_{X_{l}}\right)
\leq\lambda_{\max}\left(  \Phi_{X_{l}}\right)  \leq B_{2}m_{n}^{-1}%
,\label{EQ:lam1}\\
B_{3}m_{n}^{-1}  &  \leq E\left[  \phi_{k}\left(  X_{l}\right)  ^{2}\right]
\leq B_{4}m_{n}^{-1},\text{ and }\nonumber\\
E\left[  \left\vert \phi_{k}\left(  X_{l}\right)  \right\vert \right]   &
\propto m_{n}^{-1},\nonumber
\end{align}
for all $k$ and $l$. For any $f_{nl}\left(  X_{l}\right)  =P^{m_{n}}\left(
X_{l}\right)  ^{\prime}\boldsymbol{\beta}$ where $X_{l}$ satisfies Assumption
\ref*{A:supp}, we have%
\[
\left\{  E\left[  f_{nl}\left(  X_{l}\right)  ^{2}\right]  \right\}
^{1/2}\propto m_{n}^{-1/2}\left\Vert \boldsymbol{\beta}\right\Vert .
\]

\end{lemma}

For the first and second equations in the above lemma, see Lemma 6.2 in
\cite{Zhou_S_W}. The third equation is a direct result of those on pp. 91 and
133 in \cite{deBoor}. The last equation follows from (\ref{EQ:lam1}) and the
fact%
\begin{align*}
\left\{  E\left[  f_{nl}\left(  X_{l}\right)  ^{2}\right]  \right\}  ^{1/2}
&  =\left\{  \boldsymbol{\beta}^{\prime}E\left[  P^{m_{n}}\left(
X_{l}\right)  P^{m_{n}}\left(  X_{l}\right)  ^{\prime}\right]
\boldsymbol{\beta}\right\}  ^{1/2},\text{ and}\\
\left[  \lambda_{\min}\left(  \Phi_{X_{l}}\right)  \right]  ^{1/2}\left\Vert
\boldsymbol{\beta}\right\Vert  &  \leq\left\{  \boldsymbol{\beta}^{\prime
}E\left[  P^{m_{n}}\left(  X_{l}\right)  P^{m_{n}}\left(  X_{l}\right)
^{\prime}\right]  \boldsymbol{\beta}\right\}  ^{1/2}\leq\left[  \lambda_{\max
}\left(  \Phi_{X_{l}}\right)  \right]  ^{1/2}\left\Vert \boldsymbol{\beta
}\right\Vert .
\end{align*}

\begin{lemma}
\label{LE:XX'}Suppose that Assumptions \ref*{A:iid} and \ref*{A:supp} hold.
Then%
\begin{equation}
\Pr\left(  \left\vert \left\Vert \left(  n^{-1}\mathbb{X}_{l}^{\prime
}\mathbb{X}_{l}\right)  ^{-1}\right\Vert -\left\Vert \Phi_{X_{l}}%
^{-1}\right\Vert \right\vert >\left.  \left\Vert \Phi_{X_{l}}^{-1}\right\Vert
\right/  2\right)  \leq2m_{n}^{2}\exp\left\{  -C_{1}nm_{n}^{-3}\right\}
\label{EQ:xx1}%
\end{equation}
for some positive constant $C_{1}.$ If Assumption \ref*{A:mn} also holds,
then
\[
\Pr\left(  \left\vert \left\Vert \left(  n^{-1}\mathbb{X}_{l}^{\prime
}\mathbb{X}_{l}\right)  ^{-1}\right\Vert -\left\Vert \Phi_{X_{l}}%
^{-1}\right\Vert \right\vert >\left.  \left\Vert \Phi_{X_{l}}^{-1}\right\Vert
\right/  2\right)  \leq C_{2}\exp\left\{  -C_{3}n^{C_{4}}\right\}
\]
for some positive $C_{2},$ $C_{3},$ and $C_{4}.$
\end{lemma}

\begin{lemma}
\label{LE:error_var}Suppose that Assumptions \ref*{A:iid}, \ref*{A:supp},
\ref*{A:epsilon}, \ref*{A:tech}, and \ref*{A:mn} hold. Let $\sigma_{l}%
^{2}=E\left(  U_{l}^{2}\right)  ,$ $\omega_{l}^{4}=E\left(  U_{l}^{4}\right)
-$ $\left[  E\left(  U_{l}^{2}\right)  \right]  ^{2},$ $u_{li}=y_{i}%
-f_{nl}\left(  x_{li}\right)  ,$ $\hat{u}_{li}=y_{i}-P^{m_{n}}\left(
x_{li}\right)  ^{\prime}\boldsymbol{\hat{\beta}}_{l},$ and $\hat{\sigma}%
_{l}^{2}=n^{-1}\sum_{i=1}^{n}\hat{u}_{li}^{2}.$ Then, for $v_{n}\propto
n^{\lambda}$ with $\lambda>1/2,$%
\[
\Pr\left(  \left\vert \hat{\sigma}_{l}^{2}-\sigma_{l}^{2}\right\vert \geq
v_{n}/n\right)  \leq C_{1}\exp\left(  -C_{2}n^{C_{3}}\right)
\]
holds for some $C_{1},$ $C_{2},$ and $C_{3}>0$.
\end{lemma}

\begin{lemma}
\label{LE:rank2}Suppose Assumptions \ref*{A:iid}, \ref*{A:supp},
\ref*{A:epsilon}, \ref*{A:tech}, and \ref*{A:mn} hold.

(i) If $v_{n}\lesssim n^{s/\left(  s+2\right)  }m_{n}^{2}$, then
\[
\Pr\left(  \left\vert \boldsymbol{u}_{l}^{\prime}\mathbb{X}_{l}\left(
\hat{\sigma}_{l}^{2}\mathbb{X}_{l}^{\prime}\mathbb{X}_{l}\right)
^{-1}\mathbb{X}_{l}^{\prime}\boldsymbol{u}_{l}\right\vert \geq\frac{4}%
{3}\sigma_{l}^{2}v_{n}\right)  \leq\exp\left(  -C_{1}m_{n}^{-1}v_{n}+\log
m_{n}\right)  +C_{2}\exp\left(  -C_{3}n^{C_{4}}\right)
\]
for some positive constants $C_{1},$ $C_{2},$ $C_{3},$ and $C_{4}$.

(ii) If $v_{n}\gtrsim n^{s/\left(  s+2\right)  }m_{n}^{2},$ then%
\[
\Pr\left(  \left\vert \boldsymbol{u}_{l}^{\prime}\mathbb{X}_{l}\left(
\hat{\sigma}_{l}^{2}\mathbb{X}_{l}^{\prime}\mathbb{X}_{l}\right)
^{-1}\mathbb{X}_{l}^{\prime}\boldsymbol{u}_{l}\right\vert \geq\frac{4}%
{3}\sigma_{l}^{2}v_{n}\right)  \leq C_{5}\exp\left(  -C_{6}n^{C_{7}}\right)
\]
for some positive constants $C_{5},$ $C_{6},$ and $C_{7}.$
\end{lemma}

\begin{lemma}
\label{LE:hidden}Suppose that Assumptions \ref*{A:supp}, \ref*{A:fl}', and
\ref*{A:mn} hold and that we pre-select $\boldsymbol{Z}=\left(  Z_{1}%
,\ldots,Z_{\iota_{n}}\right)  ^{\prime}$ in previous stages and all $Z$s are
either signals or pseudo-signals.

(i) Assume that $b$ signals, without loss of generality say $\boldsymbol{X}%
_{1}^{b}\equiv\left(  X_{1},\ldots,X_{b}\right)  ^{\prime},$ are not in the
list of the pre-selected variables. Let%
\[
\boldsymbol{\eta}_{1}^{b}=E\left\{  P^{m_{n}}\left(  \boldsymbol{X}_{1}%
^{b}\right)  \left\{  Y-P^{m_{n}}\left(  \boldsymbol{Z}\right)  ^{\prime
}\left\{  E\left[  P^{m_{n}}\left(  \boldsymbol{Z}\right)  P^{m_{n}}\left(
\boldsymbol{Z}\right)  ^{\prime}\right]  \right\}  ^{-1}E\left[  P^{m_{n}%
}\left(  \boldsymbol{Z}\right)  Y\right]  \right\}  \right\}  .
\]
If there exist a $j,$ $1\leq j\leq b,$ such that $\left\{  E\left[
f_{nj}^{\ast}\left(  X_{j}\right)  ^{2}\right]  \right\}  ^{1/2}\gtrsim
\kappa_{n}\log\left(  m_{n}\right)  ^{1/2}\left(  m_{n}/n\right)  ^{1/2}$ for
some $\kappa_{n}\rightarrow\infty$, we have%
\[
\left\Vert \boldsymbol{\eta}_{1}^{b}\right\Vert \gtrsim\kappa_{n}\log\left(
m_{n}\right)  ^{1/2}n^{-1/2}.
\]

(ii) For stage 1 when $Z$ is empty, similarly let%
\[
\boldsymbol{\eta}_{1}^{p^{\ast}}=E\left[  P^{m_{n}}\left(  \boldsymbol{X}%
_{1}^{p^{\ast}}\right)  Y\right]  .
\]
If there exist a $j,$ $1\leq j\leq p^{\ast},$ such that $\left\{  E\left[
f_{nj}^{\ast}\left(  X_{j}\right)  ^{2}\right]  \right\}  ^{1/2}\gtrsim
\kappa_{n}\log\left(  m_{n}\right)  ^{1/2}\left(  m_{n}/n\right)  ^{1/2}$ for
some $\kappa_{n}\rightarrow\infty$, we have%
\[
\left\Vert \boldsymbol{\eta}_{1}^{p^{\ast}}\right\Vert \gtrsim\kappa_{n}%
\log\left(  m_{n}\right)  ^{1/2}n^{-1/2}.
\]

\end{lemma}

\begin{lemma}
\label{LE:XXQQ}Suppose Assumptions \ref*{A:iid}, \ref*{A:supp}, and
\ref*{A:full_rank2} hold. Further suppose all covariates in $\boldsymbol{Z}$
($\iota_{n}\times1$) are either signals or pseudo signals. For a random
variable $X_{l}$ that is not included in $\boldsymbol{Z},$ it holds that%
\[
\Pr\left(  \left\vert \left\Vert \left(  n^{-1}\left(
\begin{array}
[c]{c}%
\mathbb{Z}^{\prime}\\
\mathbb{X}_{l}^{\prime}%
\end{array}
\right)  \left(
\begin{array}
[c]{cc}%
\mathbb{Z} & \mathbb{X}_{l}%
\end{array}
\right)  \right)  ^{-1}\right\Vert -\left\Vert \Phi^{-1}\right\Vert
\right\vert >\left\Vert \Phi^{-1}\right\Vert /2\right)  \leq2m_{n}^{2}%
\iota_{n}^{2}\exp\left\{  -C_{1}nm_{n}^{-3}\iota_{n}^{-1}\right\}
\]
for some positive constant $C_{1},$ where $\Phi=\left[
\begin{array}
[c]{cc}%
\Phi_{\boldsymbol{Z}} & \Phi_{\boldsymbol{Z}X_{l}}\\
\Phi_{\boldsymbol{Z}X_{l}}^{\prime} & \Phi_{X_{l}}%
\end{array}
\right]  ,$ $\Phi_{\boldsymbol{Z}X_{l}}=E\left[  P^{m_{n}}(\boldsymbol{Z}%
)^{\prime}P^{m_{n}}(X_{l})\right]  ,$ and recall that $\Phi_{\boldsymbol{Z}%
}=E[P^{m_{n}}\left(  \boldsymbol{Z}\right)  P^{m_{n}}\left(  \boldsymbol{Z}%
\right)  ^{\prime}]$ and $\Phi_{X_{l}}=E\left[  P^{m_{n}}\left(  X_{l}\right)
P^{m_{n}}\left(  X_{l}\right)  ^{\prime}\right]  .$ If in addition Assumptions
\ref*{A:p}' and \ref*{A:mn}\ hold, the right-hand side of the inequality can
be replaced by $C_{2}\exp\left(  -C_{3}n^{C_{4}}\right)  .$
\end{lemma}

\begin{lemma}
\label{LE:lambdamax}Suppose Assumptions \ref*{A:iid}, \ref*{A:p}',
\ref*{A:supp}, \ref*{A:epsilon}, \ref*{A:tech}, \ref*{A:mn},\ and
\ref*{A:full_rank2} hold. Further suppose all covariates in $\boldsymbol{Z}$
($\iota_{n}\times1$) are either signals or pseudo signals. For a random
variable $X_{l}$ that is not included in $\boldsymbol{Z},$ then
\begin{align*}
\Pr\left(  \lambda_{\max}\left\{  \left(  \hat{\sigma}_{l,\boldsymbol{Z}}%
^{2}n^{-1}\mathbb{X}_{l}^{\prime}M_{\mathbb{Z}}\mathbb{X}_{l}\right)
^{-1}\right\}  \geq2\sigma^{-2}B_{X1}^{-1}m_{n}\right)   &  \leq C_{1}%
\exp(-C_{2}n^{C_{3}})\text{ and}\\
\Pr\left(  \lambda_{\min}\left\{  \left(  \hat{\sigma}_{l,\boldsymbol{Z}}%
^{2}n^{-1}\mathbb{X}_{l}^{\prime}M_{\mathbb{Z}}\mathbb{X}_{l}\right)
^{-1}\right\}  \leq\frac{1}{4}\sigma^{-2}B_{X2}^{-1}m_{n}\right)   &  \leq
C_{4}\exp(-C_{5}n^{C_{6}}),
\end{align*}
for some positive $C_{1},\ldots,C_{6},$ where $\sigma^{2}$ is defined in Lemma
\ref{LE:error_var_Z}.
\end{lemma}

\begin{lemma}
\label{LE:error_bound}Suppose Assumptions \ref*{A:iid}, \ref*{A:p}',
\ref*{A:supp}, \ref*{A:epsilon}, \ref*{A:fl}', \ref*{A:tech}, \ref*{A:mn}, and
\ref*{A:full_rank2} hold. Further suppose all covariates in $\boldsymbol{Z}$
($\iota_{n}\times1$) are either signals or pseudo signals. For a random
variable $X_{l}$ that is not included in $\boldsymbol{Z},$then%
\begin{align*}
&  \Pr\left(  \left\Vert n^{-1/2}\left(  \boldsymbol{u}_{X_{l},\boldsymbol{Z}%
}^{\prime}\boldsymbol{u}_{Y,\boldsymbol{Z}}-\boldsymbol{\eta}%
_{l,\boldsymbol{Z}}\right)  \right\Vert ^{2}\geq C_{1}m_{n}^{-1}v_{n}\right)
\\
&  \leq\left\{
\begin{array}
[c]{c}%
2m_{n}\exp\left\{  -nm_{n}^{-2}v_{n}/[C_{2}(nm_{n}^{-1}+\iota_{n}n^{1/2}%
v_{n}^{1/2})]\right\}  +m_{n}\exp\left(  -C_{3}m_{n}^{-1}v_{n}\right) \\
2m_{n}\exp\left\{  -nm_{n}^{-2}v_{n}/[C_{2}(nm_{n}^{-1}+\iota_{n}n^{1/2}%
v_{n}^{1/2})]\right\}  +C_{4}\exp\left(  -C_{5}n^{C_{6}}\right)
\end{array}
\right.
\begin{array}
[c]{c}%
\text{if }v_{n}\lesssim n^{s/\left(  s+2\right)  }m_{n}^{2}\\
\text{if }v_{n}\gtrsim n^{s/\left(  s+2\right)  }m_{n}^{2}%
\end{array}
,
\end{align*}
and%
\begin{align*}
&  \Pr\left(  \left\Vert n^{-1/2}\boldsymbol{u}_{X_{l},\boldsymbol{Z}}%
^{\prime}\mathbb{Z}\left(  \mathbb{Z}^{\prime}\mathbb{Z}\right)
^{-1}\mathbb{Z}^{\prime}\boldsymbol{u}_{Y,\boldsymbol{Z}}\right\Vert ^{2}\geq
C_{1}m_{n}^{-1}v_{n}\right) \\
&  \leq\left\{
\begin{array}
[c]{l}%
C_{13}\exp\left(  -C_{14}n^{C_{15}}\right)  +m_{n}^{2}\iota_{n}\exp\left(
-C_{16}n^{1/2}m_{n}^{-1/6}\iota_{n}^{-1}\left(  m_{n}^{-1}v_{n}\right)
^{1/2}\right) \\
C_{17}\exp\left(  -C_{18}n^{C_{19}}\right)
\end{array}
\right.
\begin{array}
[c]{c}%
\text{if }v_{n}\lesssim n^{\left(  s-2\right)  /\left[  4\left(  s+2\right)
\right]  }m_{n}^{5/6}\iota_{n}^{2}\\
\text{if }v_{n}\gtrsim n^{\left(  s-2\right)  /\left[  4\left(  s+2\right)
\right]  }m_{n}^{5/6}\iota_{n}^{2}%
\end{array}
,
\end{align*}
for some positive $C_{2},C_{3},C_{4},C_{5},$ and $C_{6}$. If we set
$v_{n}=\varsigma_{n},$ then%
\[
\Pr\left(  \left\Vert n^{-1/2}\boldsymbol{\tilde{u}}_{l,\boldsymbol{Z}%
}\right\Vert ^{2}\geq C_{1}m_{n}^{-1}\varsigma_{n}\right)  \leq n^{-M}%
+C_{20}\exp\left(  -C_{21}n^{C_{22}}\right)
\]
for any large positive constant $M$ and some positive constants $C_{20}%
,C_{21},$ and $C_{22}.$
\end{lemma}

\begin{lemma}
\label{LE:error_var_Z}Suppose Assumptions \ref*{A:iid}, \ref*{A:p}',
\ref*{A:supp}, \ref*{A:epsilon}, \ref*{A:fl}', \ref*{A:mn}, and
\ref*{A:full_rank2} hold. Suppose all covariates in $\boldsymbol{Z}$
($\iota_{n}\times1$) are either signals or pseudo signals and $X_{l}$ is a
variable that is not included in $\boldsymbol{Z}.$ Define%
\[
U=Y-\left(
\begin{array}
[c]{c}%
P^{m_{n}}\left(  \boldsymbol{Z}\right) \\
P^{m_{n}}\left(  X_{l}\right)
\end{array}
\right)  ^{\prime}\Phi^{-1}E\left[  \left(
\begin{array}
[c]{c}%
P^{m_{n}}\left(  \boldsymbol{Z}\right) \\
P^{m_{n}}\left(  X_{l}\right)
\end{array}
\right)  Y\right]  ,
\]
Let $\sigma^{2}=E\left(  U^{2}\right)  ,$ $\omega^{4}=E\left(  U^{4}\right)
-\left[  E\left(  U^{2}\right)  \right]  ^{2}$. Then, for $v_{n}\propto
n^{\lambda}$ with $\lambda>1/2,$ we have%
\[
\Pr\left(  \left\vert \hat{\sigma}_{l,\boldsymbol{Z}}^{2}-\sigma
^{2}\right\vert \geq v_{n}/n\right)  \leq C_{1}\exp\left(  -C_{2}n^{C_{3}%
}\right)  \text{ for some }C_{1},C_{2},C_{3}>0.
\]

\end{lemma}

\begin{lemma}
\label{LE:bound_error}Suppose Assumptions \ref*{A:iid}, \ref*{A:p}',
\ref*{A:supp}, \ref*{A:epsilon}, \ref*{A:tech}, \ref*{A:mn},\ and
\ref*{A:full_rank2} hold. Further suppose all covariates in $\boldsymbol{Z}$
($\iota_{n}\times1$) are either signals or pseudo signals. For a random
variable $X_{l}$ that is not included in $\boldsymbol{Z},$ we have%
\[
\Pr\left(  \boldsymbol{\tilde{u}}_{l,\boldsymbol{Z}}^{\prime}\left(
\hat{\sigma}_{l,\boldsymbol{Z}}^{2}\mathbb{X}_{l}^{\prime}M_{\mathbb{Z}%
}\mathbb{X}_{l}\right)  ^{-1}\boldsymbol{\tilde{u}}_{l,\boldsymbol{Z}}%
\geq\varsigma_{n}\right)  \leq\exp\left(  -C_{1}m_{n}^{-1}\varsigma_{n}+\log
m_{n}\right)  +C_{2}\exp(-C_{3}n^{C_{4}}).
\]

\end{lemma}

\begin{lemma}
\label{LE:allsignals}Suppose Assumptions \ref*{A:iid}, \ref*{A:p}',
\ref*{A:supp}, \ref*{A:epsilon}, \ref*{A:fl}', \ref*{A:tech}, \ref*{A:mn},
\ref*{A:xi_n}, and \ref*{A:full_rank2} hold. Then
\[
\Pr\left(  \mathcal{D}_{k}\right)  \geq1-kn^{-M_{1}}-kC_{4}\exp\left(
-C_{5}n^{C_{6}}\right)
\]
for any finite fixed positive integer $k,$ any fixed large positive constant
$M_{1},$ and some positive constants $C_{4},C_{5},$ and $C_{6}.$ When $k$ is
fixed or divergent to infinity at a rate no faster than $n^{a}$ for some
$a>0,$ we can write
\[
\Pr\left(  \mathcal{D}_{k}\right)  \geq1-n^{-M_{2}}-C_{7}\exp\left(
-C_{8}n^{C_{9}}\right)  \text{ for any }k\leq n^{a},
\]
for some large positive constant $M_{2}$ and some positive constants
$C_{7},C_{8},$ and $C_{9}.$
\end{lemma}

\section{Proofs of the Main Results\label{APP:real_main_proof}}

\begin{proof}
[Proof of Proposition \ref{TH:main1}]To see the intuition, we begin with the
case where $\theta_{l}=0$. When $\theta_{l}=0,$ $\left\Vert \boldsymbol{\beta
}_{l}\right\Vert =0$ and $\mathcal{\hat{X}}_{l}=\boldsymbol{u}_{l}^{\prime
}\mathbb{X}_{l}\left(  \hat{\sigma}_{l}^{2}\mathbb{X}_{l}^{\prime}%
\mathbb{X}_{l}\right)  ^{-1}\mathbb{X}_{l}^{\prime}\boldsymbol{u}_{l}.$ Then,%
\begin{align*}
\Pr\left(  \mathcal{\hat{X}}_{l}\geq\varsigma_{n}\right)   &  \leq\Pr\left(
\boldsymbol{u}_{l}^{\prime}\mathbb{X}_{l}\left(  \hat{\sigma}_{l}%
^{2}\mathbb{X}_{l}^{\prime}\mathbb{X}_{l}\right)  ^{-1}\mathbb{X}_{l}^{\prime
}\boldsymbol{u}_{l}\geq\varsigma_{n}\right) \\
&  =\Pr\left(  \boldsymbol{u}_{l}^{\prime}\mathbb{X}_{l}\left(  \hat{\sigma
}_{l}^{2}\mathbb{X}_{l}^{\prime}\mathbb{X}_{l}\right)  ^{-1}\mathbb{X}%
_{l}^{\prime}\boldsymbol{u}_{l}\geq\frac{4}{3}\sigma_{l}^{2}\left(  \frac
{3}{4}\sigma_{l}^{-2}\varsigma_{n}\right)  \right) \\
&  \leq\exp\left(  -C_{1}m_{n}^{-1}\varsigma_{n}+\log m_{n}\right)  +C_{2}%
\exp\left(  -C_{3}n^{C_{4}}\right)
\end{align*}
for some positive constants $C_{1},C_{2},$ and $C_{3},$ where the last line
holds by applying the first part of Lemma \ref{LE:rank2} with $v_{n}=\frac
{3}{4}\sigma_{l}^{-2}\varsigma_{n}$. Since $\varsigma_{n}\propto\kappa_{n}%
\log\left(  m_{n}\right)  m_{n},$ $C_{1}m_{n}^{-1}\varsigma_{n}-\log m_{n}\gg
M\log m_{n}$ for any fixed positive $M.$ This shows the first part of the
proposition when $\theta_{l}=0.$

Note by Lemma \ref{LE:rank}
\[
\left\{  E\left[  f_{nl}\left(  X_{l}\right)  ^{2}\right]  \right\}
^{1/2}\propto m_{n}^{-1/2}\left\Vert \boldsymbol{\beta}_{l}\right\Vert .
\]
If $\theta_{l}\lesssim\log\left(  m_{n}\right)  ^{1/2}\left(  m_{n}/n\right)
^{1/2},$ by the triangle inequality, equation (\ref{EQ:fnl_appro}) and
Assumption \ref*{A:mn},
\begin{align*}
\left\{  E\left[  f_{nl}\left(  X_{l}\right)  ^{2}\right]  \right\}  ^{1/2}
&  \leq\left\{  E\left[  f_{l}\left(  X_{l}\right)  ^{2}\right]  \right\}
^{1/2}+\left\{  E\left[  \left\vert f_{nl}\left(  X_{l}\right)  -f_{l}\left(
X_{l}\right)  \right\vert ^{2}\right]  \right\}  ^{1/2}\\
&  \lesssim\log\left(  m_{n}\right)  ^{1/2}\left(  m_{n}/n\right)
^{1/2}+C_{1}m_{n}^{-d}\lesssim\log\left(  m_{n}\right)  ^{1/2}\left(
m_{n}/n\right)  ^{1/2},
\end{align*}
because the approximation bias is asymptotically negligible under Assumption
\ref*{A:mn}. Therefore, $\left\Vert \boldsymbol{\beta}_{l}\right\Vert
\lesssim\log\left(  m_{n}\right)  ^{1/2}m_{n}n^{-1/2}.$ Then,
\begin{align*}
&  \Pr\left(  \mathcal{\hat{X}}_{l}\geq\varsigma_{n}\right) \\
&  =\Pr\left(  \left(  n^{1/2}\boldsymbol{\beta}_{l}+n^{1/2}\left(
\mathbb{X}_{l}^{\prime}\mathbb{X}_{l}\right)  ^{-1}\mathbb{X}_{l}^{\prime
}\boldsymbol{u}_{l}\right)  ^{\prime}\left(  \hat{\sigma}_{l}^{-2}%
n^{-1}\mathbb{X}_{l}^{\prime}\mathbb{X}_{l}\right)  \left(  n^{1/2}%
\boldsymbol{\beta}_{l}+n^{1/2}\left(  \mathbb{X}_{l}^{\prime}\mathbb{X}%
_{l}\right)  ^{-1}\mathbb{X}_{l}^{\prime}\boldsymbol{u}_{l}\right)
\geq\varsigma_{n}\right) \\
&  \leq\Pr\left(  \left\Vert n^{1/2}\boldsymbol{\beta}_{l}+n^{1/2}\left(
\mathbb{X}_{l}^{\prime}\mathbb{X}_{l}\right)  ^{-1}\mathbb{X}_{l}^{\prime
}\boldsymbol{u}_{l}\right\Vert ^{2}\hat{\sigma}_{l}^{-2}\lambda_{\max}\left\{
n^{-1}\mathbb{X}_{l}^{\prime}\mathbb{X}_{l}\right\}  \geq\varsigma_{n}\right)
\\
&  \leq\Pr\left(  \left\Vert n^{1/2}\boldsymbol{\beta}_{l}+n^{1/2}\left(
\mathbb{X}_{l}^{\prime}\mathbb{X}_{l}\right)  ^{-1}\mathbb{X}_{l}^{\prime
}\boldsymbol{u}_{l}\right\Vert ^{2}\geq\frac{1}{3}\sigma_{l}^{2}B_{2}%
^{-1}m_{n}\varsigma_{n}\right)  +\Pr\left(  \hat{\sigma}_{l}^{-2}\geq
2\sigma_{l}^{-2}\right) \\
&  +\Pr\left(  \lambda_{\max}\left\{  \left(  \mathbb{X}_{l}^{\prime
}\mathbb{X}_{l}\right)  \right\}  \geq\frac{3}{2}B_{2}m_{n}^{-1}\right) \\
&  \equiv A_{1}+A_{2}+A_{3}.
\end{align*}
By Lemma \ref{LE:error_var} and Lemma \ref{LE:XX'}, $A_{l}=\frac{1}{4}%
C_{2}\exp\left(  -C_{3}n^{C_{4}}\right)  $ for some constants $C_{2},C_{3}$
and $C_{4}$ and $l=2,3.$ For $A_{1},$ we have with $C\equiv\frac{1}{3}%
\sigma_{l}^{2}B_{2}^{-1},$%
\begin{align*}
A_{1}  &  \leq\Pr\left(  \left\Vert n^{1/2}\left(  \mathbb{X}_{l}^{\prime
}\mathbb{X}_{l}\right)  ^{-1}\mathbb{X}_{l}^{\prime}\boldsymbol{u}%
_{l}\right\Vert +\left\Vert n^{1/2}\boldsymbol{\beta}_{l}\right\Vert \geq
Cm_{n}^{1/2}\varsigma_{n}^{1/2}\right) \\
&  =\Pr\left(  \left\Vert \left(  n^{-1}\mathbb{X}_{l}^{\prime}\mathbb{X}%
_{l}\right)  ^{-1}n^{-1/2}\mathbb{X}_{l}^{\prime}\boldsymbol{u}_{l}\right\Vert
\geq\frac{1}{2}Cm_{n}^{1/2}\varsigma_{n}^{1/2}+\frac{1}{2}Cm_{n}%
^{1/2}\varsigma_{n}^{1/2}-\left\Vert n^{1/2}\boldsymbol{\beta}_{l}\right\Vert
\right) \\
&  \leq\Pr\left(  \left\Vert \left(  n^{-1}\mathbb{X}_{l}^{\prime}%
\mathbb{X}_{l}\right)  ^{-1}n^{-1/2}\mathbb{X}_{l}^{\prime}\boldsymbol{u}%
_{l}\right\Vert \geq\frac{1}{2}Cm_{n}^{1/2}\varsigma_{n}^{1/2}\right) \\
&  \leq\Pr\left(  \left\Vert n^{-1/2}\mathbb{X}_{l}^{\prime}\boldsymbol{u}%
_{l}\right\Vert \geq\frac{1}{2}Cm_{n}^{-1/2}\varsigma_{n}^{1/2}\right)
+\Pr\left(  \left\Vert \lambda_{\max}\left(  n^{-1}\mathbb{X}_{l}^{\prime
}\mathbb{X}_{l}\right)  ^{-1}\right\Vert \geq\frac{3}{2}B_{1}^{-1}m_{n}\right)
\\
&  \leq n^{-M}+\frac{1}{2}C_{2}\exp\left(  -C_{3}n^{C_{4}}\right)
\end{align*}
for any arbitrarily large constant $M$ and some constants $C_{2},C_{3}$ and
$C_{4},$ where for the fourth line holds by the fact that $m_{n}%
^{1/2}\varsigma_{n}^{1/2}=\kappa_{n}^{1/2}\log\left(  m_{n}\right)
^{1/2}m_{n}\gg\left\Vert n^{1/2}\boldsymbol{\beta}_{l}\right\Vert ,$ and the
last line holds by equation (\ref{EQ:p1_2}), the fact that $C_{1}m_{n}%
^{-1}\varsigma_{n}-\log m_{n}\gg M\log m_{n},$ and Lemma \ref{LE:XX'}. In sum,
we have $\Pr\left(  \mathcal{\hat{X}}_{l}\geq\varsigma_{n}\right)  \leq
n^{-M}+C_{2}\exp\left(  -C_{3}n^{C_{4}}\right)  .$

We turn to the second part of the proposition. If $\theta_{l}\gtrsim\kappa
_{n}\log\left(  m_{n}\right)  ^{1/2}\left(  m_{n}/n\right)  ^{1/2},$ by
equation (\ref{EQ:fnl_appro}) and Assumption \ref*{A:mn},
\[
\left\{  E\left[  f_{nl}\left(  X_{l}\right)  ^{2}\right]  \right\}
^{1/2}\gtrsim\kappa_{n}\log\left(  m_{n}\right)  ^{1/2}\left(  m_{n}/n\right)
^{1/2},
\]
where we use the fact that the approximation bias is asymptotically negligible
under Assumption \ref*{A:mn}. Therefore, $\left\Vert \boldsymbol{\beta}%
_{l}\right\Vert \gtrsim\kappa_{n}\log\left(  m_{n}\right)  ^{1/2}m_{n}%
n^{-1/2}.$

We bound $\Pr\left(  \mathcal{\hat{X}}_{l}<\varsigma_{n}\right)  $ as follows:%
\begin{align*}
&  \Pr\left(  \mathcal{\hat{X}}_{l}<\varsigma_{n}\right) \\
&  =\Pr\left(  \left(  n^{1/2}\boldsymbol{\beta}_{l}+n^{1/2}\left(
\mathbb{X}_{l}^{\prime}\mathbb{X}_{l}\right)  ^{-1}\mathbb{X}_{l}^{\prime
}\boldsymbol{u}_{l}\right)  ^{\prime}\left(  \hat{\sigma}_{l}^{-2}%
n^{-1}\mathbb{X}_{l}^{\prime}\mathbb{X}_{l}\right)  \left(  n^{1/2}%
\boldsymbol{\beta}_{l}+n^{1/2}\left(  \mathbb{X}_{l}^{\prime}\mathbb{X}%
_{l}\right)  ^{-1}\mathbb{X}_{l}^{\prime}\boldsymbol{u}_{l}\right)
<\varsigma_{n}\right) \\
&  \leq\Pr\left(  \left\Vert n^{1/2}\boldsymbol{\beta}_{l}+n^{1/2}\left(
\mathbb{X}_{l}^{\prime}\mathbb{X}_{l}\right)  ^{-1}\mathbb{X}_{l}^{\prime
}\boldsymbol{u}_{l}\right\Vert ^{2}\lambda_{\min}\left\{  \left(  \hat{\sigma
}_{l}^{-2}n^{-1}\mathbb{X}_{l}^{\prime}\mathbb{X}_{l}\right)  \right\}
<\varsigma_{n}\right) \\
&  \leq\Pr\left(  \left\Vert n^{1/2}\boldsymbol{\beta}_{l}+n^{1/2}\left(
\mathbb{X}_{l}^{\prime}\mathbb{X}_{l}\right)  ^{-1}\mathbb{X}_{l}^{\prime
}\boldsymbol{u}_{l}\right\Vert ^{2}<4\sigma_{l}^{2}B_{1}^{-1}m_{n}%
\varsigma_{n}\right) \\
&  +\Pr\left(  \lambda_{\min}\left\{  \left(  n^{-1}\mathbb{X}_{l}^{\prime
}\mathbb{X}_{l}\right)  \right\}  <\frac{1}{2}B_{1}m_{n}^{-1}\right)
+\Pr\left(  \hat{\sigma}_{l}^{-2}<\sigma_{l}^{-2}/2\right) \\
&  \equiv A_{4}+A_{5}+A_{6}.
\end{align*}
By Lemma \ref{LE:error_var} and Lemma \ref{LE:XX'}, $A_{l}=\frac{1}{4}%
C_{5}\exp\left(  -C_{6}n^{C_{7}}\right)  $ for some positive constants
$C_{5},C_{6},$ and $C_{7}$ and $l=5,6.$ For $A_{4}$, we have with
$C=2\sigma_{l}B_{1}^{-1/2},$%
\begin{align*}
A_{4}  &  =\Pr\left(  \left\Vert n^{1/2}\boldsymbol{\beta}_{l}+n^{1/2}\left(
\mathbb{X}_{l}^{\prime}\mathbb{X}_{l}\right)  ^{-1}\mathbb{X}_{l}^{\prime
}\boldsymbol{u}_{l}\right\Vert <Cm_{n}^{1/2}\varsigma_{n}^{1/2}\right) \\
&  \leq\Pr\left(  \left\Vert n^{1/2}\boldsymbol{\beta}_{l}\right\Vert
-\left\Vert n^{1/2}\left(  \mathbb{X}_{l}^{\prime}\mathbb{X}_{l}\right)
^{-1}\mathbb{X}_{l}^{\prime}\boldsymbol{u}_{l}\right\Vert <Cm_{n}%
^{1/2}\varsigma_{n}^{1/2}\right) \\
&  =\Pr\left(  \left\Vert \left(  n^{-1}\mathbb{X}_{l}^{\prime}\mathbb{X}%
_{l}\right)  ^{-1}n^{-1/2}\mathbb{X}_{l}^{\prime}\boldsymbol{u}_{l}\right\Vert
>\frac{1}{2}\left\Vert n^{1/2}\boldsymbol{\beta}_{l}\right\Vert +\frac{1}%
{2}\left\Vert n^{1/2}\boldsymbol{\beta}_{l}\right\Vert -Cm_{n}^{1/2}%
\varsigma_{n}^{1/2}\right) \\
&  \leq\Pr\left(  \left\Vert \left(  n^{-1}\mathbb{X}_{l}^{\prime}%
\mathbb{X}_{l}\right)  ^{-1}n^{-1/2}\mathbb{X}_{l}^{\prime}\boldsymbol{u}%
_{l}\right\Vert >\frac{1}{2}\left\Vert n^{1/2}\boldsymbol{\beta}%
_{l}\right\Vert \right) \\
&  \leq\Pr\left(  \left\Vert n^{-1/2}\mathbb{X}_{l}^{\prime}\boldsymbol{u}%
_{l}\right\Vert >\frac{1}{3}B_{1}m_{n}^{-1}\left\Vert n^{1/2}\boldsymbol{\beta
}_{l}\right\Vert \right)  +\Pr\left(  \lambda_{\max}\left\{  \left(
n^{-1}\mathbb{X}_{l}^{\prime}\mathbb{X}_{l}\right)  ^{-1}\right\}  >\frac
{3}{2}B_{1}^{-1}m_{n}\right) \\
&  \leq n^{-M}+\frac{1}{2}C_{5}\exp\left(  -C_{6}n^{C_{7}}\right)
\end{align*}
for any arbitrarily large constant $M$ and some positive constants
$C_{5},C_{6},$ and $C_{7},$ where the second inequality follows from the fact
that $\frac{1}{2}\left\Vert n^{1/2}\boldsymbol{\beta}_{l}\right\Vert \gg
m_{n}^{1/2}\varsigma_{n}^{1/2}$, and the last line holds by taking
$v_{n}\propto\left(  m_{n}^{-1/2}\left\Vert n^{1/2}\boldsymbol{\beta}%
_{l}\right\Vert \right)  ^{2}\gtrsim\kappa_{n}^{2}\log\left(  m_{n}\right)
m_{n}$ in equation (\ref{EQ:p1_2}) and applying Lemma \ref{LE:XX'}. In sum, we
have
\[
\Pr\left(  \mathcal{\hat{X}}_{l}<\varsigma_{n}\right)  \leq n^{-M}+C_{5}%
\exp\left(  -C_{6}n^{C_{7}}\right)  ,
\]
or equivalently, $\Pr\left(  \mathcal{\hat{X}}_{l}\geq\varsigma_{n}\right)
>1-n^{-M}-C_{5}\exp\left(  -C_{6}n^{C_{7}}\right)  .\medskip$
\end{proof}

\begin{proof}
[Proof of Theorem \ref{TH:main2}]Proposition \ref{TH:main1} will be used
repeatedly in the proof. First, we study TPR$_{n}.$ Note that%
\begin{align*}
E\left(  \text{TPR}_{n}\right)   &  =p^{\ast-1}\sum_{l=1}^{p_{n}}E\left[
1\left(  \widehat{\mathcal{J}}_{l}=1\text{ and }\left\{  E\left[  f_{l}^{\ast
}\left(  X_{l}\right)  ^{2}\right]  \right\}  ^{1/2}\neq0\right)  \right] \\
&  =p^{\ast-1}\sum_{l=1}^{p^{\ast}}\Pr\left(  \mathcal{\hat{X}}_{l}%
\geq\varsigma_{n}\right)  \geq1-n^{-M}-C_{1}\exp\left(  -C_{2}n^{C_{3}%
}\right)
\end{align*}
for some positive constants $M,$ $C_{1},C_{2},$ and $C_{3},$ where the
inequality holds by the second part of Proposition \ref{TH:main1} and
Assumption \ref*{A;no_hidden}.

Next, we study FPR$_{n}.$
\begin{align*}
E\left(  \text{FPR}_{n}\right)   &  =\left(  p_{n}-p^{\ast}\right)  ^{-1}%
\sum_{l=p^{\ast}+1}^{p_{n}}E\left[  1\left(  \widehat{\mathcal{J}}_{l}=1\text{
and }\left\{  E\left[  f_{l}^{\ast}\left(  X_{l}\right)  ^{2}\right]
\right\}  ^{1/2}=0\right)  \right] \\
&  =\left(  p_{n}-p^{\ast}\right)  ^{-1}\sum_{l=p^{\ast}+1}^{p^{\ast}%
+p^{\ast\ast}}E\left[  1\left(  \widehat{\mathcal{J}}_{l}=1\text{ and
}\left\{  E\left[  f_{l}^{\ast}\left(  X_{l}\right)  ^{2}\right]  \right\}
^{1/2}=0\right)  \right] \\
&  +\left(  p_{n}-p^{\ast}\right)  ^{-1}\sum_{l=p^{\ast}+p^{\ast\ast}%
+1}^{p_{n}}E\left[  1\left(  \widehat{\mathcal{J}}_{l}=1\text{ and }\left\{
E\left[  f_{l}^{\ast}\left(  X_{l}\right)  ^{2}\right]  \right\}
^{1/2}=0\right)  \right] \\
&  \leq p^{\ast\ast}/\left(  p_{n}-p^{\ast}\right)  +\left(  p_{n}-p^{\ast
}\right)  ^{-1}\sum_{l=p^{\ast}+p^{\ast\ast}+1}^{p_{n}}E\left[  1\left(
\widehat{\mathcal{J}}_{l}=1\text{ and }\left\{  E\left[  f_{l}^{\ast}\left(
X_{l}\right)  ^{2}\right]  \right\}  ^{1/2}=0\right)  \right] \\
&  \leq p^{\ast\ast}/\left(  p_{n}-p^{\ast}\right)  +C_{4}n^{-M}+C_{5}%
\exp\left(  -C_{6}n^{C_{7}}\right)
\end{align*}
for some positive $C_{4},C_{5},$ and $C_{6}$ and any large positive constant
$M,$ and the last inequality holds by the first part of Proposition
\ref{TH:main1} and Assumption \ref*{A;no_hidden}.

Now, we turn to FDR$_{n}.$ Note that%
\begin{align*}
&  E\left[  \sum_{l=1}^{p_{n}}1\left(  \widehat{\mathcal{J}}_{l}=1,\text{
}\left\{  E\left[  f_{l}^{\ast}\left(  X_{l}\right)  ^{2}\right]  \right\}
^{1/2}=0,\text{ and }\theta_{l}\lesssim\log\left(  m_{n}\right)  ^{1/2}\left(
m_{n}/n\right)  ^{1/2}\right)  \right] \\
&  =\sum_{l=p^{\ast}+p^{\ast\ast}+1}^{p_{n}}\Pr\left(  \mathcal{\hat{X}}%
_{l}\geq\varsigma_{n}\right) \\
&  \leq\left(  p_{n}-p^{\ast}-p^{\ast\ast}\right)  \left[  n^{-M}+C_{8}%
\exp\left(  -C_{9}n^{C_{10}}\right)  \right]
\end{align*}
for any large positive number $M$ and some positive constants $C_{8},C_{9},$
and $C_{10}$ by the first part of Proposition \ref{TH:main1}$.$ Taking a $M$
sufficiently large, we have
\[
E\left[  \sum_{l=1}^{p_{n}}1\left(  \widehat{\mathcal{J}}_{l}=1\text{,
}\left\{  E\left[  f_{l}^{\ast}\left(  X_{l}\right)  ^{2}\right]  \right\}
^{1/2}=0,\text{ and }\theta_{l}\lesssim\log\left(  m_{n}\right)  ^{1/2}\left(
m_{n}/n\right)  ^{1/2}\right)  \right]  \rightarrow0.
\]
Then
\[
\text{FDR}_{n}=\frac{\sum_{l=1}^{p_{n}}1\left(  \widehat{\mathcal{J}}%
_{l}=1,\text{ }\left\{  E\left[  f_{l}^{\ast}\left(  X_{l}\right)
^{2}\right]  \right\}  ^{1/2}=0,\text{ and }\theta_{l}\lesssim\log\left(
m_{n}\right)  ^{1/2}\left(  m_{n}/n\right)  ^{1/2}\right)  }{\sum_{l=1}%
^{p_{n}}\widehat{\mathcal{J}}_{l}+1}\overset{P}{\rightarrow}0
\]
by Markov inequality and the fact that $\sum_{l=1}^{p_{n}}\widehat{\mathcal{J}%
}_{l}+1\geq1$.\medskip
\end{proof}

\begin{proof}
[Proof of Proposition \ref{TH:main1_ms}]Note that $\left\Vert \boldsymbol{\eta
}_{l,\boldsymbol{Z}}\right\Vert \propto m_{n}^{-1/2}\left\Vert \theta
_{l,\boldsymbol{Z}}\right\Vert $ by equation (\ref{EQ:theta_ita}). It is
equivalent to showing the results for the case when $\boldsymbol{\eta
}_{l,\boldsymbol{Z}}\lesssim\log\left(  m_{n}\right)  ^{1/2}n^{-1/2}$ and the
case when $\boldsymbol{\eta}_{l,\boldsymbol{Z}}\gtrsim\kappa_{n}\left[
\log\left(  m_{n}\right)  \right]  ^{1/2}n^{-1/2}$.

For clarity, we begin with the case when $\boldsymbol{\eta}_{l,\boldsymbol{Z}%
}=0.$ In this case, the $\mathcal{\hat{X}}_{l,\boldsymbol{Z}}$ in equation
(\ref{EQ:chi_moreS}) reduces to $\mathcal{\hat{X}}_{l,\boldsymbol{Z}%
}=\boldsymbol{\tilde{u}}_{l,\boldsymbol{Z}}^{\prime}\left(  \hat{\sigma
}_{l,\boldsymbol{Z}}^{2}\mathbb{X}_{l}^{\prime}M_{\mathbb{Z}}\mathbb{X}%
_{l}\right)  ^{-1}\boldsymbol{\tilde{u}}_{l,\boldsymbol{Z}},$ where
$M_{\mathbb{Z}}=I_{n}-\mathbb{Z}\left(  \mathbb{Z}^{\prime}\mathbb{Z}\right)
^{-1}\mathbb{Z}^{\prime}.$ Then by Lemma \ref{LE:bound_error},
\begin{align*}
\Pr\left(  \mathcal{\hat{X}}_{l,\boldsymbol{Z}}\geq\varsigma_{n}\right)   &
\leq\exp\left(  -C_{1}m_{n}^{-1}\varsigma_{n}+\log m_{n}\right)  +C_{2}%
\exp(-C_{3}n^{C_{4}})\\
&  =\exp\left(  -C_{5}\kappa_{n}\log\left(  m_{n}\right)  +\log m_{n}\right)
+C_{2}\exp(-C_{3}n^{C_{4}})\\
&  \leq n^{-M}+C_{2}\exp\left(  -C_{3}n^{C_{4}}\right)  ,
\end{align*}
for any arbitrarily large positive constant $M,$ where the second line holds
by the fact that $\varsigma_{n}\propto\kappa_{n}\log\left(  m_{n}\right)
m_{n}$ and the last line holds because $C_{5}\kappa_{n}\log\left(
m_{n}\right)  -\log m_{n}\gg M\log m_{n}.$

When $\boldsymbol{\eta}_{l,\boldsymbol{Z}}\neq0,$ we analyze $\mathcal{\hat
{X}}_{l,\boldsymbol{Z}}$ similarly as we do for $\mathcal{\hat{X}}_{l}$ in the
proof of Proposition \ref{TH:main1}. We first consider the case where
$\boldsymbol{\eta}_{l,\boldsymbol{Z}}\lesssim\log\left(  m_{n}\right)
^{1/2}n^{-1/2}.$ Note that
\begin{align}
\Pr\left(  \mathcal{\hat{X}}_{l,\boldsymbol{Z}}\geq\varsigma_{n}\right)   &
=\Pr\left(  \left(  n^{1/2}\boldsymbol{\eta}_{l,\boldsymbol{Z}}+n^{-1/2}%
\boldsymbol{\tilde{u}}_{l,\boldsymbol{Z}}\right)  ^{\prime}\left(  \hat
{\sigma}_{l,\boldsymbol{Z}}^{2}n^{-1}\mathbb{X}_{l}^{\prime}M_{\mathbb{Z}%
}\mathbb{X}_{l}\right)  ^{-1}\left(  n^{1/2}\boldsymbol{\eta}%
_{l,\boldsymbol{Z}}+n^{-1/2}\boldsymbol{\tilde{u}}_{l,\boldsymbol{Z}}\right)
\geq\varsigma_{n}\right) \nonumber\\
&  \leq\Pr\left(  \left\Vert n^{1/2}\boldsymbol{\eta}_{l,\boldsymbol{Z}%
}+n^{-1/2}\boldsymbol{\tilde{u}}_{l,\boldsymbol{Z}}\right\Vert ^{2}%
\lambda_{\max}\left\{  \left(  \hat{\sigma}_{l,\boldsymbol{Z}}^{2}%
n^{-1}\mathbb{X}_{l}^{\prime}M_{\mathbb{Z}}\mathbb{X}_{l}\right)
^{-1}\right\}  \geq\varsigma_{n}\right) \nonumber\\
&  \leq\Pr\left(  \left\Vert n^{1/2}\boldsymbol{\eta}_{l,\boldsymbol{Z}%
}+n^{-1/2}\boldsymbol{\tilde{u}}_{l,\boldsymbol{Z}}\right\Vert ^{2}\geq
\frac{1}{2}\sigma^{2}B_{X1}m_{n}^{-1}\varsigma_{n}\right) \nonumber\\
&  +\Pr\left(  \lambda_{\max}\left\{  \left(  \hat{\sigma}_{l,\boldsymbol{Z}%
}^{2}n^{-1}\mathbb{X}_{l}^{\prime}M_{\mathbb{Z}}\mathbb{X}_{l}\right)
^{-1}\right\}  \geq2\sigma^{-2}B_{X1}^{-1}m_{n}\right) \nonumber\\
&  \equiv A_{7}+A_{8}. \label{EQ:chi1}%
\end{align}
For $A_{7},$ we have%
\begin{align}
A_{7}  &  =\Pr\left(  \left\Vert n^{1/2}\boldsymbol{\eta}_{l,\boldsymbol{Z}%
}+n^{-1/2}\boldsymbol{\tilde{u}}_{l,\boldsymbol{Z}}\right\Vert \geq
2^{-1/2}\sigma B_{X1}^{1/2}m_{n}^{-1/2}\varsigma_{n}^{1/2}\right) \nonumber\\
&  \leq\Pr\left(  \left\Vert n^{1/2}\boldsymbol{\eta}_{l,\boldsymbol{Z}%
}\right\Vert +\left\Vert n^{-1/2}\boldsymbol{\tilde{u}}_{l,\boldsymbol{Z}%
}\right\Vert \geq2^{-1/2}\sigma B_{X1}^{1/2}m_{n}^{-1/2}\varsigma_{n}%
^{1/2}\right) \nonumber\\
&  =\Pr\left(  \left\Vert n^{-1/2}\boldsymbol{\tilde{u}}_{l,\boldsymbol{Z}%
}\right\Vert \geq\frac{1}{2}2^{-1/2}\sigma B_{X1}^{1/2}m_{n}^{-1/2}%
\varsigma_{n}^{1/2}+\frac{1}{2}2^{-1/2}\sigma B_{X1}^{1/2}m_{n}^{-1/2}%
\varsigma_{n}^{1/2}-\left\Vert n^{1/2}\boldsymbol{\eta}_{l,\boldsymbol{Z}%
}\right\Vert \right) \nonumber\\
&  \leq\Pr\left(  \left\Vert n^{-1/2}\boldsymbol{\tilde{u}}_{l,\boldsymbol{Z}%
}\right\Vert \geq\frac{1}{2}2^{-1/2}\sigma B_{X1}^{1/2}m_{n}^{-1/2}%
\varsigma_{n}^{1/2}\right) \nonumber\\
&  \leq n^{-M}+\frac{1}{2}C_{2}\exp\left(  -C_{3}n^{C_{4}}\right)  ,
\label{EQ:chi1_1}%
\end{align}
where the second inequality holds by the fact that $m_{n}^{-1/2}\varsigma
_{n}^{1/2}\propto\kappa_{n}^{1/2}\left[  \log\left(  m_{n}\right)  \right]
^{1/2}\gg\left[  \log\left(  m_{n}\right)  \right]  ^{1/2}\gtrsim\left\Vert
n^{1/2}\boldsymbol{\eta}_{l,\boldsymbol{Z}}\right\Vert ,$ and the last
inequality holds by Lemma \ref{LE:error_bound}. By Lemma \ref{LE:lambdamax},
\begin{equation}
A_{8}\leq\frac{1}{2}C_{2}\exp\left(  -C_{3}n^{C_{4}}\right)  .
\label{EQ:chi1_2}%
\end{equation}
Combining (\ref{EQ:chi1_1}), (\ref{EQ:chi1_2}) and (\ref{EQ:chi1}) yields
$\Pr\left(  \mathcal{\hat{X}}_{l,\boldsymbol{Z}}\geq\varsigma_{n}\right)  \leq
n^{-M}+\frac{1}{2}C_{2}\exp\left(  -C_{3}n^{C_{4}}\right)  .$

Now, we consider the case where $\boldsymbol{\eta}_{l,\boldsymbol{Z}}%
\gtrsim\kappa_{n}\left[  \log\left(  m_{n}\right)  \right]  ^{1/2}n^{-1/2}.$
Note that%
\begin{align}
\Pr\left(  \mathcal{\hat{X}}_{l,\boldsymbol{Z}}<\varsigma_{n}\right)   &
=\Pr\left(  \left(  n\boldsymbol{\eta}_{l,\boldsymbol{Z}}+\boldsymbol{\tilde
{u}}_{l,\boldsymbol{Z}}\right)  ^{\prime}\left(  \hat{\sigma}%
_{l,\boldsymbol{Z}}^{2}\mathbb{X}_{l}^{\prime}M_{\mathbb{Z}}\mathbb{X}%
_{l}\right)  ^{-1}\left(  n\boldsymbol{\eta}_{l,\boldsymbol{Z}}%
+\boldsymbol{\tilde{u}}_{l,\boldsymbol{Z}}\right)  <\varsigma_{n}\right)
\nonumber\\
&  \leq\Pr\left(  \left\Vert n^{1/2}\boldsymbol{\eta}_{l,\boldsymbol{Z}%
}+n^{-1/2}\boldsymbol{\tilde{u}}_{l,\boldsymbol{Z}}\right\Vert ^{2}%
\lambda_{\min}\left\{  \left(  \hat{\sigma}_{l,\boldsymbol{Z}}^{2}%
n^{-1}\mathbb{X}_{l}^{\prime}M_{\mathbb{Z}}\mathbb{X}_{l}\right)
^{-1}\right\}  <\varsigma_{n}\right) \nonumber\\
&  \leq\Pr\left(  \left\Vert n^{1/2}\boldsymbol{\eta}_{l,\boldsymbol{Z}%
}+n^{-1/2}\boldsymbol{\tilde{u}}_{l,\boldsymbol{Z}}\right\Vert ^{2}%
<4\sigma^{2}B_{X2}m_{n}^{-1}\varsigma_{n}\right) \nonumber\\
&  +\Pr\left(  \lambda_{\min}\left\{  \left(  \hat{\sigma}_{l,\boldsymbol{Z}%
}^{2}n^{-1}\mathbb{X}_{l}^{\prime}M_{\mathbb{Z}}\mathbb{X}_{l}\right)
^{-1}\right\}  <\frac{1}{4}\sigma^{-2}B_{X2}^{-1}m_{n}\right) \\
&  \equiv A_{9}+A_{10}. \label{EQ:sosoP}%
\end{align}
Noting that $\left\Vert n^{1/2}\boldsymbol{\eta}_{l,\boldsymbol{Z}}%
+n^{-1/2}\boldsymbol{\tilde{u}}_{l,\boldsymbol{Z}}\right\Vert \geq\left\Vert
n^{1/2}\boldsymbol{\eta}_{l,\boldsymbol{Z}}\right\Vert -\left\Vert
n^{-1/2}\boldsymbol{\tilde{u}}_{l,\boldsymbol{Z}}\right\Vert ,$ we have%
\begin{align}
A_{9}  &  =\Pr\left(  \left\Vert n^{1/2}\boldsymbol{\eta}_{l,\boldsymbol{Z}%
}+n^{-1/2}\boldsymbol{\tilde{u}}_{l,\boldsymbol{Z}}\right\Vert <2\sigma
B_{X2}^{1/2}m_{n}^{-1/2}\varsigma_{n}^{1/2}\right) \nonumber\\
&  \leq\Pr\left(  \left\Vert n^{-1/2}\boldsymbol{\tilde{u}}_{l,\boldsymbol{Z}%
}\right\Vert >\frac{1}{2}\left\Vert n^{1/2}\boldsymbol{\eta}_{l,\boldsymbol{Z}%
}\right\Vert +\frac{1}{2}\left\Vert n^{1/2}\boldsymbol{\eta}_{l,\boldsymbol{Z}%
}\right\Vert -2\sigma B_{X2}^{1/2}m_{n}^{-1/2}\varsigma_{n}^{1/2}\right)
\nonumber\\
&  \leq\Pr\left(  \left\Vert n^{-1/2}\boldsymbol{\tilde{u}}_{l,\boldsymbol{Z}%
}\right\Vert >\frac{1}{2}\left\Vert n^{1/2}\boldsymbol{\eta}_{l,\boldsymbol{Z}%
}\right\Vert \right) \nonumber\\
&  \leq\Pr\left(  \left\Vert n^{-1/2}\boldsymbol{\tilde{u}}_{l,\boldsymbol{Z}%
}\right\Vert >C\kappa_{n}^{1/2}m_{n}^{-1/2}\varsigma_{n}^{1/2}\right)
\nonumber\\
&  \leq n^{-M}+\frac{1}{2}C_{5}\exp\left(  -C_{6}n^{C_{7}}\right)  ,
\label{EQ:soso2}%
\end{align}
for any arbitrarily large positive constant $M$ and some positive constants
$C_{5},C_{7},$ and $C_{7},$ where the second and third inequalities follow
from the fact that $\left\Vert n^{1/2}\boldsymbol{\eta}_{l,\boldsymbol{Z}%
}\right\Vert \gtrsim\kappa_{n}\left[  \log\left(  m_{n}\right)  \right]
^{1/2}\propto\kappa_{n}^{1/2}m_{n}^{-1/2}\varsigma_{n}^{1/2}\gg2\sigma
^{-1}B_{X2}^{1/2}m_{n}^{-1/2}\varsigma_{n}^{1/2},$ and the last inequality
holds by Lemma \ref{LE:error_bound} to get the last inequality. By the second
part of Lemma \ref{LE:lambdamax}, $A_{10}\leq\frac{1}{2}C_{5}\exp\left(
-C_{6}n^{C_{7}}\right)  .$ It follows that
\[
\Pr\left(  \mathcal{\hat{X}}_{l,\boldsymbol{Z}}<\varsigma_{n}\right)
<n^{-M}+C_{5}\exp\left(  -C_{6}n^{C_{7}}\right)  .
\]
That is, $\Pr\left(  \mathcal{\hat{X}}_{l,\boldsymbol{Z}}\geq\varsigma
_{n}\right)  \geq1-n^{-M}-C_{5}\exp\left(  -C_{6}n^{C_{7}}\right)  .$
\end{proof}

\bigskip

To prove Theorem \ref{TH:stoppingFDR}, we introduce some notations presented
in Table \ref{Table7} below.

\begin{table}[h]
\caption{Notations used in the proof of Theorem \ref{TH:stoppingFDR}}%
\label{Table7}%
\centering{}\centering{ }
\begin{tabular}
[c]{cc}\hline
Notation & Meaning\\\hline
\multicolumn{1}{l}{$\mathcal{B}_{l,k}$} & \multicolumn{1}{l}{variable $l$ is
selected at the $k$th stage of the OCMT procedure.}\\
\multicolumn{1}{l}{$\mathcal{L}_{l,k}=\cup_{h=1}^{k}B_{l,h}$} &
\multicolumn{1}{l}{variable $l$ is selected up to and including the $k$th
stage}\\
\multicolumn{1}{l}{$\mathcal{N}_{k}=\cup_{l=p^{\ast}+p^{\ast\ast}+1}^{p_{n}%
}L_{l,k}$} & \multicolumn{1}{l}{one or more noise variables are selected up to
the $k$th stage}\\
\multicolumn{1}{l}{$\mathcal{A}_{k}=\cap_{l=1}^{p^{\ast}}L_{l,k}$} &
\multicolumn{1}{l}{all signal variables are selected up to the $k$th stage}\\
\multicolumn{1}{l}{$\mathcal{H}_{k}=\cap_{l=p^{\ast}+1}^{p^{\ast}+p^{\ast\ast
}}L_{l,k}$} & \multicolumn{1}{l}{all pseudo-signal variable are selected up to
the $k$th stage.}\\
\multicolumn{1}{l}{$\mathcal{D}_{k}$} & \multicolumn{1}{l}{variables selected
up to the $k$th stage are signals or pseudo signals}\\
\multicolumn{1}{l}{$\mathcal{T}_{k}$} & \multicolumn{1}{l}{The OCMT procedure
concludes at or before stage $k$}\\\hline
\end{tabular}
\end{table}

Apparently, $\mathcal{D}_{k}=\mathcal{N}_{k}^{c},$ the complement of
$\mathcal{N}_{k}.\medskip$

\begin{proof}
[Proof of Theorem \ref{TH:stoppingFDR}]Recall that the constants $C_{l}$'s and
$M_{l}$'s may vary across lines. By the definition in Table 2, all
pseudo-signals have strong net effects on $Y$ in the first stage with
$\theta_{l}\gtrsim\kappa_{n}\log\left(  m_{n}\right)  ^{1/2}\left(
m_{n}/n\right)  ^{1/2}$ for a slowly divergent series $\kappa_{n}%
\rightarrow\infty.$ Under this condition, the second part of Proposition
\ref{TH:main1} implies all pseudo-signals can be selected in stage 1 with very
high probability. That is, $\mathcal{H}_{1}$ happens with very high
probability. Specifically,%
\begin{align}
\Pr\left(  \mathcal{H}_{1}\right)   &  =\Pr\left(  \cap_{l=p^{\ast}%
+1}^{p^{\ast}+p^{\ast\ast}}\mathcal{B}_{l,1}\right)  =1-\Pr\left(
\cup_{l=p^{\ast}+1}^{p^{\ast}+p^{\ast\ast}}\mathcal{B}_{l,1}^{c}\right)
\geq1-\sum_{l=p^{\ast}+1}^{p^{\ast}+p^{\ast\ast}}\Pr\left(  \mathcal{B}%
_{l,1}^{c}\right) \nonumber\\
&  \geq1-p^{\ast\ast}\left(  n^{-M}+C_{1}\exp\left(  -C_{2}n^{C_{3}}\right)
\right) \nonumber\\
&  \geq1-n^{-M_{1}}-C_{4}\exp\left(  -C_{5}n^{C_{6}}\right)  \label{EQ:h1}%
\end{align}
for some arbitrarily large constant $M$, some $M_{1}<M-B_{p^{\ast\ast}},$ and
some positive constants $C_{1},\ldots,C_{6},$\ where the second inequality
holds by the second part of Proposition \ref{TH:main1}, and the last
inequality holds by Assumption \ref*{A:p}' by restricting $p^{\ast\ast}\propto
n^{B_{p^{\ast\ast}}}.$ Since $M$ can be arbitrarily large, $M_{1}$ can be
arbitrarily large too. The above result, in conjunction with the fact that
$\mathcal{H}_{1}\subset\mathcal{H}_{k}$ for any any $k>1,$ implies that
\begin{equation}
\Pr\left(  \mathcal{H}_{k}\right)  \geq\Pr\left(  \mathcal{H}_{1}\right)
\geq1-n^{-M_{1}}-C_{4}\exp\left(  -C_{5}n^{C_{6}}\right)  . \label{EQ:hk}%
\end{equation}
That is, the probability of that all pseudo-signals are selected up to stage
$k$ is larger than the probability of that all pseudo-signals are selected up
to stage 1.

Recall that $p^{\ast}$ is the number of signals that is assumed to be fixed in
Assumption \ref*{A:p}. We consider the probability of $\mathcal{A}_{p^{\ast}%
}\cap\mathcal{H}_{p^{\ast}},$ the event that all signals and pseudo-signals
are selected up to stage $p^{\ast}.$ The key to bound this probability is
Proposition \ref{TH:main1_ms}. To apply the proposition, we need the condition
that all pre-selected variables are signals or pseudo-signals, which is the
event $\mathcal{D}_{p^{\ast}-1}$. In view of this observation, we conduct the
analysis as follows. First,
\begin{align}
\Pr\left(  \mathcal{A}_{p^{\ast}}\cap\mathcal{H}_{p^{\ast}}\right)   &
\geq\Pr\left(  \mathcal{A}_{p^{\ast}}\cap\mathcal{H}_{p^{\ast}}\cap
\mathcal{D}_{p^{\ast}-1}\right) \label{EQ:aphp}\\
&  =\Pr\left(  \mathcal{A}_{p^{\ast}}|\mathcal{H}_{p^{\ast}}\cap
\mathcal{D}_{p^{\ast}-1}\right)  \Pr\left(  \mathcal{H}_{p^{\ast}}%
\cap\mathcal{D}_{p^{\ast}-1}\right)  .\nonumber
\end{align}
For $\Pr\left(  \mathcal{H}_{p^{\ast}}\cap\mathcal{D}_{p^{\ast}-1}\right)  ,$
we have%
\begin{align}
\Pr\left(  \mathcal{H}_{p^{\ast}}\cap\mathcal{D}_{p^{\ast}-1}\right)   &
=1-\Pr\left(  \mathcal{H}_{p^{\ast}}^{c}\cup\mathcal{D}_{p^{\ast}}^{c}\right)
\nonumber\\
&  \geq\left[  1-\Pr\left(  \mathcal{H}_{p^{\ast}}^{c}\right)  \right]
-\Pr\left(  \mathcal{D}_{p^{\ast}-1}^{c}\right) \nonumber\\
&  \geq1-n^{-M_{2}}-C_{7}\exp\left(  -C_{8}n^{C_{9}}\right)  \label{EQ:hp}%
\end{align}
for an arbitrarily large constant $M_{2}$ and some positive constants
$C_{7},C_{8}$ and $C_{9}$, where the last inequality follows by equation
(\ref{EQ:hk}) and Lemma \ref{LE:allsignals}. Since we only have $p^{\ast}$
signals, with Assumption \ref*{A;no_hidden}', Lemma \ref{LE:hidden} implies
that the population effect of at least one signal on $Y$ conditional on
$\mathcal{D}_{k},$ $k\leq p^{\ast}-1$, would become large enough to be picked
up by our procedure with very high probability. For the $l$th signal$,$ $1\leq
l\leq p^{\ast}$, we denote the stage where its effect on $Y$ becomes large
enough to be picked up by $k_{l}^{\ast},$ that is $\boldsymbol{\theta
}_{l,\boldsymbol{Z}_{\left(  k_{l}^{\ast}-1\right)  }}\gtrsim\kappa_{n}%
\log\left(  m_{n}\right)  ^{1/2}m_{n}^{1/2}n^{-1/2}$. Without loss of
generality, we assume $k_{1}^{\ast}\leq k_{2}^{\ast}\leq k_{3}^{\ast}%
\ldots\leq k_{p^{\ast}}^{\ast}$. By Lemma \ref{LE:hidden} again, $k_{1}^{\ast
}=1$ and $k_{l}^{\ast}-k_{l-1}^{\ast}\leq1.$ So the OCMT procedure with very
high probability does not stop until all the signals are selected.

We now bound the probability formally. By definition,
\[
\mathcal{A}_{p^{\ast}}=\cap_{l=1}^{p^{\ast}}\mathcal{L}_{l,p^{\ast}}%
=\cap_{l=1}^{p^{\ast}}\left(  \cup_{h=1}^{p^{\ast}}\mathcal{B}_{l,h}\right)
\supseteq\cap_{l=1}^{p^{\ast}}\mathcal{B}_{l,k_{l}^{\ast}}.
\]
Noting that $k_{1}^{\ast}=1$, $k_{l}^{\ast}-k_{l-1}^{\ast}\leq1,$ and
$k_{1}^{\ast}\leq k_{2}^{\ast}\leq k_{3}^{\ast}\ldots\leq k_{p^{\ast}}^{\ast}$
if $\cap_{l=1}^{p^{\ast}}\mathcal{B}_{l,k_{l}^{\ast}}$ occurs, the OCMT
procedure does not stop before stage $k_{p^{\ast}}^{\ast}$. Then$,$%
\begin{align}
\Pr\left(  \mathcal{A}_{p^{\ast}}|\mathcal{H}_{p^{\ast}}\cap\mathcal{D}%
_{p^{\ast}-1}\right)   &  \geq\Pr\left(  \cap_{l=1}^{p^{\ast}}\mathcal{B}%
_{l,k_{l}^{\ast}}|\mathcal{H}_{p^{\ast}}\cap\mathcal{D}_{p^{\ast}-1}\right)
\nonumber\\
&  =1-\Pr\left(  \cup_{l=1}^{p^{\ast}}\mathcal{B}_{l,k_{l}^{\ast}}%
^{c}|\mathcal{H}_{p^{\ast}}\cap\mathcal{D}_{p^{\ast}-1}\right) \nonumber\\
&  \geq1-\sum_{l=1}^{p^{\ast}}\Pr\left(  \mathcal{B}_{l,k_{l}^{\ast}}%
^{c}|\mathcal{H}_{p^{\ast}}\cap\mathcal{D}_{p^{\ast}-1}\right)  .
\label{EQ:ap1}%
\end{align}
Conditional on the event $\mathcal{H}_{p^{\ast}}\cap\mathcal{D}_{p^{\ast}-1},$
all variables in $\boldsymbol{Z}_{\left(  k_{l}^{\ast}-1\right)  }$ are either
signals or pseudo-signals with $\boldsymbol{\theta}_{l,\boldsymbol{Z}_{\left(
k_{l}^{\ast}-1\right)  }}\gtrsim\kappa_{n}\log\left(  m_{n}\right)
^{1/2}m_{n}^{1/2}n^{-1/2}.$ Then we can apply the second part of Proposition
\ref{TH:main1_ms} on $\Pr\left(  \mathcal{B}_{l,k_{l}^{\ast}}^{c}%
|\mathcal{H}_{p^{\ast}}\cap\mathcal{D}_{p^{\ast}-1}\right)  $ to obtain%
\[
\Pr\left(  \mathcal{B}_{l,k_{l}^{\ast}}^{c}|\mathcal{H}_{p^{\ast}}%
\cap\mathcal{D}_{p^{\ast}-1}\right)  \leq n^{-M}+C_{1}\exp\left(
-C_{2}n^{C_{3}}\right)  .
\]
Substituting this into equation (\ref{EQ:ap1}) yields%
\begin{align}
\Pr\left(  \mathcal{A}_{p^{\ast}}|\mathcal{H}_{p^{\ast}}\cap\mathcal{D}%
_{p^{\ast}-1}\right)   &  \geq1-\sum_{l=1}^{p^{\ast}}\left[  n^{-M}+C_{1}%
\exp\left(  -C_{2}n^{C_{3}}\right)  \right] \nonumber\\
&  \geq1-n^{-M_{3}}-C_{10}\exp\left(  -C_{11}n^{C_{12}}\right)  \label{EQ:apf}%
\end{align}
for some arbitrarily large constant $M_{3}$ and some positive constants
$C_{10},C_{11}$ and $C_{12}.$ Substituting equations (\ref{EQ:hp}) and
(\ref{EQ:apf}) into equation (\ref{EQ:aphp}) yields%
\begin{align}
\Pr\left(  \mathcal{A}_{p^{\ast}}\cap\mathcal{H}_{p^{\ast}}\right)   &
\geq\left(  1-n^{-M_{3}}-C_{10}\exp\left(  -C_{11}n^{C_{12}}\right)  \right)
\left(  1-n^{-M_{2}}-C_{7}\exp\left(  -C_{8}n^{C_{9}}\right)  \right)
\nonumber\\
&  \geq1-n^{-M_{3}}-C_{10}\exp\left(  -C_{11}n^{C_{12}}\right)  -n^{-M_{2}%
}-C_{7}\exp\left(  -C_{8}n^{C_{9}}\right) \nonumber\\
&  \geq1-n^{-M_{4}}-C_{13}\exp\left(  -C_{14}n^{C_{15}}\right)
\label{EQ:aphpf}%
\end{align}
for some large positive constant $M_{4},$ and some positive constants
$C_{13},$ $C_{14},$ and $C_{15}$.

With (\ref{EQ:aphpf}), we are ready to complete the proof of the first part of
the theorem. Notice that%
\begin{align}
\Pr\left(  \mathcal{T}_{p^{\ast}}\right)   &  \geq\Pr\left(  \mathcal{T}%
_{p^{\ast}}\cap\mathcal{A}_{p^{\ast}}\cap\mathcal{H}_{p^{\ast}}\cap
\mathcal{D}_{p^{\ast}}\right) \nonumber\\
&  =\Pr\left(  \mathcal{T}_{p^{\ast}}|\mathcal{A}_{p^{\ast}}\cap
\mathcal{H}_{p^{\ast}}\cap\mathcal{D}_{p^{\ast}}\right)  \Pr\left(
\mathcal{A}_{p^{\ast}}\cap\mathcal{H}_{p^{\ast}}\cap\mathcal{D}_{p^{\ast}%
}\right) \nonumber\\
&  =\Pr\left(  \mathcal{T}_{p^{\ast}}|\mathcal{A}_{p^{\ast}}\cap
\mathcal{H}_{p^{\ast}}\cap\mathcal{D}_{p^{\ast}}\right)  \left[  1-\Pr\left(
\mathcal{A}_{p^{\ast}}^{c}\cup\mathcal{H}_{p^{\ast}}^{c}\cup\mathcal{D}%
_{p^{\ast}}^{c}\right)  \right] \nonumber\\
&  \geq\Pr\left(  \mathcal{T}_{p^{\ast}}|\mathcal{A}_{p^{\ast}}\cap
\mathcal{H}_{p^{\ast}}\cap\mathcal{D}_{p^{\ast}}\right)  \left[  1-\Pr\left(
\mathcal{A}_{p^{\ast}}^{c}\cup\mathcal{H}_{p^{\ast}}^{c}\right)  -\Pr\left(
\mathcal{D}_{p^{\ast}}^{c}\right)  \right]  . \label{EQ:tp*}%
\end{align}
Note that $\mathcal{T}_{p^{\ast}}^{c}|\left\{  \mathcal{A}_{p^{\ast}}%
\cap\mathcal{H}_{p^{\ast}}\cap\mathcal{D}_{p^{\ast}}\right\}  $ is the event
that the OCMT procedure does not stop after $p^{\ast}$ stages and all signals
and pseudo-signals have been selected. It is equivalent to one or more noise
variables being selected at stage $p^{\ast}+1$ conditional on $\mathcal{A}%
_{p^{\ast}}\cap\mathcal{H}_{p^{\ast}}\cap\mathcal{D}_{p^{\ast}}$, which is
$\cup_{l=p^{\ast}+p^{\ast\ast}+1}^{p_{n}}\mathcal{B}_{l,p^{\ast}+1}|\left\{
\mathcal{A}_{p^{\ast}}\cap\mathcal{H}_{p^{\ast}}\cap\mathcal{D}_{p^{\ast}%
}\right\}  $. By the analysis in Lemma \ref{LE:allsignals},%
\begin{align}
\Pr\left(  \mathcal{T}_{p^{\ast}}|\mathcal{A}_{p^{\ast}}\cap\mathcal{H}%
_{p^{\ast}}\cap\mathcal{D}_{p^{\ast}}\right)   &  =1-\Pr\left(  \cup
_{l=p^{\ast}+p^{\ast\ast}+1}^{p_{n}}\mathcal{B}_{l,p^{\ast}+1}|\left\{
\mathcal{A}_{p^{\ast}}\cap\mathcal{H}_{p^{\ast}}\cap\mathcal{D}_{p^{\ast}%
}\right\}  \right) \nonumber\\
&  \geq1-\sum_{l=p^{\ast}+p^{\ast\ast}+1}^{p_{n}}\Pr\left(  \mathcal{B}%
_{l,p^{\ast}+1}|\left\{  \mathcal{A}_{p^{\ast}}\cap\mathcal{H}_{p^{\ast}}%
\cap\mathcal{D}_{p^{\ast}}\right\}  \right) \nonumber\\
&  \geq1-p_{n}n^{-M}-p_{n}C_{1}\exp\left(  -C_{2}n^{C_{3}}\right) \nonumber\\
&  \geq1-n^{-M_{5}}-C_{16}\exp\left(  -C_{17}n^{C_{18}}\right)  ,
\label{EQ:tpc}%
\end{align}
where the last line holds for some arbitrarily large positive number $M_{5}$
and some positive constants $C_{16},$ $C_{17}$ and $C_{18}$ by the fact that
$p_{n}\propto n^{B_{p}}$\ in Assumption \ref*{A:p}'. Substituting equations
(\ref{EQ:aphpf}) and (\ref{EQ:tpc}) into equation (\ref{EQ:tp*}) and applying
Lemma \ref{LE:allsignals} on $\Pr\left(  \mathcal{D}_{p^{\ast}}^{c}\right)  $,
we obtain%
\begin{align}
\Pr\left(  \mathcal{T}_{p^{\ast}}\right)   &  \geq\left[  1-n^{-M_{5}}%
-C_{16}\exp(-C_{17}n^{C_{18}})\right] \nonumber\\
&  \text{ \ \ \ }\times\left[  1-n^{-M_{4}}-C_{13}\exp(-C_{14}n^{C_{15}%
})-n^{-M}-C_{1}\exp(-C_{2}n^{C_{3}})\right] \nonumber\\
&  \geq1-n^{-M_{6}}-C_{19}\exp(-C_{20}n^{C_{21}}), \label{EQ:tpcf}%
\end{align}
for some arbitrarily large positive number $M_{6}$ and some positive constants
$C_{19},$ $C_{20},$ and $C_{21}.$ Consequently,%
\[
\Pr\left(  \hat{k}_{s}>p^{\ast}\right)  =\Pr\left(  \mathcal{T}_{p^{\ast}}%
^{c}\right)  =1-\Pr\left(  \mathcal{T}_{p^{\ast}}\right)  .
\]

\medskip

The second part of the theorem is quite straightforward given the analysis so
far. First, we show the result for TPR$_{n}$. We conduct the analysis
conditional on $\mathcal{D}_{p^{\ast}}$ that all selected variables up to
stage $p^{\ast}$ are either signals or pseudo-signals. In the first stage, by
Lemma \ref{LE:hidden}, there exists at least one signal with $\theta
_{l}\gtrsim\kappa_{n}\log\left(  m_{n}\right)  ^{1/2}m_{n}^{1/2}n^{-1/2}$. Let
$\Psi_{\left(  1\right)  }^{0}=\left\{  l:1\leq l\leq p^{\ast}\text{ and
}\theta_{l}\gtrsim\kappa_{n}\log\left(  m_{n}\right)  ^{1/2}m_{n}%
^{1/2}n^{-1/2}\right\}  $, the collection of signals that have large effects
on $Y$ in the first stage. By the same logic as we have used for the proof to
the first part of the theorem,
\begin{align*}
\Pr\left(  \cap_{l\in\Psi_{\left(  1\right)  }^{0}}\left\{
\widehat{\mathcal{J}}_{l,\left(  1\right)  }=1\right\}  \right)   &
=1-\Pr\left(  \cup_{l\in\Psi_{\left(  1\right)  }^{0}}\left\{
\widehat{\mathcal{J}}_{l,\left(  1\right)  }=0\right\}  \right) \\
&  \geq1-n^{-M}-C_{1}\exp\left(  -C_{2}n^{C_{3}}\right)
\end{align*}
for some large positive constant $M$ and some positive constants $C_{1},C_{2}$
and $C_{3},$ where we use the fact that $p^{\ast}$ is a fixed number and the
number of elements in $\Psi_{\left(  1\right)  }^{0}$ is $p^{\ast}$ at most.
Conditional on $\mathcal{D}_{p^{\ast}},$ denote the pre-selected variables as
$\boldsymbol{Z}_{1}$ and the index set of $\boldsymbol{Z}_{1}$ as $S_{\left(
1\right)  }$\ for stage 2. Note that there may be some hidden signals
accidentally selected in stage 1. We include all those in $\boldsymbol{Z}_{1}$
too. As long as we conduct the analysis conditional on $\mathcal{D}_{p^{\ast}%
}$ that no noise variables are selected, we can proceed the analysis as
usual.\ Let $\Psi_{\left(  2\right)  }^{0}=\left\{  l:1\leq l\leq p^{\ast
},\text{ }l\notin S_{\left(  1\right)  }\text{\ and }\theta_{l,\boldsymbol{Z}%
_{1}}\gtrsim\kappa_{n}\log\left(  m_{n}\right)  ^{1/2}m_{n}^{1/2}%
n^{-1/2}\right\}  ,$ the set of signals that have large effects on $Y$ with
pre-selected variable $\boldsymbol{Z}_{1}.$ By Lemma \ref{LE:hidden},
$\Psi_{\left(  2\right)  }^{0}$ is not empty as long as $S_{\left(  1\right)
}\neq\left\{  1,2,\ldots,p^{\ast}\right\}  $ that there are some hidden
signals. For the same reason,%
\[
\Pr\left(  \left.  \cap_{l\in\Psi_{\left(  2\right)  }^{0}}\left\{
\widehat{\mathcal{J}}_{l,\left(  2\right)  }=1\right\}  \right\vert \cap
_{l\in\Psi_{\left(  1\right)  }^{0}}\left\{  \widehat{\mathcal{J}}_{l,\left(
1\right)  }=1\right\}  ,\mathcal{D}_{p^{\ast}}\right)  \geq1-n^{-M}-C_{1}%
\exp\left(  -C_{2}n^{C_{3}}\right)  .
\]
So on and so forth. Suppose there are some remaining signals as $k_{s}^{\ast}$
in the last stage. Because we only have $p^{\ast}$ signals, $k_{s}^{\ast}\leq
p^{\ast}$. We similarly use $\Psi_{\left(  k\right)  }^{0}$ denote set of the
signals that have large effects on $Y$ with pre-selected variables
$\boldsymbol{Z}_{k-1}$ and conditional on $\mathcal{D}_{p^{\ast}}.$ Then%
\[
\Pr\left(  \left.  \cap_{l\in\Psi_{\left(  k_{s}^{\ast}\right)  }^{0}}\left\{
\widehat{\mathcal{J}}_{l,\left(  k_{s}^{\ast}\right)  }=1\right\}  \right\vert
\cap_{k=1}^{k_{s}^{\ast}-1}\cap_{l\in\Psi_{\left(  k\right)  }^{0}}\left\{
\widehat{\mathcal{J}}_{l,\left(  k\right)  }=1\right\}  ,\mathcal{D}_{p^{\ast
}}\right)  \geq1-n^{-M}-C_{1}\exp\left(  -C_{2}n^{C_{3}}\right)  .
\]
Further,%
\begin{align*}
\Pr\left(  \cap_{k=1}^{k_{s}^{\ast}}\cap_{l\in\Psi_{\left(  k\right)  }^{0}%
}\left\{  \widehat{\mathcal{J}}_{l,\left(  k\right)  }=1\right\}  \right)   &
\geq\Pr\left(  \cap_{k=1}^{k_{s}^{\ast}}\cap_{l\in\Psi_{\left(  k\right)
}^{0}}\left\{  \widehat{\mathcal{J}}_{l,\left(  k\right)  }=1\right\}
\cap\mathcal{D}_{p^{\ast}}\right) \\
&  =\Pr\left(  \left.  \cap_{k=1}^{k_{s}^{\ast}}\cap_{l\in\Psi_{\left(
k\right)  }^{0}}\left\{  \widehat{\mathcal{J}}_{l,\left(  k\right)
}=1\right\}  \right\vert \mathcal{D}_{p^{\ast}}\right)  \Pr\left(
\mathcal{D}_{p^{\ast}}\right) \\
&  =\Pr\left(  \left.  \cap_{l\in\Psi_{\left(  k_{s}^{\ast}\right)  }^{0}%
}\left\{  \widehat{\mathcal{J}}_{l,\left(  k_{s}^{\ast}\right)  }=1\right\}
\right\vert \cap_{k=1}^{k_{s}^{\ast}-1}\cap_{l\in\Psi_{\left(  k\right)  }%
^{0}}\left\{  \widehat{\mathcal{J}}_{l,\left(  k\right)  }=1\right\}
,\mathcal{D}_{p^{\ast}}\right) \\
&  \times\Pr\left(  \left.  \cap_{l\in\Psi_{\left(  k_{s}^{\ast}-1\right)
}^{0}}\left\{  \widehat{\mathcal{J}}_{l,\left(  k_{s}^{\ast}-1\right)
}=1\right\}  \right\vert \cap_{k=1}^{k_{s}^{\ast}-2}\cap_{l\in\Psi_{\left(
k\right)  }^{0}}\left\{  \widehat{\mathcal{J}}_{l,\left(  k\right)
}=1\right\}  ,\mathcal{D}_{p^{\ast}}\right) \\
&  \times\cdots\\
&  \times\Pr\left(  \left.  \cap_{l\in\Psi_{\left(  1\right)  }^{0}}\left\{
\widehat{\mathcal{J}}_{l,\left(  1\right)  }=1\right\}  \right\vert
\mathcal{D}_{p^{\ast}}\right)  \Pr\left(  \mathcal{D}_{p^{\ast}}\right) \\
&  \geq\left[  1-n^{-M}-C_{1}\exp\left(  -C_{2}n^{C_{3}}\right)  \right]
^{k_{s}^{\ast}+1}\\
&  \geq1-\left(  k_{s}^{\ast}+1\right)  n^{-M}-\left(  k_{s}^{\ast}+1\right)
C_{1}\exp\left(  -C_{2}n^{C_{3}}\right) \\
&  \geq1-n^{-M_{7}}-C_{22}\exp\left(  -C_{23}n^{C_{24}}\right)
\end{align*}
for some large positive constant $M_{7}$ and some positive constants
$C_{22},C_{23},$ and $C_{24},$ where we use the inequalities developed right
before this equation and Lemma \ref{LE:allsignals} to obtain the second
inequality.\ Consequently, we have%
\begin{align*}
E\left(  \text{TPR}_{n}\right)   &  =p^{\ast-1}\sum_{l=1}^{p_{n}}E\left\{
1\left(  \widehat{\mathcal{J}}_{l}=1\text{ and }\left\{  E\left[  f_{l}^{\ast
}\left(  X_{l}\right)  ^{2}\right]  \right\}  ^{1/2}\neq0\right)  \right\} \\
&  =p^{\ast-1}\sum_{l=1}^{p^{\ast}}\Pr\left(  \widehat{\mathcal{J}}%
_{l}=1\right) \\
&  \geq p^{\ast-1}\sum_{l=1}^{p^{\ast}}\Pr\left(  \left\{
\widehat{\mathcal{J}}_{l}=1\right\}  \cap\left\{  \cap_{k=1}^{k_{s}^{\ast}%
}\cap_{l\in\Psi_{\left(  k\right)  }^{0}}\left\{  \widehat{\mathcal{J}%
}_{l,\left(  k\right)  }=1\right\}  \right\}  \right) \\
&  =p^{\ast-1}\sum_{l=1}^{p^{\ast}}\Pr\left(  \cap_{k=1}^{k_{s}^{\ast}}%
\cap_{l\in\Psi_{\left(  k\right)  }^{0}}\left\{  \widehat{\mathcal{J}%
}_{l,\left(  k\right)  }=1\right\}  \right) \\
&  \geq1-n^{-M_{7}}-C_{22}\exp\left(  -C_{23}n^{C_{24}}\right)  ,
\end{align*}
where the fourth line follows from the fact that
\[
\left\{  \widehat{\mathcal{J}}_{l}=1\right\}  \cap\left\{  \cap_{k=1}%
^{k_{s}^{\ast}}\cap_{l\in\Psi_{\left(  k\right)  }^{0}}\left\{
\widehat{\mathcal{J}}_{l,\left(  k\right)  }=1\right\}  \right\}  =\left\{
\cap_{k=1}^{k_{s}^{\ast}}\cap_{l\in\Psi_{\left(  k\right)  }^{0}}\left\{
\widehat{\mathcal{J}}_{l,\left(  k\right)  }=1\right\}  \right\}
\]
for $1\leq l\leq p^{\ast}$. The above equality holds by the way we define
$k_{s}^{\ast}$\ to ensure the effect of each signal on $Y$ becomes large
enough at a certain stage.

Next, we turn to FPR$_{n}.$%
\begin{align*}
E\left(  \text{FPR}_{n}\right)   &  =\left(  p_{n}-p^{\ast}\right)  ^{-1}%
\sum_{l=p^{\ast}+1}^{p_{n}}E\left\{  1\left(  \widehat{\mathcal{J}}%
_{l}=1\text{ and }\left\{  E\left[  f_{l}^{\ast}\left(  X_{l}\right)
^{2}\right]  \right\}  ^{1/2}=0\right)  \right\} \\
&  =\left(  p_{n}-p^{\ast}\right)  ^{-1}\sum_{l=p^{\ast}+1}^{p^{\ast}%
+p^{\ast\ast}}E\left[  1\left(  \widehat{\mathcal{J}}_{l}=1\right)  \right]
+\left(  p_{n}-p^{\ast}\right)  ^{-1}\sum_{l=p^{\ast}+p^{\ast\ast}+1}^{n}%
\Pr\left(  \widehat{\mathcal{J}}_{l}=1\right) \\
&  \leq\frac{p^{\ast\ast}}{p_{n}-p^{\ast}}+\left(  p_{n}-p^{\ast}\right)
^{-1}\sum_{l=p^{\ast}+p^{\ast\ast}+1}^{p_{n}}\left[  \Pr\left(
(\widehat{\mathcal{J}}_{l}=1)\cap\mathcal{T}_{p^{\ast}}\right)  +\Pr\left(
(\widehat{\mathcal{J}}_{l}=1)\cap\mathcal{T}_{p^{\ast}}^{c}\right)  \right] \\
&  \leq\frac{p^{\ast\ast}}{p_{n}-p^{\ast}}+\left(  p_{n}-p^{\ast}\right)
^{-1}\sum_{l=p^{\ast}+p^{\ast\ast}+1}^{p_{n}}\left[  \Pr\left(  \left.
(\widehat{\mathcal{J}}_{l}=1)\right\vert \mathcal{T}_{p^{\ast}}\right)
\Pr\left(  \mathcal{T}_{p^{\ast}}\right)  +\Pr\left(  \mathcal{T}_{p^{\ast}%
}^{c}\right)  \right] \\
&  \leq\frac{p^{\ast\ast}}{p_{n}-p^{\ast}}+\Pr\left(  \mathcal{N}_{p^{\ast}%
}|\mathcal{T}_{p^{\ast}}\right)  \Pr\left(  \mathcal{T}_{p^{\ast}}\right)
+\Pr\left(  \mathcal{T}_{p^{\ast}}^{c}\right) \\
&  =\frac{p^{\ast\ast}}{p_{n}-p^{\ast}}+\left[  1-\Pr\left(  \mathcal{D}%
_{p^{\ast}}|\mathcal{T}_{p^{\ast}}\right)  \right]  \Pr\left(  \mathcal{T}%
_{p^{\ast}}\right)  +\Pr\left(  \mathcal{T}_{p^{\ast}}^{c}\right) \\
&  \leq\frac{p^{\ast\ast}}{p_{n}-p^{\ast}}+n^{-M_{7}}+C_{22}\exp\left(
-C_{23}n^{C_{24}}\right)  ,
\end{align*}
for some large positive constant $M_{7}$ and positive constants $C_{22}%
,C_{23},$ and $C_{24},$ where the fifth line holds by the fact that
conditioning on $\mathcal{T}_{p^{\ast}},$ $\cup_{l=p^{\ast}+p^{\ast\ast}%
+1}^{p_{n}}\left\{  \widehat{\mathcal{J}}_{l}=1\right\}  \subseteq
\mathcal{N}_{p^{\ast}}$ (conditional on the OCMT procedure stops before stage
$p^{\ast},$ whether one or more noise variables selected by OCMT is a subset
of whether one or more noise variables selected before stage $p^{\ast}$), and
the last inequality holds by applying Lemma \ref{LE:allsignals} and equation
(\ref{EQ:tpc}).

Now, we turn to FDR$_{n}.$ By the same analysis for FPR$_{n},$ we have%
\begin{align*}
&  E\left[  \sum_{l=1}^{p_{n}}1\left(  \widehat{\mathcal{J}}_{l}=1,\text{
}\left\{  E\left[  f_{l}^{\ast}\left(  X_{l}\right)  ^{2}\right]  \right\}
^{1/2}=0,\text{ and }\theta_{l}\lesssim\log\left(  m_{n}\right)  ^{1/2}\left(
m_{n}/n\right)  ^{1/2}\right)  \right] \\
&  =\sum_{l=p^{\ast}+p^{\ast\ast}+1}^{p_{n}}E\left(  \widehat{\mathcal{J}}%
_{l}=1\right) \\
&  \leq\sum_{l=p^{\ast}+p^{\ast\ast}+1}^{p_{n}}\left[  \Pr\left(
(\widehat{\mathcal{J}}_{l}=1)\cap\mathcal{T}_{p^{\ast}}\right)  +\Pr\left(
(\widehat{\mathcal{J}}_{l}=1)\cap\mathcal{T}_{p^{\ast}}^{c}\right)  \right] \\
&  \leq n^{-M_{7}}+C_{22}\exp\left(  -C_{23}n^{C_{24}}\right)  =o(1).
\end{align*}
Then%
\[
\text{FDR}_{n}=\frac{\sum_{l=1}^{p_{n}}1\left(  \widehat{\mathcal{J}}%
_{l}=1,\text{ }\left\{  E\left[  f_{l}^{\ast}\left(  X_{l}\right)
^{2}\right]  \right\}  ^{1/2}=0,\text{ and }\theta_{l}\lesssim\log\left(
m_{n}\right)  ^{1/2}\left(  m_{n}/n\right)  ^{1/2}\right)  }{\sum_{l=1}%
^{p_{n}}\widehat{\mathcal{J}}_{l}+1}\overset{P}{\rightarrow}0.
\]
by Markov inequality and the fact that $\sum_{l=1}^{p_{n}}\widehat{\mathcal{J}%
}_{l}+1\geq1$.\medskip
\end{proof}

\begin{proof}
[Proof of Theorem \ref{TH:AGLasso}]For (i), we denote the event of the
desirable results from the OCMT procedure as
\[
\mathcal{M}_{\text{OCMT}}=\mathcal{T}_{p^{\ast}}\cap\mathcal{A}_{p^{\ast}}%
\cap\mathcal{H}_{p^{\ast}}\cap\mathcal{D}_{p^{\ast}},
\]
which is the event that the OCMT concludes at or before stage $p^{\ast},\ $all
the signals and pseudo signals are selected, and none of the noise variables
are included.$\ $By the proof of Theorem \ref{TH:stoppingFDR} (see esp.
equations (\ref{EQ:aphpf}) and (\ref{EQ:tpc})) and by applying Lemma
\ref{LE:allsignals} on $\Pr\left(  \mathcal{D}_{p^{\ast}}^{c}\right)  ,$ we
have%
\begin{equation}
\Pr\left(  \mathcal{M}_{\text{OCMT}}\right)  =1-o\left(  1\right)  .
\label{EQ:MOCMT}%
\end{equation}
Conditional on $\mathcal{M}_{\text{OCMT}},$ we claim that all the technical
conditions in Theorem 3 of \cite{HuangEtal2010} (HHW) are satisfied. For a
better exposition, we defer the proof of this claim to the end of the proof.

Now, we denote the desirable event of the adaptive group Lasso as%
\[
\mathcal{M}_{\text{AGLASSO}}=\left\{  \text{All signals are selected, no
pseudo-signals or noise variables are selected}\right\}  \text{.}%
\]
Theorem 3 of HHW implies that
\begin{equation}
\Pr\left(  \left.  \mathcal{M}_{\text{AGLASSO}}\right\vert \mathcal{M}%
_{\text{OCMT}}\right)  =1-o\left(  1\right)  . \label{EQ:MAGLasso}%
\end{equation}
Consequently we have%
\begin{align*}
\Pr\left(  \mathcal{M}_{\text{AGLASSO}}\right)   &  =\Pr\left(  \left.
\mathcal{M}_{\text{AGLASSO}}\right\vert \mathcal{M}_{\text{OCMT}}\right)
\Pr\left(  \mathcal{M}_{\text{OCMT}}\right)  +\Pr\left(  \left.
\mathcal{M}_{\text{AGLASSO}}\right\vert \mathcal{M}_{\text{OCMT}}^{c}\right)
\Pr\left(  \mathcal{M}_{\text{OCMT}}^{c}\right) \\
&  \geq\Pr\left(  \left.  \mathcal{M}_{\text{AGLASSO}}\right\vert
\mathcal{M}_{\text{OCMT}}\right)  \Pr\left(  \mathcal{M}_{\text{OCMT}}\right)
+o(1)\\
&  =1-o\left(  1\right)  ,
\end{align*}
where we use the results in equations (\ref{EQ:MOCMT}) and (\ref{EQ:MAGLasso}).

We turn to (ii). By the result in (i), $\Pr\left(  \mathcal{M}_{\text{AGLASSO}%
}\right)  >1-\varepsilon/2$ for any fixed small positive value $\varepsilon$
after some large $n$. Conditional on $\mathcal{M}_{\text{AGLASSO}},$ the post
OCMT estimation is simply a special case of \cite{Stone} because $p^{\ast}$ is
fixed. Further, the bias term is of the order $m_{n}^{-d}$. Since $m_{n}%
^{-d}\ll\left(  m_{n}/n\right)  ^{1/2}$ by Assumption \ref{A:mn}. That is, the
bias term is asymptotically negligible. Then the results in \cite{Stone}, we
have conditional on $\mathcal{M}_{\text{AGLASSO}},$%
\begin{equation}
P^{m_{n}}\left(  \boldsymbol{Z}_{\text{AGLASSO}}\right)  ^{\prime
}\boldsymbol{\hat{\beta}}_{\text{post}}-%
{\displaystyle\sum_{j=1}^{p^{\ast}}}
f_{j}^{\ast}\left(  X_{j}\right)  =O_{P}\left(  \left(  m_{n}/n\right)
^{1/2}\right)  . \label{EQ:state_Iwant}%
\end{equation}
Then unconditionally, we have $P^{m_{n}}\left(  \boldsymbol{Z}_{\text{AGLASSO}%
}\right)  ^{\prime}\boldsymbol{\hat{\beta}}_{\text{post}}-%
{\displaystyle\sum_{j=1}^{p^{\ast}}}
f_{j}^{\ast}\left(  X_{j}\right)  =O_{P}\left(  \left(  m_{n}/n\right)
^{1/2}\right)  $ as we can make $\varepsilon$ arbitrarily small.\footnote{This
holds because for any events $A$ and $B\ $with $\Pr\left(  B\right)
>1-\varepsilon/2$ and $\Pr\left(  A|B\right)  >1-\varepsilon/2,$ we have%
\[
\Pr\left(  A\right)  =\Pr\left(  A|B\right)  \Pr\left(  B\right)  +\Pr\left(
A|B^{c}\right)  \Pr\left(  B^{c}\right)  \geq\Pr\left(  A|B\right)  \Pr\left(
B\right)  >1-\varepsilon.
\]
}.

We complete the proof by demonstrating the initial claim. To achieve this, we
need to verify the conditions for the tuning parameters and conditions A1--A4
as specified in HHW.

To facilitate reading, we present A1--A4 in HHW using our notation.

\begin{enumerate}
\item[A1] The number of nonzero components $p^{\ast}$ is fixed, and there is a
constant $C_{f}>0$ such that $\min_{1\leq j\leq p^{\ast}}\left\{  E\left[
f_{j}^{\ast}\left(  X_{j}\right)  ^{2}\right]  \right\}  ^{1/2}\geq C_{f}$.

\item[A2] The random variables $\varepsilon_{1},\ldots,\varepsilon_{n}$ are
independent and identically distributed with $E\left(  \varepsilon_{i}\right)
=0$ and $\operatorname{Var}\left(  \varepsilon_{i}\right)  =\sigma^{2}$.
Furthermore, their tail probabilities satisfy $\Pr\left(  \left\vert
\varepsilon_{i}\right\vert >t\right)  \leq K\exp\left(  -Ct^{2}\right)  ,$
$i=1,\ldots,n$, for all $t\geq0$ and for constants $K$ and $C$.

\item[A3] $Ef_{j}^{\ast}\left(  X_{j}\right)  =0,$ $j=1,\ldots,p^{\ast}$.

\item[A4] The covariate vector $X$ has a continuous density and there exist
constants $C_{1}$ and $C_{2}$ such that the density function $g_{j}$ of
$X_{j}$ satisfies $0<C_{1}\leq g_{j}(x)\leq C_{2}<$ $\infty$ on $[0,1]$ for
every $1\leq j\leq p_{n}.$
\end{enumerate}

We verify those conditions as follows.

\begin{itemize}
\item First, conditional on $\mathcal{M}_{\text{OCMT}}$, all signals and
pseudo signals are selected, while none of the noise variables are included.
Consequently, the regular rank condition, as imposed in Assumption
\ref{A:full_rank2}, holds. Furthermore, the number of covariates is $p^{\ast
}+p^{\ast\ast}$, which is of the same order as $p^{\ast\ast}$ when
$p^{\ast\ast}>0$ because $p^{\ast}$ is fixed.

\item Second, we discuss the tuning parameters. The condition $\lambda
_{n1}\geq C\sqrt{n\log(p^{\ast\ast}m_{n})}$ was directly imposed in Theorem 1
of HHW. The condition $\lambda_{n1}\ll\sqrt{n/m_{n}}$ is stronger than the one
imposed in part (ii) of Theorem 1 of HHW. Thus, $\lambda_{n1}$ satisfies the
requirements in HHW. For $\lambda_{n2}$, we only need to verify whether it
satisfies condition (B2) in HHW. The condition $\lambda_{n2}\ll nm_{n}^{-1/4}$
clearly satisfies B2(a) in HHW. For B2(b) in HHW, firstly, $r_{n}\propto
\sqrt{n/[m_{n}\log(p^{\ast\ast}m_{n})]}$ in our case (the convergence rate
ensured by Theorem 1 of HHW), because $\lambda_{n1}\ll\sqrt{n/m_{n}}$ and the
bias term $m_{n}^{-d}\ll\sqrt{m_{n}/n}$ as ensured by Assumption \ref{A:mn}.
Using the rate of $r_{n}$, some simple calculations show that $\lambda_{n2}\gg
m_{n}^{1/2}\log(p^{\ast\ast}m_{n})$ satisfies B2(b), again using Assumption
\ref{A:mn}. We have verified the conditions required for the tuning parameters.

\item Third, condition A3 in HHW is merely a normalization, and we conduct
such a normalization as well. The major parts of conditions A1, A2, and A4 in
HHW have been imposed in this paper except for two conditions:

\begin{enumerate}
\item HHW impose $\min_{1\leq j\leq p^{\ast}}\left\Vert f_{j}^{\ast
}\right\Vert \geq C>0$, while our counterpart is $\left\{  E[f_{j}^{\ast
}(X_{j})^{2}]\right\}  ^{1/2}\gtrsim\kappa_{n}\log(m_{n})^{1/2}(m_{n}%
/n)^{1/2}$ for $j=1,2,...,p^{\ast}$ as in Assumption 9'.

\item HHW assume $\Pr(\left\vert \varepsilon_{i}\right\vert >t)\leq
K\exp(-Ct^{2})$ for some positive constants $K$ and $C$ and for all $t$, while
we assume $\Pr(\left\vert \varepsilon_{i}\right\vert >t)\leq K\exp(-Ct^{s})$
for some constant $s>0$ in Assumption \ref{A:epsilon}. Clearly, our conditions
are weaker.
\end{enumerate}

We are going to show that these two conditions have no impact on the final results.

\begin{enumerate}
\item[For 1.] As stated in the second bullet point, the rate of convergence of
the Lasso estimator is $r_{n}^{-1}\propto\sqrt{m_{n}\log(p^{\ast\ast}m_{n}%
)/n}$. Note that $\log(p^{\ast\ast}m_{n})\propto\log(m_{n})$ under our
framework, because $m_{n}\propto n^{B_{m}}$ and $p^{\ast\ast}\propto
n^{B_{p^{\ast\ast}}}$. Therefore, $\{E[f_{j}^{\ast}(X_{j})^{2}]\}^{1/2}\gg
r_{n}^{-1}$ for $j=1,2,...,p^{\ast}$, due to the fact that $\kappa
_{n}\rightarrow\infty$. This result implies that the Lasso estimator of
signals will not be zero with very high probability because the signal
strength is much greater than the convergence rate. Therefore, this weaker
condition has no impact on the Lasso shrinkage result.

\item[For 2.] Our tail condition slightly generalizes the one in HHW. The
major use of this condition in HHW is the probability bounds derived in Lemma
2 of HHW. In particular, only the second part of this lemma is useful for us,
viz., the result when $\frac{m_{n}\log(pm_{m})}{n}\rightarrow0$, the case in
this paper. To guarantee this result for our case, we can apply Lemma A.2 in
this paper by setting $v_{n}=Cn^{1/2}m_{m}^{-1/2}\sqrt{\log(p_{n}m_{n})}$ for
a sufficiently large constant $C$ on the series $\sum_{i=1}^{n}\phi_{j}%
(x_{ki})\varepsilon_{i}$ for $j=1,...,m_{n}$ and $k=1,...,p_{n}$. The
probability bound from this lemma is $\exp\left[  -\frac{C^{2}(1-\pi)^{2}%
\log(p_{n}m_{n})}{2C_{1}}\right]  $ for another uniformly bounded $C_{1}$, for
$j=1,...,m_{n}$ and $k=1,...,p_{n}$. We set $C$ large enough so that the
probability bound is smaller than $(p_{n}m_{n})^{-M}$ for some $M>1$. Then
\begin{align*}
\Pr\left(  \max_{j=1,...,m_{n},k=1,...,p_{n}}\left\vert \sum_{i=1}^{n}\phi
_{j}(x_{ki})\varepsilon_{i}\right\vert >v_{n}\right)   &  \leq\sum
_{j=1}^{m_{n}}\sum_{k=1}^{p_{n}}\Pr\left(  \left\vert \sum_{i=1}^{n}\phi
_{j}(x_{ki})\varepsilon_{i}\right\vert >v_{n}\right) \\
&  \leq m_{n}p_{n}(p_{n}m_{n})^{-M}\rightarrow0,
\end{align*}
which implies $\max_{j=1,...,m_{n},k=1,...,p_{n}}\left\vert \sum_{i=1}^{n}%
\phi_{j}(x_{ki})\varepsilon_{i}\right\vert =O_{P}\left[  n^{1/2}m_{m}%
^{-1/2}\sqrt{\log(p_{n}m_{n})}\right]  $, an equivalent result to the second
part of Lemma 2 in HHW.
\end{enumerate}
\end{itemize}
\end{proof}

\begin{proof}
[Proof of Theorem \ref{TH:stoppingFDR2}]The uniform boundedness of $\sum
_{j=1}^{p^{\ast}}f_{j}\left(  X_{j}\right)  $ implies that $U_{l}$ also
satisfies the tail condition in Assumption \ref*{A:epsilon}. [For details, see
the proof of Lemma \ref{LE:error_var}.] This, in conjunction with the
additional conditions in Assumption \ref*{A:p}\textquotedblright, implies that
all technical lemmas can go through. Consequently, we only need to show that
the probability bounds for the events like $\mathcal{T}_{p^{\ast}}%
,\mathcal{A}_{p^{\ast}},\mathcal{H}_{p^{\ast}},$ and $\mathcal{D}_{p^{\ast}}$
are still $1-o\left(  1\right)  $ with a diverging $p^{\ast}.$ This statement
holds due to the fact that the error bounds obtained before are of the order
either $n^{-M}$ or $\exp\left(  -n^{C}\right)  $ for some large positive
constant $M$ and some $C>0,$ and the error probabilities accumulated for a
diverging $p^{\ast}$\ are of the order $p^{\ast}\left(  n^{-M}+\exp\left(
-n^{C}\right)  \right)  =o\left(  1\right)  $. We provide some details below.

We continue to use $M$ to denote some large positive constant and $C$ to
denote some generic positive constant. For $\mathcal{H}_{p^{\ast}},$ equations
(\ref{EQ:h1}) and (\ref{EQ:hk}) imply that
\begin{align}
\Pr\left(  \mathcal{H}_{p^{\ast}}\right)   &  \geq\Pr\left(  \mathcal{H}%
_{1}\right)  =\Pr\left(  \cap_{l=p^{\ast}+1}^{p^{\ast}+p^{\ast\ast}%
}\mathcal{B}_{l,1}\right)  =1-\Pr\left(  \cup_{l=p^{\ast}+1}^{p^{\ast}%
+p^{\ast\ast}}\mathcal{B}_{l,1}^{c}\right)  \geq1-\sum_{l=p^{\ast}+1}%
^{p^{\ast}+p^{\ast\ast}}\Pr\left(  \mathcal{B}_{l,1}^{c}\right) \nonumber\\
&  \geq1-p^{\ast\ast}\left(  n^{-M}+C_{1}\exp\left(  -C_{2}n^{C_{3}}\right)
\right)  \geq1-n^{-M_{1}}-C_{4}\exp\left(  -C_{5}n^{C_{6}}\right)  .\nonumber
\end{align}
For $\mathcal{D}_{p^{\ast}},$\
\[
\Pr\left(  \mathcal{D}_{p^{\ast}}\right)  \geq1-n^{-M_{2}}-C_{7}\exp\left(
-C_{8}n^{C_{9}}\right)
\]
holds by Lemma \ref{LE:allsignals} because $p^{\ast}\lesssim n^{B_{p^{\ast}}}%
$. Similarly,
\[
\Pr\left(  \mathcal{A}_{p^{\ast}}\cap\mathcal{H}_{p^{\ast}}\right)
\geq1-n^{-M_{3}}-C_{10}\exp\left(  -C_{11}n^{C_{12}}\right)
\]
for the same reason as we obtain equation (\ref{EQ:aphpf}) and $p^{\ast
}\lesssim n^{B_{p^{\ast}}}$. The probability bound for $\mathcal{T}_{p^{\ast}%
}$ was derived based on the bounds for $\mathcal{A}_{p^{\ast}},\mathcal{H}%
_{p^{\ast}},$ and $\mathcal{D}_{p^{\ast}}.$ Since the probability bounds for
$\mathcal{A}_{p^{\ast}},\mathcal{H}_{p^{\ast}},$ and $\mathcal{D}_{p^{\ast}}$
have been shown, the probability bound for $\mathcal{T}_{p^{\ast}}$ can be
obtained using the same idea as we get equation (\ref{EQ:tpcf}). That is,%
\[
\Pr\left(  \mathcal{T}_{p^{\ast}}\right)  \geq1-n^{-M_{4}}-C_{13}\exp\left(
-C_{14}n^{C_{15}}\right)  .
\]
Then the first result is reached by observing that $\Pr\left(  \hat{k}%
_{s}>p^{\ast}\right)  =1-\Pr\left(  \mathcal{T}_{p^{\ast}}\right)  .$ The
results of TPR, FPR, and FDR can be proved similarly as in the proof of
Theorem \ref{TH:stoppingFDR} and thus omitted.
\end{proof}

%

%

\newpage\appendix\setcounter{footnote}{0} \setcounter{table}{0}
\setcounter{figure}{0} \setcounter{section}{0}
\renewcommand{\thesection}{S\arabic{section}}
\numberwithin{equation}{section}
\renewcommand{\thefigure}{S\arabic{figure}}
\renewcommand{\thetable}{S\arabic{table}}
\renewcommand{\thelemma}{S\arabic{lemma}}
\setcounter{page}{1}%
%

\setlength{\baselineskip}{15pt}%

\begin{center}
{\Large Online Appendix to}$\vspace{0.08in}$

{\Large \textquotedblleft A One Covariate at a Time Multiple Testing Approach
to Variable Selection in Additive Models\textquotedblright}

(NOT for Publication)

$\medskip$

Liangjun Su\footnote{School of Economics and Management, Tsinghua University},
Thomas Tao Yang\footnote{Research School of Economics, Australian National
University}, Yonghui Zhang\footnote{School of Economics, Renmin University of
China}, and Qiankun Zhou\footnote{Department of Economics, Louisiana State
University}

{\small \ \ \ \ \ \ \ \ \ }
\end{center}

\noindent This\ supplement is composed of two parts. Section
\ref{APP:techproof} contains the proofs of Lemmas A.4-A.12. Section
\ref{APP:tables} provides some additional simulation\ and application results.

\section{Proofs of the Technical Lemmas\label{APP:techproof}}

\begin{proof}
[Proof of Lemma \ref{LE:XX'}]To prove the lemma, we first present two
inequalities: for any $m_{n}\times m_{n}$\ symmetric matrices $\boldsymbol{A}$
and $\boldsymbol{B}$, we have%
\begin{equation}
\max\left\{  \left\vert \lambda_{\min}\left(  \boldsymbol{A}\right)
\right\vert ,\left\vert \lambda_{\max}\left(  \boldsymbol{A}\right)
\right\vert \right\}  \leq m_{n}\left\Vert \boldsymbol{A}\right\Vert _{\infty
}, \label{EQ:lambda2}%
\end{equation}
and%
\begin{equation}
\left\vert \lambda_{\min}\left(  \boldsymbol{A}\right)  -\lambda_{\min}\left(
\boldsymbol{B}\right)  \right\vert \leq\max\left\{  \left\vert \lambda_{\min
}\left(  \boldsymbol{A-B}\right)  \right\vert ,\left\vert \lambda_{\min
}\left(  \boldsymbol{B-A}\right)  \right\vert \right\}  . \label{EQ:lambda1}%
\end{equation}
We will prove them at the end of this proof.

Recall that $\Phi_{X_{l}}=E\left[  P^{m_{n}}\left(  X_{l}\right)  P^{m_{n}%
}\left(  X_{l}\right)  ^{\prime}\right]  .$ Let $\Phi_{n,l}=n^{-1}%
\mathbb{X}_{l}^{\prime}\mathbb{X}_{l}.$ The $\left(  j,k\right)  $-th entry of
$\Phi_{n,l}-\Phi_{l}$ is $\xi_{jk,l}\equiv n^{-1}\sum_{i=1}^{n}\phi_{j}\left(
x_{li}\right)  \phi_{k}\left(  x_{li}\right)  -E\left[  \phi_{j}\left(
X_{l}\right)  \phi_{k}\left(  X_{l}\right)  \right]  .$ Notice that%
\[
\text{var}\left(  \phi_{j}\left(  X_{l}\right)  \phi_{k}\left(  X_{l}\right)
\right)  \leq E\left[  \phi_{j}\left(  X_{l}\right)  ^{2}\phi_{k}\left(
X_{l}\right)  ^{2}\right]  \leq E\left[  \phi_{j}\left(  X_{l}\right)
^{2}\right]  \leq B_{4}m_{n}^{-1},
\]
where the second inequality holds by $\left\Vert \phi_{k}\right\Vert _{\infty
}\leq1,$ and the last inequality holds by Lemma \ref{LE:rank}. By Lemma
\ref{LE:bernstein},
\[
\Pr\left(  \left\vert \xi_{jk,l}\right\vert \geq v_{n}/n\right)  \leq
2\exp\left\{  -v_{n}^{2}/\left[  2\left(  B_{4}nm_{n}^{-1}+v_{n}/3\right)
\right]  \right\}  .
\]
Then by the union bound,%
\begin{align}
\Pr\left(  \left\Vert \Phi_{n,l}-\Phi_{X_{l}}\right\Vert _{\infty}\geq
v_{n}/n\right)   &  =\Pr\left(  \max_{j,k=1,2,\ldots,m_{n}}\left\vert
\xi_{jk,l}\right\vert \geq v_{n}/n\right) \nonumber\\
&  \leq\sum_{j=1}^{m_{n}}\sum_{k=1}^{m_{n}}\Pr\left(  \left\vert \xi
_{jk,l}\right\vert \geq v_{n}/n\right) \nonumber\\
&  \leq2m_{n}^{2}\exp\left\{  -v_{n}^{2}/[2\left(  B_{4}nm_{n}^{-1}%
+v_{n}/3\right)  ]\right\}  . \label{EQ:bound1}%
\end{align}

Note that$\left\Vert \Phi_{n,l}^{-1}\right\Vert -\left\Vert \Phi_{X_{l}}%
^{-1}\right\Vert =\left[  \lambda_{\min}\left(  \Phi_{n,l}\right)  \right]
^{-1}-\left[  \lambda_{\min}\left(  \Phi_{X_{l}}\right)  \right]  ^{-1}.$ By
(\ref{EQ:lambda2}), (\ref{EQ:lambda1}), and (\ref{EQ:bound1}),
\begin{align}
&  \Pr\left(  \left\vert \lambda_{\min}\left(  \Phi_{n,l}\right)
-\lambda_{\min}\left(  \Phi_{X_{l}}\right)  \right\vert \geq m_{n}%
v_{n}/n\right) \nonumber\\
&  \leq\Pr\left(  \max\left\{  \left\vert \lambda_{\min}\left(  \Phi
_{n,l}-\Phi_{X_{l}}\right)  \right\vert ,\left\vert \lambda_{\min}\left(
\Phi_{X_{l}}-\Phi_{n,l}\right)  \right\vert \right\}  \geq m_{n}v_{n}/n\right)
\nonumber\\
&  \leq\Pr\left(  \left\Vert \Phi_{n,l}-\Phi_{X_{l}}\right\Vert _{\infty}\geq
v_{n}/n\right) \nonumber\\
&  \leq2m_{n}^{2}\exp\left\{  -v_{n}^{2}/\left[  2\left(  B_{4}nm_{n}%
^{-1}+v_{n}/3\right)  \right]  \right\}  . \label{EQ:upto}%
\end{align}
Set $v_{n}=nm_{n}^{-2}B_{1}/3.$ Then the above inequality becomes%
\begin{equation}
\Pr\left(  \left\vert \lambda_{\min}\left(  \Phi_{n,l}\right)  -\lambda_{\min
}\left(  \Phi_{X_{l}}\right)  \right\vert \geq B_{1}m_{n}^{-1}/3\right)
\leq2m_{n}^{2}\exp\left\{  -C_{1}nm_{n}^{-3}\right\}  \label{EQ:bound2}%
\end{equation}
for some constant $C_{1}$ because $nm_{n}^{-1}\gg v_{n}=nm_{n}^{-2}B_{1}/3.$
This, in conjunction with the fact that $B_{1}m_{n}^{-1}\leq\lambda_{\min
}\left(  \Phi_{l}\right)  \leq B_{2}m_{n}^{-1}\ $by Lemma \ref{LE:rank},
implies%
\begin{equation}
\Pr\left(  \left\vert \lambda_{\min}\left(  \Phi_{n,l}\right)  -\lambda_{\min
}\left(  \Phi_{X_{l}}\right)  \right\vert \geq\lambda_{\min}\left(
\Phi_{X_{l}}\right)  /3\right)  \leq2m_{n}^{2}\exp\left\{  -C_{1}nm_{n}%
^{-3}\right\}  . \label{EQ:bound3a}%
\end{equation}
Note that for two positive random variables $a$ and $b,$ $\left\vert
a-b\right\vert \geq b/2$ is equivalent to $\{a-3b/2\geq0$ or $a-b/2$
$\leq0\},$ which is equivalent to
\[
b^{-1}-a^{-1}\geq b^{-1}/3\text{ \ \ or \ \ }b^{-1}-a^{-1}\leq-b^{-1}.
\]
Therefore, $\left\{  \left\vert a-b\right\vert \geq b/2\right\}  $ implies
$\left\{  \left\vert b^{-1}-a^{-1}\right\vert \geq b^{-1}/3\right\}  ,$ and
\[
\Pr\left(  \left\{  \left\vert a-b\right\vert \geq b/2\right\}  \right)
\leq\Pr\left(  \left\vert b^{-1}-a^{-1}\right\vert \geq b^{-1}/3\right)  .
\]
Taking $a=\left\{  \lambda_{\min}\left(  \Phi_{n,l}\right)  \right\}  ^{-1}$
and $b=\left\{  \lambda_{\min}\left(  \Phi_{X_{l}}\right)  \right\}  ^{-1},$
we have
\begin{align*}
&  \Pr\left(  \left\vert \left[  \lambda_{\min}\left(  \Phi_{n,l}\right)
\right]  ^{-1}-\left[  \lambda_{\min}\left(  \Phi_{X_{l}}\right)  \right]
^{-1}\right\vert \geq\left\{  \lambda_{\min}\left(  \Phi_{l}\right)  \right\}
^{-1}/2\right) \\
&  \leq\Pr\left\{  \left\vert \lambda_{\min}\left(  \Phi_{n,l}\right)
-\lambda_{\min}\left(  \Phi_{X_{l}}\right)  \right\vert \geq\lambda_{\min
}\left(  \Phi_{l}\right)  /3\right\} \\
&  \leq2m_{n}^{2}\exp\left\{  -C_{2}nm_{n}^{-3}\right\}  ,
\end{align*}
where the last inequality holds by (\ref{EQ:bound3a}). This shows the first
part of the lemma.

If Assumption \ref*{A:mn} also holds, then%
\begin{align*}
m_{n}^{2}\exp\left\{  -C_{2}nm_{n}^{-3}\right\}   &  =\exp\left\{
-C_{2}nm_{n}^{-3}+2\log m_{n}\right\} \\
&  =\exp\left\{  -C_{2}n^{1-3B_{m}}+2B_{m}\log n\right\}  \leq\exp\left\{
-C_{3}n^{C_{4}}\right\}
\end{align*}
for some $C_{4}\in(0,1-3B_{m})$ and $C_{3}\in(0,C_{2}]$. Then the second part
of the lemma follows.

To complete the proof the lemma, we now show (\ref{EQ:lambda1}) and
(\ref{EQ:lambda2}). To see equation (\ref{EQ:lambda1}), note that for any
vector $\boldsymbol{x=}\left(  x_{1},...,x_{m_{n}}\right)  ^{\prime}$ with
$\left\Vert \boldsymbol{x}\right\Vert =1,$%
\[
\min_{\left\Vert \boldsymbol{x}\right\Vert =1}\boldsymbol{x}^{\prime
}\boldsymbol{Ax}=\min_{\left\Vert \boldsymbol{x}\right\Vert =1}\left(
\boldsymbol{x}^{\prime}\boldsymbol{Bx+x}^{\prime}\left(  \boldsymbol{A}%
-\boldsymbol{B}\right)  \boldsymbol{x}\right)  \geq\min_{\left\Vert
\boldsymbol{x}\right\Vert =1}\boldsymbol{x}^{\prime}\boldsymbol{Bx}%
+\min_{\left\Vert \boldsymbol{x}\right\Vert =1}\boldsymbol{x}^{\prime
}\boldsymbol{\left(  \boldsymbol{A}-\boldsymbol{B}\right)  x},
\]
which is equivalent to
\[
\lambda_{\min}\left(  \boldsymbol{A}\right)  \geq\lambda_{\min}\left(
\boldsymbol{B}\right)  +\lambda_{\min}\left(  \boldsymbol{A-B}\right)
\ \text{or equivalently }\lambda_{\min}\left(  \boldsymbol{A-B}\right)
\leq\lambda_{\min}\left(  \boldsymbol{A}\right)  -\lambda_{\min}\left(
\boldsymbol{B}\right)  .
\]
Switching $\boldsymbol{A}$ and $\boldsymbol{B}$ yields $\lambda_{\min}\left(
\boldsymbol{B-A}\right)  \leq\lambda_{\min}\left(  \boldsymbol{B}\right)
-\lambda_{\min}\left(  \boldsymbol{A}\right)  .$ Therefore%
\[
\lambda_{\min}\left(  \boldsymbol{A-B}\right)  \leq\lambda_{\min}\left(
\boldsymbol{A}\right)  -\lambda_{\min}\left(  \boldsymbol{B}\right)
\leq-\lambda_{\min}\left(  \boldsymbol{B-A}\right)  ,
\]
and equation (\ref{EQ:lambda1}) follows. For (\ref{EQ:lambda2}), we have by
Jensen inequality
\[
\left\vert \lambda_{\max}\left(  \boldsymbol{A}\right)  \right\vert
=\left\vert \max_{\left\Vert \boldsymbol{x}\right\Vert =1}\boldsymbol{x}%
^{\prime}\boldsymbol{Ax}\right\vert \leq\left\Vert \boldsymbol{A}\right\Vert
_{\infty}\max_{\left\Vert \boldsymbol{x}\right\Vert =1}\left(  \sum
_{j=1}^{m_{n}}\left\vert x_{j}\right\vert \right)  ^{2}\leq m_{n}\left\Vert
\boldsymbol{A}\right\Vert _{\infty}\max_{\left\Vert \boldsymbol{x}\right\Vert
=1}\sum_{j=1}^{m_{n}}x_{j}^{2}=m_{n}\left\Vert \boldsymbol{A}\right\Vert
_{\infty},
\]
and similarly
\[
\left\vert \lambda_{\min}\left(  \boldsymbol{A}\right)  \right\vert
=\left\vert \min_{\left\Vert \boldsymbol{x}\right\Vert =1}\boldsymbol{x}%
^{\prime}\boldsymbol{Ax}\right\vert \leq\left\Vert \boldsymbol{A}\right\Vert
_{\infty}\min_{\left\Vert \boldsymbol{x}\right\Vert =1}\left(  \sum
_{j=1}^{m_{n}}\left\vert x_{j}\right\vert \right)  ^{2}\leq m_{n}\left\Vert
\boldsymbol{A}\right\Vert _{\infty}\min_{\left\Vert \boldsymbol{x}\right\Vert
=1}\sum_{j=1}^{m_{n}}x_{j}^{2}=m_{n}\left\Vert \boldsymbol{A}\right\Vert
_{\infty}.
\]
This completes the proof of the lemma.\medskip
\end{proof}

\begin{proof}
[Proof of Lemma \ref{LE:error_var}]Noting that $\hat{\sigma}_{l}^{2}%
=n^{-1}\boldsymbol{u}_{l}^{\prime}(I_{n}-\mathbb{X}_{l}\left(  \mathbb{X}%
_{l}^{\prime}\mathbb{X}_{l}\right)  ^{-1}\mathbb{X}_{l}^{\prime}%
)\boldsymbol{u}_{l},$ we have%
\[
\hat{\sigma}_{l}^{2}-\sigma_{l}^{2}=n^{-1}\boldsymbol{u}_{l}^{\prime
}\boldsymbol{u}_{l}-\sigma_{l}^{2}-n^{-1}\boldsymbol{u}_{l}^{\prime}%
\mathbb{X}_{l}\left(  \mathbb{X}_{l}^{\prime}\mathbb{X}_{l}\right)
^{-1}\mathbb{X}_{l}^{\prime}\boldsymbol{u}_{l}.
\]
It follows that%
\begin{align}
\Pr\left(  \left\vert \hat{\sigma}_{l}^{2}-\sigma_{l}^{2}\right\vert \geq
v_{n}\right)   &  \leq\Pr\left(  \left\vert n^{-1}\boldsymbol{u}_{l}^{\prime
}\boldsymbol{u}_{l}-\sigma_{l}^{2}\right\vert \geq\left(  1-\pi_{1}\right)
v_{n}/n\right) \label{EQ:pr_sigma}\\
&  +\Pr\left(  \left\vert n^{-1}\boldsymbol{u}_{l}^{\prime}\mathbb{X}%
_{l}\left(  \mathbb{X}_{l}^{\prime}\mathbb{X}_{l}\right)  ^{-1}\mathbb{X}%
_{l}^{\prime}\boldsymbol{u}_{l}\right\vert \geq\pi_{1}v_{n}/n\right) \nonumber
\end{align}
for any $\pi_{1}\in\left(  0,1\right)  .$

We first bound the first term on the right hand side of (\ref{EQ:pr_sigma}).
By Assumption \ref{A:tech} and
\begin{align*}
U_{l}  &  =Y-P^{m_{n}}\left(  X_{l}\right)  ^{\prime}\boldsymbol{\beta}_{l}\\
&  =\sum_{j=1}^{p^{\ast}}f_{j}^{\ast}\left(  X_{j}\right)  -P^{m_{n}}\left(
X_{l}\right)  ^{\prime}\boldsymbol{\beta}_{l}+\varepsilon
\end{align*}
with%
\[
\boldsymbol{\beta}_{l}=\left[  E\left(  P^{m_{n}}\left(  X_{l}\right)
P^{m_{n}}\left(  X_{l}\right)  ^{\prime}\right)  \right]  ^{-1}E\left(
P^{m_{n}}\left(  X_{l}\right)  Y\right)  ,
\]
we can say that $U_{l}$\ is just a uniformly bounded random term plus
$\varepsilon$ because elements in $\boldsymbol{\beta}_{l}$ are uniformly
bounded due to Lemma \ref{LE:rank} and the uniform boundedness of $\sum
_{j=1}^{p^{\ast}}f_{j}^{\ast}\left(  X_{j}\right)  ,$ and for each $x$ only a
finite elements of $P^{m_{n}}\left(  x\right)  $ are nonzero due to the usage
of finite order B-splines. This implies that $U_{l}$ shares the same tail
behavior as $\varepsilon$ and satisfies Assumption \ref{A:epsilon}. Then, the
conditions in Lemma \ref{LE:main_inequality} hold for $n^{-1}\boldsymbol{u}%
_{l}^{\prime}\boldsymbol{u}_{l}-\sigma_{l}^{2}$ with $\alpha=s/2$ by
Assumption \ref*{A:epsilon}. Applying Lemma \ref{LE:main_inequality} yields
that for any $v_{n}\propto n^{\lambda}$ with $1/2<\lambda\leq\left(
1+s/2\right)  /(2+s/2),$ $\pi<\pi_{1}$%
\begin{align}
\Pr\left(  \left\vert n^{-1}\boldsymbol{u}_{l}^{\prime}\boldsymbol{u}%
_{l}-\sigma_{l}^{2}\right\vert \geq\left(  1-\pi_{1}\right)  v_{n}/n\right)
&  \leq\exp\left[  \left.  -\left(  1-\pi\right)  ^{2}v_{n}^{2}\right/
\left(  2n\omega_{l}^{4}\right)  \right] \label{EQ:t1p1}\\
&  =\frac{C_{1}}{2}\exp\left(  -C_{2}v_{n}^{2}/n\right)  =\frac{C_{1}}{2}%
\exp\left(  -C_{2}n^{C_{3}}\right) \nonumber
\end{align}
with $C_{1}=2,$ $C_{2}=\left.  \left(  1-\pi\right)  ^{2}\right/  \left(
2\omega_{l}^{4}\right)  ,$ $C_{3}=2\lambda-1$, and last line holds by the fact
that $\lambda>1/2$. For $v_{n}\propto n^{\lambda}$ with $\lambda>\left(
1+s/2\right)  /(2+s/2),$ Lemma \ref{LE:main_inequality} implies%
\begin{equation}
\Pr\left(  \left\vert n^{-1}\boldsymbol{u}_{l}^{\prime}\boldsymbol{u}%
_{l}-\sigma_{l}^{2}\right\vert \geq\left(  1-\pi_{1}\right)  v_{n}/n\right)
\leq\frac{C_{1}}{2}\exp\left(  -C_{2}n^{C_{3}}\right)  , \label{EQ:t1p2}%
\end{equation}
for some $C_{1},C_{2},C_{3}>0.$

We turn to the second term on the right hand side of (\ref{EQ:pr_sigma}). Note
that Lemma \ref{LE:rank2} (to be shown below) still holds if we remove
$\hat{\sigma}_{l}^{2}$ and $\sigma_{l}^{2}$. Then by Lemma \ref{LE:rank2},
\begin{equation}
\Pr\left(  \left\vert n^{-1}\boldsymbol{u}_{l}^{\prime}\mathbb{X}_{l}\left(
\mathbb{X}_{l}^{\prime}\mathbb{X}_{l}\right)  ^{-1}\mathbb{X}_{l}^{\prime
}\boldsymbol{u}_{l}\right\vert \geq\pi_{1}v_{n}/n\right)  \leq\frac{C_{1}}%
{2}\exp\left(  -C_{2}n^{C_{3}}\right)  \label{EQ:t2}%
\end{equation}
for some positive constants $C_{1},C_{2}$ and $C_{3}$ for $v_{n}\propto
n^{\lambda}$ with $\lambda>1/2.$\footnote{It holds no matter which part of
Lemma \ref{LE:rank2} we apply.}

Combining (\ref{EQ:pr_sigma}), (\ref{EQ:t1p1}), (\ref{EQ:t1p2}), and
(\ref{EQ:t2}) completes the proof.\medskip
\end{proof}

\begin{proof}
[Proof of Lemma \ref{LE:rank2}]We show the first part first. Note that%
\begin{equation}
\left\vert \boldsymbol{u}_{l}^{\prime}\mathbb{X}_{l}\left(  \hat{\sigma}%
_{l}^{2}\mathbb{X}_{l}^{\prime}\mathbb{X}_{l}\right)  ^{-1}\mathbb{X}%
_{l}^{\prime}\boldsymbol{u}_{l}\right\vert \leq\hat{\sigma}_{l}^{-2}\left\Vert
\left(  n^{-1}\mathbb{X}_{l}^{\prime}\mathbb{X}_{l}\right)  ^{-1}\right\Vert
\left\Vert n^{-1/2}\mathbb{X}_{l}^{\prime}\boldsymbol{u}_{l}\right\Vert ^{2},
\label{EQ:p1}%
\end{equation}
and we bound the three terms on the right hand side of (\ref{EQ:p1}) in turn.
For the first term, with $v_{n}=\frac{1}{4}n\sigma_{l}^{2},$ Lemma
\ref{LE:error_var} implies%
\begin{equation}
\Pr\left(  \hat{\sigma}_{l}^{-2}>\frac{4}{3}\sigma_{l}^{-2}\right)  \leq
\frac{C_{2}}{2}\exp\left(  -C_{3}n^{C_{4}}\right)  , \label{EQ:p1_extra}%
\end{equation}
for some positive constants $C_{2},$ $C_{3},$ and $C_{4}$. For the second
term, equations (\ref{EQ:lam1}) and (\ref{EQ:xx1}) and Assumption
\ref*{A:mn}\ imply that%
\begin{equation}
\Pr\left(  \left\Vert \left(  n^{-1}\mathbb{X}_{l}^{\prime}\mathbb{X}%
_{l}\right)  ^{-1}\right\Vert \geq\frac{4}{3}B_{1}^{-1}m_{n}\right)
\leq2m_{n}^{2}\exp\left(  -C_{3}nm_{n}^{-3}\right)  \leq\frac{C_{2}}{2}%
\exp\left(  -C_{3}n^{C_{4}}\right)  , \label{EQ:p1_1}%
\end{equation}
for some positive constants $C_{2},$ $C_{3},$ and $C_{4}$, and $C_{4}%
\leq1-3B_{m}.$\ For the third term,
\[
\left\Vert n^{-1/2}\mathbb{X}_{l}^{\prime}\boldsymbol{u}_{l}\right\Vert
^{2}=\sum_{j=1}^{m_{n}}\left\{  \sum_{i=1}^{n}n^{-1/2}\phi_{j}\left(
x_{li}\right)  u_{li}\right\}  ^{2}.
\]
As we discuss in the proof of Lemma \ref{LE:error_var}, $U_{l}$ also satisfies
Assumption \ref*{A:epsilon}. Because $\phi_{j}\left(  x\right)  $ is uniformly
bounded for all $j,$ $\phi_{j}\left(  X_{l}\right)  U_{l}$ also satisfies
Assumption \ref*{A:epsilon}. So%
\begin{align}
\Pr\left(  \left\Vert n^{-1/2}\mathbb{X}_{l}^{\prime}\boldsymbol{u}%
_{l}\right\Vert ^{2}\geq\left(  \frac{4}{3}B_{1}^{-1}m_{n}\right)  ^{-1}%
v_{n}\right)   &  =\Pr\left(  \sum_{j=1}^{m_{n}}\left\{  \sum_{i=1}%
^{n}n^{-1/2}\phi_{j}\left(  x_{li}\right)  u_{li}\right\}  ^{2}\geq\frac{3}%
{4}B_{1}m_{n}^{-1}v_{n}\right) \nonumber\\
&  \leq\sum_{j=1}^{m_{n}}\Pr\left(  \left\{  n^{-1/2}\sum_{i=1}^{n}\phi
_{j}\left(  x_{li}\right)  u_{li}\right\}  ^{2}\geq\frac{1}{m_{n}}\frac{3}%
{4}B_{1}m_{n}^{-1}v_{n}\right) \nonumber\\
&  =\sum_{j=1}^{m_{n}}\Pr\left(  \left\vert \sum_{i=1}^{n}\phi_{j}\left(
x_{li}\right)  u_{li}\right\vert \geq\left(  \frac{3}{4}B_{1}\right)
^{1/2}m_{n}^{-1}v_{n}^{1/2}n^{1/2}\right) \nonumber\\
&  \leq m_{n}\exp\left(  -\left.  \left(  1-\pi\right)  ^{2}\frac{3}{4}%
B_{1}m_{n}^{-2}v_{n}n\right/  \left[  2n\left(  C_{8}m_{n}^{-1}\right)
\right]  \right) \nonumber\\
&  \leq\exp\left(  -\left(  1-\pi\right)  ^{2}\frac{3}{8}B_{1}C_{8}^{-1}%
m_{n}^{-1}v_{n}+\log m_{n}\right)  \label{EQ:p1_2}%
\end{align}
for any $\pi\in\left(  0,1\right)  ,$\ where the second inequality holds by
the first part of Lemma \ref{LE:main_inequality} and Assumption
\ref*{A:epsilon}\ with $\alpha=s$\ and the fact that
\[
\max_{j}\left\{  E\left[  \phi_{j}\left(  X_{l}\right)  ^{2}U_{l}^{2}\right]
\right\}  \leq\max_{j}\left\{  E\left[  \phi_{j}\left(  X_{l}\right)
^{2}E(U_{l}^{2}|X_{l})\right]  \right\}  \leq C\max_{j}\left\{  E\left[
\phi_{j}\left(  X_{l}\right)  ^{2}\right]  \right\}  \leq C_{8}m_{n}^{-1}%
\]
for some constants $C$ and $C_{8}$ by Assumption \ref{A:tech}. Here the second
inequality in the last displayed line follows from the fact that that $U_{l}$
is $\varepsilon$ plus a uniformly bounded term, and the last inequality holds
by Lemma \ref{LE:rank}.

By (\ref{EQ:p1}), (\ref{EQ:p1_extra}), (\ref{EQ:p1_1}), and (\ref{EQ:p1_2}),
we have%
\begin{align*}
\Pr\left(  \left\vert \boldsymbol{u}_{l}^{\prime}\mathbb{X}_{l}\left(
\hat{\sigma}_{l}^{2}\mathbb{X}_{l}^{\prime}\mathbb{X}_{l}\right)
^{-1}\mathbb{X}_{l}^{\prime}\boldsymbol{u}_{l}\right\vert \geq\frac{4}%
{3}\sigma_{l}^{-2}v_{n}\right)   &  \leq\Pr\left(  \hat{\sigma}_{l}%
^{-2}\left\Vert n^{-1/2}\mathbb{X}_{l}^{\prime}\boldsymbol{u}_{l}\right\Vert
^{2}\left\Vert \left(  n^{-1}\mathbb{X}_{l}^{\prime}\mathbb{X}_{l}\right)
^{-1}\right\Vert \geq\frac{4}{3}\sigma_{l}^{-2}v_{n}\right) \\
&  \leq\Pr\left(  \hat{\sigma}_{l}^{-2}>\frac{4}{3}\sigma_{l}^{-2}\right) \\
&  +\Pr\left(  \left\Vert \left(  n^{-1}\mathbb{X}_{l}^{\prime}\mathbb{X}%
_{l}\right)  ^{-1}\right\Vert \geq\frac{4}{3}B_{1}^{-1}m_{n}\right) \\
&  +\Pr\left(  \left\Vert n^{-1/2}\mathbb{X}_{l}^{\prime}\boldsymbol{u}%
_{l}\right\Vert ^{2}\geq\left(  \frac{4}{3}B_{1}^{-1}m_{n}\right)  ^{-1}%
v_{n}\right) \\
&  \leq\exp\left(  -C_{1}m_{n}^{-1}v_{n}+\log m_{n}\right)  +C_{2}\exp\left(
-C_{3}n^{C_{4}}\right)
\end{align*}
for positive constants $C_{1}=\left(  1-\pi\right)  ^{2}\frac{3}{8}B_{1}%
C_{8}^{-1},$ $C_{2}$, $C_{3}$ and $C_{4}.$ This completes the proof of the
first part of the lemma.

The second part holds similarly. The only difference is that we apply the
second part of Lemma \ref{LE:main_inequality} to the fourth line of equation
(\ref{EQ:p1_2}).\medskip
\end{proof}

\begin{proof}
[Proof of Lemma \ref{LE:hidden}]\textbf{(i)} With Assumption \ref*{A:fl}',
\cite{Stone} implies that there exists $\boldsymbol{\breve{\beta}}_{j}$ such
that
\begin{equation}
\sup_{x\in\left[  0,1\right]  }\left\vert P^{m_{n}}\left(  x\right)  ^{\prime
}\boldsymbol{\breve{\beta}}_{j}-f_{j}^{\ast}\left(  x\right)  \right\vert
=O\left(  m_{n}^{-d}\right)  \label{EQ:Rstronger}%
\end{equation}
for $j=1,2,...,p^{\ast}$. We rewrite $Y$ as%
\begin{align}
Y  &  =\sum_{j=1}^{p^{\ast}}P^{m_{n}}\left(  X_{j}\right)  ^{\prime
}\boldsymbol{\breve{\beta}}_{j}+\sum_{j=1}^{p^{\ast}}\left[  f_{j}^{\ast
}\left(  X_{j}\right)  -P^{m_{n}}\left(  X_{j}\right)  ^{\prime}%
\boldsymbol{\breve{\beta}}_{j}\right]  +\varepsilon\nonumber\\
&  \equiv\sum_{j=1}^{p^{\ast}}P^{m_{n}}\left(  X_{j}\right)  ^{\prime
}\boldsymbol{\breve{\beta}}_{j}+\breve{R}_{n}+\varepsilon.
\label{EQ:yapproStronger}%
\end{align}
Note that $\boldsymbol{\breve{\beta}}_{j}$ is similar to $\boldsymbol{\tilde
{\beta}}_{j}$ defined Section \ref{SEC:multi_tech}. We employ
$\boldsymbol{\breve{\beta}}_{j}$ here (and only here) because $\breve{R}%
_{n}=O\left(  m_{n}^{-d}\right)  $ by equation (\ref{EQ:Rstronger}) instead of
$R_{n}=O_{P}\left(  m_{n}^{-d}\right)  $ implied by equation (\ref{EQ:appro}).
The stronger result on $\breve{R}_{n}$ facilitates our proof below; see, e.g.,
the inequality in equation (\ref{EQ:Tn1_R}). Moreover, if $\left\{  E\left[
f_{nj}^{\ast}\left(  X_{j}\right)  ^{2}\right]  \right\}  ^{1/2}\gtrsim
\kappa_{n}\log\left(  m_{n}\right)  \left(  m_{n}/n\right)  ^{1/2}$ for some
$j$, $\left\{  E\left[  P^{m_{n}}\left(  X_{j}\right)  ^{\prime}%
\boldsymbol{\breve{\beta}}_{j}\right]  ^{2}\right\}  ^{1/2}\approx\left\{
E\left[  f_{nj}^{\ast}\left(  X_{j}\right)  ^{2}\right]  \right\}  ^{1/2}$
because $R_{n}=O_{P}\left(  m_{n}^{-d}\right)  $ and $\breve{R}_{n}=O\left(
m_{n}^{-d}\right)  $ and the bias term is asymptotically negligible. Thus, by
the third part of Lemma \ref{LE:rank},%
\begin{equation}
\left\Vert \boldsymbol{\breve{\beta}}_{j}\right\Vert \propto m_{n}%
^{1/2}\left\{  E\left[  f_{nj}^{\ast}\left(  X_{j}\right)  ^{2}\right]
\right\}  ^{1/2}\gtrsim\kappa_{n}m_{n}\log\left(  m_{n}\right)  n^{-1/2}.
\label{EQ:betaj_tide2}%
\end{equation}

Let $\Phi_{\boldsymbol{Z}}=E[P^{m_{n}}\left(  \boldsymbol{Z}\right)  P^{m_{n}%
}\left(  \boldsymbol{Z}\right)  ^{\prime}],$ $\Phi_{\boldsymbol{X}_{1}^{b}%
}=E[P^{m_{n}}(\boldsymbol{X}_{1}^{b})P^{m_{n}}(\boldsymbol{X}_{1}^{b}%
)^{\prime}],$ and $\Phi_{\boldsymbol{X}_{1}^{b}\boldsymbol{Z}}=E[P^{m_{n}%
}(\boldsymbol{X}_{1}^{b})$ $\cdot P^{m_{n}}(\boldsymbol{Z})^{\prime}].$
Substituting equation (\ref{EQ:yapproStronger}) into $\boldsymbol{\eta}%
_{1}^{b},$ we have%
\begin{align}
\boldsymbol{\eta}_{1}^{b}  &  =E\left[  P^{m_{n}}\left(  \boldsymbol{X}%
_{1}^{b}\right)  \left(  P^{m_{n}}\left(  \boldsymbol{X}_{1}^{b}\right)
^{\prime}\boldsymbol{\breve{\beta}}_{1}^{b}-P^{m_{n}}\left(  \boldsymbol{Z}%
\right)  ^{\prime}\Phi_{\boldsymbol{Z}}^{-1}E\left[  P^{m_{n}}\left(
\boldsymbol{Z}\right)  P^{m_{n}}\left(  \boldsymbol{X}_{1}^{b}\right)
^{\prime}\boldsymbol{\breve{\beta}}_{1}^{b}\right]  \right)  \right]
\nonumber\\
&  +E\left[  P^{m_{n}}\left(  \boldsymbol{X}_{1}^{b}\right)  \left(  \breve
{R}_{n}-P^{m_{n}}\left(  \boldsymbol{Z}\right)  ^{\prime}\Phi_{\boldsymbol{Z}%
}^{-1}E\left[  P^{m_{n}}\left(  \boldsymbol{Z}\right)  \breve{R}_{n}\right]
\right)  \right] \nonumber\\
&  =\left(  \Phi_{\boldsymbol{X}_{1}^{b}}-\Phi_{\boldsymbol{X}_{1}%
^{b}\boldsymbol{Z}}\Phi_{\boldsymbol{Z}}^{-1}\Phi_{\boldsymbol{X}_{1}%
^{b}\boldsymbol{Z}}^{\prime}\right)  \boldsymbol{\breve{\beta}}_{1}%
^{b}+\left\{  E\left[  P^{m_{n}}\left(  \boldsymbol{X}_{1}^{b}\right)
\breve{R}_{n}\right]  -\Phi_{\boldsymbol{X}_{1}^{b}\boldsymbol{Z}}%
\Phi_{\boldsymbol{Z}}^{-1}E\left[  P^{m_{n}}\left(  \boldsymbol{Z}\right)
\breve{R}_{n}\right]  \right\}  , \label{EQ:ita1b_decomposed}%
\end{align}
where $\boldsymbol{\breve{\beta}}_{1}^{b}\equiv\left(  \boldsymbol{\breve
{\beta}}_{1}^{\prime},\ldots,\boldsymbol{\breve{\beta}}_{b}^{\prime}\right)
^{\prime},$ the terms associated with $P^{m_{n}}\left(  \boldsymbol{X}%
_{b+1}^{p^{\ast}}\right)  ^{\prime}\boldsymbol{\breve{\beta}}_{b+1}^{p^{\ast}%
}$ are dropped out because $\boldsymbol{X}_{b+1}^{p^{\ast}}$ are included in
$\boldsymbol{Z}$, and the terms associated with $\varepsilon$ have zero
expectation due to Assumption \ref*{A:iid}. We analyze the above two terms in
$\boldsymbol{\eta}_{1}^{b}$ one by one.

We first study the first term in equation (\ref{EQ:ita1b_decomposed}). Let
$\Phi=\left[
\begin{array}
[c]{cc}%
\Phi_{\boldsymbol{X}_{1}^{b}} & \Phi_{\boldsymbol{X}_{1}^{b}\boldsymbol{Z}}\\
\Phi_{\boldsymbol{X}_{1}^{b}\boldsymbol{Z}}^{\prime} & \Phi_{\boldsymbol{Z}}%
\end{array}
\right]  .$ Noting that $\boldsymbol{X}_{1}^{b}$ and $\boldsymbol{Z}$ are
signals or pseudo-signals, we have by Assumption \ref{A:full_rank2} that%
\begin{equation}
\lambda_{\min}\left(  \Phi\right)  \propto m_{n}^{-1}\text{ and }\lambda
_{\max}\left(  \Phi\right)  \propto m_{n}^{-1}. \label{EQ:PXZ}%
\end{equation}
This implies that $\lambda_{\max}\left(  \Phi^{-1}\right)  =\left[
\lambda_{\min}\left(  \Phi\right)  \right]  ^{-1}\propto m_{n}$ and
$\lambda_{\min}\left(  \Phi^{-1}\right)  =\left[  \lambda_{\max}\left(
\Phi\right)  \right]  ^{-1}\propto m_{n},$ which along with the
\textit{inclusion principle} (e.g., Theorem 8.4.5 in Bernstein (2005)) further
implies that
\[
\lambda_{\min}\left(  \Phi_{\boldsymbol{X}_{1}^{b}}-\Phi_{\boldsymbol{X}%
_{1}^{b}\boldsymbol{Z}}\Phi_{\boldsymbol{Z}}^{-1}\Phi_{\boldsymbol{X}_{1}%
^{b}\boldsymbol{Z}}^{\prime}\right)  =\left[  \lambda_{\max}\left(
(\Phi_{\boldsymbol{X}_{1}^{b}}-\Phi_{\boldsymbol{X}_{1}^{b}\boldsymbol{Z}}%
\Phi_{\boldsymbol{Z}}^{-1}\Phi_{\boldsymbol{X}_{1}^{b}\boldsymbol{Z}}^{\prime
})^{-1}\right)  \right]  ^{-1}\propto m_{n}^{-1},
\]
because $(\Phi_{\boldsymbol{X}_{1}^{b}}-\Phi_{\boldsymbol{X}_{1}%
^{b}\boldsymbol{Z}}\Phi_{\boldsymbol{Z}}^{-1}\Phi_{\boldsymbol{X}_{1}%
^{b}\boldsymbol{Z}}^{\prime})^{-1}$ is the leading principal submatrix of
$\Phi^{-1}$. Then%
\begin{equation}
\left\Vert \left(  \Phi_{\boldsymbol{X}_{1}^{b}}-\Phi_{\boldsymbol{X}_{1}%
^{b}\boldsymbol{Z}}\Phi_{\boldsymbol{Z}}^{-1}\Phi_{\boldsymbol{X}_{1}%
^{b}\boldsymbol{Z}}^{\prime}\right)  \boldsymbol{\breve{\beta}}_{1}%
^{b}\right\Vert \gtrsim m_{n}^{-1}\left\Vert \boldsymbol{\breve{\beta}}%
_{1}^{b}\right\Vert \gtrsim m_{n}^{-1}\left\Vert \boldsymbol{\breve{\beta}%
}_{j}\right\Vert \label{EQ:Rankfirst}%
\end{equation}
for any $1\leq j\leq b$.

For the second term in equation (\ref{EQ:ita1b_decomposed}), we write
$E[P^{m_{n}}(\boldsymbol{X}_{1}^{b})\breve{R}_{n}]-\Phi_{\boldsymbol{X}%
_{1}^{b}\boldsymbol{Z}}\Phi_{\boldsymbol{Z}}^{-1}E\left[  P^{m_{n}}\left(
\boldsymbol{Z}\right)  \breve{R}_{n}\right]  $ $\equiv T_{n1}-T_{n2}.$ For
$T_{n1},$
\begin{equation}
\left\Vert T_{n1}\right\Vert \leq\left\Vert E\left[  \left\vert P^{m_{n}%
}\left(  \boldsymbol{X}_{1}^{b}\right)  \right\vert \left\vert \breve{R}%
_{n}\right\vert \right]  \right\Vert \lesssim m_{n}^{-d}\left\Vert E\left[
\left\vert P^{m_{n}}\left(  \boldsymbol{X}_{1}^{b}\right)  \right\vert
\right]  \right\Vert =m_{n}^{-d}\left(  \sum_{l=1}^{b}\sum_{k=1}^{m_{n}%
}\left\{  E\left\vert \phi_{k}\left(  X_{l}\right)  \right\vert \right\}
^{2}\right)  ^{1/2}, \label{EQ:Tn1_R}%
\end{equation}
where $\lesssim$ holds by equation (\ref{EQ:Rstronger}), and the equality
holds by the definitions of $P^{m_{n}}(\boldsymbol{X}_{1}^{b})$ and
$\left\Vert \cdot\right\Vert $. Here, $\left\vert A\right\vert \equiv\left(
\left\vert a_{1}\right\vert ,...,\left\vert a_{l}\right\vert \right)
^{\prime}$ for a vector $A=\left(  a_{1},...,a_{l}\right)  ^{\prime}.$ By the
property of B-splines in Lemma \ref{LE:rank}), $E\left\vert \phi_{k}\left(
X_{l}\right)  \right\vert \propto m_{n}^{-1}$. Then%
\begin{equation}
\left\Vert T_{n1}\right\Vert \lesssim m_{n}^{-d}\left(  \sum_{l=1}^{b}%
\sum_{k=1}^{m_{n}}\left\{  E\left\vert \phi_{k}\left(  X_{l}\right)
\right\vert \right\}  ^{2}\right)  ^{1/2}\propto m_{n}^{-d-1/2},
\label{EQ:Tn1}%
\end{equation}
where we use the fact that $b$ is finite. For $T_{n2},$ we first claim
\begin{equation}
\left\Vert \Phi_{\boldsymbol{X}_{1}^{b}\boldsymbol{Z}}\right\Vert \lesssim
m_{n}^{-1}. \label{EQ:ZXXZ}%
\end{equation}
Then by the submultiplicative property of $\left\Vert \cdot\right\Vert ,$%
\begin{equation}
\left\Vert T_{n2}\right\Vert \leq\left\Vert E\left[  P^{m_{n}}\left(
\boldsymbol{Z}\right)  \breve{R}_{n}\right]  \right\Vert \left\Vert
\Phi_{\boldsymbol{Z}}^{-1}\right\Vert \left\Vert \Phi_{\boldsymbol{X}_{1}%
^{b}\boldsymbol{Z}}\right\Vert \lesssim m_{n}^{-d-1/2}m_{n}m_{n}^{-1}%
=m_{n}^{-d-1/2}, \label{EQ:Tn2}%
\end{equation}
where second inequality holds by (\ref{EQ:Tn1}), (\ref{EQ:ZXXZ}), and
Assumption \ref{A:full_rank2}. Consequently, we have%
\begin{equation}
\left\Vert T_{n1}-T_{n2}\right\Vert \lesssim m_{n}^{-d-1/2}. \label{EQ:Rank2}%
\end{equation}

If we have $\left\{  E\left[  f_{nj}^{\ast}\left(  X_{j}\right)  ^{2}\right]
\right\}  ^{1/2}\gtrsim\kappa_{n}\log\left(  m_{n}\right)  \left(
m_{n}/n\right)  ^{1/2}$ for some $j,$ then by (\ref{EQ:Rankfirst}) and
(\ref{EQ:betaj_tide2}),%
\[
\left\Vert \left(  \Phi_{\boldsymbol{X}_{1}^{b}}-\Phi_{\boldsymbol{X}_{1}%
^{b}\boldsymbol{Z}}\Phi_{\boldsymbol{Z}}^{-1}\Phi_{\boldsymbol{X}_{1}%
^{b}\boldsymbol{Z}}^{\prime}\right)  \boldsymbol{\breve{\beta}}_{1}%
^{b}\right\Vert \gtrsim m_{n}^{-1}\left\Vert \boldsymbol{\breve{\beta}}%
_{j}\right\Vert \propto m_{n}^{-1/2}\left\{  E\left[  f_{nj}^{\ast}\left(
X_{j}\right)  ^{2}\right]  \right\}  ^{1/2}\gtrsim\kappa_{n}\log\left(
m_{n}\right)  n^{-1/2}.
\]
This, in conjunction with (\ref{EQ:Rank2}) and the fact that $\kappa_{n}%
\log\left(  m_{n}\right)  n^{-1/2}\gg m_{n}^{-d-1/2}$ by Assumption
\ref{A:mn}, implies that $\left\Vert T_{n1}-T_{n2}\right\Vert $ is of smaller
order than $\left\Vert (\Phi_{\boldsymbol{X}_{1}^{b}}-\Phi_{\boldsymbol{X}%
_{1}^{b}\boldsymbol{Z}}\Phi_{\boldsymbol{Z}}^{-1}\Phi_{\boldsymbol{X}_{1}%
^{b}\boldsymbol{Z}}^{\prime})\boldsymbol{\breve{\beta}}_{1}^{b}\right\Vert $
and thus%
\[
\left\Vert \boldsymbol{\eta}_{1}^{b}\right\Vert \gtrsim\kappa_{n}\log\left(
m_{n}\right)  ^{1/2}n^{-1/2}.
\]

Now we show equation (\ref{EQ:ZXXZ}). Recall that $\Phi=\left[
\begin{array}
[c]{cc}%
\Phi_{\boldsymbol{X}_{1}^{b}} & \Phi_{\boldsymbol{X}_{1}^{b}\boldsymbol{Z}}\\
\Phi_{\boldsymbol{X}_{1}^{b}\boldsymbol{Z}}^{\prime} & \Phi_{\boldsymbol{Z}}%
\end{array}
\right]  .$ By equation (\ref{EQ:PXZ}), we have%
\begin{equation}
\left\Vert \Phi\right\Vert =\max_{\substack{\left\Vert \left(  \boldsymbol{a}%
_{1}^{\prime},\boldsymbol{a}_{2}^{\prime}\right)  ^{\prime}\right\Vert
=1,\\\boldsymbol{a}_{1}\in\mathbb{R}^{bm_{n}},\boldsymbol{a}_{2}\in
\mathbb{R}^{\iota_{n}m_{n}}}}\left(  \boldsymbol{a}_{1}^{\prime}%
,\boldsymbol{a}_{2}^{\prime}\right)  \Phi\left(
\begin{array}
[c]{c}%
\boldsymbol{a}_{1}\\
\boldsymbol{a}_{2}%
\end{array}
\right)  \propto m_{n}^{-1}. \label{EQ:aXZa1}%
\end{equation}
Noting that $\left(  \boldsymbol{a}_{1}^{\prime},\boldsymbol{a}_{2}^{\prime
}\right)  \Phi\left(
\begin{array}
[c]{c}%
\boldsymbol{a}_{1}\\
\boldsymbol{a}_{2}%
\end{array}
\right)  =\boldsymbol{a}_{1}^{\prime}\Phi_{\boldsymbol{X}_{1}^{b}%
}\boldsymbol{a}_{1}+\boldsymbol{a}_{2}^{\prime}\Phi_{\boldsymbol{Z}%
}\boldsymbol{a}_{2}+2\boldsymbol{a}_{1}^{\prime}\Phi_{\boldsymbol{X}_{1}%
^{b}\boldsymbol{Z}}\boldsymbol{a}_{2},$ we have by the nonnegative
definiteness of $\Phi_{\boldsymbol{X}_{1}^{b}}$ and $\Phi_{\boldsymbol{Z}},$
\begin{equation}
\max_{\substack{\left\Vert \left(  \boldsymbol{a}_{1}^{\prime},\boldsymbol{a}%
_{2}^{\prime}\right)  ^{\prime}\right\Vert =1,\\\boldsymbol{a}_{1}%
\in\mathbb{R}^{bm_{n}},\boldsymbol{a}_{2}\in\mathbb{R}^{\iota_{n}m_{n}}%
}}\left\vert \boldsymbol{a}_{1}^{\prime}\Phi_{\boldsymbol{X}_{1}%
^{b}\boldsymbol{Z}}\boldsymbol{a}_{2}\right\vert \leq\max
_{\substack{\left\Vert \left(  \boldsymbol{a}_{1}^{\prime},\boldsymbol{a}%
_{2}^{\prime}\right)  ^{\prime}\right\Vert =1,\\\boldsymbol{a}_{1}%
\in\mathbb{R}^{bm_{n}},\boldsymbol{a}_{2}\in\mathbb{R}^{\iota_{n}m_{n}}%
}}\left(  \boldsymbol{a}_{1}^{\prime},\boldsymbol{a}_{2}^{\prime}\right)
\Phi\left(
\begin{array}
[c]{c}%
\boldsymbol{a}_{1}\\
\boldsymbol{a}_{2}%
\end{array}
\right)  \propto m_{n}^{-1}, \label{EQ:PMZPMX}%
\end{equation}
Then
\begin{align*}
\left\Vert \Phi_{\boldsymbol{X}_{1}^{b}\boldsymbol{Z}}\right\Vert  &
=\max_{\substack{\left\Vert \boldsymbol{a}_{1}\right\Vert =1,\left\Vert
\boldsymbol{a}_{2}^{\prime}\right\Vert =1,\\\boldsymbol{a}_{1}\in
\mathbb{R}^{bm_{n}},\boldsymbol{a}_{2}\in\mathbb{R}^{\iota_{n}m_{n}}%
}}\left\vert \boldsymbol{a}_{1}^{\prime}\Phi_{\boldsymbol{X}_{1}%
^{b}\boldsymbol{Z}}\boldsymbol{a}_{2}\right\vert \leq\max
_{\substack{\left\Vert \left(  \boldsymbol{a}_{1}^{\prime},\boldsymbol{a}%
_{2}^{\prime}\right)  ^{\prime}\right\Vert \leq\sqrt{2}\\\boldsymbol{a}_{1}%
\in\mathbb{R}^{bm_{n}},\boldsymbol{a}_{2}\in\mathbb{R}^{\iota_{n}m_{n}}%
}}\left\vert \boldsymbol{a}_{1}^{\prime}\Phi_{\boldsymbol{X}_{1}%
^{b}\boldsymbol{Z}}\boldsymbol{a}_{2}\right\vert \\
&  =2\max_{\substack{\left\Vert \left(  \boldsymbol{a}_{1}^{\prime
},\boldsymbol{a}_{2}^{\prime}\right)  ^{\prime}\right\Vert =1\\\boldsymbol{a}%
_{1}\in\mathbb{R}^{bm_{n}},\boldsymbol{a}_{2}\in\mathbb{R}^{\iota_{n}m_{n}}%
}}\left\vert \boldsymbol{a}_{1}^{\prime}\Phi_{\boldsymbol{X}_{1}%
^{b}\boldsymbol{Z}}\boldsymbol{a}_{2}\right\vert \propto m_{n}^{-1}%
\end{align*}
where the first equality follows from Fact 9.11.2 in \cite{Bernstein2005}, and
the first inequality holds because $\left\{  \left\Vert \boldsymbol{a}%
_{1}\right\Vert =1,\left\Vert \boldsymbol{a}_{2}^{\prime}\right\Vert
=1\right\}  \subset\{\left\Vert \left(  \boldsymbol{a}_{1}^{\prime
},\boldsymbol{a}_{2}^{\prime}\right)  \right\Vert \leq\sqrt{2}\}.$

(ii) The result is straightforward given (i). Similar to
(\ref{EQ:ita1b_decomposed}),%
\begin{equation}
\boldsymbol{\eta}_{1}^{p^{\ast}}=\Phi_{\boldsymbol{X}_{1}^{p^{\ast}}%
}\boldsymbol{\breve{\beta}}_{1}^{p^{\ast}}+E\left[  P^{m_{n}}\left(
\boldsymbol{X}_{1}^{p^{\ast}}\right)  \breve{R}_{n}\right]  .\nonumber
\end{equation}
Due to the rank condition, similar to (\ref{EQ:Rankfirst}),%
\[
\left\Vert \Phi_{\boldsymbol{X}_{1}^{p^{\ast}}}\boldsymbol{\breve{\beta}}%
_{1}^{p^{\ast}}\right\Vert \gtrsim m_{n}^{-1}\left\Vert \boldsymbol{\breve
{\beta}}_{1}^{p^{\ast}}\right\Vert \gtrsim m_{n}^{-1}\left\Vert
\boldsymbol{\breve{\beta}}_{j}\right\Vert .
\]
The bound on $E\left[  P^{m_{n}}\left(  \boldsymbol{X}_{1}^{p^{\ast}}\right)
\breve{R}_{n}\right]  $ is a direct result from (\ref{EQ:Tn1_R}) and
(\ref{EQ:Tn1}), specifically,
\[
\left\Vert E\left[  \left\vert P^{m_{n}}\left(  \boldsymbol{X}_{1}^{p^{\ast}%
}\right)  \right\vert \left\vert \breve{R}_{n}\right\vert \right]  \right\Vert
\lesssim m_{n}^{-d}\left\Vert E\left[  \left\vert P^{m_{n}}\left(
\boldsymbol{X}_{1}^{p^{\ast}}\right)  \right\vert \right]  \right\Vert \propto
m_{n}^{-d-1/2}.
\]
Using a similar argument as in (i), if $\left\{  E\left[  f_{nj}^{\ast}\left(
X_{j}\right)  ^{2}\right]  \right\}  ^{1/2}\gtrsim\kappa_{n}\log\left(
m_{n}\right)  \left(  m_{n}/n\right)  ^{1/2}$ for some $j,$%
\[
\left\Vert \Phi_{\boldsymbol{X}_{1}^{p^{\ast}}}\boldsymbol{\breve{\beta}}%
_{1}^{p^{\ast}}\right\Vert \gtrsim m_{n}^{-1}\left\Vert \boldsymbol{\breve
{\beta}}_{j}\right\Vert \gtrsim\kappa_{n}\log\left(  m_{n}\right)  n^{-1/2},
\]
and $\left\Vert E\left[  \left\vert P^{m_{n}}\left(  \boldsymbol{X}%
_{1}^{p^{\ast}}\right)  \right\vert \left\vert \breve{R}_{n}\right\vert
\right]  \right\Vert $ is a small order term due to the fact that $\kappa
_{n}\log\left(  m_{n}\right)  n^{-1/2}\gg m_{n}^{-d-1/2}$ by Assumption
\ref{A:mn}. Thus%
\[
\left\Vert \boldsymbol{\eta}_{1}^{p^{\ast}}\right\Vert \gtrsim\kappa_{n}%
\log\left(  m_{n}\right)  ^{1/2}n,
\]
as desired.
\end{proof}

\begin{proof}
[Proof of Lemma \ref{LE:XXQQ}]As in the proof of the last lemma, let
$\Phi_{\boldsymbol{Z}}=E\left[  P^{m_{n}}\left(  \boldsymbol{Z}\right)
P^{m_{n}}\left(  \boldsymbol{Z}\right)  ^{\prime}\right]  .$ We prove the
lemma by showing that
\[
\Pr\left(  \left\vert \left\Vert \left(  n^{-1}\mathbb{Z}^{\prime}%
\mathbb{Z}\right)  ^{-1}\right\Vert -\left\Vert \Phi_{\boldsymbol{Z}}%
^{-1}\right\Vert \right\vert >\left.  \left\Vert \Phi_{\boldsymbol{Z}}%
^{-1}\right\Vert \right/  2\right)  \leq2m_{n}^{2}\iota_{n}^{2}\exp\left\{
-C_{1}nm_{n}^{-3}\iota_{n}^{-1}\right\}  ,
\]
for a $\left(  \iota_{n}+1\right)  \times1$ random vector $\boldsymbol{Z}$,
where we have absorbed $X_{l}$ into $\boldsymbol{Z}.$ Assumption
\ref*{A:full_rank2} applies to $\boldsymbol{Z}$. Because of this, the proof is
essentially the same as that for Lemma \ref{LE:XX'}. The only difference is
that the dimension of the matrix here is $m_{n}\left(  \iota_{n}+1\right)
\times m_{n}\left(  \iota_{n}+1\right)  $.

Following the proof in Lemma \ref{LE:XX'} up to equation (\ref{EQ:upto}), we
have%
\begin{align*}
&  \Pr\left(  \left\vert \lambda_{\min}\left(  n^{-1}\mathbb{Z}^{\prime
}\mathbb{Z}\right)  -\lambda_{\min}\left(  \Phi_{\boldsymbol{Z}}\right)
\right\vert \geq m_{n}\left(  \iota_{n}+1\right)  v_{n}/n\right) \\
&  \leq\Pr\left(  \max\left\{  \left\vert \lambda_{\min}\left(  n^{-1}%
\mathbb{Z}^{\prime}\mathbb{Z}-\Phi_{\boldsymbol{Z}}\right)  \right\vert
,\left\vert \lambda_{\min}\left(  \Phi_{\boldsymbol{Z}}-n^{-1}\mathbb{Z}%
^{\prime}\mathbb{Z}\right)  \right\vert \right\}  \geq m_{n}\left(  \iota
_{n}+1\right)  v_{n}/n\right) \\
&  \leq\Pr\left(  \left\Vert n^{-1}\mathbb{Z}^{\prime}\mathbb{Z}%
-\Phi_{\boldsymbol{Z}}\right\Vert _{\infty}\geq v_{n}/n\right)  \leq2m_{n}%
^{2}\left(  \iota_{n}+1\right)  ^{2}\exp\left\{  -v_{n}^{2}/2\left(
B_{4}nm_{n}^{-1}+v_{n}/3\right)  \right\}  .
\end{align*}
Set $v_{n}=nm_{n}^{-2}\left(  \iota_{n}+1\right)  ^{-1}B_{X1}/3,$ the above
inequality becomes
\[
\Pr\left(  \left\vert \lambda_{\min}\left(  n^{-1}\mathbb{Z}^{\prime
}\mathbb{Z}\right)  -\lambda_{\min}\left(  \Phi_{\boldsymbol{Z}}\right)
\right\vert \geq B_{X1}m_{n}^{-1}/3\right)  \leq2m_{n}^{2}\iota_{n}^{2}%
\exp\left\{  -C_{1}nm_{n}^{-3}\iota_{n}^{-1}\right\}
\]
for some positive constant $C_{1}.$

If in addition Assumption \ref*{A:p}' and \ref*{A:mn}\ holds, $nm_{n}%
^{-3}\iota_{n}^{-1}\geq nm_{n}^{-3}p_{n}^{\ast\ast-1}\propto n^{1-3B_{m}%
-B_{p^{\ast\ast}}}=n^{C_{4}}$ and $C_{4}>0.$ The conclusion is
immediate.\medskip
\end{proof}

\begin{proof}
[Proof of Lemma \ref{LE:lambdamax}]First
\begin{align}
&  \Pr\left(  \lambda_{\max}\left\{  \left(  \hat{\sigma}_{l,\boldsymbol{Z}%
}^{2}n^{-1}\mathbb{X}_{l}^{\prime}M_{\mathbb{Z}}\mathbb{X}_{l}\right)
^{-1}\right\}  \geq2\sigma^{-2}B_{X1}^{-1}m_{n}\right) \label{EQ:l1}\\
&  \leq\Pr\left(  \hat{\sigma}_{l,\boldsymbol{Z}}^{-2}\geq\frac{4}{3}%
\sigma^{-2}\right)  +\Pr\left(  \lambda_{\max}\left\{  \left(  n^{-1}%
\mathbb{X}_{l}^{\prime}M_{\mathbb{Z}}\mathbb{X}_{l}\right)  ^{-1}\right\}
\geq\frac{3}{2}B_{X1}^{-1}m_{n}\right)  .\nonumber
\end{align}
Taking $v_{n}=\frac{1}{4}n\sigma^{2}$\ in Lemma \ref{LE:error_var_Z} yields%
\begin{equation}
\Pr\left(  \hat{\sigma}_{l,\boldsymbol{Z}}^{-2}\geq\frac{4}{3}\sigma
^{-2}\right)  =\Pr\left(  \hat{\sigma}_{l,\boldsymbol{Z}}^{2}\leq\frac{3}%
{4}\sigma^{2}\right)  \leq\frac{1}{2}C_{1}\exp(-C_{2}n^{C_{3}}) \label{EQ:l2}%
\end{equation}
for some positive $C_{1},C_{2},$ and $C_{3}.$ Noting that $\left(
\mathbb{X}_{l}^{\prime}M_{\mathbb{Z}}\mathbb{X}_{l}\right)  ^{-1}$ is the
lower block diagonal of $\left(  \left(
\begin{array}
[c]{c}%
\mathbb{Z}^{\prime}\\
\mathbb{X}_{l}^{\prime}%
\end{array}
\right)  \left(
\begin{array}
[c]{cc}%
\mathbb{Z} & \mathbb{X}_{l}%
\end{array}
\right)  \right)  ^{-1}$, by Lemma \ref{LE:XXQQ} and Assumption
\ref*{A:full_rank2}, we have%
\[
\Pr\left(  \lambda_{\max}\left\{  \left(  n^{-1}\mathbb{X}_{l}^{\prime
}M_{\mathbb{Z}}\mathbb{X}_{l}\right)  ^{-1}\right\}  \geq\frac{3}{2}%
B_{X1}^{-1}m_{n}\right)  \leq2m_{n}^{2}\iota_{n}^{2}\exp\left\{  -C_{1}%
nm_{n}^{-3}\iota_{n}^{-1}\right\}  .
\]
Since $\iota_{n}\leq p^{\ast\ast},$ $nm_{n}^{-3}\iota_{n}\leq n^{1-3B_{m}%
+B_{p^{\ast\ast}}}=n^{C_{3}}$ and $C_{3}>0$\ by Assumption \ref*{A:p}'. Then
\begin{equation}
\Pr\left(  \lambda_{\max}\left\{  \left(  n^{-1}\mathbb{X}_{l}^{\prime
}M_{\mathbb{Z}}\mathbb{X}_{l}\right)  ^{-1}\right\}  \geq\frac{3}{2}%
B_{X1}^{-1}m_{n}\right)  \leq\frac{1}{2}C_{1}\exp(-C_{2}n^{C_{3}}).
\label{EQ:l3}%
\end{equation}
Combining (\ref{EQ:l1}), (\ref{EQ:l2}), and (\ref{EQ:l3}) yields%
\[
\Pr\left(  \lambda_{\max}\left\{  \left(  \hat{\sigma}_{l,\boldsymbol{Z}}%
^{2}n^{-1}\mathbb{X}_{l}^{\prime}M_{\mathbb{Z}}\mathbb{X}_{l}\right)
^{-1}\right\}  \geq2\sigma^{-2}B_{X1}^{-1}m_{n}\right)  \leq C_{1}\exp
(-C_{2}n^{C_{3}}).
\]

For the second part, we first make the following decomposition:%
\begin{align*}
&  \Pr\left(  \lambda_{\min}\left\{  \left(  \hat{\sigma}_{l,\boldsymbol{Z}%
}^{2}n^{-1}\mathbb{X}_{l}^{\prime}M_{\mathbb{Z}}\mathbb{X}_{l}\right)
^{-1}\right\}  \leq\frac{1}{4}\sigma^{-2}B_{X2}^{-1}m_{n}\right) \\
&  \leq\Pr\left(  \hat{\sigma}_{l,\boldsymbol{Z}}^{-2}\leq\frac{1}{2}%
\sigma^{-2}\right)  +\Pr\left(  \lambda_{\min}\left\{  \left(  n^{-1}%
\mathbb{X}_{l}^{\prime}M_{\mathbb{Z}}\mathbb{X}_{l}\right)  ^{-1}\right\}
\leq\frac{1}{2}\sigma^{-2}B_{X2}^{-1}m_{n}\right)  .
\end{align*}
Then we apply Lemma \ref{LE:error_var_Z} and Lemma \ref{LE:XXQQ} to the two
terms on the right hand side of the last displayed equation. The conclusion
follows as in the proof of the first part. \medskip
\end{proof}

\begin{proof}
[Proof of Lemma \ref{LE:error_bound}]We deal with $n^{-1/2}\left(
\boldsymbol{u}_{X_{l},\boldsymbol{Z}}^{\prime}\boldsymbol{u}_{Y,\boldsymbol{Z}%
}-\boldsymbol{\eta}_{l,\boldsymbol{Z}}\right)  $ first. Note%
\[
n^{-1/2}\left(  \boldsymbol{u}_{X_{l},\boldsymbol{Z}}^{\prime}\boldsymbol{u}%
_{Y,\boldsymbol{Z}}-\boldsymbol{\eta}_{l,\boldsymbol{Z}}\right)  =n^{-1/2}%
\sum_{i=1}^{n}\left(  \boldsymbol{u}_{X_{l},\boldsymbol{Z},i}%
u_{Y,\boldsymbol{Z},i}-\boldsymbol{\eta}_{l,\boldsymbol{Z}}\right)
\]
Denote the $j$-th element of $\boldsymbol{U}_{X_{l},\boldsymbol{Z}}$ by
$U_{\phi_{j}\left(  X_{l}\right)  ,\boldsymbol{Z}}.$\ By equation
(\ref{EQ:u_xlz}) and Lemma \ref{LE:rank}, $E\left(  U_{\phi_{j}\left(
X_{l}\right)  ,\boldsymbol{Z}}^{2}\right)  $ $\propto m_{n}^{-1},$ and
$U_{\phi_{j}\left(  X_{l}\right)  ,\boldsymbol{Z}}^{{}}$ is uniformly bounded
due to the uniform boundedness of $\phi_{j}\left(  X_{l}\right)  .$

By equation (\ref{EQ:u_xlz}),
\[
U_{Y,\boldsymbol{Z}}=\sum_{j=1}^{p^{\ast}}f_{j}^{\ast}\left(  X_{j}\right)
+\varepsilon-\boldsymbol{\gamma}_{Y,\boldsymbol{Z}}^{\prime}P^{m_{n}}\left(
\boldsymbol{Z}\right)  ,
\]
where $\boldsymbol{\gamma}_{Y,\boldsymbol{Z}}\equiv\Phi_{\boldsymbol{Z}}%
^{-1}E\left[  P^{m_{n}}\left(  \boldsymbol{Z}\right)  Y\right]  $. The rank
conditions in Assumption \ref{A:full_rank2} and uniform boundedness of
$\sum_{j=1}^{p^{\ast}}f_{j}^{\ast}\left(  X_{j}\right)  $ imply that each
element in $\boldsymbol{\gamma}_{Y,\boldsymbol{Z}}$ is uniformly bounded. Note
$p^{\ast}$ is fixed and the dimension of $\boldsymbol{\gamma}%
_{Y,\boldsymbol{Z}}$ is $\iota_{n}m_{n}$, and thus
\[
U_{Y,\boldsymbol{Z}}=C_{Y,\boldsymbol{Z}}+\varepsilon,
\]
with $C_{Y,\boldsymbol{Z}}\equiv\sum_{j=1}^{p^{\ast}}f_{j}^{\ast}\left(
X_{j}\right)  -\boldsymbol{\gamma}_{Y,\boldsymbol{Z}}^{\prime}P^{m_{n}}\left(
\boldsymbol{Z}\right)  $ and $\left\vert C_{Y,\boldsymbol{Z}}\right\vert \leq
C\iota_{n}m_{n}$ for a positive $C$. The $j$-th element of $\boldsymbol{u}%
_{X_{l},\boldsymbol{Z}}^{\prime}\boldsymbol{u}_{Y,\boldsymbol{Z}%
}-\boldsymbol{\eta}_{l,\boldsymbol{Z}}$ is%
\[
\sum_{i=1}^{n}\left(  u_{\phi_{j}\left(  X_{l}\right)  ,\boldsymbol{Z}%
,i}u_{Y,\boldsymbol{Z},i}-\boldsymbol{\eta}_{l,\boldsymbol{Z,}j}\right)
=\sum_{i=1}^{n}\left(  u_{\phi_{j}\left(  X_{l}\right)  ,\boldsymbol{Z}%
,i}C_{Y,\boldsymbol{Z,}i}-\boldsymbol{\eta}_{l,\boldsymbol{Z,}j}\right)
+\sum_{i=1}^{n}u_{\phi_{j}\left(  X_{l}\right)  ,\boldsymbol{Z},i}%
\varepsilon_{i}.
\]
Since $u_{\phi_{j}\left(  X_{l}\right)  ,\boldsymbol{Z},i}$ is uniformly
bounded, $\left\vert u_{\phi_{j}\left(  X_{l}\right)  ,\boldsymbol{Z}%
,i}C_{Y,\boldsymbol{Z,}i}\right\vert \leq C\iota_{n}m_{n}$ for a positive $C,$
and $u_{\phi_{j}\left(  X_{l}\right)  ,\boldsymbol{Z},i}\varepsilon_{i}$
satisfies the tail restriction in Assumption \ref{A:epsilon}. Apply the
inequalities in Lemmas \ref{LE:bernstein} on the first term in the above, we
obtain%
\begin{align}
&  \Pr\left(  \left\vert \sum_{i=1}^{n}\left(  u_{\phi_{j}\left(
X_{l}\right)  ,\boldsymbol{Z},i}u_{Y,\boldsymbol{Z},i}-\boldsymbol{\eta
}_{l,\boldsymbol{Z,}j}\right)  \right\vert >v_{n}\right) \nonumber\\
&  \leq\Pr\left(  \left\vert \sum_{i=1}^{n}\left(  u_{\phi_{j}\left(
X_{l}\right)  ,\boldsymbol{Z},i}C_{Y,\boldsymbol{Z,}i}-\boldsymbol{\eta
}_{l,\boldsymbol{Z,}j}\right)  \right\vert >\frac{v_{n}}{2}\right)
+\Pr\left(  \left\vert \sum_{i=1}^{n}u_{\phi_{j}\left(  X_{l}\right)
,\boldsymbol{Z},i}\varepsilon_{i}\right\vert >\frac{v_{n}}{2}\right)
\nonumber\\
&  \leq2\exp\left\{  -v_{n}^{2}/\left(  C\left(  nm_{n}^{-1}+\iota_{n}%
m_{n}v_{n}\right)  \right)  \right\}  +\Pr\left(  \left\vert \sum_{i=1}%
^{n}u_{\phi_{j}\left(  X_{l}\right)  ,\boldsymbol{Z},i}\varepsilon
_{i}\right\vert >\frac{v_{n}}{2}\right)  . \label{EQ:uuita_sep}%
\end{align}

Using the above identity, we have%
\begin{align}
&  \Pr\left(  \left\Vert n^{-1/2}\sum_{i=1}^{n}\left(  \boldsymbol{u}%
_{X_{l},\boldsymbol{Z},i}u_{Y,\boldsymbol{Z},i}-\boldsymbol{\eta
}_{l,\boldsymbol{Z}}\right)  \right\Vert ^{2}\geq C_{1}m_{n}^{-1}v_{n}\right)
\nonumber\\
&  =\Pr\left(  \sum_{j=1}^{m_{n}}\left\{  n^{-1/2}\sum_{i=1}^{n}\left(
u_{\phi_{j}\left(  X_{l}\right)  ,\boldsymbol{Z},i}u_{Y,\boldsymbol{Z}%
,i}-\boldsymbol{\eta}_{l,\boldsymbol{Z,}j}\right)  \right\}  ^{2}\geq
C_{1}m_{n}^{-1}v_{n}\right) \nonumber\\
&  \leq\sum_{j=1}^{m_{n}}\Pr\left(  \left\vert \sum_{i=1}^{n}\left(
u_{\phi_{j}\left(  X_{l}\right)  ,\boldsymbol{Z},i}u_{Y,\boldsymbol{Z}%
,i}-\boldsymbol{\eta}_{l,\boldsymbol{Z,}j}\right)  \right\vert \geq
C_{1}n^{1/2}m_{n}^{-1}v_{n}^{1/2}\right) \nonumber\\
&  \leq\sum_{j=1}^{m_{n}}\left\{  \Pr\left(  \left\vert \sum_{i=1}^{n}\left(
u_{\phi_{j}\left(  X_{l}\right)  ,\boldsymbol{Z},i}C_{Y,\boldsymbol{Z,}%
i}-\boldsymbol{\eta}_{l,\boldsymbol{Z,}j}\right)  \right\vert \geq\frac{1}%
{2}C_{1}n^{1/2}m_{n}^{-1}v_{n}^{1/2}\right)  \right. \nonumber\\
&  \text{ \ \ \ \ }\left.  +\Pr\left(  \left\vert \sum_{i=1}^{n}u_{\phi
_{j}\left(  X_{l}\right)  ,\boldsymbol{Z},i}\varepsilon_{i}\right\vert
\geq\frac{1}{2}C_{1}n^{1/2}m_{n}^{-1}v_{n}^{1/2}\right)  \right\} \nonumber\\
&  \leq\left\{
\begin{array}
[c]{c}%
2m_{n}\exp\left\{  -nm_{n}^{-2}v_{n}/[C_{2}(nm_{n}^{-1}+\iota_{n}n^{1/2}%
v_{n}^{1/2})]\right\}  +m_{n}\exp\left(  -C_{3}m_{n}^{-1}v_{n}\right) \\
2m_{n}\exp\left\{  -nm_{n}^{-2}v_{n}/[C_{2}(nm_{n}^{-1}+\iota_{n}n^{1/2}%
v_{n}^{1/2})]\right\}  +C_{4}\exp\left(  -C_{5}n^{C_{6}}\right)
\end{array}
\right.
\begin{array}
[c]{c}%
\text{if }v_{n}\lesssim n^{s/\left(  s+2\right)  }m_{n}^{2}\\
\text{if }v_{n}\gtrsim n^{s/\left(  s+2\right)  }m_{n}^{2}%
\end{array}
, \label{EQ:u1}%
\end{align}
for any positive constant $C_{1}$ and some positive constants $C_{2}%
,C_{3},C_{4}$ and $C_{5}$, where the last inequality follows from the same
arguments as used to obtain (\ref{EQ:p1_2}) along with applying equation
(\ref{EQ:uuita_sep}).

We turn to $n^{-1/2}\boldsymbol{u}_{X_{l},\boldsymbol{Z}}^{\prime}%
\mathbb{Z}\left(  \mathbb{Z}^{\prime}\mathbb{Z}\right)  ^{-1}\mathbb{Z}%
^{\prime}\boldsymbol{u}_{Y,\boldsymbol{Z}}.$ Note this is a $m_{n}\times1$
vector whose $j$-th element is $n^{-1/2}\boldsymbol{u}_{\phi_{j}\left(
X_{l}\right)  ,\boldsymbol{Z}}^{\prime}\mathbb{Z}\left(  \mathbb{Z}^{\prime
}\mathbb{Z}\right)  ^{-1}\mathbb{Z}^{\prime}\boldsymbol{u}_{Y,\boldsymbol{Z}}$
with the $n\times1$ vector $\boldsymbol{u}_{\phi_{j}\left(  X_{l}\right)
,\boldsymbol{Z}}$ being an collecting of the $n$ observations for $U_{\phi
_{j}\left(  X_{l}\right)  ,\boldsymbol{Z}}.$ Then%
\begin{align}
&  \Pr\left(  \left\Vert n^{-1/2}\boldsymbol{u}_{X_{l},\boldsymbol{Z}}%
^{\prime}\mathbb{Z}\left(  \mathbb{Z}^{\prime}\mathbb{Z}\right)
^{-1}\mathbb{Z}^{\prime}\boldsymbol{u}_{Y,\boldsymbol{Z}}\right\Vert ^{2}\geq
C_{1}m_{n}^{-1}v_{n}\right) \nonumber\\
&  =\Pr\left(  \sum_{j=1}^{m_{n}}\left(  n^{-1/2}\boldsymbol{u}_{\phi
_{j}\left(  X_{l}\right)  ,\boldsymbol{Z}}^{\prime}\mathbb{Z}\left(
\mathbb{Z}^{\prime}\mathbb{Z}\right)  ^{-1}\mathbb{Z}^{\prime}\boldsymbol{u}%
_{Y,\boldsymbol{Z}}\right)  ^{2}\geq C_{1}m_{n}^{-1}v_{n}\right) \nonumber\\
&  \leq\sum_{j=1}^{m_{n}}\Pr\left(  \left(  n^{-1/2}\boldsymbol{u}_{\phi
_{j}\left(  X_{l}\right)  ,\boldsymbol{Z}}^{\prime}\mathbb{Z}\left(
n^{-1}\mathbb{Z}^{\prime}\mathbb{Z}\right)  ^{-1}n^{-1/2}\mathbb{Z}^{\prime
}\boldsymbol{u}_{Y,\boldsymbol{Z}}\right)  ^{2}\geq C_{1}nm_{n}^{-2}%
v_{n}\right) \nonumber\\
&  \leq\sum_{j=1}^{m_{n}}\Pr\left(  \left\Vert n^{-1/2}\mathbb{Z}^{\prime
}\boldsymbol{u}_{\phi_{j}\left(  X_{l}\right)  ,\boldsymbol{Z}}\right\Vert
^{2}\left\Vert n^{-1/2}\mathbb{Z}^{\prime}\boldsymbol{u}_{Y,\boldsymbol{Z}%
}\right\Vert ^{2}\lambda_{\max}\left\{  \left(  n^{-1}\mathbb{Z}^{\prime
}\mathbb{Z}\right)  ^{-1}\right\}  ^{2}\geq C_{1}nm_{n}^{-2}v_{n}\right)  .
\label{EQ:c}%
\end{align}
We deal with the three random terms in the last line in turn. First, by
Assumptions \ref*{A:full_rank2} and \ref*{A:p}'\ and Lemma \ref{LE:XXQQ}, we
have%
\begin{equation}
\Pr\left(  \lambda_{\max}\left\{  \left(  n^{-1}\mathbb{Z}^{\prime}%
\mathbb{Z}\right)  ^{-1}\right\}  ^{2}\geq\left(  \frac{4}{3}B_{X1}^{-1}%
m_{n}\right)  ^{2}\right)  \leq C_{6}\exp\left(  -C_{7}n^{C_{8}}\right)
\label{EQ:c1}%
\end{equation}
for some positive constants $C_{6},C_{7}$ and $C_{8}.$ For the second term,%
\begin{align}
&  \Pr\left(  \left\Vert n^{-1/2}\mathbb{Z}^{\prime}\boldsymbol{u}_{\phi
_{j}\left(  X_{l}\right)  ,\boldsymbol{Z}}\right\Vert ^{2}\geq C_{1}%
^{1/2}\frac{3}{4}B_{X1}n^{1/2}m_{n}^{-2}v_{n}^{1/2}\right) \nonumber\\
&  =\Pr\left(  \sum_{k_{1}=1}^{\iota_{n}}\sum_{k_{2}=1}^{m_{n}}\left(
\sum_{i=1}^{n}n^{-1/2}\phi_{k_{2}}\left(  z_{k_{1},i}\right)  u_{\phi
_{j}\left(  X_{l}\right)  ,\boldsymbol{Z},i}\right)  ^{2}\geq C_{1}^{1/2}%
\frac{3}{4}B_{X1}n^{1/2}m_{n}^{-2}v_{n}^{1/2}\right) \nonumber\\
&  \leq\sum_{k_{1}=1}^{\iota_{n}}\sum_{k_{2}=1}^{m_{n}}\Pr\left(  \left(
\sum_{i=1}^{n}n^{-1/2}\phi_{k_{2}}\left(  z_{k_{1},i}\right)  u_{\phi
_{j}\left(  X_{l}\right)  ,\boldsymbol{Z},i}\right)  ^{2}\geq C_{1}^{1/2}%
\frac{3}{4}B_{X1}n^{1/2}m_{n}^{-3}\iota_{n}^{-1}v_{n}^{1/2}\right) \nonumber\\
&  =\sum_{k_{1}=1}^{\iota_{n}}\sum_{k_{2}=1}^{m_{n}}\Pr\left(  \left\vert
\sum_{i=1}^{n}\phi_{k_{2}}\left(  z_{k_{1},i}\right)  u_{\phi_{j}\left(
X_{l}\right)  ,\boldsymbol{Z},i}\right\vert \geq C_{9}n^{3/4}m_{n}^{-3/2}%
\iota_{n}^{-1/2}v_{n}^{1/4}\right) \nonumber\\
&  \leq\left\{
\begin{array}
[c]{c}%
m_{n}\iota_{n}\exp\left(  -C_{10}n^{1/2}m_{n}^{-3/2}\iota_{n}^{-1}\left(
m_{n}^{-1}v_{n}\right)  ^{1/2}\right) \\
C_{11}\exp\left(  -C_{12}n^{C_{13}}\right)
\end{array}
\right.
\begin{array}
[c]{c}%
\text{if }v_{n}\lesssim n^{\left(  s-2\right)  /\left[  4\left(  s+2\right)
\right]  }m_{n}^{3/2}\iota_{n}^{2}\\
\text{if }v_{n}\gtrsim n^{\left(  s-2\right)  /\left[  4\left(  s+2\right)
\right]  }m_{n}^{3/2}\iota_{n}^{2}%
\end{array}
, \label{EQ:c2}%
\end{align}
for some positive constants $C_{9},$ $C_{10},C_{11}$ and $C_{12}$, where in
the last line we apply Lemma \ref{LE:main_inequality} and we use the fact that
var$\left(  \phi_{k_{2}}\left(  z_{k_{1},i}\right)  u_{\phi_{j}\left(
X_{l}\right)  ,\boldsymbol{Z},i}\right)  \propto m_{n}^{-1}.$ For the third
term, similarly,%
\begin{align}
&  \Pr\left(  \left\Vert n^{-1/2}\mathbb{Z}^{\prime}\boldsymbol{u}%
_{Y,\boldsymbol{Z}}\right\Vert ^{2}\geq C_{1}^{1/2}\frac{3}{4}B_{X1}%
n^{1/2}m_{n}^{-2}v_{n}^{1/2}\right) \nonumber\\
&  =\Pr\left(  \sum_{k_{1}=1}^{\iota_{n}}\sum_{k_{2}=1}^{m_{n}}\left(
\sum_{i=1}^{n}n^{-1/2}\phi_{k_{2}}\left(  z_{k_{1},i}\right)
u_{Y,\boldsymbol{Z},i}\right)  ^{2}\geq C_{1}^{1/2}\frac{3}{4}B_{X1}%
n^{1/2}m_{n}^{-2}v_{n}^{1/2}\right) \nonumber\\
&  \leq\left\{
\begin{array}
[c]{c}%
m_{n}\iota_{n}\exp\left(  -C_{10}n^{1/2}m_{n}^{-3/2}\iota_{n}^{-1}\left(
m_{n}^{-1}v_{n}\right)  ^{1/2}\right) \\
C_{11}\exp\left(  -C_{12}n^{C_{13}}\right)
\end{array}
\right.
\begin{array}
[c]{c}%
\text{if }v_{n}\lesssim n^{\left(  s-2\right)  /\left[  4\left(  s+2\right)
\right]  }m_{n}^{3/2}\iota_{n}^{2}\\
\text{if }v_{n}\gtrsim n^{\left(  s-2\right)  /\left[  4\left(  s+2\right)
\right]  }m_{n}^{3/2}\iota_{n}^{2}%
\end{array}
, \label{EQ:c3}%
\end{align}
where without loss of generality we use the same constants in equation
(\ref{EQ:c2}).

Equation (\ref{EQ:c}) implies%
\begin{align*}
&  \Pr\left(  \left\Vert n^{-1/2}\boldsymbol{u}_{X_{l},\boldsymbol{Z}}%
^{\prime}\mathbb{Z}\left(  \mathbb{Z}^{\prime}\mathbb{Z}\right)
^{-1}\mathbb{Z}^{\prime}\boldsymbol{u}_{Y,\boldsymbol{Z}}\right\Vert ^{2}\geq
C_{1}m_{n}^{-1}v_{n}\right) \\
&  \leq\sum_{j=1}^{m_{n}}\left\{  \Pr\left(  \lambda_{\max}\left\{  \left(
n^{-1}\mathbb{Z}^{\prime}\mathbb{Z}\right)  ^{-1}\right\}  ^{2}\geq\left(
\frac{4}{3}B_{X1}^{-1}m_{n}\right)  ^{2}\right)  \right. \\
&  +\Pr\left(  \left\Vert n^{-1/2}\mathbb{Z}^{\prime}\boldsymbol{u}_{\phi
_{j}\left(  X_{l}\right)  ,\boldsymbol{Z}}\right\Vert ^{2}\geq C_{1}%
^{1/2}\frac{3}{4}B_{X1}n^{1/2}m_{n}^{-2}v_{n}^{1/2}\right) \\
&  \left.  +\Pr\left(  \left\Vert n^{-1/2}\mathbb{Z}^{\prime}\boldsymbol{u}%
_{\phi_{j}\left(  X_{l}\right)  ,\boldsymbol{Z}}\right\Vert ^{2}\geq
C_{1}^{1/2}\frac{3}{4}B_{X1}n^{1/2}m_{n}^{-2}v_{n}^{1/2}\right)  \right\}  .
\end{align*}
Substituting the results of equations (\ref{EQ:c1}), (\ref{EQ:c2}), and
(\ref{EQ:c3}) into the above yields%
\begin{align}
&  \Pr\left(  \left\Vert n^{-1/2}\boldsymbol{u}_{X_{l},\boldsymbol{Z}}%
^{\prime}\mathbb{Z}\left(  \mathbb{Z}^{\prime}\mathbb{Z}\right)
^{-1}\mathbb{Z}^{\prime}\boldsymbol{u}_{Y,\boldsymbol{Z}}\right\Vert ^{2}\geq
C_{1}m_{n}^{-1}v_{n}\right) \label{EQ:part2}\\
&  \leq\left\{
\begin{array}
[c]{l}%
\sum_{j=1}^{m_{n}}\left[  C_{6}\exp\left(  -C_{7}n^{C_{8}}\right)
+2m_{n}\iota_{n}\exp\left(  -C_{10}n^{1/2}m_{n}^{-3/2}\iota_{n}^{-1}\left(
m_{n}^{-1}v_{n}\right)  ^{1/2}\right)  \right] \\
\sum_{j=1}^{m_{n}}\left[  C_{6}\exp\left(  -C_{7}n^{C_{8}}\right)
+2C_{11}\exp\left(  -C_{12}n^{C_{13}}\right)  \right]
\end{array}
\right.
\begin{array}
[c]{c}%
\text{if }v_{n}\lesssim n^{\frac{s-2}{4\left(  s+2\right)  }}m_{n}^{3/2}%
\iota_{n}^{2}\\
\text{if }v_{n}\gtrsim n^{\frac{s-2}{4\left(  s+2\right)  }}m_{n}^{3/2}%
\iota_{n}^{2}%
\end{array}
\nonumber\\
&  \leq\left\{
\begin{array}
[c]{l}%
C_{13}\exp\left(  -C_{14}n^{C_{15}}\right)  +m_{n}^{2}\iota_{n}\exp\left(
-C_{16}n^{1/2}m_{n}^{-3/2}\iota_{n}^{-1}\left(  m_{n}^{-1}v_{n}\right)
^{1/2}\right) \\
C_{17}\exp\left(  -C_{18}n^{C_{19}}\right)
\end{array}
\right.
\begin{array}
[c]{c}%
\text{if }v_{n}\lesssim n^{\frac{s-2}{4\left(  s+2\right)  }}m_{n}^{3/2}%
\iota_{n}^{2}\\
\text{if }v_{n}\gtrsim n^{\frac{s-2}{4\left(  s+2\right)  }}m_{n}^{3/2}%
\iota_{n}^{2}%
\end{array}
,\nonumber
\end{align}
for some positive constants $C_{13},\ldots,C_{19}.$

Note that when $\upsilon_{n}=\varsigma_{n},$ $n^{1/2}m_{n}^{-3/2}\iota
_{n}^{-1}\left(  m_{n}^{-1}\varsigma_{n}\right)  ^{1/2}\gg n^{1/2-1/6B_{m}%
-B_{p^{\ast\ast}}}=n^{C_{22}}$ for a positive $C_{22}.$ $C_{22}$ exists due to
Assumption \ref*{A:p}'. Then,
\begin{align}
&  \Pr\left(  \left\Vert n^{-1/2}\boldsymbol{u}_{X_{l},\boldsymbol{Z}}%
^{\prime}\mathbb{Z}\left(  \mathbb{Z}^{\prime}\mathbb{Z}\right)
^{-1}\mathbb{Z}^{\prime}\boldsymbol{u}_{Y,\boldsymbol{Z}}\right\Vert ^{2}\geq
C_{1}m_{n}^{-1}\varsigma_{n}\right) \label{EQ:u2}\\
&  \leq C_{20}\exp\left(  -C_{21}n^{C_{22}}\right) \nonumber
\end{align}
for some positive constants $C_{20},C_{21},$ and $C_{22},$ due to equation
(\ref{EQ:part2}) and the above observation. Therefore, by equations
(\ref{EQ:u1}) and (\ref{EQ:u2}),
\begin{align*}
&  \Pr\left(  \left\Vert n^{-1/2}\boldsymbol{\tilde{u}}_{l,\boldsymbol{Z}%
}\right\Vert ^{2}\geq C_{1}m_{n}^{-1}\varsigma_{n}\right) \\
&  =\Pr\left(  \left\Vert n^{-1/2}\left(  \boldsymbol{u}_{X_{l},\boldsymbol{Z}%
}^{\prime}\boldsymbol{u}_{Y,\boldsymbol{Z}}-\boldsymbol{\eta}%
_{l,\boldsymbol{Z}}\right)  -n^{-1/2}\boldsymbol{u}_{X_{l},\boldsymbol{Z}%
}^{\prime}\mathbb{Z}\left(  \mathbb{Z}^{\prime}\mathbb{Z}\right)
^{-1}\mathbb{Z}^{\prime}\boldsymbol{u}_{Y,\boldsymbol{Z}}\right\Vert ^{2}\geq
C_{1}m_{n}^{-1}\varsigma_{n}\right) \\
&  \leq\Pr\left(  \left\Vert n^{-1/2}\left(  \boldsymbol{u}_{X_{l}%
,\boldsymbol{Z}}^{\prime}\boldsymbol{u}_{Y,\boldsymbol{Z}}-\boldsymbol{\eta
}_{l,\boldsymbol{Z}}\right)  \right\Vert ^{2}\geq\frac{1}{2}C_{1}m_{n}%
^{-1}\varsigma_{n}\right)  +\Pr\left(  \left\Vert n^{-1/2}\boldsymbol{u}%
_{X_{l},\boldsymbol{Z}}^{\prime}\mathbb{Z}\left(  \mathbb{Z}^{\prime
}\mathbb{Z}\right)  ^{-1}\mathbb{Z}^{\prime}\boldsymbol{u}_{Y,\boldsymbol{Z}%
}\right\Vert ^{2}\geq\frac{1}{2}C_{1}m_{n}^{-1}\varsigma_{n}\right) \\
&  \leq2\exp\left(  -C_{2}m_{n}^{-1}\varsigma_{n}+\log m_{n}\right)
+C_{20}\exp\left(  -C_{21}n^{C_{22}}\right) \\
&  \leq n^{-M}+C_{20}\exp\left(  -C_{21}n^{C_{22}}\right)  ,
\end{align*}
for an arbitrarily large constant $M,$ where in the fourth line we use the
fact that the two terms in equation are of the same order due to $\iota
_{n}n_{n}^{1/2}\varsigma_{n}^{1/2}\ll nm_{n}^{-1}$ by Assumption \ref*{A:p}'.
This completes the proof of the lemma.\medskip
\end{proof}

\begin{proof}
[Proof of Lemma \ref{LE:error_var_Z}]To save notations, we absorb $X_{l}$ into
$\boldsymbol{Z}$ which now becomes a $\left(  \iota_{n}+1\right)  \times1$
vector. Assumption \ref*{A:full_rank2} holds for $\boldsymbol{Z.}$ Recall
$\Phi_{\boldsymbol{Z}}=E\left[  P^{m_{n}}\left(  \boldsymbol{Z}\right)
P^{m_{n}}\left(  \boldsymbol{Z}\right)  ^{\prime}\right]  $ and $M_{\mathbb{Z}%
}=I_{n}-\mathbb{Z}\left(  \mathbb{Z}^{\prime}\mathbb{Z}\right)  ^{-1}%
\mathbb{Z}^{\prime}.$ By definition,
\[
U=Y-P^{m_{n}}\left(  \boldsymbol{Z}\right)  ^{\prime}\Phi_{\boldsymbol{Z}%
}^{-1}E\left[  P^{m_{n}}\left(  \boldsymbol{Z}\right)  Y\right]  .
\]
Let $u_{i}=Y_{i}-P^{m_{n}}\left(  \boldsymbol{z}_{i}\right)  ^{\prime}%
\Phi_{\boldsymbol{Z}}^{-1}E\left[  P^{m_{n}}\left(  \boldsymbol{Z}\right)
Y\right]  $ and $\boldsymbol{u}=\left(  u_{1},\ldots,u_{n}\right)  .$ By the
definition of $\hat{\sigma}_{l,\boldsymbol{Z}}^{2},$\footnote{We defined
$\hat{\sigma}_{l,\boldsymbol{Z}}^{2}$ based on $\hat{u}_{i}$ from partitioned
regressions. This $\hat{u}_{i}$ is the same as the residual from the joint
regression.}%
\[
\hat{\sigma}_{l,\boldsymbol{Z}}^{2}=n^{-1}\boldsymbol{u}^{\prime}%
M_{\mathbb{Z}}\boldsymbol{u}.
\]
Then we are able to reach the same conclusion as in the proof of Lemma
\ref{LE:error_var}, with the help of Lemma \ref{LE:error_bound}.\medskip
\end{proof}

\begin{proof}
[Proof of Lemma \ref{LE:bound_error}]Note that%
\begin{align*}
\Pr\left(  \boldsymbol{\tilde{u}}_{l,\boldsymbol{Z}}^{\prime}\left(
\hat{\sigma}_{l,\boldsymbol{Z}}^{2}\mathbb{X}_{l}^{\prime}M_{\mathbb{Z}%
}\mathbb{X}_{l}\right)  ^{-1}\boldsymbol{\tilde{u}}_{l,\boldsymbol{Z}}%
\geq\varsigma_{n}\right)   &  \leq\Pr\left(  \lambda_{\max}\left\{  \left(
\hat{\sigma}_{l,\boldsymbol{Z}}^{2}n^{-1}\mathbb{X}_{l}^{\prime}M_{\mathbb{Z}%
}\mathbb{X}_{l}\right)  ^{-1}\right\}  \left\Vert n^{-1/2}\boldsymbol{\tilde
{u}}_{l,\boldsymbol{Z}}\right\Vert ^{2}\geq\varsigma_{n}\right) \\
&  \leq\Pr\left(  \lambda_{\max}\left\{  \left(  \hat{\sigma}%
_{l,\boldsymbol{Z}}^{2}n^{-1}\mathbb{X}_{l}^{\prime}M_{\mathbb{Z}}%
\mathbb{X}_{l}\right)  ^{-1}\right\}  \geq2\sigma^{-2}B_{X1}^{-1}m_{n}\right)
\\
&  +\Pr\left(  \left\Vert n^{-1/2}\boldsymbol{\tilde{u}}_{l,\boldsymbol{Z}%
}\right\Vert ^{2}\geq\left(  2\sigma^{-2}\right)  ^{-1}B_{X1}m_{n}%
^{-1}\varsigma_{n}\right)  .
\end{align*}
Then we can reach the conclusion by applying Lemma \ref{LE:lambdamax} and the
last part of Lemma \ref{LE:error_bound}.\medskip
\end{proof}

\begin{proof}
[Proof of Lemma \ref{LE:allsignals}]We prove the lemma by mathematical
induction. First, we show that the conclusion holds when $k=1,$ viz., for the
first stage. We bound the probability of $\mathcal{N}_{1}$ (one or more noise
variables are selected at stage 1) as%
\begin{equation}
\Pr\left(  \mathcal{N}_{1}\right)  =\Pr\left(  \cup_{l=p^{\ast}+p^{\ast\ast
}+1}^{p_{n}}\mathcal{L}_{l,1}\right)  =\Pr\left(  \cup_{l=p^{\ast}+p^{\ast
\ast}+1}^{p_{n}}\mathcal{B}_{l,1}\right)  \leq\sum_{l=p^{\ast}+p^{\ast\ast}%
+1}^{p_{n}}\Pr\left(  \mathcal{B}_{l,1}\right)  . \label{EQ:n1}%
\end{equation}
By the definition in Table 2, $\theta_{l}\lesssim\log\left(  m_{n}\right)
^{1/2}\left(  m_{n}/n\right)  ^{1/2}$ for noise variables, viz., for $p^{\ast
}+p^{\ast\ast}+1\leq l\leq p_{n}$. Then we are able to apply the first part of
Proposition \ref{TH:main1} to bound $\Pr\left(  \mathcal{B}_{l,1}\right)  $
for each $l\in\lbrack p^{\ast}+p^{\ast\ast}+1,p_{n}]$ to obtain%
\[
\Pr\left(  \mathcal{B}_{l,1}\right)  \leq n^{-M}+C_{1}\exp\left(
-C_{2}n^{C_{3}}\right)
\]
for any arbitrarily large constant $M$ and some positive constants $C_{1},$
$C_{2}$ and $C_{3}.$ Then by Assumption \ref*{A:p} and (\ref{EQ:n1}), we have%
\begin{align}
\Pr\left(  \mathcal{N}_{1}\right)   &  \leq\sum_{l=p^{\ast}+p^{\ast\ast}%
+1}^{p_{n}}\left(  n^{-M}+C_{1}\exp\left(  -C_{2}n^{C_{3}}\right)  \right)
\nonumber\\
&  \leq p_{n}n^{-M}+p_{n}C_{1}\exp\left(  -C_{2}n^{C_{3}}\right) \nonumber\\
&  \leq n^{-M_{1}}+C_{7}\exp\left(  -C_{5}n^{C_{6}}\right)  , \label{EQ:n1_1}%
\end{align}
where the last equality holds by some $0<M_{1}\leq M-B_{p}$ and some constants
$C_{5},C_{6}$ and $C_{7}.$ Since $M$ is arbitrarily large and $B_{p}$ is a
fixed number, $M_{1}$ can also be arbitrarily large. Notice that $\Pr\left(
\mathcal{D}_{1}\right)  =1-\Pr\left(  \mathcal{N}_{1}\right)  ,$ we get%
\[
\Pr\left(  \mathcal{D}_{1}\right)  \geq1-n^{-M_{1}}-C_{7}\exp\left(
-C_{5}n^{C_{6}}\right)  .
\]

We now show the conclusion for $k=2.$ Notice that $\mathcal{N}_{2}%
=\mathcal{N}_{1}\cup\left\{  \cup_{l=p^{\ast}+p^{\ast\ast}+1}^{p_{n}%
}\mathcal{B}_{l,2}\right\}  .$ That is, some noise variables selected up to
and including stage 2 are either selected at stage 1 or at stage 2. To apply
Proposition \ref{TH:main1_ms}, we need $\mathcal{D}_{1}$: the pre-selected
variables are either signals or pseudo-signals at stage 2. In view of this and
the fact that $\mathcal{D}_{1}=\mathcal{N}_{1}^{c}$, we have%
\begin{align}
\Pr\left(  \mathcal{N}_{2}\right)   &  =\Pr\left(  \mathcal{N}_{1}%
\cup(\left\{  \cup_{l=p^{\ast}+p^{\ast\ast}+1}^{p_{n}}\mathcal{B}%
_{l,2}\right\}  \cap\mathcal{D}_{1})\right) \nonumber\\
&  \leq\Pr\left(  \mathcal{N}_{1}\right)  +\Pr\left(  \left\{  \cup
_{l=p^{\ast}+p^{\ast\ast}+1}^{p_{n}}\mathcal{B}_{l,2}\right\}  \cap
\mathcal{D}_{1}\right) \nonumber\\
&  =\Pr\left(  \mathcal{N}_{1}\right)  +\Pr\left(  \cup_{l=p^{\ast}%
+p^{\ast\ast}+1}^{p_{n}}\mathcal{B}_{l,2}\cap\mathcal{D}_{1}\right)
\nonumber\\
&  \leq\Pr\left(  \mathcal{N}_{1}\right)  +\Pr\left(  \cup_{l=p^{\ast}%
+p^{\ast\ast}+1}^{p_{n}}\mathcal{B}_{l,2}|\mathcal{D}_{1}\right)  \Pr\left(
\mathcal{D}_{1}\right) \nonumber\\
&  \leq\Pr\left(  \mathcal{N}_{1}\right)  +\Pr\left(  \cup_{l=p^{\ast}%
+p^{\ast\ast}+1}^{p_{n}}\mathcal{B}_{l,2}|\mathcal{D}_{1}\right)  .
\label{EQ:n2}%
\end{align}
where the first equality follows from the fact $\Pr\left(  A\cup B\right)
=\Pr\left(  A\cup\left(  B\cap A^{c}\right)  \right)  $ for any two events $A$
and $B$. Since the second term in (\ref{EQ:n2}) is conditional on
$\mathcal{D}_{1}$, the pre-selected variables in $\boldsymbol{Z}_{\left(
1\right)  }$ are either signals or pseudo signals. By Assumption
\ref*{A:noisevariable}, the net effect of the noise variable $X_{l}$ on $Y$
with $\boldsymbol{Z}_{\left(  1\right)  }$ as pre-selected variables satisfies
$\boldsymbol{\theta}_{l,\boldsymbol{Z}_{\left(  1\right)  }}\lesssim
\log\left(  m_{n}\right)  ^{1/2}m_{n}^{1/2}n^{-1/2}$. Then we are able to
apply the first part of Proposition \ref{TH:main1_ms} to obtain%
\begin{align}
\Pr\left(  \cup_{l=p^{\ast}+p^{\ast\ast}+1}^{p_{n}}\mathcal{B}_{l,2}%
|\mathcal{D}_{1}\right)   &  \leq\sum_{l=p^{\ast}+p^{\ast\ast}+1}^{p_{n}}%
\Pr\left(  \mathcal{B}_{l,2}|\mathcal{D}_{1}\right)  \leq\sum_{l=p^{\ast
}+p^{\ast\ast}+1}^{p_{n}}\left(  n^{-M}+C_{1}\exp\left(  -C_{2}n^{C_{3}%
}\right)  \right) \nonumber\\
&  \leq n^{-M_{1}}+C_{7}\exp\left(  -C_{5}n^{C_{6}}\right)  , \label{EQ:n3}%
\end{align}
where without loss of generality we use the same constants as in equation
(\ref{EQ:n1_1}). Combining (\ref{EQ:n1_1})--(\ref{EQ:n3}) yields
\[
\Pr\left(  \mathcal{N}_{2}\right)  \leq2n^{-M_{1}}+2C_{7}\exp\left(
-C_{5}n^{C_{6}}\right)
\]
and
\[
\Pr\left(  \mathcal{D}_{2}\right)  =1-\Pr\left(  \mathcal{N}_{2}\right)
\geq1-2n^{-M_{1}}-2C_{7}\exp\left(  -C_{5}n^{C_{6}}\right)  .
\]

Now suppose we have shown the results for stage $k-1$ with $k\geq3,$ and
obtained
\begin{equation}
\Pr\left(  \mathcal{N}_{k-1}\right)  \leq\left(  k-1\right)  n^{-M_{1}%
}+\left(  k-1\right)  C_{7}\exp\left(  -C_{5}n^{C_{6}}\right)  .
\label{EQ:nk-1}%
\end{equation}
Note that $\mathcal{N}_{k}=\mathcal{N}_{k-1}\cup\left\{  \cup_{l=p^{\ast
}+p^{\ast\ast}+1}^{p_{n}}\mathcal{B}_{l,k}\right\}  ,$ viz., some noise
variables selected up to and including stage $k$ are either selected at stage
$k$ or before stage $k$. Noting that $\mathcal{D}_{k-1}=\mathcal{N}_{k-1}^{c}$
and following the derivation of equation (\ref{EQ:n2}), we have%
\begin{align}
\Pr\left(  \mathcal{N}_{k}\right)   &  =\Pr\left(  \mathcal{N}_{k-1}%
\cup\left(  \left\{  \cup_{l=p^{\ast}+p^{\ast\ast}+1}^{p_{n}}\mathcal{B}%
_{l,k}\right\}  \cap\mathcal{D}_{k-1}\right)  \right) \nonumber\\
&  \leq\Pr\left(  \mathcal{N}_{k-1}\right)  +\Pr\left(  \cup_{l=p^{\ast
}+p^{\ast\ast}+1}^{p_{n}}\mathcal{B}_{l,k}|\mathcal{D}_{k-1}\right)  .
\label{EQ:nk}%
\end{align}
Conditioning on $\mathcal{D}_{k-1}$ implies that the pre-selected variables
$\boldsymbol{Z}_{\left(  k-1\right)  }$ are either signals or pseudo signals
for stage $k$. By Assumption \ref*{A:noisevariable} again, the net effect of
the noise variable $X_{l}$ on $Y$ with $\boldsymbol{Z}_{\left(  k-1\right)  }$
as pre-selected variables satisfies $\boldsymbol{\theta}_{l,\boldsymbol{Z}%
_{\left(  k-1\right)  }}\lesssim\log\left(  m_{n}\right)  ^{1/2}m_{n}%
^{1/2}n^{-1/2}$. Applying the first part of Proposition \ref{TH:main1_ms} to
the second term in (\ref{EQ:nk}) yields%
\begin{align}
\Pr\left(  \cup_{l=p^{\ast}+p^{\ast\ast}+1}^{p_{n}}\mathcal{B}_{l,k}%
|\mathcal{D}_{k-1}\right)   &  \leq\sum_{l=p^{\ast}+p^{\ast\ast}+1}^{p_{n}}%
\Pr\left(  \mathcal{B}_{l,k}|\mathcal{D}_{k-1}\right) \nonumber\\
&  \leq\sum_{l=p^{\ast}+p^{\ast\ast}+1}^{p_{n}}\left(  n^{-M}+C_{1}\exp\left(
-C_{2}n^{C_{3}}\right)  \right) \nonumber\\
&  \leq n^{-M_{1}}+C_{7}\exp\left(  -C_{5}n^{C_{6}}\right)  . \label{EQ:nk1}%
\end{align}
Combining (\ref{EQ:nk-1})--(\ref{EQ:nk1}), we obtain%
\[
\Pr\left(  \mathcal{N}_{k}\right)  \leq kn^{-M_{1}}+kC_{7}\exp\left(
-C_{5}n^{C_{6}}\right)  ,
\]
and\textbf{ }%
\[
\Pr\left(  \mathcal{D}_{k}\right)  =1-\Pr\left(  \mathcal{N}_{k}\right)
\geq1-kn^{-M_{1}}-kC_{7}\exp\left(  -C_{5}n^{C_{6}}\right)  .
\]
Apparently, when $k$ is fixed or divergent to infinity at a rate no faster
than $n^{a}$ for some $a>0,$ we can write
\[
\Pr\left(  \mathcal{D}_{k}\right)  \geq1-n^{-M_{2}}-C_{7}\exp\left(
-C_{8}n^{C_{9}}\right)  \text{ for any }k\leq n^{a},
\]
where $M_{2}$ is a large positive constant, $M_{2}\leq M_{1}-a.$
\end{proof}

\bigskip

\section{Additional Simulation and Application Results\label{APP:tables}}

In this section we present some additional results for the simulation and application.

\subsection{Additional Simulation Results}

In this subsection we report the simulation results for DGPs 2 - 5 and 7 - 10
in Tables \ref{tabledgp2}--\ref{tabledgp10}. The description of these DGPs and
the summary of the simulation results are given in Sections 4.1 and 4.3 of the
paper, respectively.

In order to investigate the finite sample performance of POST-OCMT procedure
compared to that of OCMT procedure under linear model, we consider the
following DGP.

\textbf{DGP 11.} \textbf{Four signals, many pseudo-signals, and one hidden
signal}. The covariates are generated as DGP 6$:$%
\begin{align*}
X_{j}  &  =W_{j},\text{ }j=1,2,\text{ }\\
X_{j}  &  =\frac{W_{j}+U_{1}}{2},\text{ }j=3,4,\\
X_{5}  &  =U_{1},\\
X_{j}  &  =\frac{4X_{1}+\left(  j-5\right)  W_{j-1}}{j-1},\text{
}j=6,10,14,18,\ldots,\\
X_{j}  &  =\frac{4X_{2}+\left(  j-5\right)  W_{j-1}}{j-1},\text{
}j=7,11,15,19,\ldots,\\
X_{j}  &  =\frac{4X_{3}+\left(  j-5\right)  W_{j-1}}{j-1},\text{
}j=8,12,16,20,\ldots,\text{ and}\\
X_{j}  &  =\frac{4X_{4}+\left(  j-5\right)  W_{j-1}}{j-1},\text{
}j=9,13,17,21,\ldots
\end{align*}
where $W_{j},$ $j=1,\ldots,p_{n}-1,$ and $U_{1}$\ are independent draws from
$U(0,1)$. The dependent variable $Y$ is generated as
\[
Y=3(X_{1}-0.5)+3(X_{2}-0.5)+6(X_{3}-0.25)+6\left(  X_{4}-0.25\right)
-3X_{5}+\varepsilon
\]
where $\varepsilon$ are independent draws from $N(0,1)$. Then, $p^{\ast}=5$
with one hidden signal $(X_{5})$. In this DGP, we only report the results of
our procedure (OCMT and Post-OCMT). The simulation results for DGP 11 are
reported in Table \ref{tabledgp11}. Clearly, from this table, we can see that
POST-OCMT works very well in picking the true signals including the hidden one
and has good out-sample forecasting performance, due to the fact that it
eliminates the impact of pseudo-signals.

\begin{table}[p]
\caption{DGP 2}%
\label{tabledgp2}%
\centering{}\centering{ } \resizebox{!}{0.67\textwidth}{
\begin{tabular}
[c]{l|ccccccc}\hline\hline
\multicolumn{8}{c}{Panel 1: $n=200$, $p=100$}\\\hline
& NV & TPR & FPR & FDR & CS & STEP & RMSFE\\\hline
POST OCMT & 4.1190 & 0.9670 & 0.0026 & 0.0452 & 0.7620 & - & 1.0795\\
OCMT & 5.2950 & 0.9862 & 0.0141 & 0.2061 & 0.0800 & 1.1330 & 1.1899\\
AGLASSO & 4.0710 & 0.9537 & 0.0027 & 0.0469 & 0.7490 & - & 1.0872\\
BAGGING   & -      & -      & -      & -      & -      & -      & 1.4335   \\
RF        & -      & -      & -      & -      & -      & -      & 1.4581   \\\hline
\multicolumn{8}{c}{Panel 2: $n=200$, $p=200$}\\\hline
& NV & TPR & FPR & FDR & CS & STEP & RMSFE\\\hline
POST OCMT & 4.1180 & 0.9690 & 0.0012 & 0.0438 & 0.7640 & - & 1.0809\\
OCMT & 5.3560 & 0.9908 & 0.0071 & 0.2107 & 0.0740 & 1.1510 & 1.2064\\
AGLASSO & 4.0350 & 0.9447 & 0.0013 & 0.0468 & 0.7470 & - & 1.0911\\
BAGGING   & -      & -      & -      & -      & -      & -      & 1.4621   \\
RF        & -      & -      & -      & -      & -      & -      & 1.4914   \\\hline
\multicolumn{8}{c}{Panel 3: $n=200$, $p=1000$}\\\hline
& NV & TPR & FPR & FDR & CS & STEP & RMSFE\\\hline
POST OCMT & 4.1200 & 0.9617 & 0.0003 & 0.0494 & 0.7420 & - & 1.0825\\
OCMT & 5.3350 & 0.9828 & 0.0014 & 0.2134 & 0.0710 & 1.1910 & 1.2387\\
AGLASSO & 3.8194 & 0.8898 & 0.0003 & 0.0471 & 0.6863 & - & 1.1267\\
BAGGING   & -      & -      & -      & -      & -      & -      & 1.5210   \\
RF        & -      & -      & -      & -      & -      & -      & 1.5591   \\\hline
\multicolumn{8}{c}{Panel 4: $n=400$, $p=100$}\\\hline
& NV & TPR & FPR & FDR & CS & STEP & RMSFE\\\hline
POST OCMT & 4.0490 & 0.9900 & 0.0009 & 0.0161 & 0.9130 & - & 1.0483\\
OCMT & 5.1380 & 0.9942 & 0.0121 & 0.1802 & 0.1380 & 1.1000 & 1.1152\\
AGLASSO & 4.0450 & 0.9892 & 0.0009 & 0.0161 & 0.9120 & - & 1.0468\\
BAGGING   & -      & -      & -      & -      & -      & -      & 1.3167   \\
RF        & -      & -      & -      & -      & -      & -      & 1.3406   \\\hline
\multicolumn{8}{c}{Panel 5: $n=400$, $p=200$}\\\hline
& NV & TPR & FPR & FDR & CS & STEP & RMSFE\\\hline
POST OCMT & 4.0500 & 0.9848 & 0.0006 & 0.0204 & 0.8930 & - & 1.0495\\
OCMT & 5.2130 & 0.9912 & 0.0064 & 0.1926 & 0.1040 & 1.0990 & 1.1244\\
AGLASSO & 4.0540 & 0.9888 & 0.0005 & 0.0180 & 0.9020 & - & 1.0530\\
BAGGING   & -      & -      & -      & -      & -      & -      & 1.3475   \\
RF        & -      & -      & -      & -      & -      & -      & 1.3752   \\\hline
\multicolumn{8}{c}{Panel 6: $n=400$, $p=1000$}\\\hline
& NV & TPR & FPR & FDR & CS & STEP & RMSFE\\\hline
POST OCMT & 4.0480 & 0.9878 & 0.0001 & 0.0177 & 0.9050 & - & 1.0535\\
OCMT & 5.1990 & 0.9948 & 0.0012 & 0.1875 & 0.1290 & 1.0980 & 1.1286\\
AGLASSO & 4.0889 & 0.9874 & 0.0001 & 0.0246 & 0.8666 & - &
1.0530\\
BAGGING   & -      & -      & -      & -      & -      & -      & 1.4060   \\
RF        & -      & -      & -      & -      & -      & -      & 1.4371   \\\hline\hline
\end{tabular}} \bigskip\end{table}

\begin{table}[p]
\caption{DGP 3}%
\label{tabledgp3}%
\centering{}\centering{ } \resizebox{!}{0.67\textwidth}{
\begin{tabular}
[c]{l|ccccccc}\hline\hline
\multicolumn{8}{c}{Panel 1: $n=200$, $p=100$}\\\hline
& NV & TPR & FPR & FDR & CS & STEP & RMSFE\\\hline
POST OCMT & 4.4790 & 0.8678 & 0.0015 & 0.0202 & 0.4730 & - & 1.1418\\
OCMT & 4.6100 & 0.8678 & 0.0028 & 0.0376 & 0.3940 & 1.5800 & 1.2230\\
AGLASSO & 3.5440 & 0.6950 & 0.0007 & 0.0110 & 0.0100 & - & 1.2593\\
BAGGING   & -      & -      & -      & -      & -      & -      & 1.4274   \\
RF        & -      & -      & -      & -      & -      & -      & 1.4668   \\\hline
\multicolumn{8}{c}{Panel 2: $n=200$, $p=200$}\\\hline
& NV & TPR & FPR & FDR & CS & STEP & RMSFE\\\hline
POST OCMT & 4.3010 & 0.8284 & 0.0008 & 0.0235 & 0.3490 & - & 1.1612\\
OCMT & 4.4600 & 0.8284 & 0.0016 & 0.0443 & 0.2710 & 1.4590 & 1.2429\\
AGLASSO & 3.2990 & 0.6442 & 0.0004 & 0.0128 & 0.0030 & - & 1.2679\\
BAGGING   & -      & -      & -      & -      & -      & -      & 1.4584   \\
RF        & -      & -      & -      & -      & -      & -      & 1.4971   \\\hline
\multicolumn{8}{c}{Panel 3: $n=200$, $p=1000$}\\\hline
& NV & TPR & FPR & FDR & CS & STEP & RMSFE\\\hline
POST OCMT & 3.4590 & 0.6730 & 0.0001 & 0.0170 & 0.0680 & - & 1.2349\\
OCMT & 3.4900 & 0.6730 & 0.0001 & 0.0215 & 0.0590 & 1.0930 & 1.2481\\
AGLASSO & 2.6024 & 0.5105 & 0.0001 & 0.0087 & 0.0000 & - & 1.3464\\
BAGGING   & -      & -      & -      & -      & -      & -      & 1.5109   \\
RF        & -      & -      & -      & -      & -      & -      & 1.5544   \\\hline
\multicolumn{8}{c}{Panel 4: $n=400$, $p=100$}\\\hline
& NV & TPR & FPR & FDR & CS & STEP & RMSFE\\\hline
POST OCMT & 5.0050 & 0.9990 & 0.0001 & 0.0014 & 0.9860 & - & 1.0446\\
OCMT & 5.0330 & 0.9990 & 0.0004 & 0.0052 & 0.9650 & 1.9950 & 1.0574\\
AGLASSO & 4.3790 & 0.8642 & 0.0006 & 0.0084 & 0.2800 & - & 1.1136\\
BAGGING   & -      & -      & -      & -      & -      & -      & 1.3272   \\
RF        & -      & -      & -      & -      & -      & -      & 1.3683   \\\hline
\multicolumn{8}{c}{Panel 5: $n=400$, $p=200$}\\\hline
& NV & TPR & FPR & FDR & CS & STEP & RMSFE\\\hline
POST OCMT & 5.0020 & 0.9972 & 0.0001 & 0.0022 & 0.9750 & - & 1.0460\\
OCMT & 5.0460 & 0.9972 & 0.0003 & 0.0077 & 0.9470 & 1.9960 & 1.0643\\
AGLASSO & 4.2520 & 0.8408 & 0.0002 & 0.0072 & 0.1850 & - & 1.1305\\
BAGGING   & -      & -      & -      & -      & -      & -      & 1.3574   \\
RF        & -      & -      & -      & -      & -      & -      & 1.4001   \\\hline
\multicolumn{8}{c}{Panel 6: $n=400$, $p=1000$}\\\hline
& NV & TPR & FPR & FDR & CS & STEP & RMSFE\\\hline
POST OCMT & 4.8820 & 0.9684 & 0.0000 & 0.0057 & 0.8550 & - & 1.0632\\
OCMT & 4.9460 & 0.9684 & 0.0001 & 0.0136 & 0.8180 & 1.8920 & 1.0924\\
AGLASSO & 4.0764 & 0.8044 & 0.0000 & 0.0087 & 0.0255 & - &
1.1478\\
BAGGING   & -      & -      & -      & -      & -      & -      & 1.4128   \\
RF        & -      & -      & -      & -      & -      & -      & 1.4531   \\\hline\hline
\end{tabular}}\end{table}

\begin{table}[p]
\caption{DGP 4}%
\label{tabledgp4}%
\centering{}\centering{ } \resizebox{!}{0.67\textwidth}{
\begin{tabular}
[c]{l|ccccccc}\hline\hline
\multicolumn{8}{c}{Panel 1: $n=200$, $p=100$}\\\hline
& NV & TPR & FPR & FDR & CS & STEP & RMSFE\\\hline
POST OCMT & 4.3760 & 0.7934 & 0.0043 & 0.0683 & 0.3210 & - & 1.1735\\
OCMT & 5.9320 & 0.8120 & 0.0195 & 0.2552 & 0.0000 & 1.5270 & 1.5330\\
AGLASSO & 3.1090 & 0.5736 & 0.0025 & 0.0477 & 0.0050 & - & 1.3469\\
BAGGING   & -      & -      & -      & -      & -      & -      & 1.4197   \\
RF        & -      & -      & -      & -      & -      & -      & 1.4362   \\\hline
\multicolumn{8}{c}{Panel 2: $n=200$, $p=200$}\\\hline
& NV & TPR & FPR & FDR & CS & STEP & RMSFE\\\hline
POST OCMT & 4.1070 & 0.7468 & 0.0019 & 0.0667 & 0.2580 & - & 1.1956\\
OCMT & 5.6420 & 0.7614 & 0.0094 & 0.2617 & 0.0010 & 1.4090 & 1.4935\\
AGLASSO & 2.7500 & 0.5052 & 0.0011 & 0.0457 & 0.0020 & - & 1.3299\\
BAGGING   & -      & -      & -      & -      & -      & -      & 1.4499   \\
RF        & -      & -      & -      & -      & -      & -      & 1.4663   \\\hline
\multicolumn{8}{c}{Panel 3: $n=200$, $p=1000$}\\\hline
& NV & TPR & FPR & FDR & CS & STEP & RMSFE\\\hline
POST OCMT & 3.3360 & 0.5968 & 0.0004 & 0.0745 & 0.0470 & - & 1.2652\\
OCMT & 4.5670 & 0.6164 & 0.0015 & 0.2472 & 0.0000 & 1.0850 & 1.3430\\
AGLASSO & 2.0833 & 0.3773 & 0.0002 & 0.0412 & 0.0000 & - & 1.3978\\
BAGGING   & -      & -      & -      & -      & -      & -      & 1.5041   \\
RF        & -      & -      & -      & -      & -      & -      & 1.5251   \\\hline
\multicolumn{8}{c}{Panel 4: $n=400$, $p=100$}\\\hline
& NV & TPR & FPR & FDR & CS & STEP & RMSFE\\\hline
POST OCMT & 5.1300 & 0.9946 & 0.0016 & 0.0227 & 0.8440 & - & 1.0472\\
OCMT & 7.0280 & 0.9984 & 0.0212 & 0.2531 & 0.0000 & 1.9940 & 1.5916\\
AGLASSO & 4.4680 & 0.8426 & 0.0027 & 0.0397 & 0.1610 & - & 1.1553\\
BAGGING   & -      & -      & -      & -      & -      & -      & 1.3231   \\
RF        & -      & -      & -      & -      & -      & -      & 1.3455   \\\hline
\multicolumn{8}{c}{Panel 5: $n=400$, $p=200$}\\\hline
& NV & TPR & FPR & FDR & CS & STEP & RMSFE\\\hline
POST OCMT & 5.0930 & 0.9924 & 0.0007 & 0.0190 & 0.8620 & - & 1.0476\\
OCMT & 7.0330 & 0.9960 & 0.0105 & 0.2545 & 0.0000 & 1.9930 & 1.6081\\
AGLASSO & 4.3460 & 0.8278 & 0.0011 & 0.0326 & 0.1050 & - & 1.1343\\
BAGGING   & -      & -      & -      & -      & -      & -      & 1.3522   \\
RF        & -      & -      & -      & -      & -      & -      & 1.3764   \\\hline
\multicolumn{8}{c}{Panel 6: $n=400$, $p=1000$}\\\hline
& NV & TPR & FPR & FDR & CS & STEP & RMSFE\\\hline
POST OCMT & 4.9700 & 0.9598 & 0.0002 & 0.0255 & 0.7410 & - & 1.0660\\
OCMT & 6.8910 & 0.9656 & 0.0021 & 0.2609 & 0.0000 & 1.8850 & 1.5604\\
AGLASSO & 4.1797 & 0.7901 & 0.0002 & 0.0383 & 0.0208 & - &
1.1478\\
BAGGING   & -      & -      & -      & -      & -      & -      & 1.4073   \\
RF        & -      & -      & -      & -      & -      & -      & 1.4309   \\\hline\hline
\end{tabular}} \bigskip\end{table}

\begin{table}[p]
\caption{DGP 5}%
\label{tabledgp5}%
\centering{}\centering{ } \resizebox{!}{0.67\textwidth}{
\begin{tabular}
[c]{l|ccccccc}\hline\hline
\multicolumn{8}{c}{Panel 1: $n=200$, $p=100$}\\\hline
& NV & TPR & FPR & FDR & CS & STEP & RMSFE\\\hline
POST OCMT & 3.4700 & 0.8465 & 0.0009 & 0.0141 & 0.3580 & - & 1.1198\\
OCMT & 3.5210 & 0.8465 & 0.0014 & 0.0223 & 0.3120 & 1.4090 & 1.3958\\
AGLASSO & 2.8930 & 0.7127 & 0.0004 & 0.0073 & 0.0800 & - & 1.1856\\
BAGGING   & -      & -      & -      & -      & -      & -      & 1.2298   \\
RF        & -      & -      & -      & -      & -      & -      & 1.2677   \\\hline
\multicolumn{8}{c}{Panel 2: $n=200$, $p=200$}\\\hline
& NV & TPR & FPR & FDR & CS & STEP & RMSFE\\\hline
POST OCMT & 3.4040 & 0.8283 & 0.0005 & 0.0154 & 0.2930 & - & 1.1235\\
OCMT & 3.4630 & 0.8283 & 0.0008 & 0.0249 & 0.2430 & 1.3510 & 1.3660\\
AGLASSO & 2.7600 & 0.6793 & 0.0002 & 0.0083 & 0.0490 & - & 1.1822\\
BAGGING   & -      & -      & -      & -      & -      & -      & 1.2512   \\
RF        & -      & -      & -      & -      & -      & -      & 1.2906   \\\hline
\multicolumn{8}{c}{Panel 3: $n=200$, $p=1000$}\\\hline
& NV & TPR & FPR & FDR & CS & STEP & RMSFE\\\hline
POST OCMT & 2.9990 & 0.7380 & 0.0000 & 0.0089 & 0.0500 & - & 1.1515\\
OCMT & 3.0080 & 0.7380 & 0.0001 & 0.0103 & 0.0440 & 1.0650 & 1.1970\\
AGLASSO & 2.4560 & 0.6030 & 0.0000 & 0.0087 & 0.0150 & - & 1.2163\\
BAGGING   & -      & -      & -      & -      & -      & -      & 1.2904   \\
RF        & -      & -      & -      & -      & -      & -      & 1.3293   \\\hline
\multicolumn{8}{c}{Panel 4: $n=400$, $p=100$}\\\hline
& NV & TPR & FPR & FDR & CS & STEP & RMSFE\\\hline
POST OCMT & 3.9800 & 0.9908 & 0.0002 & 0.0028 & 0.9460 & - & 1.0420\\
OCMT & 4.0300 & 0.9908 & 0.0007 & 0.0110 & 0.8990 & 1.9020 & 1.6559\\
AGLASSO & 3.8250 & 0.9460 & 0.0004 & 0.0068 & 0.7470 & - & 1.0555\\
BAGGING   & -      & -      & -      & -      & -      & -      & 1.1604   \\
RF        & -      & -      & -      & -      & -      & -      & 1.1922   \\\hline
\multicolumn{8}{c}{Panel 5: $n=400$, $p=200$}\\\hline
& NV & TPR & FPR & FDR & CS & STEP & RMSFE\\\hline
POST OCMT & 3.9670 & 0.9850 & 0.0001 & 0.0044 & 0.9150 & - & 1.0492\\
OCMT & 4.0260 & 0.9850 & 0.0004 & 0.0139 & 0.8630 & 1.9210 & 1.6745\\
AGLASSO & 3.7640 & 0.9300 & 0.0002 & 0.0075 & 0.6820 & - & 1.0676\\
BAGGING   & -      & -      & -      & -      & -      & -      & 1.1751   \\
RF        & -      & -      & -      & -      & -      & -      & 1.2147   \\\hline
\multicolumn{8}{c}{Panel 6: $n=400$, $p=1000$}\\\hline
& NV & TPR & FPR & FDR & CS & STEP & RMSFE\\\hline
POST OCMT & 3.7550 & 0.9300 & 0.0000 & 0.0060 & 0.6900 & - & 1.0663\\
OCMT & 3.8200 & 0.9300 & 0.0001 & 0.0164 & 0.6350 & 1.7200 & 1.5583\\
AGLASSO & 3.5100 & 0.8675 & 0.0000 & 0.0069 & 0.4500 & - &
1.0865\\
BAGGING   & -      & -      & -      & -      & -      & -      & 1.2045   \\
RF        & -      & -      & -      & -      & -      & -      & 1.2520   \\\hline\hline
\end{tabular}}\end{table}

\begin{table}[p]
\caption{DGP 7}%
\label{tabledgp7}%
\centering{}\centering{ } \resizebox{!}{0.67\textwidth}{
\begin{tabular}
[c]{l|ccccccc}\hline\hline
\multicolumn{8}{c}{Panel 1: $n=200$, $p=100$}\\\hline
& NV & TPR & FPR & FDR & CS & STEP & RMSFE\\\hline
POST OCMT & 4.0030 & 1.0000 & 0.0000 & 0.0005 & 0.9970 & - & 1.0793\\
OCMT & 4.0080 & 1.0000 & 0.0001 & 0.0013 & 0.9920 & 1.0240 & 1.0922\\
AGLASSO & 4.0470 & 0.9970 & 0.0006 & 0.0097 & 0.9400 & - & 1.0858\\
BAGGING   & -      & -      & -      & -      & -      & -      & 1.4393   \\
RF        & -      & -      & -      & -      & -      & -      & 1.4824   \\\hline
\multicolumn{8}{c}{Panel 2: $n=200$, $p=200$}\\\hline
& NV & TPR & FPR & FDR & CS & STEP & RMSFE\\\hline
POST OCMT & 4.0070 & 0.9992 & 0.0001 & 0.0017 & 0.9870 & - & 1.0833\\
OCMT & 4.0120 & 0.9992 & 0.0001 & 0.0025 & 0.9820 & 1.0450 & 1.1071\\
AGLASSO & 4.0500 & 0.9932 & 0.0004 & 0.0125 & 0.9240 & - & 1.0906\\
BAGGING   & -      & -      & -      & -      & -      & -      & 1.4758   \\
RF        & -      & -      & -      & -      & -      & -      & 1.5246   \\\hline
\multicolumn{8}{c}{Panel 3: $n=200$, $p=1000$}\\\hline
& NV & TPR & FPR & FDR & CS & STEP & RMSFE\\\hline
POST OCMT & 4.0050 & 0.9935 & 0.0000 & 0.0052 & 0.9500 & - & 1.0873\\
OCMT & 4.0300 & 0.9935 & 0.0001 & 0.0088 & 0.9360 & 1.0910 & 1.1260\\
AGLASSO & 3.9690 & 0.9643 & 0.0001 & 0.0180 & 0.8630 & - & 1.1108\\
BAGGING   & -      & -      & -      & -      & -      & -      & 1.5348   \\
RF        & -      & -      & -      & -      & -      & -      & 1.5919   \\\hline
\multicolumn{8}{c}{Panel 4: $n=400$, $p=100$}\\\hline
& NV & TPR & FPR & FDR & CS & STEP & RMSFE\\\hline
POST OCMT & 4.0010 & 1.0000 & 0.0000 & 0.0002 & 0.9990 & - & 1.0601\\
OCMT & 4.0010 & 1.0000 & 0.0000 & 0.0002 & 0.9990 & 1.0000 & 1.0601\\
AGLASSO & 4.0210 & 1.0000 & 0.0002 & 0.0035 & 0.9790 & - & 1.0608\\
BAGGING   & -      & -      & -      & -      & -      & -      & 1.3339   \\
RF        & -      & -      & -      & -      & -      & -      & 1.3755   \\\hline
\multicolumn{8}{c}{Panel 5: $n=400$, $p=200$}\\\hline
& NV & TPR & FPR & FDR & CS & STEP & RMSFE\\\hline
POST OCMT & 4.0000 & 1.0000 & 0.0000 & 0.0000 & 1.0000 & - & 1.0575\\
OCMT & 4.0000 & 1.0000 & 0.0000 & 0.0000 & 1.0000 & 1.0000 & 1.0575\\
AGLASSO & 4.0230 & 1.0000 & 0.0001 & 0.0038 & 0.9770 & - & 1.0585\\
BAGGING   & -      & -      & -      & -      & -      & -      & 1.3646   \\
RF        & -      & -      & -      & -      & -      & -      & 1.4083   \\\hline
\multicolumn{8}{c}{Panel 6: $n=400$, $p=1000$}\\\hline
& NV & TPR & FPR & FDR & CS & STEP & RMSFE\\\hline
POST OCMT & 4.0020 & 1.0000 & 0.0000 & 0.0003 & 0.9980 & - & 1.0582\\
OCMT & 4.0020 & 1.0000 & 0.0000 & 0.0003 & 0.9980 & 1.0000 & 1.0582\\
AGLASSO & 4.0530 & 1.0000 & 0.0001 & 0.0086 & 0.9520 & - &
1.0601\\
BAGGING   & -      & -      & -      & -      & -      & -      & 1.4165   \\
RF        & -      & -      & -      & -      & -      & -      & 1.4619   \\\hline\hline
\end{tabular}}\end{table}

\begin{table}[p]
\caption{DGP 8}%
\label{tabledgp8}%
\centering{}\centering{ } \resizebox{!}{0.67\textwidth}{
\begin{tabular}
[c]{l|ccccccc}\hline\hline
\multicolumn{8}{c}{Panel 1: $n=200$, $p=100$}\\\hline
& NV & TPR & FPR & FDR & CS & STEP & RMSFE\\\hline
POST OCMT & 3.5670 & 0.6606 & 0.0027 & 0.0527 & 0.1420 & - & 1.2378\\
OCMT & 4.7090 & 0.6784 & 0.0137 & 0.2056 & 0.0000 & 1.2270 & 1.3827\\
AGLASSO & 2.9240 & 0.5442 & 0.0021 & 0.0403 & 0.0010 & - & 1.3769\\
BAGGING   & -      & -      & -      & -      & -      & -      & 1.4590   \\
RF        & -      & -      & -      & -      & -      & -      & 1.4697   \\\hline
\multicolumn{8}{c}{Panel 2: $n=200$, $p=200$}\\\hline
& NV & TPR & FPR & FDR & CS & STEP & RMSFE\\\hline
POST OCMT & 3.4800 & 0.6370 & 0.0015 & 0.0605 & 0.1180 & - & 1.2535\\
OCMT & 4.5760 & 0.6556 & 0.0066 & 0.2076 & 0.0000 & 1.1760 & 1.3681\\
AGLASSO & 2.6250 & 0.4824 & 0.0011 & 0.0425 & 0.0010 & - & 1.3807\\
BAGGING   & -      & -      & -      & -      & -      & -      & 1.4883   \\
RF        & -      & -      & -      & -      & -      & -      & 1.4998   \\\hline
\multicolumn{8}{c}{Panel 3: $n=200$, $p=1000$}\\\hline
& NV & TPR & FPR & FDR & CS & STEP & RMSFE\\\hline
POST OCMT & 3.2590 & 0.5874 & 0.0003 & 0.0685 & 0.0420 & - & 1.2762\\
OCMT & 4.2030 & 0.6082 & 0.0012 & 0.1994 & 0.0000 & 1.0790 & 1.3341\\
AGLASSO & 1.9940 & 0.3668 & 0.0002 & 0.0336 & 0.0000 & - & 1.4378\\
BAGGING   & -      & -      & -      & -      & -      & -      & 1.5429   \\
RF        & -      & -      & -      & -      & -      & -      & 1.5574   \\\hline
\multicolumn{8}{c}{Panel 4: $n=400$, $p=100$}\\\hline
& NV & TPR & FPR & FDR & CS & STEP & RMSFE\\\hline
POST OCMT & 4.8870 & 0.9474 & 0.0016 & 0.0230 & 0.7160 & - & 1.0743\\
OCMT & 7.4610 & 0.9534 & 0.0281 & 0.3053 & 0.0000 & 1.8440 & 1.5024\\
AGLASSO & 4.3690 & 0.8278 & 0.0024 & 0.0363 & 0.1030 & - & 1.1364\\
BAGGING   & -      & -      & -      & -      & -      & -      & 1.3586   \\
RF        & -      & -      & -      & -      & -      & -      & 1.3840   \\\hline
\multicolumn{8}{c}{Panel 5: $n=400$, $p=200$}\\\hline
& NV & TPR & FPR & FDR & CS & STEP & RMSFE\\\hline
POST OCMT & 4.8820 & 0.9432 & 0.0008 & 0.0249 & 0.6890 & - & 1.0768\\
OCMT & 7.4170 & 0.9496 & 0.0136 & 0.3032 & 0.0000 & 1.8290 & 1.4988\\
AGLASSO & 4.2570 & 0.8090 & 0.0011 & 0.0351 & 0.0550 & - & 1.1455\\
BAGGING   & -      & -      & -      & -      & -      & -      & 1.3882   \\
RF        & -      & -      & -      & -      & -      & -      & 1.4130   \\\hline
\multicolumn{8}{c}{Panel 6: $n=400$, $p=1000$}\\\hline
& NV & TPR & FPR & FDR & CS & STEP & RMSFE\\\hline
POST OCMT & 4.8420 & 0.9338 & 0.0002 & 0.0262 & 0.6490 & - & 1.0825\\
OCMT & 7.3220 & 0.9382 & 0.0026 & 0.3020 & 0.0000 & 1.7920 & 1.4947\\
AGLASSO & 4.1570 & 0.7908 & 0.0002 & 0.0346 & 0.0120 & - &
1.1550\\
BAGGING   & -      & -      & -      & -      & -      & -      & 1.4532   \\
RF        & -      & -      & -      & -      & -      & -      & 1.4732   \\\hline\hline
\end{tabular}} \bigskip\end{table}

\begin{table}[p]
\caption{DGP 9}%
\label{tabledgp9}%
\centering{}\centering{ } \resizebox{!}{0.67\textwidth}{
\begin{tabular}
[c]{l|ccccccc}\hline\hline
\multicolumn{8}{c}{Panel 1: $n=200$, $p=100$}\\\hline
& NV & TPR & FPR & FDR & CS & STEP & RMSFE\\\hline
POST OCMT & 4.2770 & 0.9945 & 0.0005 & 0.0031 & 0.7380 & - & 1.0422\\
OCMT & 4.2950 & 0.9948 & 0.0005 & 0.0033 & 0.7300 & 1.1000 & 1.0436\\
AGLASSO & 5.0510 & 1.0000 & 0.0014 & 0.0109 & 0.4800 & - & 1.0537\\
BAGGING   & -      & -      & -      & -      & -      & -      & 1.3471   \\
RF        & -      & -      & -      & -      & -      & -      & 1.3909   \\\hline
\multicolumn{8}{c}{Panel 2: $n=200$, $p=200$}\\\hline
& NV & TPR & FPR & FDR & CS & STEP & RMSFE\\\hline
POST OCMT & 4.5000 & 0.9870 & 0.0004 & 0.0028 & 0.5780 & - & 1.0585\\
OCMT & 4.5310 & 0.9870 & 0.0004 & 0.0030 & 0.5650 & 1.1460 & 1.0585\\
AGLASSO & 5.7540 & 0.9995 & 0.0011 & 0.0090 & 0.3420 & - & 1.0697\\
BAGGING   & -      & -      & -      & -      & -      & -      & 1.3838   \\
RF        & -      & -      & -      & -      & -      & -      & 1.4351   \\\hline
\multicolumn{8}{c}{Panel 3: $n=200$, $p=1000$}\\\hline
& NV & TPR & FPR & FDR & CS & STEP & RMSFE\\\hline
POST OCMT & 4.8200 & 0.9165 & 0.0002 & 0.0012 & 0.2200 & - & 1.1207\\
OCMT & 4.8670 & 0.9170 & 0.0002 & 0.0012 & 0.2120 & 1.1410 & 1.1236\\
AGLASSO & 8.7040 & 0.9890 & 0.0004 & 0.0048 & 0.1060 & - & 1.1143\\
BAGGING   & -      & -      & -      & -      & -      & -      & 1.4450   \\
RF        & -      & -      & -      & -      & -      & -      & 1.5038   \\\hline
\multicolumn{8}{c}{Panel 4: $n=400$, $p=100$}\\\hline
& NV & TPR & FPR & FDR & CS & STEP & RMSFE\\\hline
POST OCMT & 4.0110 & 1.0000 & 0.0000 & 0.0001 & 0.9890 & - & 1.0220\\
OCMT & 4.0110 & 1.0000 & 0.0000 & 0.0001 & 0.9890 & 1.0080 & 1.0220\\
AGLASSO & 4.6660 & 1.0000 & 0.0010 & 0.0069 & 0.6240 & - & 1.0288\\
BAGGING   & -      & -      & -      & -      & -      & -      & 1.2348   \\
RF        & -      & -      & -      & -      & -      & -      & 1.2713   \\\hline
\multicolumn{8}{c}{Panel 5: $n=400$, $p=200$}\\\hline
& NV & TPR & FPR & FDR & CS & STEP & RMSFE\\\hline
POST OCMT & 4.0150 & 1.0000 & 0.0000 & 0.0001 & 0.9860 & - & 1.0237\\
OCMT & 4.0150 & 1.0000 & 0.0000 & 0.0001 & 0.9860 & 1.0080 & 1.0237\\
AGLASSO & 5.2000 & 1.0000 & 0.0008 & 0.0061 & 0.4460 & - & 1.0351\\
BAGGING   & -      & -      & -      & -      & -      & -      & 1.2620   \\
RF        & -      & -      & -      & -      & -      & -      & 1.3032   \\\hline
\multicolumn{8}{c}{Panel 6: $n=400$, $p=1000$}\\\hline
& NV & TPR & FPR & FDR & CS & STEP & RMSFE\\\hline
POST OCMT & 4.0550 & 1.0000 & 0.0000 & 0.0001 & 0.9490 & - & 1.0225\\
OCMT & 4.0570 & 1.0000 & 0.0000 & 0.0001 & 0.9480 & 1.0360 & 1.0225\\
AGLASSO & 7.7000 & 1.0000 & 0.0004 & 0.0037 & 0.1270 & - & 1.0551\\
BAGGING   & -      & -      & -      & -      & -      & -      & 1.3132   \\
RF        & -      & -      & -      & -      & -      & -      & 1.3576   \\\hline\hline
\end{tabular}}\end{table}

\begin{table}[p]
\caption{DGP 10}%
\label{tabledgp10}%
\centering{}\centering{ } \resizebox{!}{0.67\textwidth}{\begin{tabular}
[c]{l|ccccccc}\hline\hline
\multicolumn{8}{c}{Panel 1: $n=200$, $p=100$}\\\hline
& NV & TPR & FPR & FDR & CS & STEP & RMSFE\\\hline
POST OCMT & 4.0240 & 0.9668 & 0.0016 & 0.0260 & 0.7450 & - & 1.0329\\
OCMT & 4.0250 & 0.9668 & 0.0016 & 0.0261 & 0.7440 & 1.8780 & 1.0329\\
AGLASSO & 5.0160 & 0.9942 & 0.0108 & 0.1413 & 0.4370 & - & 1.0375\\
BAGGING   & -      & -      & -      & -      & -      & -      & 1.2856   \\
RF        & -      & -      & -      & -      & -      & -      & 1.3107   \\\hline
\multicolumn{8}{c}{Panel 2: $n=200$, $p=200$}\\\hline
& NV & TPR & FPR & FDR & CS & STEP & RMSFE\\\hline
POST OCMT & 4.0700 & 0.9535 & 0.0013 & 0.0403 & 0.6560 & - & 1.0421\\
OCMT & 4.0740 & 0.9535 & 0.0013 & 0.0409 & 0.6540 & 1.8320 & 1.0421\\
AGLASSO & 5.7520 & 0.9915 & 0.0091 & 0.2184 & 0.2760 & - & 1.0519\\
BAGGING   & -      & -      & -      & -      & -      & -      & 1.3212   \\
RF        & -      & -      & -      & -      & -      & -      & 1.3391   \\\hline
\multicolumn{8}{c}{Panel 3: $n=200$, $p=1000$}\\\hline
& NV & TPR & FPR & FDR & CS & STEP & RMSFE\\\hline
POST OCMT & 4.0410 & 0.8625 & 0.0006 & 0.0899 & 0.2740 & - & 1.0910\\
OCMT & 4.0490 & 0.8628 & 0.0006 & 0.0908 & 0.2740 & 1.5650 & 1.0922\\
AGLASSO & 7.7110 & 0.9473 & 0.0039 & 0.3834 & 0.0530 & - & 1.1009\\
BAGGING   & -      & -      & -      & -      & -      & -      & 1.3701   \\
RF        & -      & -      & -      & -      & -      & -      & 1.3878   \\\hline
\multicolumn{8}{c}{Panel 4: $n=400$, $p=100$}\\\hline
& NV & TPR & FPR & FDR & CS & STEP & RMSFE\\\hline
POST OCMT & 4.0060 & 0.9998 & 0.0001 & 0.0012 & 0.9920 & - & 1.0068\\
OCMT & 4.0060 & 0.9998 & 0.0001 & 0.0012 & 0.9920 & 1.9870 & 1.0068\\
AGLASSO & 4.6880 & 1.0000 & 0.0072 & 0.0972 & 0.5930 & - & 1.0135\\
BAGGING   & -      & -      & -      & -      & -      & -      & 1.1703   \\
RF        & -      & -      & -      & -      & -      & -      & 1.2185   \\\hline
\multicolumn{8}{c}{Panel 5: $n=400$, $p=200$}\\\hline
& NV & TPR & FPR & FDR & CS & STEP & RMSFE\\\hline
POST OCMT & 4.0050 & 0.9998 & 0.0000 & 0.0010 & 0.9930 & - & 1.0081\\
OCMT & 4.0050 & 0.9998 & 0.0000 & 0.0010 & 0.9930 & 1.9950 & 1.0081\\
AGLASSO & 5.2130 & 1.0000 & 0.0062 & 0.1571 & 0.4460 & - & 1.0195\\
BAGGING   & -      & -      & -      & -      & -      & -      & 1.1953   \\
RF        & -      & -      & -      & -      & -      & -      & 1.2466   \\\hline
\multicolumn{8}{c}{Panel 6: $n=400$, $p=1000$}\\\hline
& NV & TPR & FPR & FDR & CS & STEP & RMSFE\\\hline
POST OCMT & 4.0170 & 0.9985 & 0.0000 & 0.0038 & 0.9710 & - & 1.0089\\
OCMT & 4.0170 & 0.9985 & 0.0000 & 0.0038 & 0.9710 & 1.9920 & 1.0089\\
AGLASSO & 7.2690 & 0.9998 & 0.0033 & 0.3310 & 0.1700 & - &
1.0373\\
BAGGING   & -      & -      & -      & -      & -      & -      & 1.2492   \\
RF        & -      & -      & -      & -      & -      & -      & 1.2888   \\\hline\hline
\end{tabular}
} \bigskip\end{table}

\begin{table}[p]
\caption{DGP 11}%
\label{tabledgp11}%
\centering{}\centering{ } \resizebox{!}{0.5\textwidth}{\begin{tabular}
[c]{l|ccccccc}\hline\hline
\multicolumn{8}{c}{Panel 1: $n=200$, $p=100$}\\\hline
& NV & TPR & FPR & FDR & CS & STEP & RMSFE\\\hline
POST OCMT & 4.9320 & 0.9290 & 0.0445 & 0.0030 & 0.5490 & - & 1.0928\\
OCMT & 6.4460 & 0.9404 & 0.2162 & 0.0182 & 0.0560 & 1.5020 & 1.5047\\\hline
\multicolumn{8}{c}{Panel 2: $n=200$, $p=200$}\\\hline
& NV & TPR & FPR & FDR & CS & STEP & RMSFE\\\hline
POST OCMT & 4.9770 & 0.9328 & 0.0486 & 0.0016 & 0.5720 & - & 1.0888\\
OCMT & 6.5490 & 0.9444 & 0.2266 & 0.0093 & 0.0480 & 1.4550 & 1.4812\\\hline
\multicolumn{8}{c}{Panel 3: $n=200$, $p=1000$}\\\hline
& NV & TPR & FPR & FDR & CS & STEP & RMSFE\\\hline
POST OCMT & 4.9270 & 0.9250 & 0.0469 & 0.0003 & 0.5520 & - & 1.0960\\
OCMT &  6.5430 & 0.9382 & 0.2280 & 0.0019 & 0.0500 & 1.4590 & 1.4843\\\hline
\multicolumn{8}{c}{Panel 4: $n=400$, $p=100$}\\\hline
& NV & TPR & FPR & FDR & CS & STEP & RMSFE\\\hline
POST OCMT &5.0450 & 0.9870 & 0.0168 & 0.0011 & 0.8780 & - & 1.0298\\
OCMT & 6.6390 & 0.9894 & 0.2077 & 0.0176 & 0.0940 & 1.3610 & 1.3662\\\hline
\multicolumn{8}{c}{Panel 5: $n=400$, $p=200$}\\\hline
& NV & TPR & FPR & FDR & CS & STEP & RMSFE\\\hline
POST OCMT & 5.0350 & 0.9860 & 0.0159 & 0.0005 & 0.8660 & - & 1.0330\\
OCMT & 6.6150 & 0.9866 & 0.2056 & 0.0086 & 0.1050 & 1.3640 & 1.3687\\\hline
\multicolumn{8}{c}{Panel 6: $n=400$, $p=1000$}\\\hline
& NV & TPR & FPR & FDR & CS & STEP & RMSFE\\\hline
POST OCMT & 5.0320 & 0.9872 & 0.0145 & 0.0001 & 0.8810 & - & 1.0360\\
OCMT & 6.6970 & 0.9886 & 0.2061 & 0.0018 & 0.0890 & 1.3820 & 1.3857\\\hline\hline
\end{tabular}
} \bigskip\end{table}

\pagebreak

\subsection{Additional Application Results}

In this subsection, we first provide the description of the variables used in
the application and some summary statistics. Then we report the frequencies of
variables selected.

Table \ref{table:app_def} gives the description of the dependent variable and
covariates used in our application. Table \ref{table:summary} reports the
summary statistics for the variables. Table \ref{table:varselected} reports
the frequencies of covariates selected by the methods under investigation. For
those variables that are not listed in Table \ref{table:varselected}, they are
not selected by any of these methods. Table \ref{table:varselected} suggests
that among the 78 covariates in the application, three of them, namely, G102,
G133, and G137, are all frequently selected by the multiple-stage OCMT,
post-OCMT, group Lasso and adaptive group Lasso.

\begin{center}
{\small \begin{longtable}[p]{ll}
\caption[Definitions of Variables]{Definitions of Variables} \label{table:app_def} \\
\hline \hline
\endfirsthead
\multicolumn{2}{c}%
{{ \tablename\ \thetable{} -- continued from previous page}} \\
\hline
\endhead
\hline \multicolumn{2}{r}{{Continued on next page}} \\ \hline
\endfoot
\hline \hline
\endlastfoot
\multicolumn{2}{c}{Panel A: Dependent Variable}  \\\hline
$G136$ & Remittance to his/her original home in rural areas in 2007 \\\hline
\multicolumn{2}{c}{}  \\
\multicolumn{2}{c}{Panel B: Continuous Independent Variables}  \\\hline
$A05$  & Age \\
$A08$  & Rank of Siblings (1: the oldest, 2: the second oldest, and so on so forth) \\
$A25$  & Height (in cm)\\
$A26$  & Weight (in kg)\\
$A37$  & Cost of medical insurance  in 2007 (in RMB)\\
$A39$  & Total medical cost  in 2007 (in RMB)\\
$A40$  & Out-of-Pocket medical cost  in 2007 (in RMB)\\
$B103$  & Years of eduction \\
$C103$  & Year when starting current job\\
$C111$ & Number of employees in the company he/she is working \\
& (ranked from 1 (smallest) to 9 (largest)) \\
$C112$ & Average number of working hours per week \\
$C105$ & Date when working in current occupation \\
$C117$ & Average income in current occupation \\
$C126$ & Food compensation per month from current job (in RMB) \\
$C165$ & Date when start to work in urban areas \\
$C171$ & Time spent to find the first job (in days) \\
$C177$ & Average number of working hours per week for the fist job \\
$C178$ & The salary of the first month for the first job (in RMB) \\
$C179$ & The salary of the last month for the first job (in RMB)\\
$C183$ & Number of months for doing the first job \\
$C186$ & Number of cities in which he/she worked \\
$C188$ & Average month salary if he/she worked in his/her hometown\\
$E4\_11$ & Number of people he/she contacted during the last spring festival \\
$E4\_12$ & Number of relatives he/she contacted during the last spring festival \\
$E4\_13$ & Number of close friends he/she contacted during the last spring festival \\
$E4\_14$ & Number of people in major cities he/she contacted during the last spring festival \\
$E4\_16$ & Number of people who helped him/her in the last 12 months \\
$G119$ & Average cost of food per month including dining in restaurants (in RMB) \\
$G120$ & Average cost on clothes per month (in RMB) \\
$G121$ & Average cost of accommodation including utilities per month (in RMB) \\
$G122$ & Total living cost in 2007 (in RMB) \\
$G123$ & Total consumption of durable goods  in 2007 (in RMB) \\
$G124$ & Total consumption of non-durable goods or services   \\
& (e.g., dishes, beauty, hair cutting) in 2007 (in RMB) \\
$G125$ & Total consumption of medication, nutrition supplements, and physical therapies  \\
& in 2007 (in RMB) \\
$G126$ & Cost of transportation  in 2007 (in RMB) \\
$G127$ & Cost of phone services and postage  in 2007 (in RMB) \\
$G128$ & Cost of entertainment  in 2007 (in RMB) \\
$G132$ & Other consumption in 2007 (in RMB) \\
$G141$ & Cost per month for the minimum living standard in current city (in RMB) \\
$H219\_2$ & Value of his/her cell phones (in RMB) \\
$I101$ & Current total living area (in square meters) \\
$I112$ & Current rent per month (in RMB) \\
$J108$ & Average cost of hiring a handy man per day in his/her hometown (in RMB) \\
$J109$ & Percentage of workforce working in cities in his/her country \\
$J112$ & Value of houses owned in his/her hometown (in RMB) \\
$J113$ & Acres of land owned in his/her hometown \\
$G102$ & Monthly income (in RMB) \\
\multicolumn{2}{c}{}  \\
\multicolumn{2}{c}{Panel C: Continuous Independent Variables but Treated as Dummy Variables}  \\\hline
$C130$ & Compensation per month for accommodation from the current job  (in RMB) \\
$E4\_15$ & Number of people in possession of city Hukou he/she contacted during\\
&  the last spring festival \\
$G129$ & Cost of education for children excluding left-behind children  in 2007 (in RMB) \\
$G130$ & Cost of all non-saving insurances in 2007  (in RMB) \\
$G131$ & Other consumption in 2007 (in RMB) \\
$G133$ & Gifts to others including his/her parents in 2007 (in RMB) \\
$G134$ & Investments in stocks or bonds in 2007 (in RMB) \\
$G135$ & Cost of building new houses or renovation in his/her hometown  in 2007 (in RMB) \\
$G137$ & Education cost for left-behind children  in 2007 (in RMB) \\
$G138$ & Fines or interests charged in 2007 (in RMB) \\
$G139$ & Cost of loans in 2007 (in RMB) \\
$B110$ & Total marks of his/her college entrance exam (0 if not attended) \\
$B119$ & Cost of last job training (in RMB) \\\hline
\multicolumn{2}{c}{}  \\
\multicolumn{2}{c}{Panel D: Discrete Variables that Treated as Dummy Variables}  \\\hline
$A01$   & Number of months working outside of his/her hometown in 2007\\
$A10$   & Number of children \\
$A27$   & Health condition (excellent to bad ranking from 1 to 5) \\
$H219\_1$ & Number of cell phones owned \\
$I102$  & Number of people currently living together \\
\multicolumn{2}{c}{}  \\
\multicolumn{2}{c}{Panel E: Dummy Variables}  \\\hline
$A04$   & Gender (0: female 1: male) \\
$A09$   & Marital Status (0: not married  1:married) \\
$A15$   & Status of Hukou (0: rural 1: city) \\
$A21$   & Ownership of the insurance for unemployment (0: no 1: yes) \\
$A22$   & Ownership of the insurance for age care (0: no 1: yes) \\
$A23$   & Ownership of the insurance for injuries during work (0: no 1: yes) \\
$A24$   & Ownership of the programme of the deposit for house purchases (0: no 1: yes) \\
$A28$   & Disable or not (0: no 1:yes )\\
$B108$  & Whether he/she attended the college entrance exam or not  (0: no 1:yes )\\
$C127$  & Whether employer provides housing or not (0: no 1:yes)\\
$C184$  & Whether he/she ever lived in his/her hometown for over 3 months continuously \\
& after working in cities (0: no 1:yes) \\
Dongguan & Whether he/she works in Dong Guan city (0: no 1:yes) \\
Shenzhen & Whether he/she works in Shen Zhen city (0: no 1:yes) \\
\end{longtable}}

{\small \begin{longtable}[p]{llllllll}
\caption[Summary Statistics]{Summary Statistics} \label{table:summary} \\
\hline \hline
Variable & MEAN & STD & MIN & LQ & MD & UQ & MAX  \\\hline
\endfirsthead
\multicolumn{8}{c}%
{{ \tablename\ \thetable{} -- continued from previous page}} \\\hline
Variable & MEAN & STD & MIN & LQ & MD & UQ & MAX  \\
\hline
\endhead
\hline \multicolumn{8}{r}{{Continued on next page}} \\ \hline
\endfoot
\hline \hline
\endlastfoot
\multicolumn{8}{c}{Panel A: Dependent Variable}  \\\hline
$G136$ & 3374 & 4576 & 0 & 0 & 2000 & 5000 & 60000 \\
\multicolumn{2}{c}{}  \\
\multicolumn{8}{c}{Panel B: Continuous Independent Variables}  \\\hline
$A05$      & 30 & 9 & 17 & 23 & 27 & 36 & 66 \\
$A08$      & 2 & 1 & 1 & 1 & 2 & 3 & 8 \\
$A25$      & 166 & 7 & 147 & 160 & 168 & 171 & 186 \\
$A26$      & 59 & 9 & 36 & 52 & 60 & 65 & 82 \\
$A37$      & 109 & 413 & 0 & 0 & 10 & 10 & 5000 \\
$A39$      & 289 & 746 & 0 & 0 & 50 & 270 & 8000 \\
$A40$      & 261 & 676 & 0 & 0 & 50 & 200 & 8000 \\
$B103$     & 10 & 2 & 2 & 9 & 9 & 12 & 20 \\
$C103$     & 2006 & 3 & 1989 & 2005 & 2006 & 2007 & 2008 \\
$C111$     & 6 & 2 & 1 & 4 & 6 & 7 & 9 \\
$C112$     & 57 & 14 & 34 & 48 & 56 & 70 & 126 \\
$C105$     & 200516 & 314 & 198910 & 200404 & 200606 & 200709 & 200811 \\
$C117$     & 1620 & 795 & 500 & 1165 & 1500 & 1900 & 10000 \\
$C126$     & 102 & 112 & 0 & 0 & 95 & 160 & 450 \\
$C165$     & 199896 & 650 & 197607 & 199555 & 200012 & 200404 & 200803 \\
$C171$     & 14 & 26 & 0 & 2 & 7 & 15 & 365 \\
$C177$     & 58 & 19 & 0 & 48 & 60 & 70 & 114 \\
$C178$     & 724 & 423 & 0 & 500 & 700 & 900 & 5000 \\
$C179$     & 874 & 516 & 0 & 550 & 800 & 1100 & 5000 \\
$C183$     & 18 & 24 & 1 & 6 & 12 & 24 & 188 \\
$C186$     & 2 & 5 & 0 & 1 & 2 & 3 & 90 \\
$C188$     & 858 & 853 & 0 & 600 & 800 & 1000 & 10000 \\
$E4\_11$   & 38 & 55 & 0 & 15 & 29 & 40 & 800 \\
$E4\_12$   & 15 & 14 & 0 & 6 & 10 & 20 & 120 \\
$E4\_13$   & 24 & 50 & 0 & 5 & 11 & 20 & 700 \\
$E4\_14$   & 18 & 34 & 0 & 3 & 10 & 20 & 400 \\
$E4\_16$   & 4 & 11 & 0 & 0 & 2 & 3 & 200 \\
$G119$     & 520 & 324 & 0 & 300 & 500 & 700 & 3000 \\
$G120$     & 125 & 153 & 0 & 0 & 100 & 200 & 1500 \\
$G121$     & 207 & 302 & 0 & 0 & 100 & 300 & 2000 \\
$G122$     & 14630 & 9131 & 0 & 8795 & 12490 & 18425 & 75300 \\
$G123$     & 726 & 1385 & 0 & 0 & 200 & 1000 & 10000 \\
$G124$     & 876 & 2080 & 0 & 300 & 500 & 1000 & 37000 \\
$G125$     & 380 & 850 & 0 & 0 & 100 & 400 & 8500 \\
$G126$     & 701 & 843 & 0 & 200 & 500 & 1000 & 7200 \\
$G127$     & 971 & 989 & 0 & 515 & 900 & 1200 & 17400 \\
$G128$     & 302 & 664 & 0 & 0 & 100 & 300 & 5000 \\
$G132$     & 4813 & 6800 & 0 & 1000 & 3000 & 6000 & 74400 \\
$G141$     & 1031 & 748 & 50 & 500 & 900 & 1400 & 8700 \\
$H219\_2$  & 907 & 750 & 10 & 300 & 800 & 1200 & 5000 \\
$I101$     & 28 & 27 & 3 & 15 & 20 & 30 & 325 \\
$I112$     & 161 & 239 & 0 & 0 & 50 & 250 & 1800 \\
$J108$     & 38 & 13 & 0 & 30 & 40 & 50 & 80 \\
$J109$     & 59 & 17 & 0 & 50 & 60 & 70 & 90 \\
$J112$     & 52528 & 71220 & 0 & 10000 & 30000 & 80000 & 700000 \\
$J113$     & 1 & 2 & 0 & 1 & 1 & 1 & 20 \\
$G102$     & 1756 & 1055 & 0 & 1100 & 1500 & 2200 & 11000 \\
\multicolumn{2}{c}{}  \\
\multicolumn{8}{c}{Panel C: Continuous Independent Variables but Treated as Dummy Variables}  \\\hline
$C130$      & 72 & 97 & 0 & 0 & 0 & 150 & 500 \\
$E4\_15$    & 8 & 18 & 0 & 0 & 2 & 7 & 250 \\
$G129$      & 251 & 2164 & 0 & 0 & 0 & 0 & 40400 \\
$G130$      & 33 & 232 & 0 & 0 & 0 & 0 & 3000 \\
$G131$      & 194 & 1028 & 0 & 0 & 0 & 0 & 18800 \\
$G133$      & 422 & 1088 & 0 & 0 & 0 & 400 & 12000 \\
$G134$      & 80 & 1193 & 0 & 0 & 0 & 0 & 20000 \\
$G135$      & 164 & 2758 & 0 & 0 & 0 & 0 & 60000 \\
$G137$      & 535 & 2279 & 0 & 0 & 0 & 0 & 27000 \\
$G138$      & 149 & 625 & 0 & 0 & 0 & 0 & 7800 \\
$G139$      & 1776 & 4464 & 0 & 0 & 0 & 1000 & 45000 \\
$B110$      & 45 & 134 & 0 & 0 & 0 & 0 & 630 \\
$B119$      & 103 & 633 & 0 & 0 & 0 & 0 & 10000 \\
\multicolumn{2}{c}{}  \\
\multicolumn{8}{c}{Panel D: Discrete Variables that Treated as Dummy Variables}  \\\hline
$A01$       & 11 & 2 & 1 & 12 & 12 & 12 & 12 \\
$A10$       & 0.665 & 0.838 & 0 & 0 & 0 & 1 & 4 \\
$A27$       & 1.749 & 0.770 & 1 & 1 & 2 & 2 & 5 \\
$H219\_1$   & 1.214 & 0.563 & 0 & 1 & 1 & 1 & 4 \\
$I102$      & 4.278 & 3.268 & 1 & 2 & 3 & 6 & 20 \\
\multicolumn{2}{c}{}  \\
\multicolumn{8}{c}{Panel E: Dummy Variables}  \\\hline
$A04$       & 0.691 & 0.462 & 0 & 0 & 1 & 1 & 1 \\
$A09$       & 0.486 & 0.500 & 0 & 0 & 0 & 1 & 1 \\
$A15$       & 0.014 & 0.119 & 0 & 0 & 0 & 0 & 1 \\
$A21$       & 0.189 & 0.392 & 0 & 0 & 0 & 0 & 1 \\
$A22$       & 0.348 & 0.477 & 0 & 0 & 0 & 1 & 1 \\
$A23$       & 0.356 & 0.479 & 0 & 0 & 0 & 1 & 1 \\
$A24$       & 0.029 & 0.167 & 0 & 0 & 0 & 0 & 1 \\
$A28$       & 0.021 & 0.142 & 0 & 0 & 0 & 0 & 1 \\
$B108$      & 0.132 & 0.338 & 0 & 0 & 0 & 0 & 1 \\
$C127$      & 0.698 & 0.460 & 0 & 0 & 1 & 1 & 1 \\
$C184$      & 0.220 & 0.415 & 0 & 0 & 0 & 0 & 1 \\
Dongguan    & 0.317 & 0.466 & 0 & 0 & 0 & 1 & 1 \\
Shenzhen    & 0.374 & 0.484 & 0 & 0 & 0 & 1 & 1 \\
\end{longtable}Note: MEAN = sample mean, STD = standard deviation, MIN =
minimum, LQ = 25 percentile, MD = median, UQ = 75 percentile, MAX = maximum. }

\pagebreak
\end{center}

{\small \begin{table}[h]
\caption{Frequencies of Variables Selected}%
\label{table:varselected}%
{\small  \centering{}\centering{ }
\begin{tabular}
[c]{cccccc}%
\multicolumn{6}{c}{}\\\hline\hline
& One Stage & Multiple Stage & Post-OCMT & Group LASSO & Adaptive
GLASSO\\\hline
$G102$ & $100\%$ & $100\%$ & $100\%$ & $100\%$ & $100\%$\\
$G133$ & $0\%$ & $98\%$ & $91\%$ & $100\%$ & $100\%$\\
$G137$ & $12\%$ & $98\%$ & $91\%$ & $85\%$ & $85\%$\\
$A09$ & $0\%$ & $1\%$ & $0\%$ & $0\%$ & $0\%$\\
$C111$ & $0\%$ & $0\%$ & $0\%$ & $2\%$ & $0\%$\\
$C112$ & $0\%$ & $0\%$ & $0\%$ & $58\%$ & $55\%$\\
$C126$ & $0\%$ & $0\%$ & $0\%$ & $11\%$ & $10\%$\\
$C188$ & $0\%$ & $0\%$ & $0\%$ & $17\%$ & $15\%$\\
$A04$ & $0\%$ & $0\%$ & $0\%$ & $4\%$ & $2\%$\\
$C127$ & $0\%$ & $0\%$ & $0\%$ & $3\%$ & $2\%$\\
$C127$ & $0\%$ & $0\%$ & $0\%$ & $3\%$ & $2\%$\\
$I102$ & $0\%$ & $0\%$ & $0\%$ & $1\%$ & $0\%$\\
$G134$ & $0\%$ & $0\%$ & $0\%$ & $4\%$ & $3\%$\\
$G135$ & $0\%$ & $0\%$ & $0\%$ & $22\%$ & $20\%$\\
$G138$ & $0\%$ & $8\%$ & $6\%$ & $0\%$ & $0\%$\\
$G139$ & $0\%$ & $0\%$ & $0\%$ & $4\%$ & $0\%$\\\hline\hline
\end{tabular}
}\end{table}}

\bigskip

\bigskip%

%

\end{document}